\documentclass[11pt]{article}
\usepackage{setspace}
\usepackage{natbib}
\usepackage{graphicx}
\usepackage{lscape}
\usepackage{pdflscape}
\usepackage{afterpage}
\usepackage{pdfpages}
\usepackage{float}
\usepackage{booktabs}
\usepackage[utf8]{inputenc}
\usepackage{indentfirst}
\usepackage{arydshln}
\usepackage{arydshln}
\usepackage{tikz}
\usepackage{lipsum}
\usepackage{pgffor}
\usepackage{dsfont}
\usepackage{amsfonts}
\usepackage{amsmath}
\allowdisplaybreaks
\usepackage{mathtools}
\newcommand\independent{\protect\mathpalette{\protect\independenT}{\perp}}
\def\independenT#1#2{\mathrel{\rlap{$#1#2$}\mkern2mu{#1#2}}}

\DeclareMathOperator*{\argmin}{\arg\!\min}
\DeclareMathOperator*{\argsup}{\arg\!\sup}
\DeclareMathOperator*{\arginf}{\arg\!\inf}
\usepackage{xfrac}
\usepackage{amssymb}
\usepackage{graphicx}
\usepackage{url}				% Para escrever links
\usepackage{subfig} % Required for creating figures with multiple parts (subfigures)
\usepackage{enumitem}
\usepackage{relsize,exscale}
\usepackage{multirow}

\usepackage{titlesec}
\titleformat*{\section}{\large\bfseries}
\titleformat*{\subsection}{\large\bfseries}


\newcounter{parentnumber}
\newtheorem{theorem}{Theorem}

\newtheorem{assumption}[theorem]{Assumption}

\newtheorem{corollary}[theorem]{Corollary}

\newtheorem{proposition}[theorem]{Proposition}

\newenvironment{proof}[1][Proof]{\noindent \textbf{#1.} }{\  \rule{0.5em}{0.5em}}
\providecommand{\U}[1]{\protect\rule{.1in}{.1in}}

\usepackage{hyperref} % More descriptive referencing
\hypersetup{
	colorlinks=true, breaklinks=true, bookmarks=true,bookmarksnumbered,
	urlcolor=blue, linkcolor=blue, citecolor=blue, % Link colors
	pdftitle={}, % PDF title
	pdfauthor={\textcopyright}, % PDF Author
	pdfsubject={}, % PDF Subject
	pdfkeywords={}, % PDF Keywords
	pdfcreator={pdfLaTeX}, % PDF Creator
	pdfproducer={LaTeX with hyperref and ClassicThesis} % PDF producer
}

\begin{document}
	\setstretch{1}
	\title{{\LARGE Sharp Bounds for the Marginal Treatment Effect with Sample Selection\thanks{I would like to thank Joseph Altonji, Nathan Barker, Ivan Canay, Xuan Chen, John Finlay, Carlos A. Flores, John Eric Humphries, Yuichi Kitamura, Marianne Köhli, Helena Laneuville, Jaewon Lee, Yusuke Narita, Pedro Sant'anna, Masayuki Sawada, Azeem Shaikh, Edward Vytlacil, Stephanie Weber, Siuyuat Wong and seminar participants at Yale University for helpful suggestions.}}}
	
	\author{ \begin{tabular}{c} 
	\Large{Vitor Possebom}\thanks{\emph{Email:} vitoraugusto.possebom@yale.edu} \\
	Yale University
	\end{tabular}	
	}
	\date{~}
	
	\maketitle

	\newsavebox{\tablebox} \newlength{\tableboxwidth}
	
	%%%%%%%%%%%%%%%%%%%%%%%%%%%%%%%%%%
	%%       Abstract                %
	%%%%%%%%%%%%%%%%%%%%%%%%%%%%%%%%%%

\begin{center}
	
% \footnotesize
First Draft: October 2018

This Draft: April 2019

% \normalsize

%
%
\href{https://sites.google.com/site/vitorapossebom/research}{Please click here for the most recent version}

\

\

		\large{\textbf{Abstract}}
	\end{center}
	
I analyze treatment effects in situations when agents endogenously select into the treatment group and into the observed sample. As a theoretical contribution, I propose pointwise sharp bounds for the marginal treatment effect (MTE) of interest within the always-observed subpopulation under monotonicity assumptions. Moreover, I impose an extra mean dominance assumption to tighten the previous bounds. I further discuss how to identify those bounds when the support of the propensity score is either continuous or discrete. Using these results, I estimate bounds for the MTE of the Job Corps Training Program on hourly wages for the always-employed subpopulation and find that it is decreasing in the likelihood of attending the program within the Non-Hispanic group. For example, the Average Treatment Effect on the Treated is between \$.33 and \$.99 while the Average Treatment Effect on the Untreated is between \$.71 and \$3.00.
	
	\ 
	
	
	\ 
	
	
	\newpage
	
	\doublespacing
	
%%%%%%%%%%%%%%%%%%%%%%%%%%%%%%%%%%%%%%%%%%%%%%%
% Introduction
%%%%%%%%%%%%%%%%%%%%%%%%%%%%%%%%%%%%%%%%%%%%%%%
\section{Introduction}\label{introduction}

In the applied treatment effects literature, there are many problems that face two identification challenges: endogenous selection into treatment and endogenous sample selection. For instance, in Labor Economics, if a researcher wants to evaluate the effect of a job training program on wages, she has to understand why agents choose to enroll in the program and why agents select into her sample by being employed. In this situation, she may combine information on hourly labor earnings (the observable outcome) and employment (sample selection status) to uncover the effect on hourly wages (the outcome of interest). Similar problems appear in Labor Economics when analyzing the college wage premium and scarring effects. In the Health Sciences, a researcher faces the same identification challenges when analyzing the effect of a drug on a health quality index when the drug may save a patient's life. Moreover, in randomized control trials, researchers are concerned with non-compliance and differential attrition rates between treated and control groups. This double selection problem is also present when analyzing the effect of an educational intervention on short- and long-term outcomes and the effect of procedural laws on litigation outcomes.\footnote{Training programs are studied by \cite{Heckman1999a}, \cite{Lee2009} and \cite{Chen2015}. The college wage premium is analyzed by \cite{Altonji1993}, \cite{Card1999} and \cite{Carneiro2011}. Scarring effects are discussed by \cite{Heckman1980}, \cite{Farber1993} and \cite{Jacobson1993}. Some education interventions are studied by \cite{Krueger2001}, \cite{Angrist2006}, \cite{Angrist2009}, \cite{Chetty2011} and \cite{Dobbie2015}. Medical treatments are analyzed by \cite{CASS1984}, \cite{Sexton1984} and \cite{Health2004}. Litigation outcomes are discussed by \cite{Helland2017}. RCT with attrition are illustrated by \cite{DeMel2013} and \cite{Angelucci2015}.}

To simultaneously address both idetification challenges, I propose a Generalized Roy Model (\cite{Heckman1999}) with sample selection in which there is one outcome of interest that is observed only if the individual self-selects into the sample. Under a monotonicity assumption on the sample selection indicator, I decompose the Marginal Treatment Response (MTR) function for the potential observable outcome when treated as a weighted average of (i) the MTR on the outcome of interest for the subpopulation who is always observed and (ii) the Marginal Treatment Effect (MTE) on the observable outcome for the subpopulation who is observed only when treated. Under a bounded (in one direction) support condition, this decomposition is useful because it allows me to propose pointwise sharp bounds for the MTE on the outcome of interest within the always-observed subpopulation ($MTE^{OO}$) as a function of (i) the MTR functions on the observable outcome, (ii) the maximum and (or) minimum of the support of the potential outcome, and (iii) the proportions of always-observed individuals and observed-only-when-treated individuals. I also show that it is impossible to construct bounds without extra assumptions when the support of the potential outcome is the entire real line. After that, I impose an extra mean dominance assumption that compares the always-observed population against the observed-only-when-treated population, tightening the previous bounds. Moreover, under this new assumption, I show that those tighter bounds are also pointwise sharp and derive an informative lower bound even when the support of the potential outcome is the entire real line.

I then proceed to show that those bounds are well-identified. When the support of the propensity score is an interval, the relevant objects are point identified by applying the local instrumental variable approach (LIV, see \cite{Heckman1999}) to the expectations of the observable outcome and of the selection indicator conditional on the propensity score and the treatment status. However, in many empirical applications, the support of the propensity score is a finite set. In such a context, I can identify bounds for the $MTE^{OO}$ of interest by adapting the nonparametric bounds proposed by \cite{Mogstad2017} or the flexible parametric approach suggested by \cite{Brinch2017} to encompass a sample selection problem. When using the nonparametric approach, the bounds for the $MTE^{OO}$ of interest are simply an outer set that contains the true $MTE^{OO}$, i.e., they are not pointwise sharp anymore.

Partial identification of the $MTE^{OO}$ of interest is useful for two reasons. First and most importantly, bounds for the $MTE^{OO}$ can be used to shed light on the heterogeneity of treatment effects, allowing the researcher to understand who would benefit and who would lose from a specific treatment, as recently illustrated by \cite{Cornelissen2018} and \cite{Bhuller2019}. This knowledge can be used to optimally design policies that incentivize to agents to take a treatment. Second, bounds for the $MTE^{OO}$ can be used to construct bounds for any treatment effect parameter that is written as a weighted integral of the $MTE^{OO}$. For instance, by taking a weighted average of the pointwise sharp bounds for the $MTE^{OO}$, one can bound the average treatment effect (ATE), the average treatment effect on the treated (ATT), any local average treatment effect (LATE, \cite{Imbens1994}) and any policy-relevant treatment effect (PRTE, \cite{Heckman2001a}) within the always-observed subpopulation. Although such bounds may not be sharp for any specific parameter, they are a flexible and easy-to-apply tool for many empirical problems that depend on a varied set of treatment effects.\footnote{As a consequence of this trade-off between flexibility and sharpness, I recommend the use of a specialized tool if the parameter of interest already has specific bounds (e.g., the ITT by \cite{Lee2009} and the LATE by \cite{Chen2015}).}

% Although such bounds may not be sharp for any specific parameter, they are a general and easy-to-apply solution to many empirical problems. Therefore, if the applied researcher is interested in a parameter that already has specific bounds for it (e.g., intention-to-treat treatment effect ($ITT^{OO}$) by \cite{Lee2009} and local average treatment effect ($LATE^{OO}$) by \cite{Chen2015} within the always-observed subpopulation), he or she should use a specialized tool. However, if the applied researcher is interested in parameters without specialized bounds (e.g., ATE, ATT and the Average Treatment Effect on the Untreated (ATU) in the case of imperfect compliance), he or she may take a weighted integral of pointwise sharp bounds for the $MTE^{OO}$ of interest. In other words, facing a trade-off between empirical flexibility and sharpness, the partial identification tool proposed in this paper focus on empirical flexibility while still ensuring some notion of sharpness.

Finally, I illustrate the usefulness of the proposed bounds for the $MTE^{OO}$ of interest by analyzing the effect of the Job Corps Training Program (JCTP) on hourly wages within the Non-Hispanic always-employed subpopulation. My framework is ideal to analyze this important experiment because it simultaneously addresses the imperfect compliance issue (self-selection into treatment) by focusing on the MTE and the endogenous employment decision (sample selection) by using a partial identification strategy. Although my $MTE^{OO}$ bounds are uninformative when using only the monotonicity assumption, they are tight and positive under a mean dominance assumption, illustrating the identification power of extra assumptions in a context of partial identification. Most interestingly, I find that the bounds of the $MTE^{OO}$ on hourly wages are decreasing in the likelihood of attending the program, implying that the agents who would benefit the most from the JCTP are the least likely to attend it. As a consequence of this result, my estimates suggest that ATU is greater than the ATT for the always-employed subpopulation. Moreover, my bounds for the $LATE^{OO}$ are in line with the estimates of \cite{Chen2015} and the effect of the JCTP on employment is positive for every agent according to the test proposed by \cite{Machado2018}. Finally, as a by-product of my estimation strategy, I also find that the MTE on employment and hourly labor earnings are decreasing in the likelihood of attending the JCTP, a result that is in line with the estimated upper bounds of \cite{Chen2017}.

I make contributions to three branches of literature: identification of treatment effects using an instrument, identification of treatment effects with sample selection, and the effect of job training programs. They are all vast and only briefly summarized here.

In the literature about treatment effects with an instrument, \cite{Imbens1994} show that we can identify the LATE. \cite{Heckman1999}, \cite{Heckman2005} and \cite{Heckman2006} define the MTE and explain how to compute any treatment effect as a weighted average of the MTE. However, if the support of the propensity score is not the unit interval, then it is not possible to non-parametrically identify some common treatment effects, such as the ATE, the ATT and the ATU. A parametric solution to this problem is given by \cite{Brinch2017}, who identify a flexible polynomial function for the MTE whose degree is defined by the cardinality of the propensity score support, while a nonparametric solution is given by \cite{Mogstad2017}, who use the information contained on IV-like estimands to construct non-parametrically worst- and best- case bounds for policy-relevant treatment effects.\footnote{Other important contributions are made by \cite{Manski1990}, \cite{Manski1997}, \cite{Manski2000}, \cite{Heckman2001}, \cite{Bhattacharya2008}, \cite{Chesher2010}, \cite{Chiburis2010}, \cite{Shaikh2011}, \cite{Bhattacharya2012}, \cite{Cornelissen2016}, \cite{Chen2017}, \cite{Huber2017}, \cite{Kowalski2018}, \cite{Mourifie2018} and \cite{Zhou2019}.}

I contribute to this literature by extending the non-parametric approach by \cite{Mogstad2017} and the flexible parametric approach by \cite{Brinch2017} to encompass a sample selection problem. By doing so, I can partially identify the MTE function on the outcome of interest, which, in my framework, is different from the observable outcome.

In the literature about identification of treatment effects with sample selection, the control function approach (\cite{Heckman1979}, \cite{Ahn1993} and \cite{Das2003}) and the use of auxiliary data (\cite{Chen2008}) are two classical solutions to this problem. Another approach is to partially identify the parameter of interest by imposing weak monotonicity assumptions. For example, in a seminal paper, \cite{Lee2009} imposes that sample selection is monotone on treatment assignment to sharply bound the ITT for the subpopulation of always-observed individuals ($ITT^{OO}$).\footnote{Other relevant contributions are made by \cite{Frangakis2002}, \cite{Blundell2007}, \cite{Imai2008}, \cite{Lechner2010}, \cite{Blanco2013}, \cite{Mealli2013}, \cite{Behaghel2015} and \cite{Huber2015}.}

In the intersection of both literatures, a few authors address the problem of sample selection and endogenous treatment simultaneously. By using two instrumental variables, \cite{Fricke2015} and \cite{Lee2016} identify different treatment effects. However, since finding a credible instrument for sample selection is challenging in some cases, it is worth developing alternative tools that do not require more than an instrument for selection into treatment. \cite{Frolich2014} point identify the LATE by assuming that there is no contemporaneous relationship between the potential outcomes and the sample selection problem. \cite{Chen2015} derive bounds for Average Treatment Effect within the always-observed compliers ($LATE^{OO}$) by combining one instrument with a double exclusion restriction with monotonicity assumptions on the sample selection and the selection into treatment problems.\footnote{Other important contributions are made by \cite{Huber2014}, \cite{Steinmayr2014}, \cite{Blanco2017} and \cite{Kedagni2018}.}

I contribute to this literature by partially identifying the MTE on the always-observed subsample allowing for a contemporaneous relationship between the potential outcomes and the sample selection problem, and using only one (discrete) instrument combined with a monotonicity assumption. Deriving bounds for the $MTE^{OO}$ is theoretically important because it can unify, in one framework, the bounds for different treatment effects with sample selection. It is also empirically relevant because it allows us to partially identify any treatment effect on the outcome of interest in many empirical problems. For instance, when analyzing the effect of a job training program on wages, it is useful to compare the ATT with the ATU in order to understand whether the workers who would benefit the most from such a policy are actually the ones who receive training.

In the literature about job training programs, \cite{Heckman1999a} wrote an influential survey paper. In particular, many papers were written about the effects of the Job Corps Training Program (JCTP) after a randomized experiment funded by the U.S. Department of Labor in 1995.\footnote{For example, significant contributions are made by \cite{Schochet2001}, \cite{Schochet2008}, \cite{Flores-Lagunes2010}, \cite{Flores2012}, \cite{Blanco2013}, \cite{Blanco2013a}, \cite{Blanco2017} and \cite{Chen2017}.} Finally, my work is closer to the research done by \cite{Lee2009} and \cite{Chen2015}, who analyze the effect of the Job Corps Training Program on wages by focusing, respectively, on the ITT and the LATE parameters within the always-observed subpopulation. \cite{Lee2009} rules out a zero effect after accounting for the loss in labor market experience generated by the extra education acquired by Job Corps participants. \cite{Chen2015} find that the $LATE^{OO}$ on hourly wages four years after randomization is between 5.7\% and 13.9\% for the entire population and between 7.7\% and 17.5\% for the non-Hispanic population under monotonicity and mean dominance assumptions.

I contribute to this literature by analyzing the MTE on hourly wages within the Non-Hispanic group and formally testing whether this training program has a monotone effect on employment by implementing the test proposed by \cite{Machado2018}.

% I contribute to literature about the Job Corps Training Program by analyzing the MTE within the Non-Hispanic group, allowing me to understand heterogeneous treatment effects over the likelihood of attending the program. To summarize those results, I also compute estimates of the $ATE^{OO}$, the $LATE^{OO}$, the $ATT^{OO}$ and the $ATU^{OO}$. Moreover, I formally test whether this training program has a monotone effect on employment by implementing the test proposed by \cite{Machado2018} for the the non-Hispanic and Hispanic subpopulations. My empirical results suggests that the agents who are more likely to benefit from the JCTP are the least likely to attend the program.

This paper proceeds as follows: Section \ref{model} details the Generalized Roy Model with sample selection; Section \ref{bounds} explains how to derive bounds for the $MTE^{OO}$ of interest; Sections \ref{interval} and \ref{discrete} discuss identification of the $MTE^{OO}$ bounds when the support of the propensity score is continuous or discrete; and Section \ref{application} analyzes the effect of the Job Corps Training Program on hourly wages. Finally, Section \ref{furtherwork} concludes.

%%%%%%%%%%%%%%%%%%%%%%%%%%%%%%%%%%%%%%%%%%%%%%%
% Framework
%%%%%%%%%%%%%%%%%%%%%%%%%%%%%%%%%%%%%%%%%%%%%%%

\section{Framework}\label{model}

I begin with the classical potential outcome framework by \cite{Rubin1974} and modify it to include a sample selection problem. Let $Z$ be an instrumental variable whose support is given by $\mathcal{Z}$, $X$ be a vector of covariates whose support is given by $\mathcal{X}$, $W \coloneqq \left(X, Z\right)$ be a vector that combines the covariates and the instrument whose support is given by $\mathcal{W} \coloneqq \mathcal{X} \times \mathcal{Z}$, $D$ be a treatment status indicator, $Y_{0}^{*}$ be the potential outcome of interest when the person is not treated, and $Y_{1}^{*}$ be the potential outcome of interest when the person is treated. The outcome variable of interest (e.g., wages) is $Y^{*} \coloneqq D \cdot Y_{1}^{*} + \left(1 - D\right) \cdot Y_{0}^{*}$. Moreover, let $S_{1}$ and $S_{0}$ be potential sample selection indicators when treated and when not treated, and define $S \coloneqq D \cdot S_{1} + \left(1 - D\right) \cdot S_{0}$ as the sample selection indicator (e.g., employment status). Define $Y \coloneqq S \cdot Y^{*}$ as the observable outcome (e.g., labor earnings). I also define $Y_{1} \coloneqq S_{1} \cdot Y_{1}^{*}$ and $Y_{0} \coloneqq S_{0} \cdot Y_{0}^{*}$ as the potential observable outcomes. Observe that, following \cite{Lee2009} and \cite{Chen2015}, my notation implicitly imposes two exclusion restrictions: Z has no direct impact on the potential outcome of interest nor on the sample selection indicator. The second exclusion restriction requires attention in empirical applications. On the one hand, it may be a strong assumption in randomized control trials if sample selection is due to attrition and initial assignment has an effect on the subject's willingness to contact the researchers. On the other hand, it may be a reasonable assumption in many labor market applications, such as the evaluation of a job training program. For instance, in my empirical section, it is plausible that the initial random assignment to the Job Corps Training Program (JCTP) has no impact on future employment status.

I model sample selection and selection into treatment using the Generalized Roy Model \citep{Heckman1999}. Let $U$ and $V$ be random variables, and $P:\mathcal{W} \rightarrow \mathbb{R}$ and $Q:\left\lbrace 0, 1 \right\rbrace \times \mathcal{X} \rightarrow \mathbb{R}$ be unknown functions. I assume that:
\begin{equation}
\label{treatment}
D \coloneqq \mathbf{1}\left\lbrace P\left(W\right) \geq U \right\rbrace
\end{equation} 
and
\begin{equation}
\label{selection}
S \coloneqq \mathbf{1}\left\lbrace Q\left(D, X\right) \geq V \right\rbrace.
\end{equation}
As \cite{Vytlacil2002} shows, equations \eqref{treatment} and \eqref{selection} are equivalent to assuming monotonicity conditions on the selection-into-treatment problem (\cite{Imbens1994}) and on the sample selection problem (\cite{Lee2009}). I stress that both monotonicity assumptions are testable using the tools developed by \cite{Machado2018}. Note also that, given equation \eqref{selection}, $S_{0} = \mathbf{1}\left\lbrace Q\left(0, X\right) \geq V \right\rbrace$ and $S_{1} = \mathbf{1}\left\lbrace Q\left(1, X\right) \geq V \right\rbrace$.

The random variables $U$ and $V$ are jointly continuously distributed conditional on $X$ with density $f_{U,V \left\vert X \right.}:\mathbb{R}^{2} \times \mathcal{X} \rightarrow \mathbb{R}$ and cumulative distribution function $F_{U,V \left\vert X \right.}:\mathbb{R}^{2} \times \mathcal{X} \rightarrow \mathbb{R}$. As has been shown in the literature, equations \eqref{treatment} and \eqref{selection} can be rewritten as
\begin{align*}
D & = \mathbf{1}\left\lbrace F_{U \left\vert X \right.} \left(P\left(W\right) \left\vert X \right. \right) \geq F_{U \left\vert X \right.} \left(U \left\vert X \right.\right) \right\rbrace = \mathbf{1}\left\lbrace \tilde{P}\left(W\right) \geq \tilde{U} \right\rbrace \\
S & = \mathbf{1}\left\lbrace F_{V \left\vert X \right.} \left(Q\left(D, X\right) \left\vert X \right.\right) \geq F_{V \left\vert X \right.} \left(V \left\vert X \right.\right)  \right\rbrace = \mathbf{1}\left\lbrace \tilde{Q}\left(D, X\right) \geq \tilde{V} \right\rbrace
\end{align*}
where $\tilde{P}\left(W\right) \coloneqq F_{U \left\vert X \right.} \left(P\left(W\right) \left\vert X \right. \right)$, $\tilde{U} \coloneqq F_{U \left\vert X \right.} \left(U \left\vert X \right.\right)$, $\tilde{Q}\left(D, X\right) \coloneqq F_{V \left\vert X \right.} \left(Q\left(D, X\right) \left\vert X \right. \right)$, and $\tilde{V} \coloneqq F_{V \left\vert X \right.} \left(V \left\vert X \right.\right)$. Consequently, the marginal distributions of $\tilde{U}$ and $\tilde{V}$ conditional on $X$ follow the standard uniform distribution. Since this is merely a normalization, I drop the tilde and mantain throughout the paper the normalization that $\left(P\left(w\right), Q\left(d, x\right)\right) \in \left[0, 1\right]^{2}$ for any $\left(x, z, d\right) \in \mathcal{W} \times \left\lbrace 0, 1 \right\rbrace$ and the marginal distributions of $U$ and $V$ conditional on $X$ follow the standard uniform distribution, even though their joint distribution allows for any kind of dependency between those two variables. As a consequence of such normalization, $P\left(w\right)$ represents the propensity score and is equal to $\mathbb{P}\left[\left. D = 1 \right\vert W = w\right]$, while $Q\left(d, x\right)$ is equal to $\mathbb{P}\left[\left. S_{d} = 1 \right\vert X = x\right]$.

Moreover, I assume that:
\begin{assumption}\label{ind}
	The instrument $Z$ is independent of all latent variables given the covariates $X$, i.e., $Z \independent \left(U, V, Y_{0}^{*}, Y_{1}^{*} \right) \left\vert X \right.$.
\end{assumption}

\begin{assumption}\label{propensityscore}
	The distribution of $P\left(W\right)$ given $X$ is nondegenerate.
\end{assumption}

\begin{assumption}\label{finite}
	The first and second population moments of the  potential outcomes of interest are finite, i.e., $\mathbb{E}\left[\left\vert Y_{d}^{*} \right\vert\right] < + \infty$ and $\mathbb{E}\left[\left( Y_{d}^{*} \right)^{2}\right] < + \infty$ for any $d \in \left\lbrace 0, 1 \right\rbrace$.
\end{assumption}

\begin{assumption}\label{positive}
	Both treatment groups exist for any value of X, i.e., $0 < \mathbb{P}\left[D = 1 \left\vert X \right.\right] < 1$.
\end{assumption}

\begin{assumption}\label{invariantX}
	The covariates $X$ are invariant to counterfactual manipulations, i.e., $X_{0} = X_{1} = X$, where $X_{0}$ and $X_{1}$ are the counterfactual values of $X$ that would be observed when the person is, respectively, not treated or treated.
\end{assumption}

\begin{assumption}\label{support}
	The potential outcomes $Y_{0}^{*}$ and $Y_{1}^{*}$ have the same support, i.e., $\mathcal{Y}^{*} \coloneqq \mathcal{Y}_{0}^{*} = \mathcal{Y}_{1}^{*}$, where $\mathcal{Y}_{0}^{*} \subseteq \mathbb{R}$ is the support of $Y_{0}^{*}$ and $\mathcal{Y}_{1}^{*} \subseteq \mathbb{R}$ is the support of $Y_{1}^{*}$.
\end{assumption}

\begin{assumption}\label{bounded}
	Define $\underline{y}^{*} \coloneqq \inf \left\lbrace y \in \mathcal{Y}^{*} \right\rbrace \in \mathbb{R} \cup \left\lbrace -\infty \right\rbrace$ and $\overline{y}^{*} \coloneqq \sup \left\lbrace y \in \mathcal{Y}^{*} \right\rbrace \in \mathbb{R} \cup \left\lbrace \infty \right\rbrace$. I assume that $\underline{y}^{*}$ and $\overline{y}^{*}$ are known, and that
	\begin{enumerate}
		\item $\underline{y}^{*} > - \infty$, $\overline{y}^{*} = \infty$ and $\mathcal{Y}^{*}$ is an interval, or
		
		\item $\underline{y}^{*} = - \infty$, $\overline{y}^{*} < \infty$ and $\mathcal{Y}^{*}$ is an interval, or
		
		\item $\underline{y}^{*} > - \infty$, $\overline{y}^{*} < \infty$ and
		\begin{enumerate}
			\item $\mathcal{Y}^{*}$ is an interval or
			
			\item $\underline{y}^{*} \in \mathcal{Y}^{*}$ and $\overline{y}^{*} \in \mathcal{Y}^{*}$.
		\end{enumerate}
	\end{enumerate}
\end{assumption}

Assumption \ref{bounded} is fairly general. Case 1 covers continuous random variables whose support is convex and bounded below (e.g., wages), while Case 3.a covers continuous variables with bounded convex support (e.g., test scores). Case 3.b encompasses not only binary variables, but also any discrete variable whose support is finite (e.g., years of education). It also includes mixed random variables whose support is not an interval but achieves its maximum and minimum. Case 2 is included for theoretical complementness. Furthermore, Proposition \ref{partialnecessary} shows that assumption \ref{bounded} is partially necessary to the existence of bounds for the $MTE^{OO}$ of interest in the sense that, if $\underline{y}^{*} = -\infty$ and $\overline{y}^{*} = + \infty$, then it is impossible to bound the marginal treatment effect on the outcome of interest within the always-observed subpopulation without any extra assumptions.

\begin{assumption}\label{increasing_sample_selection}
	Treatment has a positive effect on the sample selection indicator for all individuals, i.e., $Q\left(1, x\right) > Q\left(0, x\right) > 0$ for any $x \in \mathcal{X}$.
\end{assumption}

Assumption \ref{increasing_sample_selection} goes beyond the monotonicity condition implicitly imposed by equation \eqref{selection} by assuming that the direction of the effect of treatment on the sample selection indicator is known and positive, i.e., $Q\left(1, x\right) \geq Q\left(0, x\right)$ for any $x \in \mathcal{X}$. In this sense, it is a standard assumption in the literature.\footnote{\cite{Lee2009} and \cite{Chen2015} write it in an equivalent way as $S_{1} \geq S_{0}$, while \cite{Manski1997} and \cite{Manski2000} call it the ``monotone treatment response'' assumption.} Most importantly, it is also a testable assumption using the tools developed by \cite{Machado2018}, because, under monotone sample selection (equation \eqref{selection}), identification of the sign of the ATE on the selection indicator provides a test for Assumption \ref{increasing_sample_selection}. However, Assumption \ref{increasing_sample_selection} is slightly stronger than what is usually imposed in the literature, because it additionally imposes $Q\left(0, x\right) > 0$ and $Q\left(1, x\right) > Q\left(0, x\right)$ for any $x \in \mathcal{X}$. While the first inequality implies that there is a subpopulation who is always observed, allowing me to properly define my target parameter (the marginal treatment effect on the outcome of interest within the always-observed population, $MTE^{OO}$), the second inequality implies that there is a subpopulation who is observed only when treated, making the problem theoretically interesting by eliminating trivial cases of point identification of the $MTE^{OO}$ as discussed in Proposition \ref{boundsY1Proposition}. Finally, I emphasize that all my results can be stated and derived with some straightforward changes if I impose $Q\left(0, x\right) > Q\left(1, x\right) > 0$ for any $x \in \mathcal{X}$ instead of Assumption \ref{increasing_sample_selection}, as is done in Appendix \ref{decreasing_sample_selection}. I also discuss, in Appendix \ref{agnostic}, an agnostic approach to monotonicity in the sample selection problem (equation \eqref{selection}) and show, in Appendix \ref{nomonotonicity}, that bounds derived with non-monotone sample selection are uninformative (i.e., equal to $\left(\underline{y}^{*} - \overline{y}^{*}, \overline{y}^{*} - \underline{y}^{*}\right)$) under mild regularity conditions.

In my empirical application, Assumption \ref{increasing_sample_selection} imposes that the JCTP has a positive effect on employment for all individuals, which is plausible given the objectives and services provided by this training program. As discussed by \cite{Chen2015}, the two potential threats against it --- the ``lock-in'' effect (\cite{Ours2004}) and an increase in the reservation wage of treated individuals --- are likely to become less relevant in the long run, justifying my focus on the hourly wage after 208 weeks from randomization. Most importantly, this assumption is formally tested by the method developed by \cite{Machado2018} and I reject, at the 1\%-significance level, the null hypothesis that Assumption \ref{increasing_sample_selection} is invalid within the Non-Hispanic group.

Finally, in partial identification contexts, extra assumptions may have a lot of identification power. In the specific case of identifying treatment effects with sample selection, it is common to use mean or stochastic dominance assumptions to tighten the bounds for the parameter of interest (\cite{Imai2008}, \cite{Blanco2013}, \cite{Huber2015} and \cite{Huber2017}) and justify them based on the intuitive argument that some population sub-groups have more favorable underlying characteristics than others. In particular, I discuss the identifying power of the following mean dominance assumption\footnote{In appendix \ref{MeanDominance}, I derive bounds for the MTE of interest when the above inequality holds in the other direction.}:
\begin{assumption}\label{meandominanceG}
	The potential outcome when treated within the always-observed subpopulation is greater than or equal to the same parameter within the observed-only-when-treated subpopulation:
	\begin{equation*}
	\mathbb{E}\left[Y_{1}^{*} \left\vert X = x, U = u, S_{0} = 1, S_{1} = 1 \right.\right] \geq \mathbb{E}\left[Y_{1}^{*} \left\vert X = x, U = u, S_{0} = 0, S_{1} = 1 \right.\right]
	\end{equation*}
	for any $x \in \mathcal{X}$ and $u \in \left[0, 1\right]$.
\end{assumption}
Unfortunately, this assumption is empirically untestable, implying that its use must be justified for each application based on qualitative or theoretical arguments. In particular, in my empirical application, Assumption \ref{meandominanceG} imposes that the marginal treatment response function of wages when treated for the always-employed population is greater than the same object for the employed-only-when-treated population. Intuitively, this assumption imposes that the group with better potential employment outcomes also has, on average, better potential wages, i.e., there is positive selection into employment. % Similarly to the case discussed by \citet[section 2.3]{Chen2015}, Assumption \eqref{meandominanceG} implies a positive correlation between employment and wages, which is supported by standard models of labor supply.

%%%%%%%%%%%%%%%%%%%%%%%%%%%%%%%%%%%%%%%%%%%%%%%%%%%%%
% bounds for the MTE
%%%%%%%%%%%%%%%%%%%%%%%%%%%%%%%%%%%%%%%%%%%%%%%%%%%%%
\section{Bounds for the $\mathbf{MTE^{OO}}$ on the outcome of interest} \label{bounds}
The target parameter, the MTE on the outcome of interest for the subpopulation who is always observed ($MTE^{OO}$), is given by
\begin{align}
\Delta_{Y^{*}}^{OO}\left(x, u\right) & \coloneqq \mathbb{E}\left[Y_{1}^{*} - Y_{0}^{*} \left\vert X = x, U = u, S_{0} = 1, S_{1} = 1 \right.\right] \nonumber \\
& \label{target} = \mathbb{E}\left[Y_{1}^{*} \left\vert X = x, U = u, S_{0} = 1, S_{1} = 1 \right.\right]  - \mathbb{E}\left[Y_{0}^{*} \left\vert X = x, U = u, S_{0} = 1, S_{1} = 1 \right.\right]
\end{align}
for any $u \in \left[0, 1\right]$ and any $x \in \mathcal{X}$, and is a natural parameter of interest. In labor market applications where sample selection is due to observing wages only when agents are employed, it is the effect on wages for the subpopulation who is always employed. In medical applications where sample selection is due to the death of a patient, it is the effect on health quality for the subpopulation who survives regardless of treatment status. In the education literature where sample selection is due to students quitting school, it is the effect on test scores for the subpopulation who do not drop out of school regardless of treatment status. In all those cases, the target parameter captures the intensive margin of the treatment effect.\footnote{If the researcher is interested in the extensive margin of the treatment effect, captured by the MTE on the observable outcome ($\mathbb{E}\left[Y_{1} - Y_{0} \left\vert X = x, U = u\right.\right]$) and by the MTE on the selection indicator ($\mathbb{E}\left[S_{1} - S_{0} \left\vert X = x, U = u\right.\right]$), he or she can apply the identification strategies described by \cite{Heckman2006}, \cite{Brinch2017} and \cite{Mogstad2017}.}

Other possibly interesting parameters are the MTE on the outcome of interest within the subpopulation who is never observed ($\mathbb{E}\left[Y_{1}^{*} - Y_{0}^{*} \left\vert X = x, U = u, S_{0} = 0, S_{1} = 0 \right.\right]$, $MTE^{NN}$), the MTR function under no treatment for the outcome of interest within the subpopulation who is observed only when treated ($\mathbb{E}\left[Y_{0}^{*} \left\vert X = x, U = u, S_{0} = 0, S_{1} = 1 \right.\right]$, $MTR_{0}^{NO}$) and the MTR function under treatment for the outcome of interest within the subpopulation who is observed only when treated ($\mathbb{E}\left[Y_{1}^{*} \left\vert X = x, U = u, S_{0} = 0, S_{1} = 1 \right.\right]$, $MTR_{1}^{NO}$). While the last parameter can be partially identified (Appendix \ref{observedonlywhentreated}), the first two parameters are impossible to point identify or bound in an informative way because the outcome of interest ($Y_{0}^{*}$ or $Y_{1}^{*}$) is never observed for the conditioning subpopulations.\footnote{\cite{Zhang2008} discuss this identification issue in a deeper way. Moreover, in some applications (e.g., analyzing the impact of a medical treatment on a health quality measure where selection is given by whether the patient is alive), the potential outcome $Y_{d}^{*}$ is not even properly defined when $S_{d} = 0$ for $d \in \left\lbrace 0, 1 \right\rbrace$.} As a consequence, it is not possible to point identify or bound in an informative way the Marginal Treatment Effect for the entire population ($\mathbb{E}\left[Y_{1}^{*} - Y_{0}^{*} \left\vert X = x, U = u\right.\right]$, $MTE$) either. Note also that the subpopulation who is observed only when not treated ($S_{0} = 1$ and $S_{1} = 0$) does not exist by Assumption \ref{increasing_sample_selection}. Furthermore, observe that the conditioning subpopulations in all the above-mentioned parameters are determined by post-treatment outcomes and, as a consequence, are connected to the statistical literature known as principal stratification (\cite{Frangakis2002}).

I now focus on the target parameter $\Delta_{Y^{*}}^{OO}\left(x, u\right)$ given by equation \eqref{target}. While Subsection \ref{MonotonicitySection} derives bounds for the $MTE^{OO}$ of interest (equation \eqref{target}) using only a monotonicity assumption (Assumptions \ref{ind}-\ref{increasing_sample_selection}), Subsection \ref{MeanDominanceSection} tightens those bounds by additionally imposing the Mean Dominance Assumption \ref{meandominanceG}. Finally, Subsection \ref{EmpiricalRelevance} discusses the empirical relevance of such bounds.

%%%%%%%%%%%%%%%%%%%%%%%%%%%%%%%%%%%%%%%%%%%%%%
% Partial Identification with only a Monotonicity Assumption
%%%%%%%%%%%%%%%%%%%%%%%%%%%%%%%%%%%%%%%%%%%%%%
\subsection{Partial Identification with only a Monotonicity Assumption}\label{MonotonicitySection}

Here, my goal is to derive bounds for $\Delta_{Y^{*}}^{OO}\left(x, u\right)$ under Assumptions \ref{ind}-\ref{increasing_sample_selection}. Note that the second right-hand term in equation \eqref{target} can be written as\footnote{Appendix \ref{proofUntreated} contains a proof of this claim.}
\begin{equation}\label{m0YstarOO}
\mathbb{E}\left[Y_{0}^{*} \left\vert X = x, U = u, S_{0} = 1, S_{1} = 1 \right.\right] = \dfrac{m_{0}^{Y}\left(x, u\right)}{m_{0}^{S}\left(x, u\right)},
\end{equation}
where I define $m_{0}^{Y}\left(x, u\right) \coloneqq \mathbb{E}\left[Y_{0} \left\vert X = x, U = u \right.\right]$ and $m_{0}^{S}\left(x, u\right) \coloneqq \mathbb{E}\left[S_{0} \left\vert X = x, U = u \right.\right]$ as the MTR functions associated with the counterfactual variables $Y_{0}$ and $S_{0}$, respectively. In this section, I assume that all terms in the right-hand side of equation \eqref{m0YstarOO} are point identified, postponing the discussion about their identification to Sections \ref{interval} and \ref{discrete}.

The first right-hand term in equation \eqref{target} can be written as\footnote{Appendix \ref{proofTreated} contains a proof of this claim.}
\begin{equation}\label{m1YstarOO}
\mathbb{E}\left[Y_{1}^{*} \left\vert X = x, U = u, S_{0} = 1, S_{1} = 1 \right.\right] = \dfrac{m_{1}^{Y}\left(x, u\right) - \Delta_{Y}^{NO}\left(x, u\right) \cdot \Delta_{S}\left(x, u\right)}{m_{0}^{S}\left(x, u\right)},
\end{equation}
where $m_{1}^{Y}\left(x, u\right) \coloneqq \mathbb{E}\left[Y_{1} \left\vert X = x, U = u \right.\right]$ is the MTR function associated with the counterfactual variable $Y_{1}$, $\Delta_{Y}^{NO}\left(x, u\right) \coloneqq \mathbb{E}\left[ Y_{1} - Y_{0} \left\vert X = x, U = u, S_{0} = 0, S_{1} = 1 \right.\right]$ is the MTE on the observable outcome $Y$ for the subpopulation who is observed only when treated, $\Delta_{S}\left(x, u\right) \coloneqq \mathbb{E}\left[S_{1} - S_{0} \left\vert X = x, U = u \right.\right] = m_{1}^{S}\left(x, u\right) - m_{0}^{S}\left(x, u\right)$ is the MTE on the selection indicator, and $m_{1}^{S}\left(x, u\right) \coloneqq \mathbb{E}\left[S_{1} \left\vert X = x, U = u \right.\right]$ is the MTR function associated with the counterfactual variable $S_{1}$. In this section, I also assume that $m_{1}^{Y}\left(x, u\right)$ and $\Delta_{S}\left(x, u\right)$ are point identified, postponing the discussion about their identification to Sections \ref{interval} and \ref{discrete}.

Although point identification of $\mathbb{E}\left[Y_{1}^{*} \left\vert X = x, U = u, S_{0} = 1, S_{1} = 1 \right.\right]$ is not possible due to the term $\Delta_{Y}^{NO}\left(x, u\right)$ in equation \eqref{m1YstarOO}, I can find identifiable bounds for it.\footnote{Appendix \ref{proofboundsY1} contains a proof of this proposition.}
\begin{proposition}\label{boundsY1Proposition}
Suppose that $m_{0}^{Y}\left(x, u\right)$, $m_{1}^{Y}\left(x, u\right)$, $m_{0}^{S}\left(x, u\right)$ and $\Delta_{S}\left(x, u\right)$ are point identified.
	
Under Assumptions \ref{ind}-\ref{support}, \ref{bounded}.1 and \ref{increasing_sample_selection}, $\mathbb{E}\left[Y_{1}^{*} \left\vert X = x, U = u, S_{0} = 1, S_{1} = 1 \right.\right]$ must satisfy
\begin{equation}\label{boundsY1lower}
\underline{y}^{*} \leq \mathbb{E}\left[Y_{1}^{*} \left\vert X = x, U = u, S_{0} = 1, S_{1} = 1 \right.\right] \leq \dfrac{m_{1}^{Y}\left(x, u\right) - \underline{y}^{*} \cdot \Delta_{S}\left(x, u\right)}{m_{0}^{S}\left(x, u\right)}.
\end{equation}

Under Assumptions \ref{ind}-\ref{support}, \ref{bounded}.2 and \ref{increasing_sample_selection}, $\mathbb{E}\left[Y_{1}^{*} \left\vert X = x, U = u, S_{0} = 1, S_{1} = 1 \right.\right]$ must satisfy
\begin{equation}\label{boundsY1upper}
\dfrac{m_{1}^{Y}\left(x, u\right) - \overline{y}^{*} \cdot \Delta_{S}\left(x, u\right)}{m_{0}^{S}\left(x, u\right)} \leq \mathbb{E}\left[Y_{1}^{*} \left\vert X = x, U = u, S_{0} = 1, S_{1} = 1 \right.\right] \leq \overline{y}^{*}.
\end{equation}

Under Assumptions \ref{ind}-\ref{support}, \ref{bounded}.3 (sub-case (a) or (b)) and \ref{increasing_sample_selection}, $\mathbb{E}\left[Y_{1}^{*} \left\vert X = x, U = u, S_{0} = 1, S_{1} = 1 \right.\right]$ must satisfy
\begin{align}
\dfrac{m_{1}^{Y}\left(x, u\right) - \overline{y}^{*} \cdot \Delta_{S}\left(x, u\right)}{m_{0}^{S}\left(x, u\right)} & \leq \mathbb{E}\left[Y_{1}^{*} \left\vert X = x, U = u, S_{0} = 1, S_{1} = 1 \right.\right] \nonumber \\
& \label{boundsY1} \leq \dfrac{m_{1}^{Y}\left(x, u\right) - \underline{y}^{*} \cdot \Delta_{S}\left(x, u\right)}{m_{0}^{S}\left(x, u\right)}.
\end{align}
\end{proposition}

Note that, even when the support is bounded in only one direction (Assumptions \ref{bounded}.1 and \ref{bounded}.2), it is possible to derive lower and upper bounds for $\mathbb{E}\left[Y_{1}^{*} \left\vert X = x, U = u, S_{0} = 1, S_{1} = 1 \right.\right]$.

At this point, it is worth understanding the determinants of the width of those bounds.  First, if there is no sample selection problem at all ($\mathbb{P}\left[S_{0} = 1, S_{1} = 1 \left\vert X = x, U = u \right.\right] = 1$, i.e., the always-observed group is the entire population), then $m_{0}^{S}\left(x, u\right) = 1$, $\Delta_{S}\left(x, u\right) = 0$, implying point identification in equation \eqref{m1YstarOO}. Second, if there is no problem of differential sample selection with respect to treatment status ($\mathbb{P}\left[S_{0} = 0, S_{1} = 1 \left\vert X = x, U = u \right.\right] = 0$, i.e., the observed-only-when-treated subpopulation has zero mass), then $\Delta_{S}\left(x, u\right) = 0$, once more implying point identification in equation \eqref{m1YstarOO}. Both cases are theoretically uninteresting and ruled out by Assumption \ref{increasing_sample_selection}.

Finally, combining equations \eqref{target} and \eqref{m0YstarOO} and Proposition \ref{boundsY1Proposition}, I can partially identify the target parameter $\Delta_{Y^{*}}^{OO}\left(x, u\right)$:
\begin{corollary}\label{MTEbounds}
Suppose that $m_{0}^{Y}\left(x, u\right)$, $m_{1}^{Y}\left(x, u\right)$, $m_{0}^{S}\left(x, u\right)$ and $\Delta_{S}\left(x, u\right)$ are point identified.

Under Assumptions \ref{ind}-\ref{support}, \ref{bounded}.1 and \ref{increasing_sample_selection}, the bounds for $\Delta_{Y^{*}}^{OO}\left(x, u\right)$ are given by
\begin{equation}\label{mte_lb_lower}
\Delta_{Y^{*}}^{OO}\left(x, u\right) \geq \underline{y}^{*} - \dfrac{m_{0}^{Y}\left(x, u\right)}{m_{0}^{S}\left(x, u\right)} \eqqcolon \underline{\Delta_{Y^{*}}^{OO}}\left(x, u\right)
\end{equation}
and
\begin{equation}\label{mte_ub_lower}
\Delta_{Y^{*}}^{OO}\left(x, u\right) \leq \dfrac{m_{1}^{Y}\left(x, u\right) - \underline{y}^{*} \cdot \Delta_{S}\left(x, u\right)}{m_{0}^{S}\left(x, u\right)} - \dfrac{m_{0}^{Y}\left(x, u\right)}{m_{0}^{S}\left(x, u\right)} \eqqcolon \overline{\Delta_{Y^{*}}^{OO}}\left(x, u\right).
\end{equation}

Under Assumptions \ref{ind}-\ref{support}, \ref{bounded}.2 and \ref{increasing_sample_selection}, the bounds for $\Delta_{Y^{*}}^{OO}\left(x, u\right)$ are given by
\begin{equation}\label{mte_lb_upper}
\Delta_{Y^{*}}^{OO}\left(x, u\right) \geq \dfrac{m_{1}^{Y}\left(x, u\right) - \overline{y}^{*} \cdot \Delta_{S}\left(x, u\right)}{m_{0}^{S}\left(x, u\right)} - \dfrac{m_{0}^{Y}\left(x, u\right)}{m_{0}^{S}\left(x, u\right)} \eqqcolon \underline{\Delta_{Y^{*}}^{OO}}\left(x, u\right)
\end{equation}
and
\begin{equation}\label{mte_ub_upper}
\Delta_{Y^{*}}^{OO}\left(x, u\right) \leq \overline{y}^{*} - \dfrac{m_{0}^{Y}\left(x, u\right)}{m_{0}^{S}\left(x, u\right)} \eqqcolon \overline{\Delta_{Y^{*}}^{OO}}\left(x, u\right).
\end{equation}

Under Assumptions \ref{ind}-\ref{support}, \ref{bounded}.3 (sub-case (a) or (b)) and \ref{increasing_sample_selection}, the bounds for $\Delta_{Y^{*}}^{OO}\left(x, u\right)$ are given by
\begin{equation}\label{mte_lb}
\Delta_{Y^{*}}^{OO}\left(x, u\right) \geq \max\left\lbrace \dfrac{m_{1}^{Y}\left(x, u\right) - \overline{y}^{*} \cdot \Delta_{S}\left(x, u\right)}{m_{0}^{S}\left(x, u\right)}, \underline{y}^{*}\right\rbrace - \dfrac{m_{0}^{Y}\left(x, u\right)}{m_{0}^{S}\left(x, u\right)} \eqqcolon \underline{\Delta_{Y^{*}}^{OO}}\left(x, u\right)
\end{equation}
and
\begin{equation}\label{mte_ub}
\Delta_{Y^{*}}^{OO}\left(x, u\right) \leq \min\left\lbrace \dfrac{m_{1}^{Y}\left(x, u\right) - \underline{y}^{*} \cdot \Delta_{S}\left(x, u\right)}{m_{0}^{S}\left(x, u\right)}, \overline{y}^{*}\right\rbrace - \dfrac{m_{0}^{Y}\left(x, u\right)}{m_{0}^{S}\left(x, u\right)} \eqqcolon \overline{\Delta_{Y^{*}}^{OO}}\left(x, u\right).
\end{equation}
\end{corollary}

Furthermore, I can show that\footnote{The definition of pointwise sharpness used here and in the rest of the paper follows the definition of sharpness given by \citet[Remark 2.1.]{Canay2017}. Moreover, note that, if the functions $m_{0}^{Y}$, $m_{1}^{Y}$, $m_{0}^{S}$ and $\Delta_{S}$ are point identified only in a subset of the unit interval, then pointwise sharpness holds only in that subset.}:

\begin{theorem}\label{sharpbounds}
Suppose that the functions $m_{0}^{Y}$, $m_{1}^{Y}$, $m_{0}^{S}$ and $\Delta_{S}$ are point identified at every pair $\left(x, u\right) \in \mathcal{X} \times \left[0, 1\right]$. Under Assumptions \ref{ind}-\ref{support}, \ref{bounded} (sub-cases 1, 2, 3(a) or 3(b)) and \ref{increasing_sample_selection}, the bounds $\underline{\Delta_{Y^{*}}^{OO}}$ and $\overline{\Delta_{Y^{*}}^{OO}}$, given by Corollary \ref{MTEbounds}, are pointwise sharp, i.e., for any $\overline{u} \in \left[0, 1\right]$, $\overline{x} \in \mathcal{X}$ and $\delta\left(\overline{x}, \overline{u}\right) \in \left(\underline{\Delta_{Y^{*}}^{OO}}\left(\overline{x}, \overline{u}\right), \overline{\Delta_{Y^{*}}^{OO}}\left(\overline{x}, \overline{u}\right)\right)$, there exist random variables $\left(\tilde{Y}_{0}^{*}, \tilde{Y}_{1}^{*}, \tilde{U}, \tilde{V}\right)$ such that
\begin{equation}\label{faketarget}
\Delta_{\tilde{Y}^{*}}^{OO}\left(\overline{x}, \overline{u}\right) \coloneqq \mathbb{E}\left[\tilde{Y}_{1}^{*} - \tilde{Y}_{0}^{*} \left\vert X = \overline{x}, \tilde{U} = \overline{u}, \tilde{S}_{0} = 1, \tilde{S}_{1} = 1 \right.\right] = \delta\left(\overline{x}, \overline{u}\right),
\end{equation}
\begin{equation}\label{correctsupport}
\mathbb{P}\left[\left. \left(\tilde{Y}_{0}^{*}, \tilde{Y}_{1}^{*}, \tilde{V}\right) \in \mathcal{Y}^{*} \times \mathcal{Y}^{*} \times \left[0, 1\right] \right\vert X = \overline{x}, \tilde{U} = u \right] = 1 \text{ for any } u \in \left[0, 1\right],
\end{equation}
and
\begin{equation}\label{DataRestriction}
F_{\tilde{Y}, \tilde{D}, \tilde{S}, Z, X}\left(y, d, s, z, \overline{x} \right) = F_{Y, D, S, Z, X} \left(y, d, s, z, \overline{x}\right)
\end{equation}
for any $\left(y, d, s, z\right) \in \mathbb{R}^{4}$,	where $\tilde{D} \coloneqq \mathbf{1}\left\lbrace P\left(X, Z\right) \geq \tilde{U}\right\rbrace$, $\tilde{S}_{0} = \mathbf{1}\left\lbrace Q\left(0, X\right) \geq \tilde{V}\right\rbrace$, $\tilde{S}_{1} = \mathbf{1}\left\lbrace Q\left(1, X\right) \geq \tilde{V}\right\rbrace$, $\tilde{Y}_{0} = \tilde{S}_{0} \cdot \tilde{Y}_{0}^{*}$, $\tilde{Y}_{1} = \tilde{S}_{1} \cdot \tilde{Y}_{1}^{*}$ and $\tilde{Y} = \tilde{D} \cdot \tilde{Y}_{1} + \left(1 - \tilde{D}\right) \cdot \tilde{Y}_{0}$.
\end{theorem}
\begin{proof}
	Here, I provide a sketch of the proof of Theorem \ref{sharpbounds}. Appendix \ref{proofsharp} contains its detailed version. I define candidate random variables $\left(\tilde{Y}_{0}^{*}, \tilde{Y}_{1}^{*}, \tilde{U}, \tilde{V}\right)$ through their joint cumulative distribution function $F_{\tilde{Y}_{0}^{*}, \tilde{Y}_{1}^{*}, \tilde{U}, \tilde{V}, Z, X}$ and then check that equations \eqref{faketarget}, \eqref{correctsupport} and \eqref{DataRestriction} are satisfied. Intuitively, I define this joint probability function to be equal to $F_{Y_{0}^{*}, Y_{1}^{*}, U, V, Z, X}$ at every point, but the point $\tilde{U} = \bar{u}$. By doing so, I ensure that the equation \eqref{DataRestriction} holds because $\tilde{U} = \bar{u}$ is associated to a mass zero set. I, then, define the function $F_{\tilde{Y}_{0}^{*}, \tilde{Y}_{1}^{*}, \tilde{U}, \tilde{V}, Z, X}$ at $\tilde{U} = \bar{u}$ to ensure that equations \eqref{faketarget} and \eqref{correctsupport} hold.
\end{proof}

Intuitively, Theorem \ref{sharpbounds} says that, for any $\delta\left(\overline{x}, \overline{u}\right) \in \left(\underline{\Delta_{Y^{*}}^{OO}}\left(\overline{x}, \overline{u}\right), \overline{\Delta_{Y^{*}}^{OO}}\left(\overline{x}, \overline{u}\right)\right)$, it is possible to create candidate random variables $\left(\tilde{Y}_{0}^{*}, \tilde{Y}_{1}^{*}, \tilde{U}, \tilde{V}\right)$ that generate the candidate marginal treatment effect $\delta\left(\overline{x}, \overline{u}\right)$ (equation \eqref{faketarget}), satisfy the bounded support condition --- a restriction imposed by my model (Assumption \ref{bounded}) and summarized in equation \eqref{correctsupport} --- and generate the same distribution of the observable variables --- a restriction imposed by the data and summarized in equation \eqref{DataRestriction}. In other words, the data and the model in Section \ref{model} do not generate enough restrictions to refute that the true target parameter $\Delta_{Y^{*}}^{OO}\left(\overline{x}, \overline{u}\right)$ is equal to the candidate target parameter $\delta\left(\overline{x}, \overline{u}\right)$.

Moreover, the bounded support condition (Assumption \ref{bounded}) is partially necessary to the existence of bounds for the target parameter $\Delta_{Y^{*}}^{OO}\left(\overline{x}, \overline{u}\right)$. When the support is unbounded in both directions (i.e., $\underline{y}^{*} = - \infty$ and $\overline{y}^{*} = + \infty$), then it is impossible to derive bounds for the target parameter $\Delta_{Y^{*}}^{OO}\left(\overline{x}, \overline{u}\right)$ without any extra assumption. Proposition \ref{partialnecessary} formalizes this last statement.\footnote{Appendix \ref{proofnecessary} contains the proof of this proposition, whose intuition is similar to the one provided for Theorem \ref{sharpbounds}.}
\begin{proposition}\label{partialnecessary}
Suppose that the functions $m_{0}^{Y}$, $m_{1}^{Y}$, $m_{0}^{S}$ and $\Delta_{S}$ are point identified at every pair $\left(x, u\right) \in \mathcal{X} \times \left[0, 1\right]$. Impose Assumptions \ref{ind}-\ref{support} and \ref{increasing_sample_selection}. If $\mathcal{Y}^{*} = \mathbb{R}$, then, for any $\overline{u} \in \left[0, 1\right]$, $\overline{x} \in \mathcal{X}$ and $\delta\left(\overline{x}, \overline{u}\right) \in \mathbb{R}$, there exist random variables $\left(\tilde{Y}_{0}^{*}, \tilde{Y}_{1}^{*}, \tilde{U}, \tilde{V}\right)$ such that
\begin{equation}\label{faketargetP}
\Delta_{\tilde{Y}^{*}}^{OO}\left(\overline{x}, \overline{u}\right) \coloneqq \mathbb{E}\left[\tilde{Y}_{1}^{*} - \tilde{Y}_{0}^{*} \left\vert X = \overline{x}, \tilde{U} = \overline{u}, \tilde{S}_{0} = 1, \tilde{S}_{1} = 1 \right.\right] = \delta\left(\overline{x}, \overline{u}\right),
\end{equation}
\begin{equation}\label{correctsupportP}
\mathbb{P}\left[\left. \left(\tilde{Y}_{0}^{*}, \tilde{Y}_{1}^{*}, \tilde{V}\right) \in \mathcal{Y}^{*} \times \mathcal{Y}^{*} \times \left[0, 1\right] \right\vert X = \overline{x}, \tilde{U} = u \right] = 1 \text{ for any } u \in \left[0, 1\right],
\end{equation}
and
\begin{equation}\label{DataRestrictionP}
F_{\tilde{Y}, \tilde{D}, \tilde{S}, Z, X}\left(y, d, s, z, \overline{x} \right) = F_{Y, D, S, Z, X} \left(y, d, s, z, \overline{x}\right)
\end{equation}
for any $\left(y, d, s, z\right) \in \mathbb{R}^{4}$,	where $\tilde{D} \coloneqq \mathbf{1}\left\lbrace P\left(X, Z\right) \geq \tilde{U}\right\rbrace$, $\tilde{S}_{0} = \mathbf{1}\left\lbrace Q\left(0, X\right) \geq \tilde{V}\right\rbrace$, $\tilde{S}_{1} = \mathbf{1}\left\lbrace Q\left(1, X\right) \geq \tilde{V}\right\rbrace$, $\tilde{Y}_{0} = \tilde{S}_{0} \cdot \tilde{Y}_{0}^{*}$, $\tilde{Y}_{1} = \tilde{S}_{1} \cdot \tilde{Y}_{1}^{*}$ and $\tilde{Y} = \tilde{D} \cdot \tilde{Y}_{1} + \left(1 - \tilde{D}\right) \cdot \tilde{Y}_{0}$.
\end{proposition}

In other words, when the support of the potential outcome is the entire real line, the data and the model in Section \ref{model} do not generate enough restrictions to refute that the true target parameter $\Delta_{Y^{*}}^{OO}\left(\overline{x}, \overline{u}\right)$ is equal to an arbitrarily large effect in magnitude. This impossibility result is interesting in light of the previous literature about partial identification of treatment effects with sample selection. In the case of the $ITT^{OO}$ (\cite{Lee2009}) and the $LATE^{OO}$ (\cite{Chen2015}), it is possible to construct informative bounds even when the support of the potential outcome is the entire real line. However, when focusing on a specific point of the $MTE^{OO}$ function, it is impossible to construct informative bounds when $\mathcal{Y}^{*} = \mathbb{R}$ due to the local nature of the target parameter.

There is one remark about the results I just derived. Theorem \ref{sharpbounds} and Proposition \ref{partialnecessary} do not impose any smoothness condition on the joint distribution of $\left(Y_{0}^{*}, Y_{1}^{*}, U, V, Z, X\right)$. In particular, the conditional cumulative distribution functions $F_{V \left\vert X, U \right.}$, $F_{Y_{0}^{*} \left\vert X, U, V \right.}$ and $F_{Y_{1}^{*} \left\vert X, U, V \right.}$ are allowed to be discontinuous functions of U at the point $\overline{u}$. Appendix \ref{discontinuous} states and proves a sharpness result similar to Theorem \ref{sharpbounds} and an impossibility result similar to Proposition \ref{partialnecessary} when $F_{V \left\vert X, U \right.}$, $F_{Y_{0}^{*} \left\vert X, U, V \right.}$ and $F_{Y_{1}^{*} \left\vert X, U, V \right.}$ must be continuous functions of U.

%%%%%%%%%%%%%%%%%%%%%%%%%%%%%%%%%%%%%%%%%%%%%%
% Partial Identification with an Extra Mean Dominance Assumption
%%%%%%%%%%%%%%%%%%%%%%%%%%%%%%%%%%%%%%%%%%%%%%
\subsection{Partial Identification with an Extra Mean Dominance Assumption}\label{MeanDominanceSection}

Here, I use the Mean Dominance Assumption \ref{meandominanceG} to tighten the bounds for the target parameter $\Delta_{Y^{*}}^{OO}$ (equation \eqref{target}) given by Corollary \ref{MTEbounds}. Note that Assumption \ref{meandominanceG} implies that $\Delta_{Y}^{NO}\left(x, u\right) \leq \dfrac{m_{1}^{Y}\left(x, u\right)}{m_{1}^{S}\left(x, u\right)} \leq \mathbb{E}\left[Y_{1}^{*} \left\vert X = x, U = u, S_{0} = 1, S_{1} = 1 \right.\right]$ by equations \eqref{DeltaY_NO} and \eqref{decomposition}.  As a consequence, by following the same steps of the proof of corollary \ref{MTEbounds}, I can derive:
\begin{corollary}\label{boundmeandomG}
	Fix $u \in \left[0, 1\right]$ and $x \in \mathcal{X}$ arbitrarily. Suppose that the $m_{0}^{Y}\left(x, u\right)$, $m_{1}^{Y}\left(x, u\right)$, $m_{0}^{S}\left(x, u\right)$ and $\Delta_{S}\left(x, u\right)$ are point identified.
	
	Under Assumptions \ref{ind}-\ref{support}, \ref{bounded}.1, \ref{increasing_sample_selection} and \ref{meandominanceG}, $\Delta_{Y^{*}}^{OO}\left(x, u\right)$ must satisfy
	\begin{equation}\label{mte_lb_md}
	\Delta_{Y^{*}}^{OO}\left(x, u\right) \geq \dfrac{m_{1}^{Y}\left(x, u\right)}{m_{1}^{S}\left(x, u\right)} - \dfrac{m_{0}^{Y}\left(x, u\right)}{m_{0}^{S}\left(x, u\right)} \eqqcolon \underline{\Delta_{Y^{*}}^{OO}}\left(x, u\right)
	\end{equation}
	and
	\begin{equation}\label{mte_ub_md}
	\Delta_{Y^{*}}^{OO}\left(x, u\right) \leq \dfrac{m_{1}^{Y}\left(x, u\right) - \underline{y}^{*} \cdot \Delta_{S}\left(x, u\right)}{m_{0}^{S}\left(x, u\right)} - \dfrac{m_{0}^{Y}\left(x, u\right)}{m_{0}^{S}\left(x, u\right)} \eqqcolon \overline{\Delta_{Y^{*}}^{OO}}\left(x, u\right).
	\end{equation}
	
	Under Assumptions \ref{ind}-\ref{support}, \ref{bounded}.2, \ref{increasing_sample_selection} and \ref{meandominanceG}, $\Delta_{Y^{*}}^{OO}\left(x, u\right)$ must satisfy
	\begin{equation}
	\Delta_{Y^{*}}^{OO}\left(x, u\right) \geq \dfrac{m_{1}^{Y}\left(x, u\right)}{m_{1}^{S}\left(x, u\right)} - \dfrac{m_{0}^{Y}\left(x, u\right)}{m_{0}^{S}\left(x, u\right)} \eqqcolon \underline{\Delta_{Y^{*}}^{OO}}\left(x, u\right)
	\end{equation}
	and
	\begin{equation}
	\Delta_{Y^{*}}^{OO}\left(x, u\right) \leq \overline{y}^{*} - \dfrac{m_{0}^{Y}\left(x, u\right)}{m_{0}^{S}\left(x, u\right)} \eqqcolon \overline{\Delta_{Y^{*}}^{OO}}\left(x, u\right).
	\end{equation}
	
	Under Assumptions \ref{ind}-\ref{support}, \ref{bounded}.3 (sub-case (a) or (b)), \ref{increasing_sample_selection} and \ref{meandominanceG}, $\Delta_{Y^{*}}^{OO}\left(x, u\right)$ must satisfy
	\begin{equation}
	\Delta_{Y^{*}}^{OO}\left(x, u\right) \geq \dfrac{m_{1}^{Y}\left(x, u\right)}{m_{1}^{S}\left(x, u\right)} - \dfrac{m_{0}^{Y}\left(x, u\right)}{m_{0}^{S}\left(x, u\right)} \eqqcolon \underline{\Delta_{Y^{*}}^{OO}}\left(x, u\right)
	\end{equation}
	and
	\begin{equation}
	\Delta_{Y^{*}}^{OO}\left(x, u\right) \leq \min\left\lbrace \dfrac{m_{1}^{Y}\left(x, u\right) - \underline{y}^{*} \cdot \Delta_{S}\left(x, u\right)}{m_{0}^{S}\left(x, u\right)}, \overline{y}^{*}\right\rbrace - \dfrac{m_{0}^{Y}\left(x, u\right)}{m_{0}^{S}\left(x, u\right)} \eqqcolon \overline{\Delta_{Y^{*}}^{OO}}\left(x, u\right).
	\end{equation}
	
	When $\mathcal{Y}^{*} = \mathbb{R}$ and Assumptions \ref{ind}-\ref{support}, \ref{increasing_sample_selection} and \ref{meandominanceG} hold, $\Delta_{Y^{*}}^{OO}\left(x, u\right)$ must satisfy
	\begin{equation}\label{lowermean}
	\Delta_{Y^{*}}^{OO}\left(x, u\right) \geq \dfrac{m_{1}^{Y}\left(x, u\right)}{m_{1}^{S}\left(x, u\right)} - \dfrac{m_{0}^{Y}\left(x, u\right)}{m_{0}^{S}\left(x, u\right)} \eqqcolon \underline{\Delta_{Y^{*}}^{OO}}\left(x, u\right)
	\end{equation}
	and
	\begin{equation}
	\Delta_{Y^{*}}^{OO}\left(x, u\right) \leq \infty \eqqcolon \overline{\Delta_{Y^{*}}^{OO}}\left(x, u\right).
	\end{equation}
\end{corollary}

Notice that, under Mean Dominance Assumption \ref{meandominanceG}, I can increase the lower bounds proposed in Corollary \ref{MTEbounds} under Assumption \ref{bounded} and provide an informative lower bound even when the support of the outcome of interest is the entire real line, a result in stark contrast with Proposition \ref{partialnecessary}.\footnote{Appendix \ref{comparing} discusses when Corollary \ref{boundmeandomG} provides bounds that are strictly tighter than the ones provided by Corollary \ref{MTEbounds}.} These improvements clearly show the identifying power of the Mean Dominance Assumption \ref{meandominanceG}. Moreover, the phenomenon of obtaining more informative bounds by imposing extra assumptions is common in the partial identification literature, as explained by \cite{Tamer2010} and illustrated by \cite{Kline2016}.

As in Subsection \ref{MonotonicitySection}, I assume that $m_{0}^{Y}\left(x, u\right)$, $m_{1}^{Y}\left(x, u\right)$, $m_{0}^{S}\left(x, u\right)$, $m_{1}^{S}\left(x, u\right)$, and $\Delta_{S}\left(x, u\right)$ are point identified, postponing the discussion about their identification to Sections \ref{interval} and \ref{discrete}.

Now, using the above corollary, I can combine the sharpness and the impossibility results of Subsection \ref{MonotonicitySection} in one single proposition\footnote{Appendix \ref{proofsharpboundsmeanG} contains a proof of this proposition, whose intuition is similar to the one provided for Theorem \ref{sharpbounds}. The only difference is that, now, the function $F_{\tilde{Y}_{0}^{*}, \tilde{Y}_{1}^{*}, \tilde{U}, \tilde{V}, Z, X}$ at $\tilde{U} = \bar{u}$ must also satisfy equation \eqref{fakemeandominanceG}.}:
\begin{proposition}\label{sharpboundsmeanG}
	Suppose that the functions $m_{0}^{Y}$, $m_{1}^{Y}$, $m_{0}^{S}$, $m_{1}^{S}$ and $\Delta_{S}$ are point identified at every pair $\left(x, u\right) \in \mathcal{X} \times \left[0, 1\right]$. Under Assumptions \ref{ind}-\ref{support}, \ref{increasing_sample_selection} and \ref{meandominanceG}, the bounds $\underline{\Delta_{Y^{*}}^{OO}}$ and $\overline{\Delta_{Y^{*}}^{OO}}$, given by Corollary \ref{boundmeandomG}, are pointwise sharp, i.e., for any $\overline{u} \in \left[0, 1\right]$, $\overline{x} \in \mathcal{X}$ and $\delta\left(\overline{x}, \overline{u}\right) \in \left(\underline{\Delta_{Y^{*}}^{OO}}\left(\overline{x}, \overline{u}\right), \overline{\Delta_{Y^{*}}^{OO}}\left(\overline{x}, \overline{u}\right)\right)$, there exist random variables $\left(\tilde{Y}_{0}^{*}, \tilde{Y}_{1}^{*}, \tilde{U}, \tilde{V}\right)$ such that
	\begin{equation}\label{faketargetmeanG}
	\Delta_{\tilde{Y}^{*}}^{OO}\left(\overline{x}, \overline{u}\right) \coloneqq \mathbb{E}\left[\tilde{Y}_{1}^{*} - \tilde{Y}_{0}^{*} \left\vert X = \overline{x}, \tilde{U} = \overline{u}, \tilde{S}_{0} = 1, \tilde{S}_{1} = 1 \right.\right] = \delta\left(\overline{x}, \overline{u}\right),
	\end{equation}
	\begin{equation}\label{correctsupportmeanG}
	\mathbb{P}\left[\left. \left(\tilde{Y}_{0}^{*}, \tilde{Y}_{1}^{*}, \tilde{V}\right) \in \mathcal{Y}^{*} \times \mathcal{Y}^{*} \times \left[0, 1\right] \right\vert X = \overline{x}, \tilde{U} = u \right] = 1 \text{ for any } u \in \left[0, 1\right],
	\end{equation}
	\begin{equation}\label{fakemeandominanceG}
	\mathbb{E}\left[\tilde{Y}_{1}^{*} \left\vert X = \overline{x}, \tilde{U} = \overline{u}, \tilde{S}_{0} = 1, \tilde{S}_{1} = 1 \right.\right] \geq \mathbb{E}\left[\tilde{Y}_{1}^{*} \left\vert X = \overline{x}, \tilde{U} = \overline{u}, \tilde{S}_{0} = 0, \tilde{S}_{1} = 1 \right.\right],
	\end{equation}
	and
	\begin{equation}\label{DataRestrictionmeanG}
	F_{\tilde{Y}, \tilde{D}, \tilde{S}, Z, X}\left(y, d, s, z, \overline{x} \right) = F_{Y, D, S, Z, X} \left(y, d, s, z, \overline{x}\right)
	\end{equation}
	for any $\left(y, d, s, z\right) \in \mathbb{R}^{4}$,	where $\tilde{D} \coloneqq \mathbf{1}\left\lbrace P\left(X, Z\right) \geq \tilde{U}\right\rbrace$, $\tilde{S}_{0} = \mathbf{1}\left\lbrace Q\left(0, X\right) \geq \tilde{V}\right\rbrace$, $\tilde{S}_{1} = \mathbf{1}\left\lbrace Q\left(1, X\right) \geq \tilde{V}\right\rbrace$, $\tilde{Y}_{0} = \tilde{S}_{0} \cdot \tilde{Y}_{0}^{*}$, $\tilde{Y}_{1} = \tilde{S}_{1} \cdot \tilde{Y}_{1}^{*}$ and $\tilde{Y} = \tilde{D} \cdot \tilde{Y}_{1} + \left(1 - \tilde{D}\right) \cdot \tilde{Y}_{0}$.
\end{proposition}

Note that, in addition to all the restriction imposed by Theorem \ref{sharpbounds}, the candidate random  variables $\left(\tilde{Y}_{0}^{*}, \tilde{Y}_{1}^{*}, \tilde{U}, \tilde{V}\right)$ must also satisfy an extra model restriction (equation \eqref{fakemeandominanceG}) associated with the Mean Dominance Assumption \ref{meandominanceG}. Intuitively, Proposition \ref{sharpboundsmeanG} says that the data (equation \eqref{DataRestrictionmeanG}) and the model (equations \eqref{correctsupportmeanG} and \eqref{fakemeandominanceG}) do not generate enough restrictions to refute that the true target parameter $\Delta_{Y^{*}}^{OO}\left(\overline{x}, \overline{u}\right)$ is equal to the candidate target parameter $\delta\left(\overline{x}, \overline{u}\right)$ (equation \eqref{faketargetmeanG}).

%%%%%%%%%%%%%%%%%%%%%%%%%%%%%%%%%%%%%%%%%%%%%%
% Empirical Relevance of bounds for the MTE of Interest
%%%%%%%%%%%%%%%%%%%%%%%%%%%%%%%%%%%%%%%%%%%%%%
\subsection{Empirical Relevance of bounds for the $\mathbf{MTE^{OO}}$ of Interest}\label{EmpiricalRelevance}
Now, it is worth discussing the empirical relevance of partially identifying the $MTE^{OO}$ of interest. First, bounds for the $MTE^{OO}$ can illuminate the heterogeneity of the treatment effect, allowing the researcher to understand who would benefit and who would lose with a specific treatment. This is important because common parameters (e.g., $ATE^{OO}$, $ATT^{OO}$, $ATU^{OO}$, $LATE^{OO}$) can be positive even when most people lose with a policy if the few winners have very large gains. Moreover, knowing, even partially, the $MTE^{OO}$ function can be useful to optimally design policies that provides incentives to agents to take some treatment. Second, I can use the $MTE^{OO}$ bounds to partially identify any treatment effect that is described as a weighted integral of $\Delta_{Y^{*}}^{OO}\left(x, u\right)$ because
\begin{align}
	\int_{0}^{1} \left(\underline{\Delta_{Y^{*}}^{OO}}\left(x, u\right)\right) \cdot \omega\left(x, u\right) \, \text{d} u & \leq \int_{0}^{1} \Delta_{Y^{*}}^{OO}\left(x, u\right) \cdot \omega\left(x, u\right) \, \text{d} u \nonumber \\
	& \label{integralbounds} \leq \int_{0}^{1} \left(\overline{\Delta_{Y^{*}}^{OO}}\left(x, u\right)\right) \cdot \omega\left(x, u\right) \, \text{d} u,
\end{align}
where $\omega(x, \cdot)$ is a known or identifiable weighting function. Even though such bounds may not be sharp for any specific parameter, they are a general and off-the-shelf solution to many empirical problems. As a consequence of this trade-off, I recommend the applied researcher to use a specialized tool if he or she is interested in a parameter that already has specific bounds for it (e.g., $ITT^{OO}$ by \cite{Lee2009} and $LATE^{OO}$ by \cite{Chen2015}). However, I suggest the applied researcher to easily compute a weighted integral of pointwise sharp bounds for the MTE of interest if he or she is interested in parameters without specialized bounds (e.g., ATE, ATT and ATU in the case with imperfect compliance). In other words, facing a trade-off between empirical flexibility and sharpness, the partial identification tool proposed in this paper focus on empirical flexibility while still ensuring pointwise sharpness of the bounds for the MTE of interest.

Tables \ref{integral} and \ref{weights} show some of the treatment effect parameters that can be partially identified using inequality \eqref{integralbounds}. More examples are given by \citet[Tables 1A and 1B]{Heckman2006} and \citet[Table 1]{Mogstad2017}.

\begin{table}[!htbp]
	\centering
	\caption{{Treatment Effects as Weighted Integrals of the Marginal Treatment Effect}} \label{integral}
	\begin{lrbox}{\tablebox}
		\begin{tabular}{l}
			\hline
			\hline
			
			$ATE^{OO} = \mathbb{E}\left[Y_{1}^{*} - Y_{0}^{*} \left\vert S_{0} = 1, S_{1} = 1 \right.\right] = \int_{0}^{1} \Delta_{Y^{*}}^{OO}\left(u\right) \, \text{d} u$ \\
			
			\\
			
			$ATT^{OO} = \mathbb{E}\left[Y_{1}^{*} - Y_{0}^{*} \left\vert D = 1, S_{0} = 1, S_{1} = 1 \right.\right] = \int_{0}^{1} \Delta_{Y^{*}}^{OO}\left(u\right) \cdot \omega_{ATT}\left(u\right) \, \text{d} u$ \\
			
			\\
			
			$ATU^{OO} = \mathbb{E}\left[Y_{1}^{*} - Y_{0}^{*} \left\vert D = 0, S_{0} = 1, S_{1} = 1 \right.\right] = \int_{0}^{1} \Delta_{Y^{*}}^{OO}\left(u\right) \cdot \omega_{ATU}\left(u\right) \, \text{d} u$ \\
			
			\\
			
			$LATE^{OO}(\underline{u}, \overline{u}) = \mathbb{E}\left[Y_{1}^{*} - Y_{0}^{*} \left\vert U \in \left[\underline{u}, \overline{u}\right], S_{0} = 1, S_{1} = 1 \right.\right] = \int_{0}^{1} \Delta_{Y^{*}}^{OO}\left(u\right) \cdot \omega_{LATE}\left(u\right) \, \text{d} u$ \\
					
			\hline			
		\end{tabular}
	\end{lrbox}
	\usebox{\tablebox}\\
	\settowidth{\tableboxwidth}{\usebox{\tablebox}} \parbox{\tableboxwidth}{\footnotesize{Source: \citet{Heckman2006} and \citet{Mogstad2017}. Note: Conditioning on $X$ is kept implicit in this table for brevity.}
	}
\end{table}

\begin{table}[!htbp]
	\centering
	\caption{{Weights}} \label{weights}
	\begin{lrbox}{\tablebox}
		\begin{tabular}{l}
			\hline
			\hline
			
			$\omega_{ATT}\left(x, u\right) = \dfrac{\int_{u}^{1} f_{P\left(W\right) \left\vert X \right.}\left(p \left\vert x \right.\right) \, \text{d} p}{\mathbb{E}\left[P\left(W\right) \left\vert X = x \right.\right]}$ \\
			
			\\
			
			$\omega_{ATU}\left(x, u\right) = = \dfrac{\int_{0}^{u} f_{P\left(W\right) \left\vert X \right.}\left(p \left\vert x \right.\right) \, \text{d} p}{ 1- \mathbb{E}\left[P\left(W\right) \left\vert X = x \right.\right]}$ \\
			
			\\
			
			$\omega_{LATE}\left(x, u\right) = \dfrac{\mathbf{1}\left\lbrace u \in \left[\underline{u}, \overline{u}\right]  \right\rbrace}{\overline{u} - \underline{u}}$ \\
			
			\hline			
		\end{tabular}
	\end{lrbox}
	\usebox{\tablebox}\\
	\settowidth{\tableboxwidth}{\usebox{\tablebox}} \parbox{\tableboxwidth}{\footnotesize{Source: \citet{Heckman2006} and \citet{Mogstad2017}.}
	}
\end{table}

%%%%%%%%%%%%%%%%%%%%%%%%%%%%%%%%%%%%%%%%%%%%%%%%%%%%%
% bounds for the MTE when the support of the propensity score is an interval
%%%%%%%%%%%%%%%%%%%%%%%%%%%%%%%%%%%%%%%%%%%%%%%%%%%%%

\section{Partial identification when the support of the propensity score is an interval} \label{interval}
Here, I fix $x \in \mathcal{X}$ and impose that the support of the propensity score, defined by $\mathcal{P}_{x} \coloneqq \left\lbrace P\left(x, z\right): z \in \mathcal{Z} \right\rbrace$, is an interval\footnote{$\mathcal{P}_{x}$ as an interval may be achieved by a continuous instrument $Z$ or by the existence of independent covariates \citep{Carneiro2011}.}. Then, under Assumptions \ref{ind}-\ref{invariantX}, the MTR functions associated with any variable $A \in \left\lbrace Y, S \right\rbrace$ are point identified by\footnote{Appendix \ref{proofQ0Q1} contains a proof of this claim based on the Local Instrumental Variable (LIV) approach described by \cite{Heckman2005}.}:
\begin{equation} \label{identifyQ0}
m_{0}^{A}\left(x, p\right) = \mathbb{E}\left[A \left\vert X = x, P\left(W\right) = p, D = 0 \right.\right] - \dfrac{\partial \mathbb{E}\left[A \left\vert X = x, P\left(W\right) = p, D = 0 \right.\right]}{\partial p} \cdot \left(1 - p\right),
\end{equation}
and
\begin{equation}\label{identifyQ1}
m_{1}^{A}\left(x, p\right) = \mathbb{E}\left[A \left\vert X = x, P\left(W\right) = p, D = 1 \right.\right] +  \dfrac{\partial \mathbb{E}\left[A \left\vert X = x, P\left(W\right) = p, D = 1 \right.\right]}{\partial p} \cdot p
\end{equation}
for any $p \in \mathcal{P}_{x}$.

Finally, the pointwise sharp bounds for $\Delta_{Y^{*}}^{OO}\left(x, p\right)$ are point identified by combining equations \eqref{identifyQ0} and \eqref{identifyQ1}, the fact that $\Delta_{S}\left(x, p\right) = m_{1}^{S}\left(x, p\right) - m_{0}^{S}\left(x, p\right)$, and Corollaries \ref{MTEbounds} or \ref{boundmeandomG}.

%%%%%%%%%%%%%%%%%%%%%%%%%%%%%%%%%%%%%%%%%%%%%%%%%%%%%
% bounds for the MTE when the support of the propensity score is discrete
%%%%%%%%%%%%%%%%%%%%%%%%%%%%%%%%%%%%%%%%%%%%%%%%%%%%%

\section{Partial identification when the support of the propensity score is discrete} \label{discrete}

When the support of the propensity score is not an interval, I cannot point identify $m_{0}^{Y}\left(x, u\right)$, $m_{1}^{Y}\left(x, u\right)$, $m_{0}^{S}\left(x, u\right)$, $m_{1}^{S}\left(x, u\right)$, and $\Delta_{S}\left(x, u\right)$ without extra assumptions, implying that I cannot identify the bounds for $\Delta_{Y^{*}}^{OO}\left(x, u\right)$ given by Corollaries \ref{MTEbounds} or \ref{boundmeandomG}. There are two solutions for this lack of identification: I can non-parametrically bound those four objects (\cite{Mogstad2017}) or I can impose flexible parametric assumptions (\cite{Brinch2017}) to point identify them. While the first approach is discussed in Subsection \ref{boundsMTE}, the second one is detailed in Subsection \ref{parametric}.

%%%%%%%%%%%%%%%%%%%%%%%%%%%%%%%%%%%%%%%%%%%%%%%%%%%%%
% Non-parametric Bounds
%%%%%%%%%%%%%%%%%%%%%%%%%%%%%%%%%%%%%%%%%%%%%%%%%%%%%
\subsection{Non-parametric outer set around the $\mathbf{MTE^{OO}}$ of interest}\label{boundsMTE}
For any $u \in \left[0, 1\right]$ and $x \in \mathcal{X}$, I can bound $m_{0}^{S}\left(x, u\right)$, $m_{1}^{S}\left(x, u\right)$, $\Delta_{S}\left(x, u\right)$, $m_{0}^{Y}\left(x, u\right)$, $m_{1}^{Y}\left(x, u\right)$ and $\Delta_{Y}\left(x, u\right)$ using the machinery proposed by \cite{Mogstad2017}. To do so, fix $A \in \left\lbrace S, Y \right\rbrace$ and $d \in \left\lbrace 0, 1 \right\rbrace$ and define the pair of functions $m^{A} \coloneqq \left(m_{0}^{A}, m_{1}^{A}\right)$ and the set of admissible MTR functions $\mathcal{M}^{A} \ni m^{A}$. For example, in the case of a binary function, the admissible set would be $\mathcal{M}^{A} = \left[0,1\right]^{\mathcal{X} \times \left[0, 1\right]} \times \left[0,1\right]^{\mathcal{X} \times \left[0, 1\right]}$ and, in the case of the selection indicator, this set would be further restricted by Assumption \ref{increasing_sample_selection} to $$\mathcal{M}^{A} = \left\lbrace \left(m_{0}^{A}, m_{1}^{A}\right) \in \left[0,1\right]^{\mathcal{X} \times \left[0, 1\right]} \times \left[0,1\right]^{\mathcal{X} \times \left[0, 1\right]} \colon m_{1}^{A}\left(x, u\right) \geq m_{0}^{A}\left(x, u\right) \hspace{5pt} \forall \left(x, u\right) \in \mathcal{X} \times \left[0, 1\right] \right\rbrace.$$ Moreover, define the function $\Gamma_{A}^{*} \colon \mathcal{M}^{A} \rightarrow \mathbb{R}$  as:
\begin{equation*}
\Gamma_{A}^{*}\left(\tilde{m}^{A}\right) = \tilde{m}_{1}^{A}\left(x, u\right) - \tilde{m}_{0}^{A}\left(x, u\right),
\end{equation*}
and observe that $\Gamma_{A}^{*}\left(m^{A}\right) = \Delta_{A}\left(x, u\right)$. Furthermore, define $\mathcal{G}_{A}$ to be a collection of known or identified measurable functions $g_{A} \colon \left\lbrace 0, 1 \right\rbrace \times \mathcal{Z} \rightarrow \mathbb{R}$ whose second moment is finite. For each IV-like specification $g_{A} \in \mathcal{G}_{A}$, define also $\beta_{g_{A}} \coloneqq \mathbb{E}\left[g_{A}\left(D, Z\right) A \left\vert X = x \right.\right]$. According to \citet[Proposition 1]{Mogstad2017}, the function $\Gamma_{g_{A}} \colon \mathcal{M}^{A} \rightarrow \mathbb{R}$, defined as
\begin{align*}
\Gamma_{g_{A}}\left(\tilde{m}^{A}\right) & = \mathbb{E}\left[\left. \int_{0}^{1} \tilde{m}_{0}^{A}\left(X, u\right) \cdot g_{A}\left(0, Z\right) \cdot \mathbf{1}\left\lbrace p\left(W\right) < u \right\rbrace \, \text{d} u \right\vert X = x \right] \\
& \hspace{20pt} + \mathbb{E}\left[\left. \int_{0}^{1} \tilde{m}_{1}^{A}\left(X, u\right) \cdot g_{A}\left(1, Z\right) \cdot \mathbf{1}\left\lbrace p\left(W\right) \geq u \right\rbrace \, \text{d} u \right\vert X = x \right],
\end{align*}
satisfies $\Gamma_{g_{A}}\left(m^{A}\right) = \beta_{g_{A}}$. As a result, $m^{A}$ must lie in the set $\mathcal{M}_{\mathcal{G}_{A}}$ of admissible functions that satisfy the restrictions imposed by the data through the IV-like specifications, where:
\begin{equation*}
\mathcal{M}_{\mathcal{G}_{A}} \coloneqq \left\lbrace \tilde{m}^{A} \in \mathcal{M}^{A} \colon \Gamma_{g_{A}}\left(\tilde{m}^{A}\right) = \beta_{g_{A}} \text{ for all } g_{A} \in \mathcal{G}_{A} \right\rbrace.
\end{equation*}

Assuming that $\mathcal{M}^{A}$ is convex and $M_{\mathcal{G}_{A}} \neq \emptyset$ for every $A \in \left\lbrace S, Y \right\rbrace$, \citet[Proposition 2]{Mogstad2017} show that:
\begin{equation}\label{nonparametricBoundsMTE}
\begin{array}{rcccl}
\inf\limits_{\tilde{m}^{A} \in \mathcal{M}_{\mathcal{G}_{A}}} \Gamma_{A}^{*} \left( \tilde{m}^{A} \right) & \eqqcolon \underline{\Delta_{A}\left(x, u\right)} \leq & \Delta_{A}\left(x, u\right) & \leq \overline{\Delta_{A}\left(x, u\right)} \coloneqq & \sup\limits_{\tilde{m}^{A} \in \mathcal{M}_{\mathcal{G}_{A}}} \Gamma_{A}^{*} \left( \tilde{m}^{A} \right).
\end{array}
\end{equation}
Based on this result, I can also define bounds for the MTR functions as $$\left(\overline{m_{0}^{A}\left(x, u\right)}, \underline{m_{1}^{A}\left(x, u\right)}\right) \coloneqq \arginf\limits_{\tilde{m}^{A} \in \mathcal{M}_{\mathcal{G}_{A}}} \Gamma_{A}^{*} \left( \tilde{m}^{A} \right) \text{ and } \left(\underline{m_{0}^{A}\left(x, u\right)}, \overline{m_{1}^{A}\left(x, u\right)}\right) \coloneqq \argsup\limits_{\tilde{m}^{A} \in \mathcal{M}_{\mathcal{G}_{A}}} \Gamma_{A}^{*} \left( \tilde{m}^{A} \right),$$ where
\begin{equation}\label{nonparametricBoundsMTR}
\underline{m_{d}^{A}\left(x, u\right)} \leq m_{d}^{A}\left(x, u\right) \leq \overline{m_{d}^{A}\left(x, u\right)} \text{ for any } d \in \left\lbrace 0, 1 \right\rbrace.
\end{equation}

As a consequence, I can combine Corollaries \ref{MTEbounds} and \ref{boundmeandomG} and inequalities \eqref{nonparametricBoundsMTE} and \eqref{nonparametricBoundsMTR} to provide a non-parametrically identified outer set around $\Delta_{Y^{*}}^{OO}\left(x, u\right)$, that contains the true target parameter $\Delta_{Y^{*}}^{OO}\left(x, u\right)$ by construction. However, the cost of non-parametric partial identification of $m_{0}^{S}\left(x, u\right)$, $m_{1}^{S}\left(x, u\right)$, $\Delta_{S}\left(x, u\right)$, $m_{0}^{Y}\left(x, u\right)$, $m_{1}^{Y}\left(x, u\right)$ and $\Delta_{Y}\left(x, u\right)$ is losing the pointwise sharpness of the bounds around the target parameter $\Delta_{Y^{*}}^{OO}\left(x, u\right)$.

%%%%%%%%%%%%%%%%%%%%%%%%%%%%%%%%%%%%%%%%%%%%%%%%%%%%%
% Parametric Bounds
%%%%%%%%%%%%%%%%%%%%%%%%%%%%%%%%%%%%%%%%%%%%%%%%%%%%%
\subsection{Parametric identification of the $\mathbf{MTE^{OO}}$ bounds}\label{parametric}

The fully non-parametric approach explained in Subsection \ref{boundsMTE} may provide an uninformative outer set (e.g., equal to $\overline{y}^{*} - \underline{y}^{*}$ or $\underline{y}^{*} - \overline{y}^{*}$ when the support of the potential outcome is bounded). In such cases, parametric assumptions on the marginal treatment response functions may buy a lot of identifying power. Although restrictive in principle, parametric assumptions may be flexible enough to provide credible bounds for $\Delta_{Y^{*}}^{OO}\left(x, u\right)$, as illustrated by \cite{Brinch2017}.

I fix $x \in \mathcal{X}$ and assume that the support of the propensity score $P\left(x, Z\right)$ is discrete and given by $\mathcal{P}_{x} = \left\lbrace p_{x, 1}, \ldots, p_{x,N} \right\rbrace$ for some $N \in \mathbb{N}$. I could directly apply the identification strategy proposed by \cite{Brinch2017} by assuming that the MTR functions associated with $Y$ and $S$ are polynomial functions of $U$. However, this assumption is problematic for binary variables, such as the selection indicator $S$. For this reason, I make a small modification to the procedure created by \cite{Brinch2017}: for $d \in \left\lbrace 0, 1 \right\rbrace$ and $A \in \left\lbrace Y, S \right\rbrace$, the MTR function is given by
\begin{equation}\label{polynomialMTR}
m_{d}^{A}\left(x, u\right) = M^{A}\left(u, \boldsymbol{\theta}_{x, d}^{A}\right)
\end{equation}
for any $u \in \left[0, 1\right]$, where $\Theta_{x}^{A} \subset \mathbb{R}^{2L}$ is a set of feasible parameters, $L \in \left\lbrace 1, \ldots, N \right\rbrace$ is the number of parameters for each treatment group $d$, $\left(\boldsymbol{\theta}_{x, 0}^{A}, \boldsymbol{\theta}_{x, 1}^{A}\right) \in \Theta_{x}^{A}$ is a vector of pseudo-true unknown parameters, and $M^{A} \colon \left[0, 1\right] \times \mathbb{R}^{2L} \rightarrow \mathbb{R}$ is a known function. For instance, in the case of a binary variable, a reasonable choice of $M^{A}$ is the Bernstein Polynomial $\left(M^{A}\left(u, \boldsymbol{\theta}_{x, d}^{A}\right) = \sum_{l = 0}^{L - 1} \theta_{x, d, l}^{A} \cdot {L - 1 \choose l} \cdot u^{l} \cdot \left(1 - u\right)^{L - 1 - l}\right)$ with feasible set $\Theta_{x}^{A} = \left[0, 1\right]^{2L}$. In the case of the selection indicator, the feasible set would be further restricted by Assumption \ref{increasing_sample_selection} to $\Theta_{x}^{A} = \left\lbrace \left(\boldsymbol{\tilde{\theta}}_{x, 0}^{A}, \boldsymbol{\tilde{\theta}}_{x, 1}^{A}\right) \in \left[0, 1\right]^{2L} \colon \boldsymbol{\tilde{\theta}}_{x, 1}^{A} \geq \boldsymbol{\tilde{\theta}}_{x, 0}^{A} \right\rbrace$. I stress that the only difference between the Bernstein polynomial model and the simple polynomial model proposed by \cite{Brinch2017} is that it is easier to impose feasibility restrictions on the former model.

Back to the parametric model given by equation \eqref{polynomialMTR}, I define the parameters $\left(\boldsymbol{\theta}_{x, 0}^{A}, \boldsymbol{\theta}_{x, 1}^{A}\right)$ as pseudo-true parameters in the sense that the parametric model in equation \eqref{polynomialMTR} is an approximation to the true data generating process via the moments $\mathbb{E}\left[A \left\vert X = x, P\left(W\right) = p_{n}, D = d \right.\right]$ for any $d \in \left\lbrace 0, 1 \right\rbrace$ and $n \in \left\lbrace 1, \ldots, N \right\rbrace$. Formally, I define
\begin{align}
\left(\boldsymbol{\theta}_{x, 0}^{A}, \boldsymbol{\theta}_{x, 1}^{A}\right) \coloneqq \argmin\limits_{\left(\boldsymbol{\tilde{\theta}}_{x, 0}^{A}, \boldsymbol{\tilde{\theta}}_{x, 1}^{A}\right) \in \Theta_{x}^{A}} \hspace{5pt} \sum_{n = 1}^{N} & \left\lbrace \left(\mathbb{E}\left[A \left\vert X = x, P\left(W\right) = p_{n}, D = 0 \right.\right] - \dfrac{\int_{p_{n}}^{1} M^{A}\left(u, \boldsymbol{\tilde{\theta}}_{x,0}^{A}\right) \, \text{d} u}{1 - p_{n}}\right)^{2} \right. \nonumber\\
& \label{pseudotrue} + \left. \left(\mathbb{E}\left[A \left\vert X = x, P\left(W\right) = p_{n}, D = 1 \right.\right] - \dfrac{\int_{0}^{p_{n}} M^{A}\left(u, \boldsymbol{\tilde{\theta}}_{x,1}^{A}\right) \, \text{d} u}{p_{n}}\right)^{2} \right\rbrace.
\end{align}

Note that, to estimate parameters $\left(\boldsymbol{\theta}_{x, 0}^{A}, \boldsymbol{\theta}_{x, 1}^{A}\right)$, I can simply use the sample analogue of equation \eqref{pseudotrue}, i.e., I only have to estimate a constrained OLS regression whose restrictions are given by the set $\Theta_{x}^{A}$. If the model restrictions imposed through the set of feasible parameters $\Theta_{x}^{A}$ are valid and $L = N$, then my parametric model collapses to the model proposed by \cite{Brinch2017} and I find that\footnote{Appendix \ref{parametricproof} contains a proof of this claim.}, for any $p_{n} \in \mathcal{P}_{x}$,
\begin{align}
\label{nonseparable0} \mathbb{E}\left[A \left\vert X = x, P\left(W\right) = p_{n}, D = 0 \right.\right] & = \dfrac{\int_{p_{n}}^{1} M^{A}\left(u, \boldsymbol{\theta}_{x,0}^{A}\right) \, \text{d} u}{1 - p_{n}} \\
\label{nonseparable1} \mathbb{E}\left[A \left\vert X = x, P\left(W\right) = p_{n}, D = 1 \right.\right] & = \dfrac{\int_{0}^{p_{n}} M^{A}\left(u, \boldsymbol{\theta}_{x,1}^{A}\right) \, \text{d} u}{p_{n}}.
\end{align}

I can then combine Corollaries \ref{MTEbounds} and \ref{boundmeandomG} and equations \eqref{polynomialMTR} and \eqref{pseudotrue} to bound $\Delta_{Y^{*}}^{OO}\left(x, u\right)$.

%%%%%%%%%%%%%%%%%%%%%%%%%%%%%%%%%%%%%%%%%%%%%%%%%%%%%
% Empirical Application
%%%%%%%%%%%%%%%%%%%%%%%%%%%%%%%%%%%%%%%%%%%%%%%%%%%%%

\section{Empirical Application: Job Corps Training Program}\label{application}
I focus on analyzing the Marginal Treatment Effect of the Job Corps Training Program (JCTP) on wages for the always-employed subpopulation ($MTE^{OO}$). This program provides free education and vocational training to individuals who are legal residents of the U.S., are between the ages of 16 and 24 and come from a low-income household (\cite{Schochet2001} and \cite{Lee2009}). Besides receiving education and vocational training, the trainees reside in the Job Corps center, that offers meals and a small cash allowance.

In the mid 1990's, the U.S. Department of Labor hired Mathematica Policy Resarch, Inc., to evaluate the JCTP through a randomized experiment. According to \cite{Chen2015}, eligible people who applied to JCTP for the first time between November 1994 and December 1995 (80,833 applicants) were randomly assigned into a treatment group and a control group. People in the control group (5,977) were embargoed from the program for 3 years, while those in the treatment group (74,856) were allowed to enroll in JC. However, in this randomized control trial, there was non-compliance (selection into treatment) because some individuals in the treated group decided not to participate in the program and some individuals in the control group were able to attend the JCTP even though they were officially embargoed.

To evaluate the JCTP, I start by describing the dataset, providing summary statistics and, most importantly, formally testing the assumptions that the potential treatment status is monotone on the instrument (equation \eqref{treatment}) and that the potential employment (sample selection status) is positively monotone on the treatment (Assumption \ref{increasing_sample_selection}) using the test elaborated by \cite{Machado2018}. I then estimate and discuss the marginal treatment responses and effects on employment and labor earnings using the parametric tool developed by \cite{Brinch2017}. Finally, I estimate and discuss the bounds for the $MTE^{OO}$ on wages without and with the mean dominance assumption (Assumption \ref{meandominanceG}), given, respectively, by Corollaries \ref{MTEbounds} and \ref{boundmeandomG}.

%%%%%%%%%%%%%%%%%%%%%
%% Descriptive Statistics and the Monotonicity Assumptions
%%%%%%%%%%%%%%%%%%%%%
\subsection{Descriptive Statistics and the Monotonicity Assumptions}
The publicly available National Job Corps Study (NJCS) sample contains 15,386 individuals --- all 5,977 control group individuals and 9,409 randomly selected treatment group individuals. All of them were interviewed at random assignment and at 12, 30 and 48 months after random assignment. Following \cite{Lee2009}, I only keep individuals with non-missing values for weekly earnings and weekly hours worked for every week after randomization (9,145). Following \cite{Chen2015}, my instrument ($Z$) is random treatment assignment and my treatment dummy ($D$) is an indicator variable that is equal to one if the individual was ever enrolled in the JCTP during the 208 weeks after random assignment. Since this variable has 51 missing values, the final sample size is 9,094 observations.

The dataset contains information about demographic covariates (sex, age, race, marriage, number of children, years of schooling, criminal behavior, personal income) and pre- and post-treatment labor market outcomes (employment and earnings). Following \cite{Chen2015}, hourly wages at week 208 are created by dividing weekly earnings by weekly hours worked at that week, implying that a missing wage is equivalent to zero weekly hours worked. I consider the person to be unemployed ($S = 0$) when the wage is missing and to be employed ($S = 1$) when the wage is non-missing. Differently from \cite{Lee2009} and \cite{Chen2015}, who use log hourly wages as their main outcome variable, my outcome of interest ($Y^{*}$) is the level of the hourly wage because Assumption \ref{bounded}.1 requires that the support $\mathcal{Y}^{*}$ has a finite lower bound. As a consequence, the observable outcome $Y$ is defined as hourly labor earnings. Finally, I use the NJCS design weights in my empirical analysis because some subpopulations were randomized with different, but known, probabilities (\cite{Schochet2001}).

Considering the results found by \cite{Flores-Lagunes2010}, who focus on explaining the negative but insignificant effects on employment and labor earnings for the Hispanic subpopulation, I separately analyze two subsamples from the NJCS sample: the Non-Hispanics subsample and the Hispanics subsample. Table \ref{descriptive} shows descriptive statistics for both subsamples. Note that, as expected, the pre-treatment covariates are, on average, very similar between the groups defined by the random treatment assignment. Consequently, both subsamples maintain the balance of baseline variables. However, when comparing Non-Hispanics and Hispanics, I find numerically small differences with respect to the variables \emph{female}, \emph{never married}, \emph{has children}, \emph{ever arrested}, \emph{has a job at baseline}, and \emph{had a job}.

\begin{table}[!htbp]
	\centering
	\caption{Summary Statistics of Selected Baseline Variables} \label{descriptive}
	\begin{lrbox}{\tablebox}
	\begin{tabular}{lccclccc}
		\hline \hline
		& \multicolumn{3}{c}{Non-Hispanic Sample} &  & \multicolumn{3}{c}{Hispanic Sample} \\ \cline{2-4} \cline{6-8} 
		& Z = 1 & Z = 0 & Diff. &  & Z = 1 & Z = 0 & Diff. \\ \hline
		Female &  .443 & .454 & -.011 &  & .502 & .473 & .030 \\
		& & & (.011) & & & & (.025) \\
		Age at baseline &  18.436 & 18.342 & .095* &  & 18.438 & 18.398 & .040 \\
		& & & (.049) & & & & (.109) \\
		White &  .318 & .318 & .000 &  & --- & --- & --- \\
		& & & (.011) & & & & \\
		Black &  .595 & .592 & .002 &  & --- & --- & ---  \\
		& & & (.011) & & & & \\
		Never married &  .926 & .924 & .002 &  & .875 & .874 & .001 \\
		& & & (.006) & & & & (.017) \\
		Has children &  .186 & .190 & -.004 &  & .201 & .206 & -.004 \\
		& & & (.009) & & & & (.020) \\
		Years of Schooling &  10.137 & 10.115 & .022  &  &  10.022 & 10.057 & -.034 \\
		& & & (.036) & & & & (.084) \\
		Ever arrested &  .255 & .257 & -.002 &  & .216 & .211 & .005  \\
		& & & (.010) & & & & (.021) \\
		Personal Inc.: $<$3000 & .787 & .788 & -.001 &  &  .789 & .794 & -.005 \\
		& & & (.010) & & & & (.022) \\
		Has a job at baseline &  .204 & .188 & .016*  &  &  .170 & .211 & -.041** \\
		& & & (.009) & & & & (.020) \\
		\emph{A year before baseline:} &  &  &  &  &  &  &  \\
		\hspace{5pt} Had a job &  .642 & .627 & .015  &  &  .601 & .630 & -.029 \\
		& & & (.011) & & & & (.025) \\
		\hspace{5pt} Months employed & 3.652 & 3.513 & .140  &  &  3.344 & 3.616 & -.272 \\
		& & & (.098) & & & & (.214) \\
		\hspace{5pt} Earnings & 2899.41 & 2795.62 & 103.79  &  &  2956.38 & 2885.47 & 70.91 \\
		& & & (103.81) & & & & (477.08) \\
		Observations & 4554 & 2977 & Total: 7531 &  & 942 & 621 & Total: 1563 \\
		\hline
	\end{tabular}
	\end{lrbox}
	\usebox{\tablebox}\\
	\settowidth{\tableboxwidth}{\usebox{\tablebox}} \parbox{\tableboxwidth}{\footnotesize{Note: Z indicates random treatment assignment. Robust standard errors are in parenthesis. ***, ** and * denote  that difference is statistically significant at the 1\%, at 5\% and 10\% level, respectively. Estimation uses design weights.}
	}
\end{table}

Table \ref{PreliminaryEffects} shows preliminary effects within the Non-Hispanic and the Hispanic subsamples. The first row shows that a large number of individuals did not comply to their treatment assignment. As is expected for any voluntary treatment, a large share of individuals (around 30\% for both subsamples) decided not to take the treatment even though they were assigned to the treatment group. There are also some individuals (5\% among Non-Hispanics and 3\% among Hispanics) who attended the JCTP even though they were embargoed. Moreover, the instrument (treatment assignment) is clearly strong for both subsamples, suggesting that Assumption \ref{propensityscore} is plausible in this context. When analyzing the treatment effects and similarly to the previous literature (e.g., \cite{Schochet2008}, \cite{Flores-Lagunes2010} and \cite{Chen2015}), we find that the JCTP has a positive and significant effect on Non-Hispanics and a negative but insignificant effect on Hispanics.

\begin{table}[!htbp]
	\centering
	\caption{Preliminary Effects} \label{PreliminaryEffects}
	\begin{lrbox}{\tablebox}
		\begin{tabular}{lccclccc}
			\hline \hline
			& \multicolumn{3}{c}{Non-Hispanic Sample} &  & \multicolumn{3}{c}{Hispanic Sample} \\ \cline{2-4} \cline{6-8} 
			& Z = 1 & Z = 0 & Diff. &  & Z = 1 & Z = 0 & Diff. \\ \hline
			Ever enrolled in JCTP & .737 & .047 & .689*** &  & .747 & .028 & .719*** \\
			& & & (.008) & & & & (.016) \\
			\emph{ITT estimates} &  & &  &  &  &  & \\
			\hspace{5pt} Hours per week & 28.06 & 25.54 & 2.52*** &  & 26.63 & 27.30 & -.670 \\
			& & & (.60) & & & & (1.28) \\
			\hspace{5pt} Earnings per week & 230.24 & 194.72 & 35.52*** &  & 218.34 & 228.63 & -1.29 \\
			& & & (5.49) & & & & (12.68) \\
			\hspace{5pt} Employed & .613 & .564 & .049*** &  & .605 & .607 & -.002 \\
			& & & (.011) & & & & (.025) \\
			\emph{LATE estimates} &  & &  &  &  &  & \\
			\hspace{5pt} Hours per week &  &  & 3.66*** &  &  &  & -.930 \\
			& & & (.880) & & & & (1.78) \\
			\hspace{5pt} Earnings per week &  &  & 51.52*** &  &  &  & -14.31 \\
			& & & (8.00) & & & & (17.64) \\
			\hspace{5pt} Employed &  &  & .071*** &  &  &  & -.003 \\
			& & & (.016) & & & & (.034) \\
			
			\hline
		\end{tabular}
	\end{lrbox}
	\usebox{\tablebox}\\
	\settowidth{\tableboxwidth}{\usebox{\tablebox}} \parbox{\tableboxwidth}{\footnotesize{Note: Z indicates random treatment assignment. Outcome variables are measured at week 208 after randomization. Robust standard errors are in parenthesis. ***, ** and * denote that difference is statistically significant at the 1\%, at 5\% and 10\% level, respectively. Estimation uses design weights.}
	}
\end{table}

This last result, particularly with respect to the employment status, is important for my analysis. Similarly to \cite{Lee2009} and \cite{Chen2015}, I assume that the effect of the treatment on employment (i.e., sample selection) is monotone and positive. However, a negative effect of JCTP on employment is evidence against this assumption as discussed by \cite{Flores-Lagunes2010} and \cite{Chen2015}. For this reason, I formally test Assumption \ref{increasing_sample_selection}. To do so, I implement the procedure developed by \cite{Machado2018}, that simultaneously tests instrument exogeneity (Assumption \ref{ind}), monotonicity of treatment take-up on treatment assignment (equation \eqref{treatment}) and monotonicity of employment on the treatment (equation \eqref{selection}). Their procedure also uses this last test as a gate-keeper to test that the effect of the treatment on employment is positive (Assumption \eqref{increasing_sample_selection}).

In a more detailed way, the test proposed by \cite{Machado2018} has three steps. In the first step, the null hypothesis is that the instrument is not exogenous, or treatment take-up is not monotone on treatment assignment, or employment is not monotone on treatment take-up. As a consequence, the alternative hypothesis is that Assumption \ref{ind} and equations \eqref{treatment} and \eqref{selection} hold. In the second step, that is implemented only if the first step rejects its null hypothesis, the second null hypothesis is that the effect of the treatment on employment is non-positive. Consequently, its alternative hypothesis is that Assumptions \ref{ind} and \ref{increasing_sample_selection} and equations \eqref{treatment} and \eqref{selection} hold. Finally, in the third step, that is implemented only if the second step does not reject its null hypothesis, the third null hypothesis is that the effect of the treatment on employment is non-negative. Consequently, its alternative hypothesis is that, while Assumption \ref{ind} and equations \eqref{treatment} and \eqref{selection} are valid, Assumption \ref{increasing_sample_selection} holds in the opposite direction (see Assumption \ref{decreasingassumption}).

Table \ref{msv2018} shows the results of the test described above. Within the Non-Hispanics subsample, steps 1 and 2 reject their null hypotheses at the 1\%-significance level, implying that Assumptions \ref{ind} and \ref{increasing_sample_selection} and equations \eqref{treatment} and \eqref{selection} are plausible given the data. Consequently, it is reasonable to use Corollary \ref{MTEbounds} to bound the $MTE^{OO}$ of the JCTP on wages within the Non-Hispanics subsample. For the Hispanics subsample, step 1 rejects its null hypothesis at the 1\%-significance level, while neither step 2 nor step 3 reject their null hypotheses at the 10\%-significance level. As a consequence, Assumption \ref{ind} and equations \eqref{treatment} and \eqref{selection} are plausible given the data, but it seems that there is no effect of the treatment on employment, i.e., $S_{1} = S_{0}$ for all individuals. With no differential sample selection for the Hispanic population, point identification of the MTE of interest is trivial as discussed immediately after Proposition \ref{boundsY1Proposition}. For this reason, I focus my empirical analysis on the Non-Hispanic subsample.

\begin{table}[!htbp]
	\centering
	\caption{Testing the Identification Assumptions} \label{msv2018}
	\begin{lrbox}{\tablebox}
		\begin{tabular}{lccccccccc}
			\hline \hline
			& \multicolumn{4}{c}{Non-Hispanics Subsample} & & \multicolumn{4}{c}{Hispanics Subsample} \\
			\cline{2-5} \cline{7-10}
			& Estimated & \multicolumn{3}{c}{Critical Value} & & Estimated & \multicolumn{3}{c}{Critical Value} \\
			\cline{3-5} \cline{8-10}
			& Test Statistic & 10\% & 5\% & 1\%  & & Test Statistic & 10\% & 5\% & 1\%\\
			\hline
			Step 1 & .282 & .034 & .039 & .043 & & .308 & .044 & .047 & .050 \\
			Step 2 & .070 & .033 & .036 & .039 & & -.003 & .032 & .036 & .038 \\
			Step 3 & -.070 & .033 & .036 & .039 & & .003 & .032 & .036 & .038 \\
			\hline
		\end{tabular}
	\end{lrbox}
	\usebox{\tablebox}\\
	\settowidth{\tableboxwidth}{\usebox{\tablebox}} \parbox{\tableboxwidth}{\footnotesize{Note: The alternative hypothesis of step 1 is that Assumption \ref{ind} and equations \eqref{treatment} and \eqref{selection} are valid. The alternative hypothesis of step 2 is that Assumptions \ref{ind} and \ref{increasing_sample_selection} and equations \eqref{treatment} and \eqref{selection} are valid. The alternative hypothesis of step 3 is that Assumptions \ref{ind} and \ref{decreasingassumption} and equations \eqref{treatment} and \eqref{selection} are valid. Critical values were computed using 10,000 bootstrap repetitions and are related to the 10\%, 5\% and 1\% significance levels. Estimation uses design weights.}
	}
\end{table}

%%%%%%%%%%%%%%%%%%%%%
%% MTR and MTE on Employment and Labor Earnings 
%%%%%%%%%%%%%%%%%%%%%
\subsection{MTR and MTE on Employment and Labor Earnings: Non-Hispanics subpopulation}\label{EstimationSandY}
As a preliminary step to estimate the bounds for the $MTE^{OO}$ of the JCTP on hourly wages within the Non-Hispanic subsample, I need to estimate the MTR functions on employment and hourly labor earnings, i.e., I need to estimate the functions $m_{0}^{S}$, $m_{1}^{S}$, $m_{0}^{Y}$, and $m_{1}^{Y}$. To do so, I use the procedure described in Subsection \ref{parametric}, that adapts the method developed by \cite{Brinch2017} to a constrained framework. Specifically, I model the MTR functions of $Y$ and $S$ using Bernstein polynomials with four parameters, i.e., $M^{A}\left(u, \boldsymbol{\theta}_{d}^{A}\right) = \theta_{d, 0}^{A} \cdot \left(1 - u\right) + \theta_{d, 1}^{A} \cdot u$ for any $A \in \left\lbrace Y, S \right\rbrace$ and $d \in \left\lbrace 0, 1 \right\rbrace$ with feasible sets $\Theta^{Y} = \mathbb{R}_{+}^{4}$ and $\Theta^{S} = \left\lbrace \left(\boldsymbol{\theta}_{0}^{S}, \boldsymbol{\theta}_{1}^{S}\right) \in \left[0, 1\right]^{4} \colon \boldsymbol{\theta}_{1}^{S} \geq \boldsymbol{\theta}_{0}^{S} \right\rbrace$. To estimate $\left(\boldsymbol{\theta}_{0}^{A}, \boldsymbol{\theta}_{1}^{A}\right)$. I run the following constrained OLS model:\footnote{Appendix \ref{OLSproof} connects the OLS model \eqref{OLSmodel} to the minimization problem \eqref{pseudotrue} when the instrument is binary and there are no covariates. It also provides the explicit formula for the bounds in Corollaries \ref{MTEbounds} and \ref{boundmeandomG} using the parametric model described in Subsection \ref{parametric}. Appendix \ref{montecarlo} implements a Monte Carlo Simulation that analyzes the coverage rate of confidence intervals around the MTE bounds that are based on the OLS model \eqref{OLSmodel}.}

\begin{equation}\label{OLSmodel}
A = a_{0}^{A} \cdot \left(1 - D\right) + b_{0}^{A} \cdot \left(1 - D\right) \cdot P\left(Z\right) + a_{1}^{A} \cdot D + b_{1}^{A} \cdot D \cdot P\left(Z\right) + e,
\end{equation}
where $e$ is the error term, $\theta_{0, 0}^{A} = a_{0}^{A} - b_{0}^{A}$, $\theta_{0, 1}^{A} = a_{0}^{A} + b_{0}^{A}$, $\theta_{1, 0}^{A} = a_{1}^{A}$, $\theta_{1, 1}^{A} = a_{1}^{A} + 2 \cdot b_{1}^{A}$ and the constraints on $\left(a_{0}^{A}, b_{0}^{A}, a_{1}^{A}, b_{1}^{A}\right)$ are given by $\Theta^{A}$.

Tabel \ref{bmw2017} reports the point-estimates and 90\%-confidence intervals of the parametric models for the MTR functions on employment and hourly labor earnings. Note that the feasibility constraint $\theta_{1,0}^{S} \geq \theta_{0,0}^{S}$ is binding even though Assumption \ref{increasing_sample_selection} is plausible according to the test proposed by \cite{Machado2018}. Moreover, for the upper bound of the 90\%-confidence interval, the feasibility constraint $\theta_{1,0}^{S} \leq 1$ is also binding.

\begin{table}[!htbp]
	\centering
	\caption{Parametric MTR Functions: Non-Hispanic Subsample} \label{bmw2017}
	\begin{lrbox}{\tablebox}
		\begin{tabular}{ccccc}
			\hline \hline
			Outcome & \multicolumn{4}{c}{Parameters for any $A \in \left\lbrace Y, S \right\rbrace$} \\
			\cline{2-5}
			Variable & $\theta_{0, 0}^{A}$ & $\theta_{0, 1}^{A}$ & $\theta_{1, 0}^{A}$ & $\theta_{1, 1}^{A}$ \\
			\hline
			\multirow{2}{*}{Employment (S)} & 0.46 & 0.66 & 0.46 & 0.89 \\
			& $\left[ 0.39, 0.47 \right]$ & $\left[ 0.64, 0.71 \right]$ & $\left[ 0.39, 0.47 \right]$ & $\left[ 0.84, 1.00 \right]$ \\
			\multirow{2}{*}{Labor Earnings (Y)} & 2.96 & 5.74 & 3.00 & 8.39 \\
			& $\left[ 1.45, 3.69\right]$ & $\left[ 4.98, 6.94 \right]$ & $\left[ 2.20, 3.41 \right]$ & $\left[ 7.54, 9.81 \right]$ \\
			\hline
		\end{tabular}
	\end{lrbox}
	\usebox{\tablebox}\\
	\settowidth{\tableboxwidth}{\usebox{\tablebox}} \parbox{\tableboxwidth}{\footnotesize{Note: The MTR on Employment is given by $M^{S}\left(u, \boldsymbol{\theta}_{d}^{S}\right) = \theta_{d, 0}^{S} \cdot \left(1 - u\right) + \theta_{d, 1}^{S} \cdot u$ with feasibility set given by $\Theta^{S} = \left\lbrace \left(\boldsymbol{\theta}_{0}^{S}, \boldsymbol{\theta}_{1}^{S}\right) \in \left[0, 1\right]^{4} \colon \boldsymbol{\theta}_{1}^{S} \geq \boldsymbol{\theta}_{0}^{S} \right\rbrace$. The MTR on Labor Earnings is given by $M^{Y}\left(u, \boldsymbol{\theta}_{d}^{Y}\right) = \theta_{d, 0}^{Y} \cdot \left(1 - u\right) + \theta_{d, 1}^{Y} \cdot u$ with feasibility set given by $\Theta^{Y} = \mathbb{R}_{+}^{4}$. In brackets, I report 90\%-confidence interval based on 5,000 bootstrap repetitions. Estimation uses design weights.}
	}
\end{table}

It is easier to understand and interpret those estimates using Figure \ref{ParametricMTRMTE}. The solid lines are the point-estimates of the MTR and MTE functions based on the parameters reported in Table \ref{bmw2017}. The dotted lines are pointwise 90\%-confidence intervals around the estimated functions based on 5,000 bootstrap repetitions. Blue colored lines are associated with treated potential outcomes, while red colored lines are associated with untreated outcomes. In Subfigure \ref{MTRS}, I find that, although the employment probability for the agents who are most likely to attend the JCTP is similar between treated and untreated individuals, the employment probability for the agents who are less likely to attend the JCTP is much higher for treated individuals than for untreated ones. As a consequence, the MTE on employment within the Non-Hispanic subsample (Subfigure \ref{MTES}) is increasing in the latent heterogeneity. Similarly, in Subfigure \ref{MTRY}, I find that, although expected hourly labor earnings for the agents who are most likely to attend the JCTP is similar between treated and untreated individuals, expected hourly labor earnings for the agents who are less likely to attend the JCTP is much higher for treated individuals than for untreated ones. As a consequence, the MTE on hourly labor earnings within the Non-Hispanic subsample (Subfigure \ref{MTEY}) is increasing in the latent heterogeneity. I highlight that the shape of my estimated MTE functions are in line with the results by \cite{Chen2017}, whose estimated upper bounds also suggest that the ATE on those variables is greater than the ATT.

\begin{figure}[!htbp] 
	\caption{Parametric MTR and MTE Functions: Non-Hispanic subsample} \label{ParametricMTRMTE}
	
	\begin{center}
		\subfloat[MTR on Employment\label{MTRS}]{\includegraphics[width = .45 \columnwidth]{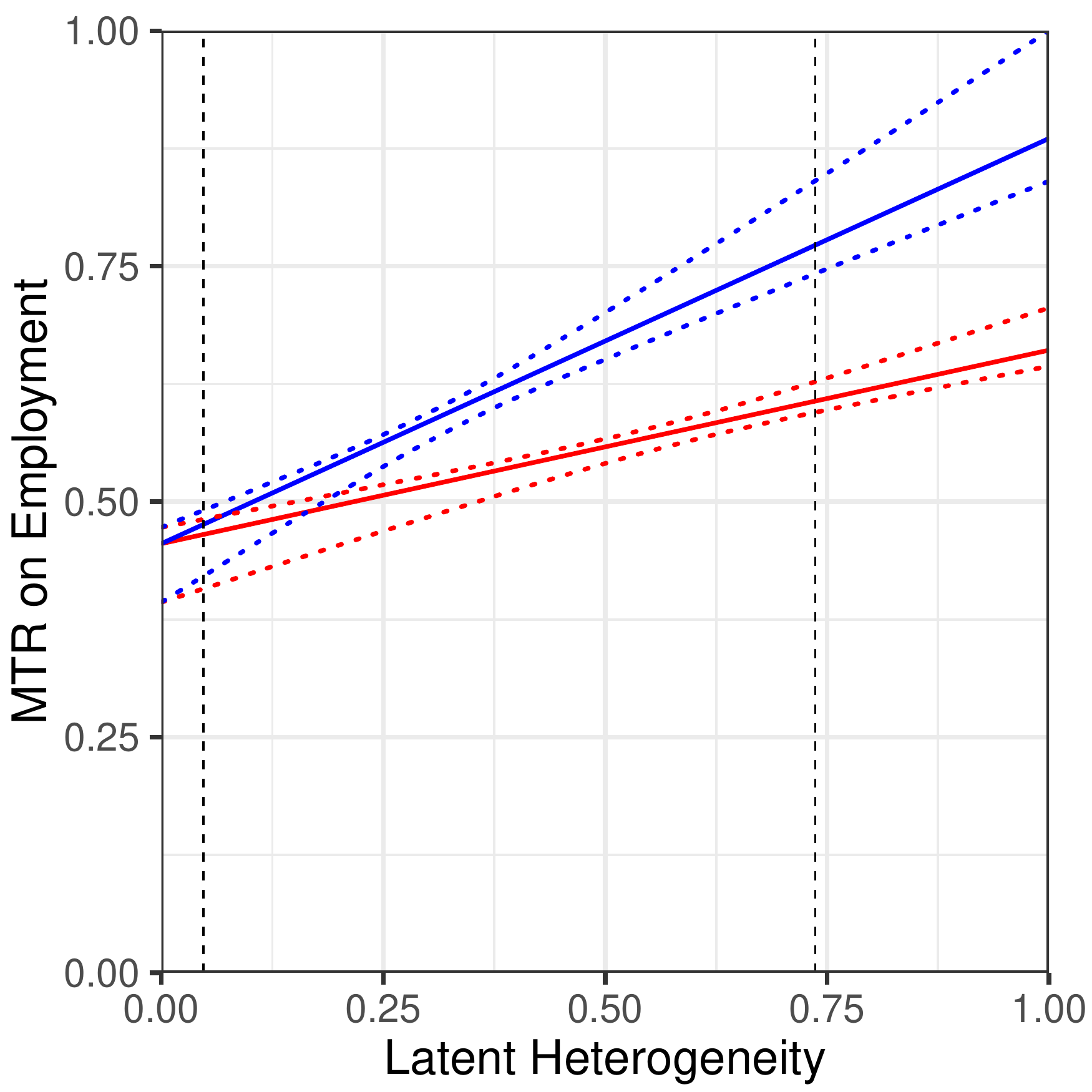}} \quad	
		\subfloat[MTE on Employment\label{MTES}]{\includegraphics[width = .45 \columnwidth]{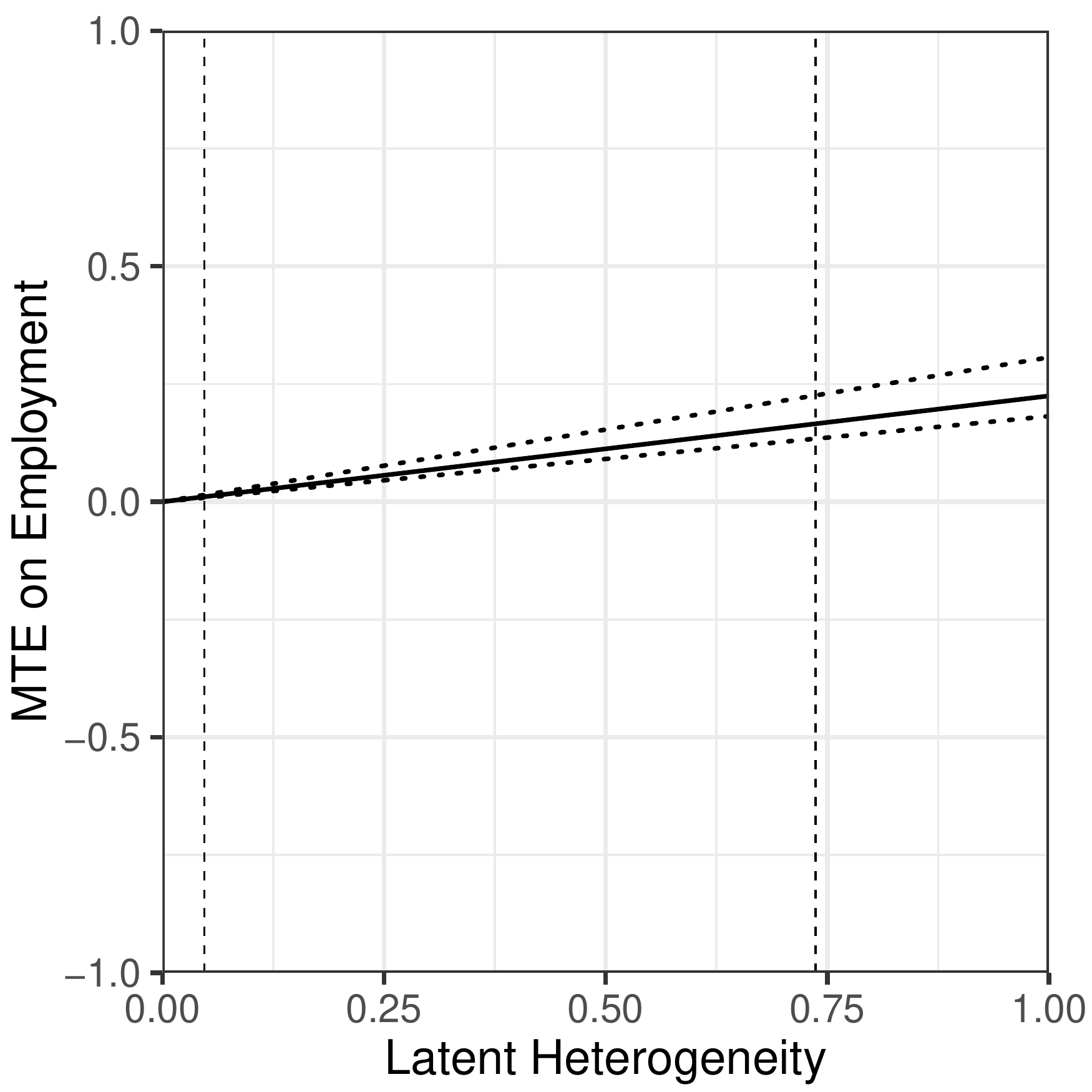}} \\
		\subfloat[MTR on Labor Earnings\label{MTRY}]{\includegraphics[width = .45 \columnwidth]{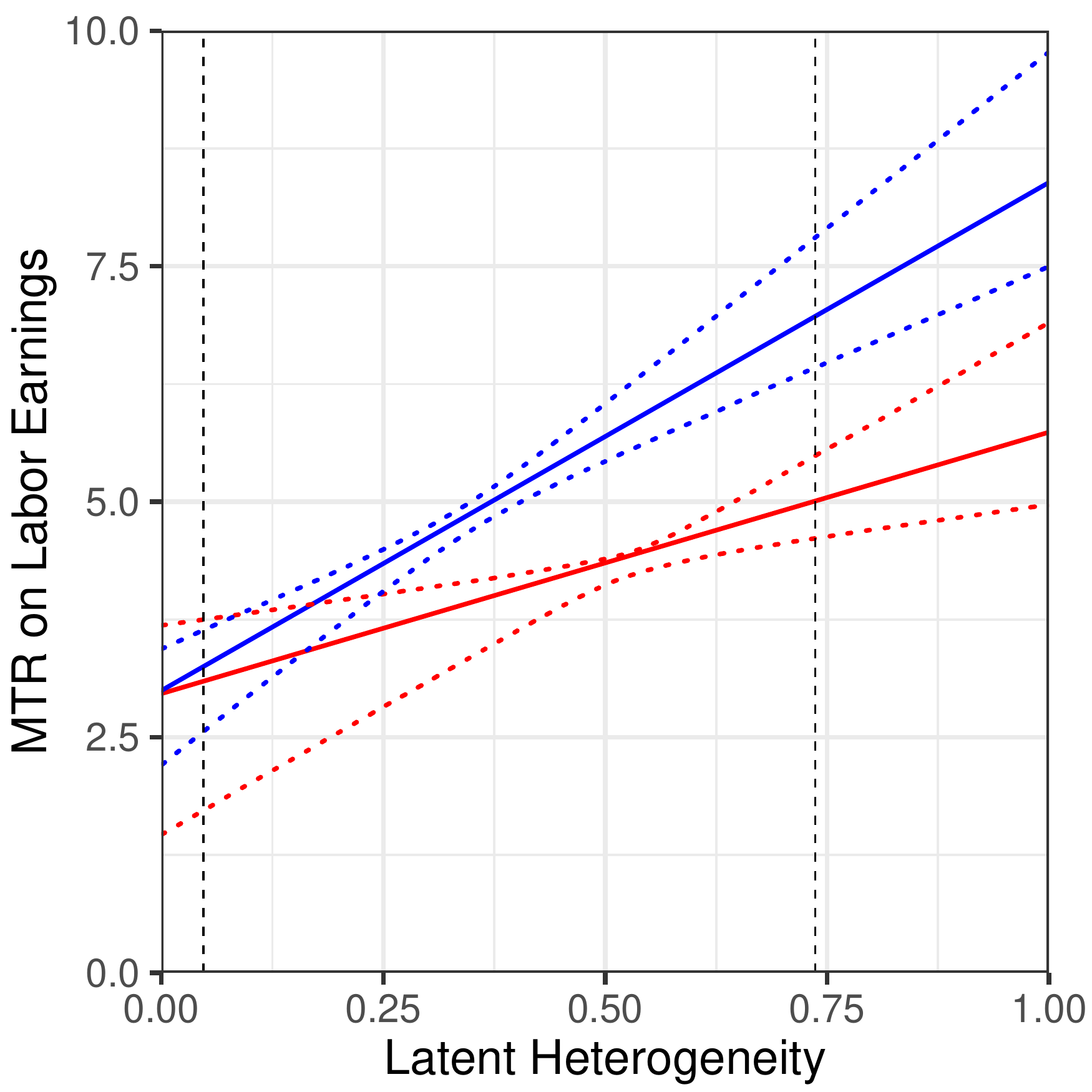}} \quad	
		\subfloat[MTE on Labor Earnings\label{MTEY}]{\includegraphics[width = .45 \columnwidth]{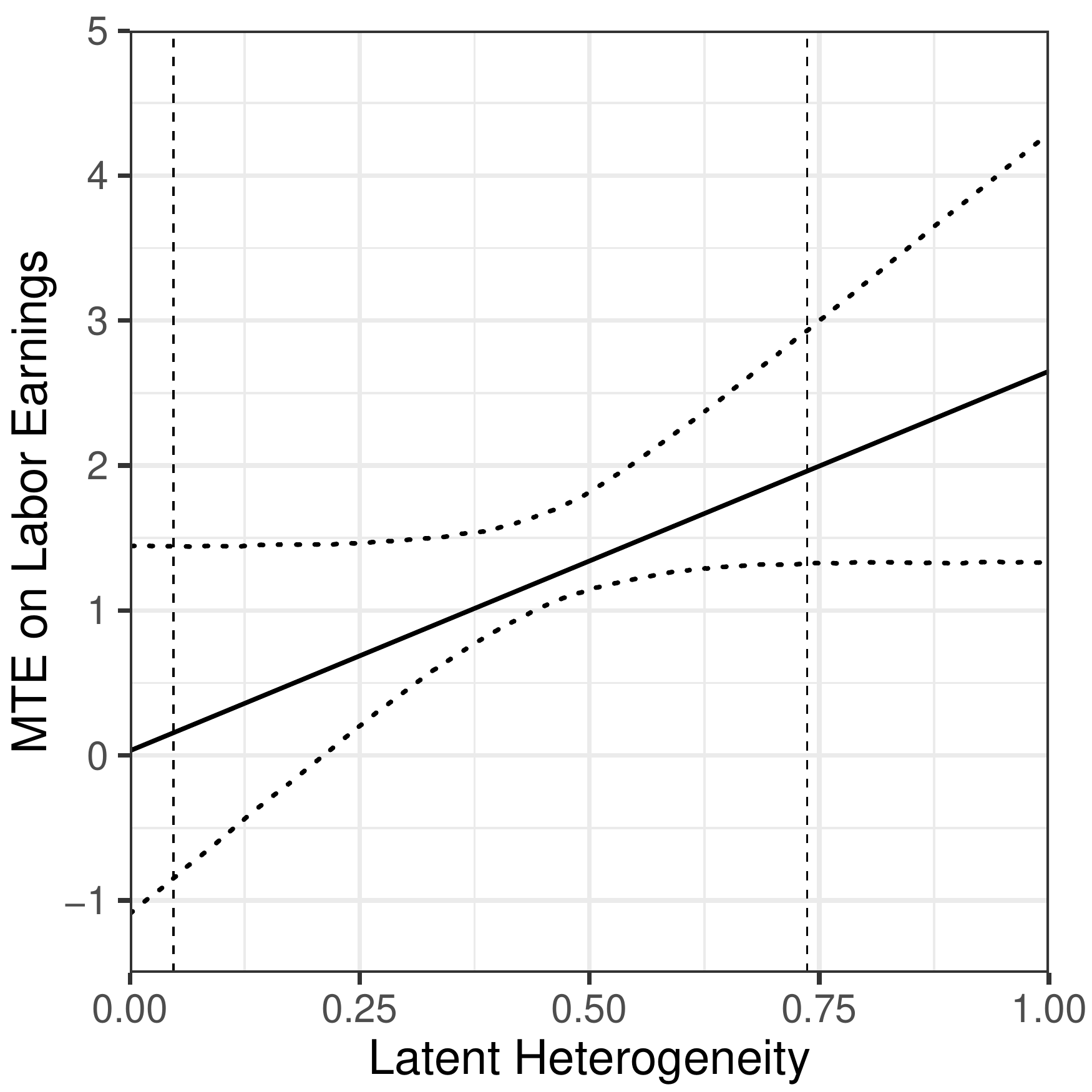}} \\
	\end{center}
	\footnotesize{Notes: The solid lines are the point-estimates of the MTR and MTE functions based on the parameters reported in Table \ref{bmw2017}. The dotted lines are pointwise 90\%-confidence intervals around the estimated functions based on 5,000 bootstrap repetitions. Blue colored lines are associated with treated potential outcomes, while red colored lines are associated with untreated outcomes. The vertical dashed lines represent the sample values of the propensity score $P\left[\left. D = 1 \right \vert Z\right]$. Estimation uses design weights.}
\end{figure}

%%%%%%%%%%%%%%%%%%%%%
%% bounds for the MTE on Wages 
%%%%%%%%%%%%%%%%%%%%%
\subsection{Bounds for the $\mathbf{MTE^{OO}}$ on Wages: Non-Hispanic subpopulation}\label{empiricalresults}
To partially identify the $MTE^{OO}$ of the JCTP on wages within the Non-Hispanic subsample, I can combine the functions estimated in Subsection \ref{EstimationSandY} with Corollaries \ref{MTEbounds} and \ref{boundmeandomG}. While the first corollary imposes only assumptions that are valid by the experimental design (Assumption \ref{ind}), technical (Assumptions \ref{finite}-\ref{bounded}) or testable (Assumptions \ref{propensityscore} and \ref{increasing_sample_selection}, and equation \ref{treatment}), Corollary \ref{boundmeandomG} additionally uses the Mean Dominance Assumption \ref{meandominanceG}. This last assumption imposes that the marginal treatment response function of wages when treated for the always-employed population is greater than the same object for the employed-only-when-treated population, implying a positive correlation between potential employment and potential wages, which is supported by standard models of labor supply.\footnote{\cite{Chen2015} discuss the connection between the Mean Dominance Assumption \ref{meandominanceG} and the Labor Economics literature in a deeper way.}.

Another issue when estimating bounds for a parameter of interest is that there are two ways to construct confidence intervals. The conservative method finds the $\zeta$-confidence intervals around the upper and lower $MTE^{OO}$ bounds and then uses their upper most and lower most bounds to construct a confidence interval that contains the identified region with probability $\zeta$. Since the parameter of interest has to be inside the identified region, this confidence interval contain the parameter of interest with probability at least $\zeta$. An alternative method is proposed by \cite{Imbens2004}, who directly construct a $\zeta$-confidence interval that contains the parameter of interest. Since they take into account that the parameter of interest has to be inside the identified region by construction, their confidence interval is tighter than the conservative method.

Figure \ref{ParametricMTEW} shows the parametric bounds of the $MTE^{OO}$ on wages using Corollary \ref{MTEbounds} (Subfigure \ref{MTEW}) and using Corollary \ref{boundmeandomG} (Subfigure \ref{MTEW-MD}). The solid lines are the point-estimates of the parametric bounds of the MTE on wages, while the dotted lines are pointwise conservative 90\%-confidence intervals around the identified region based on 5,000 bootstrap repetitions and the dashed lines are pointwise 90\%-confidence intervals of the parameter of interest (\cite{Imbens2004}) based also on 5,000 bootstraps repetitions.

As a way to understand the magnitude of the effects, I compare the estimated $MTE^{OO}$ bounds against the average observed hourly wage of the Non-Hispanics assigned to the control group, \$7.72. Note that the lower bounds that do not use the mean dominance assumption (Subfigure \ref{MTEW}) are implausibly negative. Even for the agents who are the most likely to attend the JCTP, the lower bound of the $MTE^{OO}$ on wages (-\$6.51) imply that the JCTP would drive their hourly wages almost to zero. This implausibly negative lower bound is based on the worst-case scenario that unrealistically imposes that the treated potential wage for the always-employed subpopulation is equal to zero.

By imposing the Mean Dominance Assumption \ref{meandominanceG}, I rule out this extreme case by assuming that there is positive selection into employment. As a consequence, I can increase the lower bound from equation \eqref{mte_lb_lower} to equation \eqref{mte_lb_md}, narrowing the bounds of the $MTE^{OO}$ on wages (Subfigure \ref{MTEW-MD}). Under this extra assumption, the $MTE^{OO}$ on wages is significant at the 10\%-confidence level for latent heterogeneity values between 0.34 and 0.68 when I use the conservative confidence interval and between 0.35 and 0.73 when I use the confidence interval based on \cite{Imbens2004}. Most interestingly, the point-estimate of the lower bound of the $MTE^{OO}$ on wages is decreasing in the likelihood of attending the JCTP.

\begin{figure}[!htbp] 
	\caption{Parametric Bounds of the $MTE^{OO}$ on Wages: Non-Hispanic subsample} \label{ParametricMTEW}
	
	\begin{center}
		\subfloat[Without Mean Dominance Assumption\label{MTEW}]{\includegraphics[width = .45 \columnwidth]{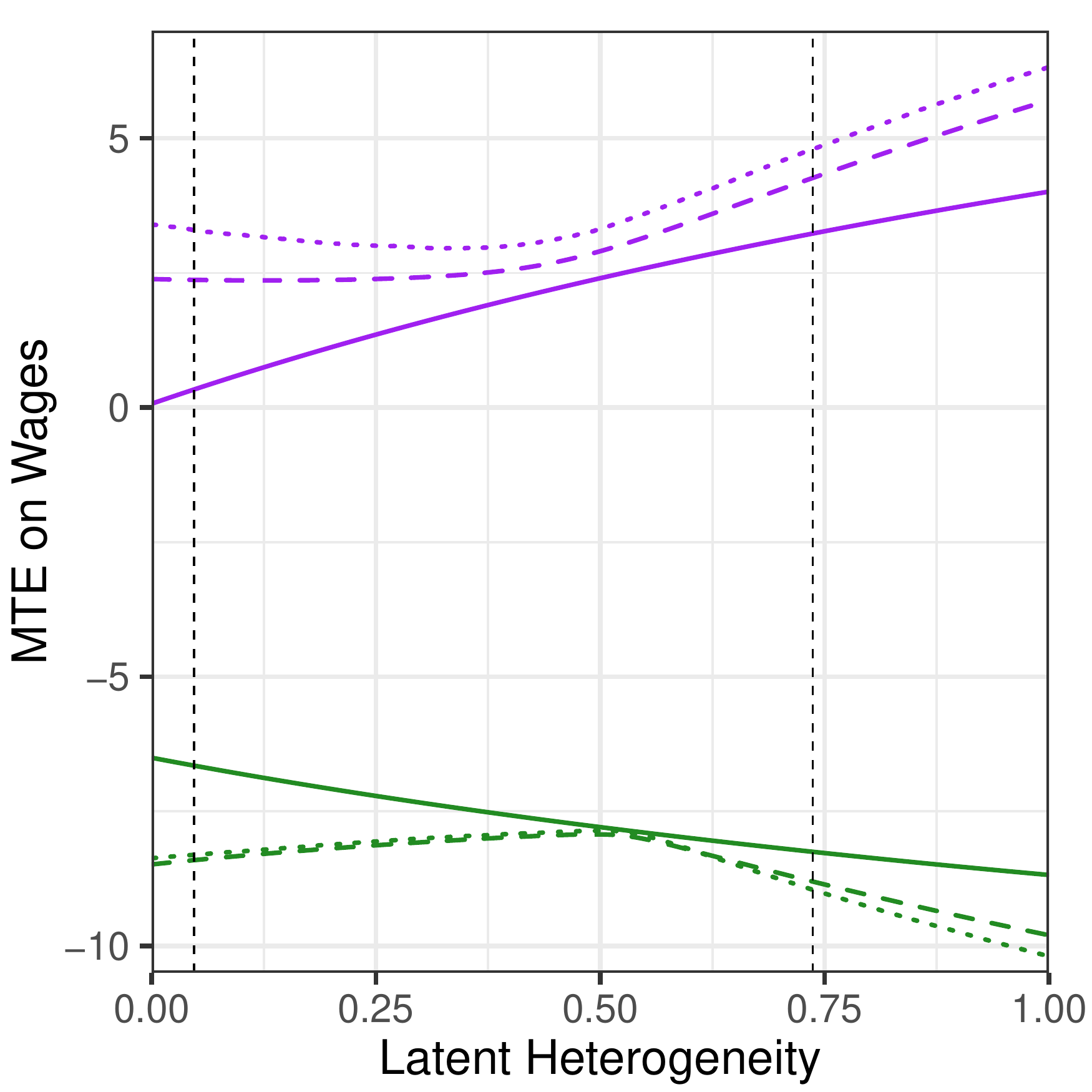}} \quad	
		\subfloat[With Mean Dominance Assumption\label{MTEW-MD}]{\includegraphics[width = .45 \columnwidth]{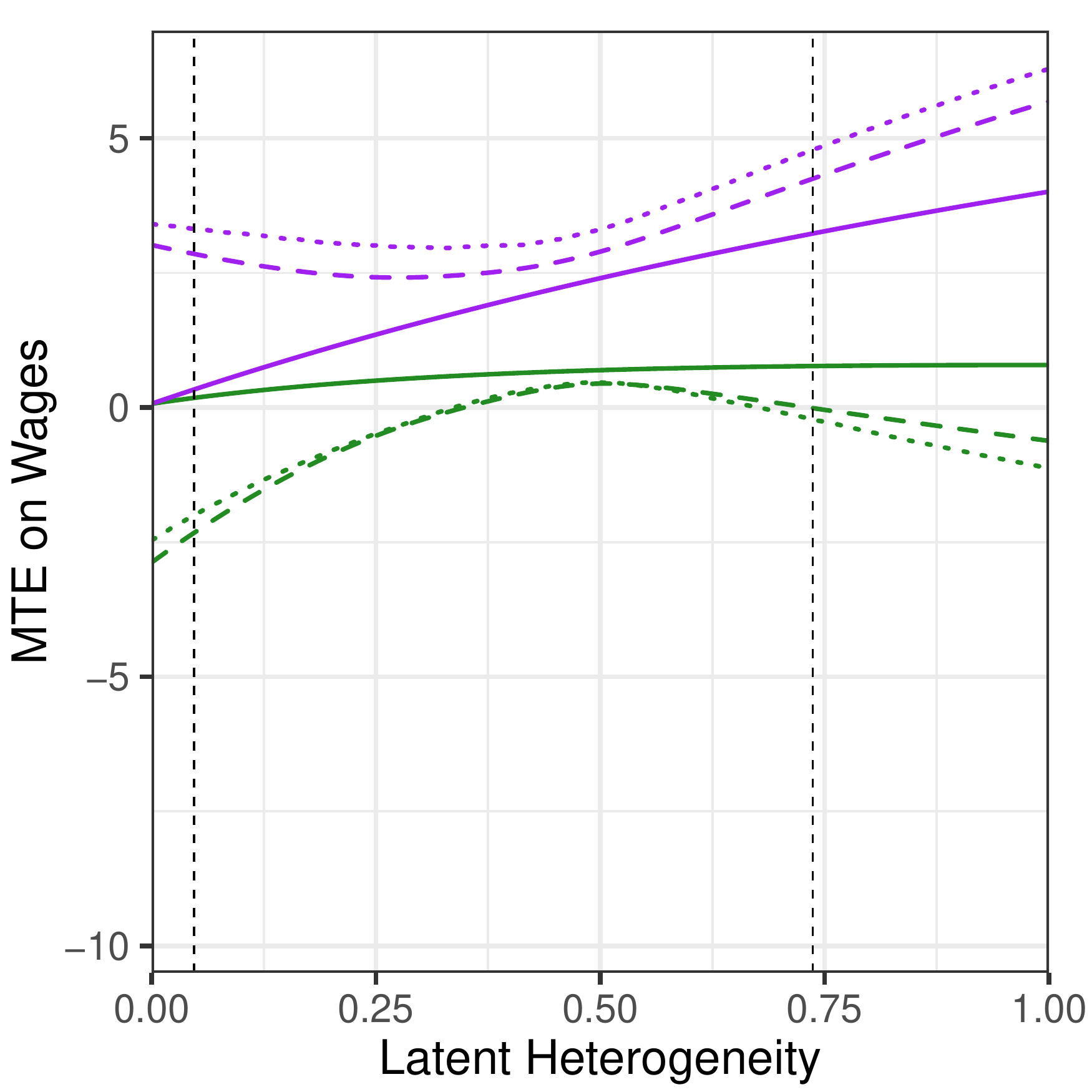}} \\
	\end{center}
	\footnotesize{Notes: The solid lines are the point-estimates of the parametric bounds of the $MTE^{OO}$ on wages. The dotted lines are pointwise conservative 90\%-confidence intervals around the identified region based on 5,000 bootstrap repetitions. The dashed lines are pointwise 90\%-confidence intervals of the parameter of interest (\cite{Imbens2004}) based on 5,000 bootstrap repetitions. The vertical dashed lines represent the sample values of the propensity score $P\left[\left. D = 1 \right \vert Z\right]$. Estimation uses design weights.}
\end{figure}

To better understand the magnitude of those effects and compare my results with the previous literature, I summarize the bounds for the $MTE^{OO}$ function using four key parameters --- $ATE^{OO}$, $ATT^{OO}$, $ATU^{OO}$ and $LATE^{OO}$ --- that are described in Tables \ref{integral} and \ref{weights} as integrals of the $MTE^{OO}$ function. Table \ref{ATEW} reports those bounds in brackets, the 90\%-conservative confidence intervals of the identified region in parenthesis and the 90\%-confidence intervals of the parameter of interest (\cite{Imbens2004}) in braces. As expected, the bounds without the mean dominance assumption are wide and uninformative, while, when imposing Assumption \ref{meandominanceG}, all parameters but the $ATT^{OO}$ are significant at 10\% according to both types of confidence intervals.

\begin{table}[!htbp]
	\centering
	\caption{Bounds of the $ATE^{OO}$, $ATT^{OO}$, $ATU^{OO}$ and $LATE^{OO}$ on Wages: Non-Hispanic subsample} \label{ATEW}
	\begin{lrbox}{\tablebox}
		\begin{tabular}{ccccc}
			\hline \hline
			Mean Dominance& \multicolumn{4}{c}{Treatment Effect} \\
			\cline{2-5}
			Asssumption \ref{meandominanceG} & $ATE^{OO}$ & $ATT^{OO}$ & $ATU^{OO}$ & $LATE^{OO}$ \\
			\hline
			\multirow{3}{*}{NO} & $\left[ -7.73, 2.28 \right]$ & $\left[ -7.11, 1.17 \right]$ & $\left[ -8.20, 3.14 \right]$ & $\left[ -7.52, 1.91 \right]$ \\
			& $\left( -7.88, 3.15 \right)$ & $\left( -8.16, 3.09 \right)$ & $\left( -8.54, 4.29 \right)$ & $\left( -7.94, 2.97 \right)$ \\
			& $\left\lbrace -7.95, 2.75 \right\rbrace$ & $\left\lbrace -8.35, 2.57 \right\rbrace$ & $\left\lbrace -8.57, 3.96 \right\rbrace$ & $\left\lbrace -8.01, 2.51 \right\rbrace$ \\
			&&&&\\
			\multirow{3}{*}{YES} & $\left[ 0.61, 2.28 \right]$ & $\left[ 0.33, 0.99 \right]$ & $\left[ 0.71, 3.00 \right]$ & $\left[ 0.58, 1.91 \right]$ \\
			& $\left( 0.38, 3.14 \right)$ & $\left( -1.42, 3.18 \right)$ & $\left( 0.18, 3.69 \right)$ & $\left( 0.12, 3.00 \right)$ \\
			& $\left\lbrace 0.35, 2.75 \right\rbrace$ & $\left\lbrace -1.43, 2.76 \right\rbrace$ & $\left\lbrace 0.27, 3.69 \right\rbrace$ & $\left\lbrace 0.07, 2.51 \right\rbrace$ \\
			\hline
		\end{tabular}
	\end{lrbox}
	\usebox{\tablebox}\\
	\settowidth{\tableboxwidth}{\usebox{\tablebox}} \parbox{\tableboxwidth}{\footnotesize{Note: In brackets, I report the bounds for the parameter of interest that are integrals of the bounds for the $MTE^{OO}$ function. In parenthesis, I report conservative 90\%-confidence intervals around the identified region based on 5,000 bootstrap repetitions, while, in braces, I report 90\%-confidence intervals of the parameter of interest (\cite{Imbens2004}) based on 5,000 bootstrap repetitions. Estimation uses design weights.}
	}
\end{table}

I stress that my $LATE^{OO}$ estimates represent an effect between 7.51\% and 24.74\% of the average observed hourly wage of the Non-Hispanics assigned to the control group, which are comparable to the bounds of the $LATE^{OO}$ parameter derived by \cite{Chen2015} --- approximately between 7.7\% and 17.5\% under a similar set of assumptions. The finding that their bounds are tighter than mine for the $LATE^{OO}$ is not surprising because their method leverages all the available information to specifically identify the $LATE^{OO}$ while my tool bounds the $MTE^{OO}$ function and then flexibly bounds the other treatment effects for the always-employed population.

As a consequence of this flexibility, I can partially identify other treatment effects that may be policy-relevant. For example, the $ATE^{OO}$ is bounded between 7.90\% and 29.53\% of the average observed hourly wage of the Non-Hispanics assigned to the control group. Most interestingly, the $ATT^{OO}$ and the $ATU^{OO}$ are, respectively, bounded between 4.27\% and 12.82\%, and 9.20\% and 38.86\%, suggesting that the agents who do not attend the JCTP might be the ones who would benefit the most from it. This result is even stronger when we analyze the confidence intervals around the $ATT^{OO}$ and the $ATU^{OO}$: while the first treatment effect is not significantly different from zero, the second parameter is significantly different from zero. To conclude, I highlight that, even though the upper bound of the treatment effects on wages may be unrealistically large, the magnitude of the lower bounds are similar to the results found by \cite{Chen2017} and are reasonable when compared to ITT effects of 16.70\% on earnings per week and of 9.87\% on hours per week that are shown in Table \ref{PreliminaryEffects}.

%%%%%%%%%%%%%%%%%%%%%%%%%%%%%%%%%%%%%%%%%%%%%%%%%%%%%
% Conclusion
%%%%%%%%%%%%%%%%%%%%%%%%%%%%%%%%%%%%%%%%%%%%%%%%%%%%%

\section{Conclusion}\label{furtherwork}
My main theoretical contribution provides pointwise sharp bounds for the MTE of interest within the always-observed subpopulation by imposing a monotonicity assumption that the treatment has a positive impact on sample selection for every agent. Those bounds are tightened by imposing an extra mean dominance assumption that the potential outcome when treated within the always-observed subpopulation is greater than or equal to the same parameter within the observed-only-when-treated subpopulation. Both bounds can be estimated using the LIV approach if the instrument is continuous, using a non-parametric outer set based on the method developed by \cite{Mogstad2017}, or using a parametric model based on the strategy proposed by \cite{Brinch2017}. Such bounds are useful to analyze many empirical problems that include endogenous self-selection into treatment and sample selection.

My main empirical findings suggest that the marginal treatment effect of the Job Corps Training Program (JCTP) on employment, hourly labor earnings and hourly wages increases with the latent heterogeneity variable within the Non-Hispanic group. More specifically, while MTEs for the agents who are the most likely to attend the JCTP are very small, the MTEs for the agents who are the least likely to attend the JCTP are considerably large. Economically, this result implies that the agents who are more likely to benefit from the JCTP are not attending it due to some unobserved constraint. A similar result is found by \cite{Chen2017}, whose empirical evidence suggests that the effects of the JCTP on employment and labor earnings for never-takers are significantly positive. They argue that those agents are not enrolling at the JCTP due to family constraints (lack of childcare services), incomplete information on JCTP's benefits, overconfidence or personal preferences for non-enrollment. A more complete analysis of why agents who would benefit from attending the JCTP are not doing so is beyond the scope of this paper, but is an important question for future research because it may help policy makers to better target the JCTP to the population who would benefit the most from this program.

%\begin{enumerate}		
%	\item Extension (multi-period treatment): D is taking the treatment in period 1, S is taking the treatment in period 2 and can be influence by the endogenous choice of D. Potential outcomes depends on D and S, i.e., $Y_{d, s}$.
%
%	\item Should I have an Monte Carlo Exercise in the appendix? I can also get an application without sample selection (but with imperfect compliance) and impose a known sample selection mechanism and try to recover the true treatment effect.
%
%\end{enumerate}

%----------------------------------------------------------------------------------------
%	BIBLIOGRAPHY
%----------------------------------------------------------------------------------------

\singlespace

\renewcommand{\refname}{References}
%For modifying the bibliography heading

\bibliographystyle{agsm}
\bibliography{references_mte_sample_selection}
%The file containing the bibliography

%%%%%%%%%%%%%%%%%%%%%%%%%%%%%%%%%%%%%%%%%%%%%%%%%%%%%%%%%%%%%%%%%%%%
%%%%%%%%%%%%%%%%%%%%%%%% APPENDIX %%%%%%%%%%%%%%%%%%%%%%%%%%%%%%%%%%
%%%%%%%%%%%%%%%%%%%%%%%%%%%%%%%%%%%%%%%%%%%%%%%%%%%%%%%%%%%%%%%%%%%%

\pagebreak

\newpage

\pagebreak

\setcounter{table}{0}
\renewcommand\thetable{A.\arabic{table}}

\setcounter{figure}{0}
\renewcommand\thefigure{A.\arabic{figure}}

\setcounter{equation}{0}
\renewcommand\theequation{A.\arabic{equation}}

\appendix

\begin{center}
	\huge
	Supporting Information
	
	(Online Appendix)
	
\end{center}

\doublespacing
\normalsize

%%%%%%%%%%%%%%%%%%%%%%%%%%%%%%%%%%%%%%%%%%%%%%%%%%%%%%%%%%%%%%%
% Proofs
%%%%%%%%%%%%%%%%%%%%%%%%%%%%%%%%%%%%%%%%%%%%%%%%%%%%%%%%%%%%%%%
\section{Proofs of the main results} \label{proofs}

\setcounter{table}{0}
\renewcommand\thetable{A.\arabic{table}}

\setcounter{figure}{0}
\renewcommand\thefigure{A.\arabic{figure}}

\setcounter{equation}{0}
\renewcommand\theequation{A.\arabic{equation}}

\subsection{Proof of Equation \eqref{m0YstarOO}}\label{proofUntreated}

Note that
\begin{align*}
\mathbb{E}\left[Y_{0}^{*} \left\vert X = x, U = u, S_{0} = 1, S_{1} = 1 \right.\right] & = \mathbb{E}\left[Y_{0}^{*} \left\vert X = x, U = u, S_{0} = 1 \right.\right] \\
& \hspace{20pt} \text{by Assumption } \ref{increasing_sample_selection} \\
& = \dfrac{\mathbb{E}\left[S_{0} \cdot Y_{0}^{*} \left\vert X = x, U = u \right.\right]}{\mathbb{P}\left[S_{0} = 1 \left\vert X = x, U = u \right.\right]} \\
& \hspace{20pt} \text{by the definition of conditional expectation } \\
& = \dfrac{\mathbb{E}\left[Y_{0} \left\vert X = x, U = u \right.\right]}{\mathbb{E}\left[S_{0} \left\vert X = x, U = u \right.\right]} \\
& = \dfrac{m_{0}^{Y}\left(x, u\right)}{m_{0}^{S}\left(x, u\right)}. & \blacksquare
\end{align*}

\subsection{Proof of Equation \eqref{m1YstarOO}}\label{proofTreated}

First, observe that
\begin{align}
	m_{0}^{S}\left(x, u\right) & \coloneqq \mathbb{E}\left[S_{0} \left\vert X = x, U = u \right.\right] \nonumber \\
	& \label{m0S} = \mathbb{P}\left[Q\left(0, X\right) \geq V \left\vert X = x, U = u \right.\right] \\
	& \hspace{20pt} \text{by equation } \eqref{selection}, \nonumber \\
	m_{1}^{S}\left(x, u\right) & \coloneqq \mathbb{E}\left[S_{1} \left\vert X = x, U = u \right.\right] \nonumber \\
	& \label{m1S} = \mathbb{P}\left[Q\left(1, X\right) \geq V \left\vert X = x, U = u \right.\right] \\
	& \hspace{20pt} \text{by equation } \eqref{selection}, \nonumber \\
	\Delta_{S}\left(x, u\right) & \coloneqq \mathbb{E}\left[S_{1} - S_{0} \left\vert X = x, U = u \right.\right] \nonumber \\
	& = m_{1}^{S}\left(x, u\right) - m_{0}^{S}\left(x, u\right) \nonumber\\
	&  = \mathbb{P}\left[Q\left(1, X\right) \geq V > Q\left(0, X\right) \left\vert X = x, U = u \right.\right] \nonumber \\
	& \hspace{20pt} \text{by equations } \eqref{m0S} \text{ and } \eqref{m1S} \text{ and Assumption } \eqref{increasing_sample_selection} \nonumber \\
	& \label{deltaS} = \mathbb{P}\left[S_{0} = 0, S_{1} = 1 \left\vert X = x, U = u \right.\right] \\
	& \hspace{20pt} \text{by equation } \eqref{selection} \text{, and} \nonumber \\
	\Delta_{Y}^{NO}\left(x, u\right) & \coloneqq \mathbb{E}\left[ Y_{1} - Y_{0} \left\vert X = x, U = u, S_{0} = 0, S_{1} = 1 \right.\right] \nonumber \\
	& = \mathbb{E}\left[S_{1} \cdot Y_{1}^{*} - S_{0} \cdot Y_{0}^{*} \left\vert X = x, U = u, S_{0} = 0, S_{1} = 1 \right.\right] \nonumber \\
	& \label{DeltaY_NO} = \mathbb{E}\left[Y_{1}^{*} \left\vert X = x, U = u, S_{0} = 0, S_{1} = 1 \right.\right].
\end{align}

Note also that:
\begin{align}
m_{1}^{Y}\left(x, u\right) & \coloneqq \mathbb{E}\left[Y_{1} \left\vert X = x, U = u \right. \right] \nonumber \\
& = \mathbb{E}\left[S_{1} \cdot Y_{1}^{*} \left\vert X = x, U = u \right. \right] \nonumber \\
& = \mathbb{E}\left[Y_{1}^{*} \left\vert X = x, U = u, S_{0} = 1, S_{1} = 1 \right. \right] \cdot \mathbb{P}\left[S_{0} = 1 \left\vert X = x, U = u \right.\right] \nonumber \\
& \hspace{20pt} + \mathbb{E}\left[Y_{1}^{*} \left\vert X = x, U = u, S_{0} = 0,  S_{1} = 1 \right. \right] \cdot \mathbb{P}\left[S_{0} = 0, S_{1} = 1 \left\vert X = x, U = u \right.\right] \nonumber \\
& \text{by Assumption } \ref{increasing_sample_selection} \text{ and the Law of Iterated Expectations} \nonumber \\
& \label{decomposition} = \mathbb{E}\left[Y_{1}^{*} \left\vert X = x, U = u, S_{0} = 1, S_{1} = 1 \right. \right] \cdot m_{0}^{S}\left(x, u\right) + \Delta_{Y}^{NO}\left(x, u\right) \cdot \Delta_{S}\left(x, u\right) \\
& \text{by equations } \eqref{m0S}\text{, } \eqref{deltaS} \text{ and } \eqref{DeltaY_NO} \nonumber,
\end{align}
implying equation \eqref{m1YstarOO} after some rearrangement. $\blacksquare$

%%%%%%%%%%%%%%%%%%%%%%%%%%%%%%%%%%%%%%%%%%%
% Proof: Bounds Y1
%%%%%%%%%%%%%%%%%%%%%%%%%%%%%%%%%%%%%%%%%%%
\subsection{Proof of Proposition \ref{boundsY1Proposition}}\label{proofboundsY1}
Note that
\begin{equation}\label{naturalbounds}
\underline{y}^{*} \leq \mathbb{E}\left[Y_{1}^{*} \left\vert X = x, U = u, S_{0} = 1, S_{1} = 1 \right.\right] \leq \overline{y}^{*}
\end{equation}
by the definition of $\underline{y}^{*}$ and $\overline{y}^{*}$. Observe also that \begin{equation*}
\underline{y}^{*} \leq \Delta_{Y}^{NO}\left(x, u\right) \leq \overline{y}^{*}
\end{equation*}
by equation \eqref{DeltaY_NO} and the definition of $\underline{y}^{*}$ and $\overline{y}^{*}$, implying, by equation \eqref{m1YstarOO}, that
\begin{equation}\label{app_boundsY1lower}
\mathbb{E}\left[Y_{1}^{*} \left\vert X = x, U = u, S_{0} = 1, S_{1} = 1 \right.\right] \leq \dfrac{m_{1}^{Y}\left(x, u\right) - \underline{y}^{*} \cdot \Delta_{S}\left(x, u\right)}{m_{0}^{S}\left(x, u\right)}
\end{equation}
under assumption \ref{bounded}.1,
\begin{equation}\label{app_boundsY1upper}
\dfrac{m_{1}^{Y}\left(x, u\right) - \overline{y}^{*} \cdot \Delta_{S}\left(x, u\right)}{m_{0}^{S}\left(x, u\right)} \leq \mathbb{E}\left[Y_{1}^{*} \left\vert X = x, U = u, S_{0} = 1, S_{1} = 1 \right.\right]
\end{equation}
under assumption \ref{bounded}.2, and
\begin{align}
\dfrac{m_{1}^{Y}\left(x, u\right) - \overline{y}^{*} \cdot \Delta_{S}\left(x, u\right)}{m_{0}^{S}\left(x, u\right)} & \leq \mathbb{E}\left[Y_{1}^{*} \left\vert X = x, U = u, S_{0} = 1, S_{1} = 1 \right.\right] \nonumber \\
& \label{app_boundsY1} \leq \dfrac{m_{1}^{Y}\left(x, u\right) - \underline{y}^{*} \cdot \Delta_{S}\left(x, u\right)}{m_{0}^{S}\left(x, u\right)}.
\end{align}
under Assumption \ref{bounded}.3 (sub-case (a) or (b)). Combining equations \eqref{naturalbounds}-\eqref{app_boundsY1}, it is easy to show that Proposition \ref{boundsY1Proposition} holds. $\blacksquare$

%%%%%%%%%%%%%%%%%%%%%%%%%%%%%%%%%%%%%%%%%%%
% Proof: Sharp Bounds
%%%%%%%%%%%%%%%%%%%%%%%%%%%%%%%%%%%%%%%%%%%
\subsection{Proof of Theorem \ref{sharpbounds}}\label{proofsharp}
First, I prove Theorem \ref{sharpbounds} under Assumption \ref{bounded}.3 (sub-cases (a) and (b)). At the end of this section, I prove Theorem \ref{sharpbounds} under assumptions \ref{bounded}.1 and \ref{bounded}.2.

\subsubsection{Proof under Assumption \ref{bounded}.3 (sub-cases (a) and (b))}\label{proofsharp3}

Fix $\overline{u} \in \left[0, 1\right]$, $\overline{x} \in \mathcal{X}$ and $\delta\left(\overline{x}, \overline{u}\right) \in \left(\underline{\Delta_{Y^{*}}^{OO}}\left(\overline{x}, \overline{u}\right), \overline{\Delta_{Y^{*}}^{OO}}\left(\overline{x}, \overline{u}\right)\right)$ arbitrarily. For brevity, define $\alpha\left(\overline{x}, \overline{u}\right) \coloneqq \delta\left(\overline{x}, \overline{u}\right) + \dfrac{m_{0}^{Y}\left(\overline{x}, \overline{u}\right)}{m_{0}^{S}\left(\overline{x}, \overline{u}\right)}$ and $\gamma\left(\overline{x}, \overline{u}\right) \coloneqq \dfrac{m_{1}^{Y}\left(\overline{x}, \overline{u}\right) - \alpha\left(\overline{x}, \overline{u}\right) \cdot m_{0}^{S}\left(\overline{x}, \overline{u}\right)}{\Delta_{S}\left(\overline{x}, \overline{u}\right)}$.

Note that
\begin{equation}\label{sanity1}
\begin{array}{cll}
& \delta\left(\overline{x}, \overline{u}\right) & \in \left(\underline{\Delta_{Y^{*}}^{OO}}\left(\overline{x}, \overline{u}\right), \overline{\Delta_{Y^{*}}^{OO}}\left(\overline{x}, \overline{u}\right)\right) \\
& & \\
\Leftrightarrow & \alpha\left(\overline{x}, \overline{u}\right) & \in \left(\max\left\lbrace \dfrac{m_{1}^{Y}\left(x, u\right) - \overline{y}^{*} \cdot \Delta_{S}\left(x, u\right)}{m_{0}^{S}\left(x, u\right)}, \underline{y}^{*}\right\rbrace, \right. \\
& & \\
& & \hspace{30pt} \left. \min\left\lbrace \dfrac{m_{1}^{Y}\left(x, u\right) - \underline{y}^{*} \cdot \Delta_{S}\left(x, u\right)}{m_{0}^{S}\left(x, u\right)}, \overline{y}^{*}\right\rbrace\right) \\
& & \\
& & \subseteq \left(\underline{y}^{*}, \overline{y}^{*}\right),
\end{array}
\end{equation}
and that
\begin{equation}\label{sanity2}
\begin{array}{cll}
& \alpha\left(\overline{x}, \overline{u}\right) & \in \left(\dfrac{m_{1}^{Y}\left(x, u\right) - \overline{y}^{*} \cdot \Delta_{S}\left(x, u\right)}{m_{0}^{S}\left(x, u\right)}, \dfrac{m_{1}^{Y}\left(x, u\right) - \underline{y}^{*} \cdot \Delta_{S}\left(x, u\right)}{m_{0}^{S}\left(x, u\right)} \right) \\
& & \\
\Leftrightarrow & \gamma\left(\overline{x}, \overline{u}\right) & \in \left(\underline{y}^{*}, \overline{y}^{*}\right).
\end{array}
\end{equation}

The strategy of this proof consists of defining candidate random variables $\left(\tilde{Y}_{0}^{*}, \tilde{Y}_{1}^{*}, \tilde{U}, \tilde{V}\right)$ through their joint cumulative distribution function $F_{\tilde{Y}_{0}^{*}, \tilde{Y}_{1}^{*}, \tilde{U}, \tilde{V}, Z, X}$ and then checking that equations \eqref{faketarget}, \eqref{correctsupport} and \eqref{DataRestriction} are satisfied. I fix $\left(y_{0}, y_{1}, u, v, z, x\right) \in \mathbb{R}^{6}$ and define $F_{\tilde{Y}_{0}^{*}, \tilde{Y}_{1}^{*}, \tilde{U}, \tilde{V}, Z, X}$ in twelve steps:
\begin{enumerate}
	\item[Step 1.] For $x \notin \mathcal{X}$, $F_{\tilde{Y}_{0}^{*}, \tilde{Y}_{1}^{*}, \tilde{U}, \tilde{V}, Z, X}\left(y_{0}, y_{1}, u, v, z, x\right) = F_{Y_{0}^{*}, Y_{1}^{*}, U, V, Z, X}\left(y_{0}, y_{1}, u, v, z, x\right)$.
	
	\item[Step 2.] From now on, consider $x \in \mathcal{X}$. Since $$F_{\tilde{Y}_{0}^{*}, \tilde{Y}_{1}^{*}, \tilde{U}, \tilde{V}, Z, X}\left(y_{0}, y_{1}, u, v, z, x\right) = F_{\tilde{Y}_{0}^{*}, \tilde{Y}_{1}^{*}, \tilde{U}, \tilde{V}, Z \left\vert X \right.}\left(y_{0}, y_{1}, u, v, z \left\vert x \right.\right) \cdot F_{X}\left(x\right),$$ it suffices to define $F_{\tilde{Y}_{0}^{*}, \tilde{Y}_{1}^{*}, \tilde{U}, \tilde{V}, Z \left\vert X \right.}\left(y_{0}, y_{1}, u, v, z \left\vert x \right.\right)$. Moreover, I impose $$\left. Z \independent \left(\tilde{Y}_{0}^{*}, \tilde{Y}_{1}^{*}, \tilde{U}, \tilde{V} \right) \right\vert X$$ by writing $$F_{\tilde{Y}_{0}^{*}, \tilde{Y}_{1}^{*}, \tilde{U}, \tilde{V}, Z \left\vert X \right.}\left(y_{0}, y_{1}, u, v, z \left\vert x \right.\right) = F_{\tilde{Y}_{0}^{*}, \tilde{Y}_{1}^{*}, \tilde{U}, \tilde{V}\left\vert X \right.}\left(y_{0}, y_{1}, u, v \left\vert x \right.\right) \cdot F_{Z \left\vert X \right.}\left(z \left\vert x \right.\right),$$ implying that it is sufficient to define $F_{\tilde{Y}_{0}^{*}, \tilde{Y}_{1}^{*}, \tilde{U}, \tilde{V}\left\vert X \right.}\left(y_{0}, y_{1}, u, v \left\vert x \right.\right)$.
	
	\item[Step 3.] For $u \notin \left[0, 1\right]$, I define $F_{\tilde{Y}_{0}^{*}, \tilde{Y}_{1}^{*}, \tilde{U}, \tilde{V}\left\vert X \right.}\left(y_{0}, y_{1}, u, v \left\vert x \right.\right) = F_{Y_{0}^{*}, Y_{1}^{*}, U, V\left\vert X \right.}\left(y_{0}, y_{1}, u, v \left\vert x \right.\right)$.
	
	\item[Step 4.] From now on, consider $u \in \left[0, 1\right]$. Since $$F_{\tilde{Y}_{0}^{*}, \tilde{Y}_{1}^{*}, \tilde{U}, \tilde{V}\left\vert X \right.}\left(y_{0}, y_{1}, u, v \left\vert x \right.\right) = F_{\tilde{Y}_{0}^{*}, \tilde{Y}_{1}^{*}, \tilde{V}\left\vert X, \tilde{U} \right.}\left(y_{0}, y_{1}, v \left\vert x, u \right.\right) \cdot F_{\tilde{U}\left\vert X \right.}\left(u \left\vert x \right.\right),$$ it suffices to define $F_{\tilde{Y}_{0}^{*}, \tilde{Y}_{1}^{*}, \tilde{V}\left\vert X, \tilde{U} \right.}\left(y_{0}, y_{1}, v \left\vert x, u \right.\right)$ and $F_{\tilde{U}\left\vert X \right.}\left(u \left\vert x \right.\right)$.
	
	\item[Step 5.] I define $F_{\tilde{U}\left\vert X \right.}\left(u \left\vert x \right.\right) = F_{U\left\vert X \right.}\left(u \left\vert x \right.\right) = u$.
	
	\item[Step 6.] For any $u \neq \overline{u}$, I define $F_{\tilde{Y}_{0}^{*}, \tilde{Y}_{1}^{*}, \tilde{V}\left\vert X, \tilde{U} \right.}\left(y_{0}, y_{1}, v \left\vert x, u\right.\right) = F_{Y_{0}^{*}, Y_{1}^{*}, V\left\vert X, U \right.}\left(y_{0}, y_{1}, v \left\vert x, u\right.\right)$.
	
	\item[Step 7.] For any $v \notin \left[0, 1\right]$, I define $F_{\tilde{Y}_{0}^{*}, \tilde{Y}_{1}^{*}, \tilde{V}\left\vert X, \tilde{U} \right.}\left(y_{0}, y_{1}, v \left\vert x, \overline{u}\right.\right) = F_{Y_{0}^{*}, Y_{1}^{*}, V\left\vert X, U \right.}\left(y_{0}, y_{1}, v \left\vert x, \overline{u}\right.\right)$.
	
	\item[Step 8.] From now on, consider $v \in \left[0, 1\right]$. Since $$F_{\tilde{Y}_{0}^{*}, \tilde{Y}_{1}^{*}, \tilde{V}\left\vert X, \tilde{U} \right.}\left(y_{0}, y_{1}, v \left\vert x, \overline{u}\right.\right) = F_{\tilde{Y}_{0}^{*}, \tilde{Y}_{1}^{*}\left\vert X, \tilde{U}, \tilde{V} \right.}\left(y_{0}, y_{1}\left\vert x, \overline{u}, v\right.\right) \cdot F_{\tilde{V}\left\vert X, \tilde{U} \right.}\left(v \left\vert x, \overline{u}\right.\right),$$ it is sufficient to define $F_{\tilde{Y}_{0}^{*}, \tilde{Y}_{1}^{*}\left\vert X, \tilde{U}, \tilde{V} \right.}\left(y_{0}, y_{1}\left\vert x, \overline{u}, v\right.\right)$ and $F_{\tilde{V}\left\vert X, \tilde{U} \right.}\left(v \left\vert x, \overline{u}\right.\right)$.
	
	\item[Step 9.] I define
	$$
	F_{\tilde{V}\left\vert X, \tilde{U} \right.}\left(v \left\vert x, \overline{u}\right.\right) = \left\lbrace
	\begin{array}{cl}
	m_{0}^{S}\left(x, \overline{u}\right) \cdot \dfrac{v}{Q\left(0, x\right)} & \text{if } v \leq Q\left(0, x\right) \\
	& \\
	m_{0}^{S}\left(x, \overline{u}\right) + \Delta_{S}\left(x, \overline{u}\right) \cdot \dfrac{v - Q\left(0, x\right)}{Q\left(1, x\right) - Q\left(0, x\right)} & \text{if } Q\left(0, x\right) < v \leq Q\left(1, x\right) \\
	& \\
	m_{1}^{S}\left(x, \overline{u}\right) + \left(1 - m_{1}^{S}\left(x, \overline{u}\right)\right)\dfrac{v - Q\left(1, x\right)}{1 - Q\left(1, x\right)} & \text{if } Q\left(1, x\right) < v
	\end{array}
	\right..
	$$
	
	\item[Step 10.] I write $F_{\tilde{Y}_{0}^{*}, \tilde{Y}_{1}^{*}\left\vert X, \tilde{U}, \tilde{V} \right.}\left(y_{0}, y_{1}\left\vert x, \overline{u}, v\right.\right) = F_{\tilde{Y}_{0}^{*}\left\vert X, \tilde{U}, \tilde{V} \right.}\left(y_{0}\left\vert x, \overline{u}, v\right.\right) \cdot F_{ \tilde{Y}_{1}^{*}\left\vert X, \tilde{U}, \tilde{V} \right.}\left(y_{1}\left\vert x, \overline{u}, v\right.\right)$, implying that I can separately define $F_{\tilde{Y}_{0}^{*}\left\vert X, \tilde{U}, \tilde{V} \right.}\left(y_{0}\left\vert x, \overline{u}, v\right.\right)$ and $F_{ \tilde{Y}_{1}^{*}\left\vert X, \tilde{U}, \tilde{V} \right.}\left(y_{1}\left\vert x, \overline{u}, v\right.\right)$.
	
	\item[Step 11.] When $\mathcal{Y}^{*}$ is a bounded interval (sub-case (a) in Assumption \ref{bounded}.3), I define
	$$
	F_{\tilde{Y}_{0}^{*}\left\vert X, \tilde{U}, \tilde{V} \right.}\left(y_{0}\left\vert x, \overline{u}, v\right.\right) = \left\lbrace
	\begin{array}{cl}
	\mathbf{1}\left\lbrace y_{0} \geq \dfrac{m_{0}^{Y}\left(\overline{x}, \overline{u}\right)}{m_{0}^{S}\left(\overline{x}, \overline{u}\right)} \right\rbrace & \text{if } v \leq Q\left(0, x\right) \\
	---------- & ------- \\
	\mathbf{1}\left\lbrace y_{0} \geq \dfrac{\underline{y}^{*} + \overline{y}^{*}}{2} \right\rbrace & \text{if } Q\left(0, x\right) < v
	\end{array}
	\right..
	$$

	When $\overline{y}^{*} = \max \left\lbrace y \in \mathcal{Y}^{*} \right\rbrace$ and $\underline{y}^{*} = \min \left\lbrace y \in \mathcal{Y}^{*} \right\rbrace$ (sub-case (b) in Assumption \ref{bounded}.3), I define
	$$
	F_{\tilde{Y}_{0}^{*}\left\vert X, \tilde{U}, \tilde{V} \right.}\left(y_{0}\left\vert x, \overline{u}, v\right.\right) = \left\lbrace
	\begin{array}{cl}
	0 & \text{if } y_{0} < \underline{y}^{*} \text{ and } v \leq Q\left(0, x\right) \\
	& \\
	1 - \dfrac{\dfrac{m_{0}^{Y}\left(\overline{x}, \overline{u}\right)}{m_{0}^{S}\left(\overline{x}, \overline{u}\right)} - \underline{y}^{*}}{\overline{y}^{*} - \underline{y}^{*}} & \text{if } \underline{y}^{*} \leq y_{0} < \overline{y}^{*} \text{ and } v \leq Q\left(0, x\right) \\
	& \\
	1 & \text{if } \overline{y}^{*} \leq y_{0} \text{ and } v \leq Q\left(0, x\right) \\
	---------- & -------------- \\
	\mathbf{1}\left\lbrace y_{0} \geq \overline{y}^{*} \right\rbrace & \text{if } Q\left(0, x\right) < v
	\end{array}
	\right..
	$$
	which are valid cumulative distribution functions because $\dfrac{m_{0}^{Y}\left(\overline{x}, \overline{u}\right)}{m_{0}^{S}\left(\overline{x}, \overline{u}\right)} \in \left[\underline{y}^{*}, \overline{y}^{*}\right]$.

	\item[Step 12.] When $\mathcal{Y}^{*}$ is a bounded interval (sub-case (a) in Assumption \ref{bounded}.3), I define
	$$
	F_{\tilde{Y}_{1}^{*}\left\vert X, \tilde{U}, \tilde{V} \right.}\left(y_{1}\left\vert x, \overline{u}, v\right.\right) = \left\lbrace
	\begin{array}{cl}
	\mathbf{1}\left\lbrace y_{1} \geq \alpha\left(\overline{x}, \overline{u}\right) \right\rbrace & \text{if } v \leq Q\left(0, x\right) \\
	-------- & ----------- \\
	\mathbf{1}\left\lbrace y_{1} \geq \gamma\left(\overline{x}, \overline{u}\right) \right\rbrace & \text{if } Q\left(0, x\right) < v \leq Q\left(1, x\right) \\
	-------- & ----------- \\
	\mathbf{1}\left\lbrace y_{1} \geq \dfrac{\underline{y}^{*} + \overline{y}^{*}}{2} \right\rbrace & \text{if } Q\left(1, x\right) < v
	\end{array}
	\right..
	$$
	
	When $\overline{y}^{*} = \max \left\lbrace y \in \mathcal{Y}^{*} \right\rbrace$ and $\underline{y}^{*} = \min \left\lbrace y \in \mathcal{Y}^{*} \right\rbrace$ (sub-case (b) in Assumption \ref{bounded}.3), I define
	$$
	F_{\tilde{Y}_{1}^{*}\left\vert X, \tilde{U}, \tilde{V} \right.}\left(y_{1}\left\vert x, \overline{u}, v\right.\right) = \left\lbrace
	\begin{array}{cl}
	0 & \text{if } y_{1} < \underline{y}^{*} \text{ and } v \leq Q\left(0, x\right) \\
	& \\
	1 - \dfrac{\alpha\left(\overline{x}, \overline{u}\right) - \underline{y}^{*}}{\overline{y}^{*} - \underline{y}^{*}} & \text{if } \underline{y}^{*} \leq y_{1} < \overline{y}^{*} \text{ and } v \leq Q\left(0, x\right) \\
	& \\
	1 & \text{if } \overline{y}^{*} \leq y_{1} \text{ and } v \leq Q\left(0, x\right) \\
	-------- & ------------------ \\
	0 & \text{if } y_{1} < \underline{y}^{*} \text{ and } Q\left(0, x\right) < v \leq Q\left(1, x\right) \\
	& \\
	1 - \dfrac{\gamma\left(\overline{x}, \overline{u}\right) - \underline{y}^{*}}{\overline{y}^{*} - \underline{y}^{*}} & \text{if } \underline{y}^{*} \leq y_{1} < \overline{y}^{*} \text{ and } Q\left(0, x\right) < v \leq Q\left(1, x\right) \\
	& \\
	1 & \text{if } \overline{y}^{*} \leq y_{1} \text{ and } Q\left(0, x\right) < v \leq Q\left(1, x\right) \\
	-------- & ------------------ \\
	\mathbf{1}\left\lbrace y_{1} \geq \overline{y}^{*} \right\rbrace & \text{if } Q\left(1, x\right) < v
	\end{array}
	\right..
	$$
	which are valid cumulative distribution functions because of equations \eqref{sanity1} and \eqref{sanity2}.
\end{enumerate}

Having defined the joint cumulative distribution function $F_{\tilde{Y}_{0}^{*}, \tilde{Y}_{1}^{*}, \tilde{U}, \tilde{V}, Z, X}$, note that equations \eqref{sanity1} and \eqref{sanity2}, $\dfrac{m_{0}^{Y}\left(\overline{x}, \overline{u}\right)}{m_{0}^{S}\left(\overline{x}, \overline{u}\right)} \in \left[\underline{y}^{*}, \overline{y}^{*}\right]$ and steps 7-12 ensure that equation \eqref{correctsupport} holds.

Now, I show, in three steps, that equation \eqref{faketarget} holds.
\begin{enumerate}	
	\item[Step 13.] Observe that
	\begin{align}
	& \mathbb{E}\left[\tilde{Y}_{1}^{*} \left\vert X = \overline{x}, \tilde{U} = \overline{u},  \tilde{S}_{0} = 1, \tilde{S}_{1} = 1 \right.\right] \nonumber \\
	& \hspace{30pt} = \mathbb{E}\left[\tilde{Y}_{1}^{*} \left\vert X = \overline{x}, \tilde{U} = \overline{u},  Q\left(0, \overline{x}\right) \geq \tilde{V} \right.\right] \nonumber \\
	& \hspace{50pt} \text{by the definition of } \tilde{S}_{0} \text{ and } \tilde{S}_{1} \nonumber \\
	& \hspace{30pt} = \dfrac{\mathbb{E}\left[\mathbf{1}\left\lbrace Q\left(0, \overline{x}\right) \geq \tilde{V} \right\rbrace \cdot \tilde{Y}_{1}^{*} \left\vert X = \overline{x}, \tilde{U} = \overline{u} \right.\right]}{\mathbb{P}\left[Q\left(0, \overline{x}\right) \geq \tilde{V} \left\vert X = \overline{x}, \tilde{U} = \overline{u} \right.\right]} \nonumber \\
	& \hspace{50pt} \text{by the definition of conditional expectation} \nonumber \\
	& \hspace{30pt} = \dfrac{\mathbb{E}\left[\mathbf{1}\left\lbrace Q\left(0, \overline{x}\right) \geq \tilde{V} \right\rbrace \cdot \mathbb{E}\left[\tilde{Y}_{1}^{*} \left\vert X = \overline{x}, \tilde{U} = \overline{u}, \tilde{V} \right. \right] \left\vert X = \overline{x}, \tilde{U} = \overline{u} \right.\right]}{\mathbb{P}\left[Q\left(0, \overline{x}\right) \geq \tilde{V} \left\vert X = \overline{x}, \tilde{U} = \overline{u} \right.\right]} \nonumber \\
	& \hspace{50pt} \text{by the Law of Iterated Expectations} \nonumber \\
	& \hspace{30pt} = \dfrac{\mathop{\mathlarger{\int\limits_{0}^{Q\left(0, \overline{x}\right)}}} \mathbb{E}\left[\tilde{Y}_{1}^{*} \left\vert X = \overline{x}, \tilde{U} = \overline{u}, \tilde{V} = v \right. \right] \, \text{d}F_{\tilde{V} \left\vert X, \tilde{U} \right.}\left(v \left\vert x, \overline{u} \right.\right)}{\mathbb{P}\left[Q\left(0, \overline{x}\right) \geq \tilde{V} \left\vert X = \overline{x}, \tilde{U} = \overline{u} \right.\right]} \nonumber \\
	& \hspace{50pt} \text{by the definition of expectation and by step 7} \nonumber \\
	& \hspace{30pt} = \dfrac{\mathop{\mathlarger{\int\limits_{0}^{Q\left(0, \overline{x}\right)}}} \alpha\left(\overline{x}, \overline{u}\right) \, \text{d}F_{\tilde{V} \left\vert X, \tilde{U} \right.}\left(v \left\vert x, \overline{u} \right.\right)}{\mathbb{P}\left[Q\left(0, \overline{x}\right) \geq \tilde{V} \left\vert X = \overline{x}, \tilde{U} = \overline{u} \right.\right]} \nonumber \\
	& \hspace{50pt} \text{by step 12} \nonumber \\
	& \label{step14a} \hspace{30pt} = \alpha\left(\overline{x}, \overline{u}\right) \\
	& \hspace{50pt} \text{by linearity of the Lebesgue Integral} \nonumber
	\end{align}
	
	\item[Step 14.] Similarly to the last step, notice that
	\begin{align}
	& \mathbb{E}\left[\tilde{Y}_{0}^{*} \left\vert X = \overline{x}, \tilde{U} = \overline{u},  \tilde{S}_{0} = 1, \tilde{S}_{1} = 1 \right.\right] \nonumber \\
	& \hspace{30pt} = \mathbb{E}\left[\tilde{Y}_{0}^{*} \left\vert X = \overline{x}, \tilde{U} = \overline{u},  Q\left(0, \overline{x}\right) \geq \tilde{V} \right.\right] \nonumber \\
	& \hspace{30pt} = \dfrac{\mathbb{E}\left[\mathbf{1}\left\lbrace Q\left(0, \overline{x}\right) \geq \tilde{V} \right\rbrace \cdot \tilde{Y}_{0}^{*} \left\vert X = \overline{x}, \tilde{U} = \overline{u} \right.\right]}{\mathbb{P}\left[Q\left(0, \overline{x}\right) \geq \tilde{V} \left\vert X = \overline{x}, \tilde{U} = \overline{u} \right.\right]} \nonumber \\
	& \hspace{30pt} = \dfrac{\mathbb{E}\left[\mathbf{1}\left\lbrace Q\left(0, \overline{x}\right) \geq \tilde{V} \right\rbrace \cdot \mathbb{E}\left[\tilde{Y}_{0}^{*} \left\vert X = \overline{x}, \tilde{U} = \overline{u}, \tilde{V} \right. \right] \left\vert X = \overline{x}, \tilde{U} = \overline{u} \right.\right]}{\mathbb{P}\left[Q\left(0, \overline{x}\right) \geq \tilde{V} \left\vert X = \overline{x}, \tilde{U} = \overline{u} \right.\right]} \nonumber \\
	& \hspace{30pt} = \dfrac{\mathop{\mathlarger{\int\limits_{0}^{Q\left(0, \overline{x}\right)}}} \mathbb{E}\left[\tilde{Y}_{0}^{*} \left\vert X = \overline{x}, \tilde{U} = \overline{u}, \tilde{V} = v \right. \right] \, \text{d}F_{\tilde{V} \left\vert X, \tilde{U} \right.}\left(v \left\vert x, \overline{u} \right.\right)}{\mathbb{P}\left[Q\left(0, \overline{x}\right) \geq \tilde{V} \left\vert X = \overline{x}, \tilde{U} = \overline{u} \right.\right]} \nonumber \\
	& \hspace{30pt} = \dfrac{ \mathop{\mathlarger{\mathlarger{\mathop{\mathlarger{\int\limits_{0}^{Q\left(0, \overline{x}\right)}}}}}} \dfrac{m_{0}^{Y}\left(\overline{x}, \overline{u}\right)}{m_{0}^{S}\left(\overline{x}, \overline{u}\right)} \, \text{d}F_{\tilde{V} \left\vert X, \tilde{U} \right.}\left(v \left\vert x, \overline{u} \right.\right)}{\mathbb{P}\left[Q\left(0, \overline{x}\right) \geq \tilde{V} \left\vert X = \overline{x}, \tilde{U} = \overline{u} \right.\right]} \text{ by step 11} \nonumber \\
	& \label{step15a} \hspace{30pt} = \dfrac{m_{0}^{Y}\left(\overline{x}, \overline{u}\right)}{m_{0}^{S}\left(\overline{x}, \overline{u}\right)}.
	\end{align}
		
	\item[Step 15.] Note that
	\begin{align*}
	\Delta_{\tilde{Y}^{*}}^{OO}\left(\overline{x}, \overline{u}\right) & \coloneqq \mathbb{E}\left[\tilde{Y}_{1}^{*} - \tilde{Y}_{0}^{*} \left\vert X = \overline{x}, \tilde{U} = \overline{u}, \tilde{S}_{0} = 1, \tilde{S}_{1} = 1 \right.\right] \\
	& = \mathbb{E}\left[\tilde{Y}_{1}^{*}\left\vert X = \overline{x}, \tilde{U} = \overline{u}, \tilde{S}_{0} = 1, \tilde{S}_{1} = 1 \right.\right] \\
	& \hspace{20pt} - \mathbb{E}\left[\tilde{Y}_{0}^{*}\left\vert X = \overline{x}, \tilde{U} = \overline{u}, \tilde{S}_{0} = 1, \tilde{S}_{1} = 1 \right.\right] \\
	& = \alpha\left(\overline{x}, \overline{u}\right) - \dfrac{m_{0}^{Y}\left(\overline{x}, \overline{u}\right)}{m_{0}^{S}\left(\overline{x}, \overline{u}\right)} \\
	& \hspace{20pt} \text{by equations } \eqref{step14a} \text{ and } \eqref{step15a} \\
	& = \delta\left(\overline{x}, \overline{u}\right) \\
	& \hspace{20pt} \text{by the definition of } \alpha\left(\overline{x}, \overline{u}\right),
	\end{align*}
	ensuring that equation \eqref{faketarget} holds.
\end{enumerate}

Finally, I show, in two steps, that equation \eqref{DataRestriction} holds.
\begin{enumerate}
	\item[Step 16.] Fix $\left(y, d, s, z\right) \in \mathbb{R}^{4}$ arbitrarily and observe that equation \eqref{DataRestriction} can be simplified to:
	\begin{align}
	& F_{\tilde{Y}, \tilde{D}, \tilde{S}, Z, X}\left(y, d, s, z, \overline{x} \right) = F_{Y, D, S, Z, X} \left(y, d, s, z, \overline{x}\right) \nonumber \\
	\Leftrightarrow & F_{\tilde{Y}, \tilde{D}, \tilde{S}, Z \left\vert X \right.}\left(y, d, s, z \left\vert \overline{x} \right.\right) \cdot F_{X}\left(\overline{x}\right) = F_{Y, D, S, Z \left\vert X \right.}\left(y, d, s, z \left\vert \overline{x} \right.\right) \cdot F_{X}\left(\overline{x}\right) \nonumber \\
	\Leftrightarrow & \label{DataRestrictionSimplified} F_{\tilde{Y}, \tilde{D}, \tilde{S}, Z \left\vert X \right.}\left(y, d, s, z \left\vert \overline{x} \right.\right) = F_{Y, D, S, Z \left\vert X \right.}\left(y, d, s, z \left\vert \overline{x} \right.\right)
	\end{align}
	
	\item[Step 17.] Notice that
	\begin{align}
	& F_{\tilde{Y}, \tilde{D}, \tilde{S}, Z \left\vert X \right.}\left(y, d, s, z \left\vert \overline{x} \right.\right) \nonumber \\
	& \hspace{30pt} = \mathbb{E}\left[\left.\mathbf{1}\left\lbrace \left(\tilde{Y}, \tilde{D}, \tilde{S}, Z\right) \leq \left(y, d, s, z\right) \right\rbrace \right\vert X = \overline{x} \right] \nonumber \\
	& \hspace{30pt} = \int \mathbf{1}\left\lbrace \left(\tilde{Y}, \tilde{D}, \tilde{S}, Z\right) \leq \left(y, d, s, z\right) \right\rbrace \, \text{d} F_{\tilde{Y}_{0}^{*}, \tilde{Y}_{1}^{*}, \tilde{U}, \tilde{V}, Z \left\vert X \right.}\left(y_{0}, y_{1}, u, v, z \left\vert \overline{x} \right.\right) \nonumber \\
	& \hspace{50pt} \text{because } \left(\tilde{Y}, \tilde{D}, \tilde{S}, Z\right) \text{ are functions of } \left(\tilde{Y}_{0}^{*}, \tilde{Y}_{1}^{*}, \tilde{U}, \tilde{V}, Z\right) \nonumber \\
	& \hspace{30pt} = \int \left[\mathbf{1}\left\lbrace \left(\tilde{Y}, \tilde{D}, \tilde{S}, Z\right) \leq \left(y, d, s, z\right) \right\rbrace \cdot \mathbf{1}\left\lbrace u \neq \overline{u} \right\rbrace\right] \, \text{d} F_{\tilde{Y}_{0}^{*}, \tilde{Y}_{1}^{*}, \tilde{U}, \tilde{V}, Z \left\vert X \right.}\left(y_{0}, y_{1}, u, v, z \left\vert \overline{x} \right.\right) \nonumber \\
	& \hspace{50pt} + \int \left[\mathbf{1}\left\lbrace \left(\tilde{Y}, \tilde{D}, \tilde{S}, Z\right) \leq \left(y, d, s, z\right) \right\rbrace \cdot \mathbf{1}\left\lbrace u = \overline{u} \right\rbrace\right] \, \text{d} F_{\tilde{Y}_{0}^{*}, \tilde{Y}_{1}^{*}, \tilde{U}, \tilde{V}, Z \left\vert X \right.}\left(y_{0}, y_{1}, u, v, z \left\vert \overline{x} \right.\right) \nonumber \\
	& \hspace{50pt} \text{by linearity of the Lebesgue Integral} \nonumber \\
	& \hspace{30pt} = \int \left[\mathbf{1}\left\lbrace \left(\tilde{Y}, \tilde{D}, \tilde{S}, Z\right) \leq \left(y, d, s, z\right) \right\rbrace \cdot \mathbf{1}\left\lbrace u \neq \overline{u} \right\rbrace\right] \, \text{d} F_{\tilde{Y}_{0}^{*}, \tilde{Y}_{1}^{*}, \tilde{U}, \tilde{V}, Z \left\vert X \right.}\left(y_{0}, y_{1}, u, v, z \left\vert \overline{x} \right.\right) \nonumber \\
	& \hspace{50pt} \text{because } \mathbb{P}\left[\left. \tilde{U} = \overline{u} \right\vert X = \overline{x}\right] = 0 \text{ by step 5} \nonumber \\
	& \hspace{30pt} = \int \left[\mathbf{1}\left\lbrace \left(Y, D, S, Z\right) \leq \left(y, d, s, z\right) \right\rbrace \cdot \mathbf{1}\left\lbrace u \neq \overline{u} \right\rbrace\right] \, \text{d} F_{Y_{0}^{*}, Y_{1}^{*}, U, V, Z \left\vert X \right.}\left(y_{0}, y_{1}, u, v, z \left\vert \overline{x} \right.\right) \nonumber \\
	& \hspace{50pt} \text{by steps 2-6} \nonumber \\
	& \hspace{30pt} = \int \left[\mathbf{1}\left\lbrace \left(Y, D, S, Z\right) \leq \left(y, d, s, z\right) \right\rbrace \cdot \mathbf{1}\left\lbrace u \neq \overline{u} \right\rbrace\right] \, \text{d} F_{Y_{0}^{*}, Y_{1}^{*}, U, V, Z \left\vert X \right.}\left(y_{0}, y_{1}, u, v, z \left\vert \overline{x} \right.\right) \nonumber \\
	& \hspace{50pt} + \int \left[\mathbf{1}\left\lbrace \left(Y, D, S, Z\right) \leq \left(y, d, s, z\right) \right\rbrace \cdot \mathbf{1}\left\lbrace u = \overline{u} \right\rbrace\right] \, \text{d} F_{Y_{0}^{*}, Y_{1}^{*}, U, V, Z \left\vert X \right.}\left(y_{0}, y_{1}, u, v, z \left\vert \overline{x} \right.\right) \nonumber \\
	& \hspace{50pt} \text{because } \mathbb{P}\left[U = \overline{u} \left\vert X = \overline{x} \right.\right] = 0 \nonumber \\
	& \hspace{30pt} = \int \mathbf{1}\left\lbrace \left(Y, D, S, Z\right) \leq \left(y, d, s, z\right) \right\rbrace \, \text{d} F_{Y_{0}^{*}, Y_{1}^{*}, U, V, Z \left\vert X \right.}\left(y_{0}, y_{1}, u, v, z \left\vert \overline{x} \right.\right) \nonumber \\
	& \hspace{50pt} \text{by linearity of the Lebesgue Integral} \nonumber \\
	& \hspace{30pt} = \mathbb{E}\left[\left.\mathbf{1}\left\lbrace \left(Y, D, S, Z\right) \leq \left(y, d, s, z\right) \right\rbrace \right\vert X = \overline{x} \right] \nonumber \\
	& \hspace{30pt} = F_{Y, D, S, Z \left\vert X \right.}\left(y, d, s, z \left\vert \overline{x} \right.\right), \nonumber
	\end{align}
	implying equation \eqref{DataRestriction} according to equation \eqref{DataRestrictionSimplified}.
\end{enumerate}

I can then conclude that Theorem \ref{sharpbounds} is true. $\blacksquare$

As a remark, the above constructive proof defines random variables $\left(\tilde{Y}_{0}^{*}, \tilde{Y}_{1}^{*}, \tilde{U}, \tilde{V}\right)$ that matches other important moments of the true data generating process besides the ones imposed by Theorem \ref{sharpbounds}.

\begin{enumerate}
	\item[Remark 1.] Note that
	\begin{align}
	\mathbb{P}\left[\tilde{S}_{0} = 1, \tilde{S}_{1} = 1 \left\vert X = \overline{x}, \tilde{U} = \overline{u} \right.\right] & = \mathbb{P}\left[Q\left(0, \overline{x}\right) \geq \tilde{V} \left\vert X = \overline{x}, \tilde{U} = \overline{u} \right.\right] \nonumber \\
	& \hspace{20pt} \text{by the definition of } \tilde{S}_{0} \text{ and } \tilde{S}_{1} \nonumber \\
	& \label{step13a} = m_{0}^{S}\left(\overline{x}, \overline{u}\right) \\
	& \hspace{20pt} \text{by step 9}, \nonumber
	\end{align}
	and, similarly, that
	\begin{align}
	\mathbb{P}\left[\tilde{S}_{0} = 0, \tilde{S}_{1} = 1 \left\vert X = \overline{x}, \tilde{U} = \overline{u} \right.\right] & = \mathbb{P}\left[Q\left(1, \overline{x}\right) \geq \tilde{V} > Q\left(0, \overline{x}\right) \left\vert X = \overline{x}, \tilde{U} = \overline{u} \right.\right] \nonumber \\
	& \label{step13b} = \Delta_{S}\left(\overline{x}, \overline{u}\right).
	\end{align}

	\item[Remark 2.] Analogously to equation \eqref{step14a}, I find that
	\begin{equation}\label{step14b}
	\mathbb{E}\left[\tilde{Y}_{1}^{*} \left\vert X = \overline{x}, \tilde{U} = \overline{u},  \tilde{S}_{0} = 0, \tilde{S}_{1} = 1 \right.\right] = \gamma\left(\overline{x}, \overline{u}\right).
	\end{equation}
	
	\item[Remark 3.] Combining equations \eqref{decomposition}, \eqref{step14a} and \eqref{step13a}-\eqref{step14b}, I have that
	\begin{equation*}
	\mathbb{E}\left[\tilde{Y}_{1} \left\vert X = x, \tilde{U} = \overline{u} \right. \right] = m_{1}^{Y}\left(\overline{x}, \overline{u}\right).
	\end{equation*}
	
	\item[Remark 4.] Similarly to step 17, I can show that $F_{\tilde{Y}_{0}^{*}, \tilde{Y}_{1}^{*}, \tilde{V}}\left(y_{0}, y_{1}, v\right) = F_{Y_{0}^{*}, Y_{1}^{*}, V}\left(y_{0}, y_{1}, v\right),$ implying that $\mathbb{E}\left[\left\vert \tilde{Y}_{d}^{*} \right\vert\right] < + \infty$ and $\mathbb{E}\left[\left( \tilde{Y}_{d}^{*} \right)^{2}\right] < + \infty$ for any $d \in \left\lbrace 0, 1 \right\rbrace$.
\end{enumerate}

\subsubsection{Proof under Assumptions \ref{bounded}.1 and \ref{bounded}.2}
I, now, prove Theorem \ref{sharpbounds} under Assumptions \ref{bounded}.1 and \ref{bounded}.2. In particular, I focus on the case $\underline{y}^{*} > - \infty$ and $\overline{y}^{*} = + \infty$ (Assumption \ref{bounded}.1) because it is more common in empirical applications. The case $\underline{y}^{*} = - \infty$ and $\overline{y}^{*} < + \infty$ (Assumption \ref{bounded}.2) is symmetric.

The proof under Assumption \ref{bounded}.1 is equal to the proof under Assumption \ref{bounded}.3(a). The only difference is that
\begin{equation}
\begin{array}{cll}
& \delta\left(\overline{x}, \overline{u}\right) & \in \left(\underline{\Delta_{Y^{*}}^{OO}}\left(\overline{x}, \overline{u}\right), \overline{\Delta_{Y^{*}}^{OO}}\left(\overline{x}, \overline{u}\right)\right) \\
& & \\
\Leftrightarrow & \alpha\left(\overline{x}, \overline{u}\right) & \in \left(\underline{y}^{*}, \dfrac{m_{1}^{Y}\left(x, u\right) - \underline{y}^{*} \cdot \Delta_{S}\left(x, u\right)}{m_{0}^{S}\left(x, u\right)}\right) \\
& & \\
& & \subseteq \left(\underline{y}^{*}, +\infty\right),
\end{array}
\end{equation}
and that
\begin{equation}
\begin{array}{cll}
& \alpha\left(\overline{x}, \overline{u}\right) & \in \left(\underline{y}^{*}, \dfrac{m_{1}^{Y}\left(x, u\right) - \underline{y}^{*} \cdot \Delta_{S}\left(x, u\right)}{m_{0}^{S}\left(x, u\right)}\right) \\
& & \\
\Leftrightarrow & \gamma\left(\overline{x}, \overline{u}\right) & \in \left(\underline{y}^{*},  + \infty\right).
\end{array}
\end{equation}

%%%%%%%%%%%%%%%%%%%%%%%%%%%%%%%%%%%%%%%%%%%
% Proof: Bounded Support is necessary
%%%%%%%%%%%%%%%%%%%%%%%%%%%%%%%%%%%%%%%%%%%
\subsection{Proof of Proposition \ref{partialnecessary}}\label{proofnecessary}
This proof is essentially the same proof of Theorem \ref{sharpbounds} under Assumption \ref{bounded}.3.(a) (appendix \ref{proofsharp3}). Fix $\overline{u} \in \left[0, 1\right]$, $\overline{x} \in \mathcal{X}$ and $\delta\left(\overline{x}, \overline{u}\right) \in \mathbb{R}$ arbitrarily. For brevity, define $\alpha\left(\overline{x}, \overline{u}\right) \coloneqq \delta\left(\overline{x}, \overline{u}\right) + \dfrac{m_{0}^{Y}\left(\overline{x}, \overline{u}\right)}{m_{0}^{S}\left(\overline{x}, \overline{u}\right)}$ and $\gamma\left(\overline{x}, \overline{u}\right) \coloneqq \dfrac{m_{1}^{Y}\left(\overline{x}, \overline{u}\right) - \alpha\left(\overline{x}, \overline{u}\right) \cdot m_{0}^{S}\left(\overline{x}, \overline{u}\right)}{\Delta_{S}\left(\overline{x}, \overline{u}\right)}$. Note that $\alpha\left(\overline{x}, \overline{u}\right) \in \mathbb{R} = \mathcal{Y}^{*}$ and $\gamma\left(\overline{x}, \overline{u}\right) \in \mathbb{R} = \mathcal{Y}^{*}$.

I define the random variables $\left(\tilde{Y}_{0}^{*}, \tilde{Y}_{1}^{*}, \tilde{U}, \tilde{V}\right)$ using the joint cumulative distribution function $F_{\tilde{Y}_{0}^{*}, \tilde{Y}_{1}^{*}, \tilde{U}, \tilde{V}, Z, X}$ described by steps 1-12 in Appendix \ref{proofsharp3} for the case of convex support $\mathcal{Y}^{*}$. Note that equation \eqref{correctsupportP} is trivially true when $\mathcal{Y}^{*} = \mathbb{R}$. Moreover, equations \eqref{faketargetP} and \eqref{DataRestrictionP} are valid by the argument described in steps 13-17 in Appendix \ref{proofsharp3}.

I can then conclude that Proposition \ref{partialnecessary} is true. $\blacksquare$

%%%%%%%%%%%%%%%%%%%%%%%%%%%%%%%%%%%%%%%%%%%%%%%%%%%%
% Proof: Comparing the bounds with and without the Mean Dominance Assumption
%%%%%%%%%%%%%%%%%%%%%%%%%%%%%%%%%%%%%%%%%%%%%%%%%%%%
\subsection{Comparing Corollaries \ref{MTEbounds} and \ref{boundmeandomG}}\label{comparing}
In order to compare Corollaries \ref{MTEbounds} and \ref{boundmeandomG}, I first prove that the second corollary provides lower bounds that are weakly larger than the lower bounds provided by the first corollary.

Fix $u \in \left[0, 1\right]$ and $x \in \mathcal{X}$ arbitrarily and note that $$\dfrac{m_{1}^{Y}\left(x, u\right)}{m_{1}^{S}\left(x, u\right)} = \dfrac{\mathbb{E}\left[\left. S_{1} \cdot Y_{1}^{*} \right\vert X = x, U = u\right]}{\mathbb{P}\left[\left. S_{1} = 1 \right\vert X = x, U = u\right]} = \mathbb{E}\left[\left. Y_{1}^{*} \right\vert X = x, U = u, S_{1} = 1\right],$$ implying that $\underline{y}^{*} \leq \dfrac{m_{1}^{Y}\left(x, u\right)}{m_{1}^{S}\left(x, u\right)} \leq \overline{y}^{*}$. Consequently, observe that $$\dfrac{m_{1}^{Y}\left(x, u\right) - \overline{y}^{*} \cdot \Delta_{S}\left(x, u\right)}{m_{0}^{S}\left(x, u\right)} \leq \dfrac{m_{1}^{Y}\left(x, u\right) - \dfrac{m_{1}^{Y}\left(x, u\right)}{m_{1}^{S}\left(x, u\right)} \cdot \Delta_{S}\left(x, u\right)}{m_{0}^{S}\left(x, u\right)} = \dfrac{m_{1}^{Y}\left(x, u\right)}{m_{1}^{S}\left(x, u\right)}.$$

The argument above shows that Corollary \ref{boundmeandomG} provides bounds that are weakly tighter than the ones provided by Corollary \ref{MTEbounds}. They will be strictly tighter if $\underline{y}^{*} < \dfrac{m_{1}^{Y}\left(x, u\right)}{m_{1}^{S}\left(x, u\right)} < \overline{y}^{*}$. Moreover, the improvement generated by the Mean Dominance Assumption \ref{meandominanceG} is proportional to $\dfrac{m_{1}^{Y}\left(x, u\right)}{m_{1}^{S}\left(x, u\right)} - \underline{y}^{*}$ and $\overline{y}^{*} - \dfrac{m_{1}^{Y}\left(x, u\right)}{m_{1}^{S}\left(x, u\right)}$ because $$\dfrac{m_{1}^{Y}\left(x, u\right)}{m_{1}^{S}\left(x, u\right)} - \dfrac{m_{1}^{Y}\left(x, u\right) - \overline{y}^{*} \cdot \Delta_{S}\left(x, u\right)}{m_{0}^{S}\left(x, u\right)} = \dfrac{\Delta_{S}\left(x, u\right) \cdot \left(\overline{y}^{*} \cdot m_{1}^{S}\left(x, u\right) - m_{1}^{Y}\left(x, u\right)\right)}{m_{0}^{S}\left(x, u\right) \cdot m_{1}^{S}\left(x, u\right)}.$$

%%%%%%%%%%%%%%%%%%%%%%%%%%%%%%%%%%%%%%%%%%%%%%%%%%%%
% Proof: Mean Dominance Assumption
%%%%%%%%%%%%%%%%%%%%%%%%%%%%%%%%%%%%%%%%%%%%%%%%%%%%
\subsection{Proof of Proposition \ref{sharpboundsmeanG}}\label{proofsharpboundsmeanG}
This proof is essentially the same proof of Theorem \ref{sharpbounds} and Proposition \ref{partialnecessary} (Appendices \ref{proofsharp} and \ref{proofnecessary}). Fix $\overline{u} \in \left[0, 1\right]$, $\overline{x} \in \mathcal{X}$ and $\delta\left(\overline{x}, \overline{u}\right) \in \left(\underline{\Delta_{Y^{*}}^{OO}}\left(\overline{x}, \overline{u}\right), \overline{\Delta_{Y^{*}}^{OO}}\left(\overline{x}, \overline{u}\right)\right)$ arbitrarily. For brevity, define $\alpha\left(\overline{x}, \overline{u}\right) \coloneqq \delta\left(\overline{x}, \overline{u}\right) + \dfrac{m_{0}^{Y}\left(\overline{x}, \overline{u}\right)}{m_{0}^{S}\left(\overline{x}, \overline{u}\right)}$ and $\gamma\left(\overline{x}, \overline{u}\right) \coloneqq \dfrac{m_{1}^{Y}\left(\overline{x}, \overline{u}\right) - \alpha\left(\overline{x}, \overline{u}\right) \cdot m_{0}^{S}\left(\overline{x}, \overline{u}\right)}{\Delta_{S}\left(\overline{x}, \overline{u}\right)}$. The only difference from the previous proofs is that, now, 
\begin{align}
\mathbb{E}\left[\tilde{Y}_{1}^{*} \left\vert X = \overline{x}, \tilde{U} = \overline{u}, \tilde{S}_{0} = 1, \tilde{S}_{1} = 1 \right.\right] & = \alpha\left(\overline{x}, \overline{u}\right) \nonumber \\
& \hspace{20pt} \text{ by equation } \eqref{step14a} \nonumber \\
& \label{MDresult} \geq \dfrac{m_{1}^{Y}\left(\overline{x}, \overline{u}\right)}{m_{1}^{S}\left(\overline{x}, \overline{u}\right)} \\
& \hspace{20pt} \text{because }  \delta\left(\overline{x}, \overline{u}\right)  \geq \underline{\Delta_{Y^{*}}^{OO}}\left(\overline{x}, \overline{u}\right) \nonumber
\end{align}
and that
\begin{align}
\mathbb{E}\left[\tilde{Y}_{1}^{*} \left\vert X = \overline{x}, \tilde{U} = \overline{u}, \tilde{S}_{0} = 0, \tilde{S}_{1} = 1 \right.\right] & = \gamma\left(\overline{x}, \overline{u}\right) \nonumber \\
& \hspace{20pt} \text{ by equation } \eqref{step14b} \nonumber \\
& = \dfrac{m_{1}^{Y}\left(\overline{x}, \overline{u}\right) - \alpha\left(\overline{x}, \overline{u}\right) \cdot m_{0}^{S}\left(\overline{x}, \overline{u}\right)}{\Delta_{S}\left(\overline{x}, \overline{u}\right)} \nonumber \\
& \leq \dfrac{m_{1}^{Y}\left(\overline{x}, \overline{u}\right) - \dfrac{m_{1}^{Y}\left(\overline{x}, \overline{u}\right)}{m_{1}^{S}\left(\overline{x}, \overline{u}\right)} \cdot m_{0}^{S}\left(\overline{x}, \overline{u}\right)}{\Delta_{S}\left(\overline{x}, \overline{u}\right)} \nonumber \\
& \hspace{20pt} \text{by equation } \eqref{MDresult} \nonumber \\
& = \dfrac{m_{1}^{Y}\left(\overline{x}, \overline{u}\right)}{m_{1}^{S}\left(\overline{x}, \overline{u}\right)}, \nonumber
\end{align}
implying that the model restriction \eqref{fakemeandominanceG} holds.

%%%%%%%%%%%%%%%%%%%%%%%%%%%%%%%%%%%%%%%%%%%%%%%%%%%%
% Proof: Identification - Continuous
%%%%%%%%%%%%%%%%%%%%%%%%%%%%%%%%%%%%%%%%%%%%%%%%%%%%
\subsection{Proof of Equations \eqref{identifyQ0} and \eqref{identifyQ1}}\label{proofQ0Q1}
I first prove that equation \eqref{identifyQ0} holds. For any $A \in \left\lbrace Y, S \right\rbrace$, observe that
\begin{align*}
\mathbb{E}\left[A \left\vert X = x, P\left(W\right) = p, D = 0 \right.\right] & = \mathbb{E}\left[A_{0} \left\vert X = x, P\left(W\right) = p, D = 0 \right.\right] \\
& = \mathbb{E}\left[A_{0} \left\vert X = x, P\left(W\right) = p, P\left(W\right) < U \right.\right] \\
& \hspace{20pt} \text{by equation } \eqref{treatment} \\
& = \mathbb{E}\left[A_{0} \left\vert X = x, P\left(W\right) = p, p < U \right.\right] \\
& = \mathbb{E}\left[A_{0} \left\vert X = x, p < U \right.\right] \\
& \hspace{20pt} \text{by assumption } \eqref{ind} \\
& = \dfrac{\mathbb{E}\left[\mathbf{1}\left\lbrace p < U \right\rbrace \cdot A_{0}  \left\vert X = x\right. \right]}{\mathbb{P}\left[p < U \left\vert X = x \right. \right]} \\
& \hspace{20pt} \text{by the definition of conditional expectation} \\
& = \dfrac{\mathbb{E}\left[\mathbf{1}\left\lbrace p < U \right\rbrace \cdot A_{0} \left\vert X = x \right. \right]}{1 - p} \\
& \hspace{20pt} \text{ by the normalization } U \left\vert X \right. \sim \text{Uniform}\left[0, 1\right] \\
& = \dfrac{\mathbb{E}\left[\mathbf{1}\left\lbrace p < U \right\rbrace \cdot \mathbb{E}\left[A_{0} \left\vert X = x, U = u \right.\right] \left\vert X = x \right. \right]}{1 - p} \\
& \hspace{20pt} \text{by the Law of Iterated Expectations} \\
& = \dfrac{\int_{p}^{1} m_{0}^{A}\left(x, u\right) \, \text{d} u}{1 - p} \\
& \hspace{20pt} \text{by the normalization } U \left\vert X \right. \sim \text{Uniform}\left[0, 1\right], \\
\end{align*}
implying that
\begin{align*}
\dfrac{\partial \mathbb{E}\left[A \left\vert X = x, P\left(W\right) = p, D = 0 \right.\right]}{\partial p} & = \dfrac{- m_{0}^{A}\left(x, p\right)}{1 - p}  + \dfrac{\mathbb{E}\left[\mathbf{1}\left\lbrace p < U \right\rbrace \cdot A_{0} \left\vert X = x \right. \right]}{\left(1 - p\right)^{2}} \\
& = \dfrac{- m_{0}^{A}\left(x, p\right)}{1 - p}  + \dfrac{\mathbb{E}\left[\mathbf{1}\left\lbrace p < U \right\rbrace \cdot A_{0} \left\vert X = x \right. \right]}{\left(1 - p\right) \cdot \mathbb{P}\left[p < U \left\vert X = x \right. \right]} \\
& \hspace{20pt} \text{by the normalization } U \left\vert X \right. \sim \text{Uniform}\left[0, 1\right] \\
& = \dfrac{- m_{0}^{A}\left(x, p\right)}{1 - p} + \dfrac{\mathbb{E}\left[A \left\vert X = x, P\left(W\right) = p, D = 0 \right.\right]}{1 - p}
\end{align*}
Rearranging the last expression, I can derive equation \eqref{identifyQ0}:
\begin{align*}
m_{0}^{A}\left(x, p\right) & = \mathbb{E}\left[A \left\vert X = x, P\left(W\right) = p, D = 0 \right.\right] \\
& \hspace{20pt} - \dfrac{\partial \mathbb{E}\left[A \left\vert X = x, P\left(W\right) = p, D = 0 \right.\right]}{\partial p} \cdot \left(1 - p\right).
\end{align*}

Equation \eqref{identifyQ1} is derived in an analogous way using $\mathbb{E}\left[A \left\vert X = x, P\left(W\right) = p, D = 1 \right.\right]$ and its derivative with respect to the propensity score. $\blacksquare$

\subsection{Proof of Equations \eqref{nonseparable0} and \eqref{nonseparable1}} \label{parametricproof}
We first prove that equation \eqref{nonseparable0} holds. For any $A \in \left\lbrace Y, S \right\rbrace$, observe that
\begin{align*}
\mathbb{E}\left[A \left\vert X = x, P\left(W\right) = p_{n}, D = 0 \right.\right] & = \dfrac{\int_{p_{n}}^{1} m_{0}^{A}\left(x, u\right) \, \text{d} u}{1 - p_{n}} \\
& \hspace{20pt} \text{according to Appendix } \ref{proofQ0Q1} \\
& = \dfrac{\int_{p_{n}}^{1} M^{A}\left(u, \boldsymbol{\theta}_{x,0}^{A}\right) \, \text{d} u}{1 - p_{n}} \\
& \hspace{20pt} \text{by equation } \eqref{polynomialMTR}.
\end{align*}

Equation \eqref{nonseparable1} is derived in an analogous way using $\mathbb{E}\left[A \left\vert X = x, P\left(W\right) = p_{n}, D = 1 \right.\right]$. $\blacksquare$

\subsection{Parametric Bounds for the $\mathbf{MTE^{OO}}$}\label{OLSproof}

\subsubsection{Connecting OLS Model \eqref{OLSmodel} to the Minimization Problem \eqref{pseudotrue}} 
Note that, for any $z \in \left\lbrace 0, 1 \right\rbrace$,
\begin{align}
\dfrac{\int_{P\left(z\right)}^{1} M^{A}\left(u, \boldsymbol{\theta}_{0}^{A}\right) \, \text{d} u}{1 - P\left(z\right)} & = \dfrac{\int_{P\left(z\right)}^{1} \left(\theta_{0, 0}^{A} \cdot \left(1 - u\right) + \theta_{0, 1}^{A} \cdot u\right) \, \text{d} u}{1 - P\left(z\right)} \nonumber \\
& = \dfrac{\theta_{0, 0}^{A} + \theta_{0, 1}^{A}}{2} + \dfrac{- \theta_{0, 0}^{A} + \theta_{0, 1}^{A}}{2} \cdot P\left(z\right) \nonumber \\
& \label{OLSd0} = a_{0}^{A} + b_{0}^{A} \cdot P\left(z\right),
\end{align}
where $a_{0}^{A} \coloneqq \dfrac{\theta_{0, 0}^{A} + \theta_{0, 1}^{A}}{2}$ and $b_{0}^{A} \coloneqq \dfrac{-\theta_{0, 0}^{A} + \theta_{0, 1}^{A}}{2}$, and
\begin{align}
\dfrac{\int_{0}^{P\left(z\right)} M^{A}\left(u, \boldsymbol{\theta}_{1}^{A}\right) \, \text{d} u}{P\left(z\right)} & = \dfrac{\int_{0}^{P\left(z\right)} \left(\theta_{1, 0}^{A} \cdot \left(1 - u\right) + \theta_{1, 1}^{A} \cdot u\right) \, \text{d} u}{P\left(z\right)} \nonumber \\
& = \theta_{1, 0}^{A} + \dfrac{-\theta_{1, 0}^{A} + \theta_{1, 1}^{A}}{2} \cdot P\left(z\right) \nonumber \\
& \label{OLSd1} = a_{1}^{A} + b_{1}^{A} \cdot P\left(z\right),
\end{align}
where $a_{1}^{A} \coloneqq \theta_{1, 0}^{A}$ and $b_{1}^{A} \coloneqq \dfrac{-\theta_{1, 0}^{A} + \theta_{1, 1}^{A}}{2}$.

When I combine equations \eqref{pseudotrue}, \eqref{OLSd0} and \eqref{OLSd1}, I find the OLS model given by equation \eqref{OLSmodel}. Moreover, by solving the linear system given by $a_{0}^{A} = \dfrac{\theta_{0, 0}^{A} + \theta_{0, 1}^{A}}{2}$, $b_{0}^{A} = \dfrac{-\theta_{0, 0}^{A} + \theta_{0, 1}^{A}}{2}$, $a_{1}^{A} = \theta_{1, 0}^{A}$ and $b_{1}^{A} = \dfrac{-\theta_{1, 0}^{A} + \theta_{1, 1}^{A}}{2}$, I find that $\theta_{0, 0}^{A} = a_{0}^{A} - b_{0}^{A}$, $\theta_{0, 1}^{A} = a_{0}^{A} + b_{0}^{A}$, $\theta_{1, 0}^{A} = a_{1}^{A}$, $\theta_{1, 1}^{A} = a_{1}^{A} + 2 \cdot b_{1}^{A}$.

\subsubsection{Explicit Formulas for the Bounds in Corollaries \ref{MTEbounds} and \ref{boundmeandomG}}

When the marginal treatment response functions are given by the parametric model described in Subsection \ref{EstimationSandY} and the outcome of interested is bounded below by zero (e.g., hourly wages), Corollary \ref{MTEbounds} implies that, for any $x \in \mathcal{X}$ and $u \in \left[0, 1\right]$,
\begin{equation}
\Delta_{Y^{*}}^{OO}\left(x, u\right) \geq - \dfrac{\theta_{0, 0}^{Y} \cdot \left(1 - u\right) + \theta_{0, 1}^{Y} \cdot u}{\theta_{0, 0}^{S} \cdot \left(1 - u\right) + \theta_{0, 1}^{S} \cdot u},
\end{equation}
and
\begin{equation}
\Delta_{Y^{*}}^{OO}\left(x, u\right) \leq \dfrac{\theta_{1, 0}^{Y} \cdot \left(1 - u\right) + \theta_{1, 1}^{Y} \cdot u}{\theta_{0, 0}^{S} \cdot \left(1 - u\right) + \theta_{0, 1}^{S} \cdot u} - \dfrac{\theta_{0, 0}^{Y} \cdot \left(1 - u\right) + \theta_{0, 1}^{Y} \cdot u}{\theta_{0, 0}^{S} \cdot \left(1 - u\right) + \theta_{0, 1}^{S} \cdot u}.
\end{equation}

In the same context, Corollary \ref{boundmeandomG} implies that
\begin{equation}
\Delta_{Y^{*}}^{OO}\left(x, u\right) \geq \dfrac{\theta_{1, 0}^{Y} \cdot \left(1 - u\right) + \theta_{1, 1}^{Y} \cdot u}{\theta_{1, 0}^{S} \cdot \left(1 - u\right) + \theta_{1, 1}^{S} \cdot u} - \dfrac{\theta_{0, 0}^{Y} \cdot \left(1 - u\right) + \theta_{0, 1}^{Y} \cdot u}{\theta_{0, 0}^{S} \cdot \left(1 - u\right) + \theta_{0, 1}^{S} \cdot u},
\end{equation}
and
\begin{equation}
\Delta_{Y^{*}}^{OO}\left(x, u\right) \leq \dfrac{\theta_{1, 0}^{Y} \cdot \left(1 - u\right) + \theta_{1, 1}^{Y} \cdot u}{\theta_{0, 0}^{S} \cdot \left(1 - u\right) + \theta_{0, 1}^{S} \cdot u} - \dfrac{\theta_{0, 0}^{Y} \cdot \left(1 - u\right) + \theta_{0, 1}^{Y} \cdot u}{\theta_{0, 0}^{S} \cdot \left(1 - u\right) + \theta_{0, 1}^{S} \cdot u}.
\end{equation}

\pagebreak

%%%%%%%%%%%%%%%%%%%%%%%%%%%%%%%%%%%%%%%%%
% Observed only when treated
%%%%%%%%%%%%%%%%%%%%%%%%%%%%%%%%%%%%%%%%%
\setcounter{table}{0}
\renewcommand\thetable{B.\arabic{table}}

\setcounter{figure}{0}
\renewcommand\thefigure{B.\arabic{figure}}

\setcounter{equation}{0}
\renewcommand\theequation{B.\arabic{equation}}

\setcounter{theorem}{0}
\renewcommand\thetheorem{B.\arabic{theorem}}

\section{Bounds for the MTR within the Observed-only-when-treated subpopulation}\label{observedonlywhentreated}
Here, I use the same notation of Section \ref{bounds} and I am interested in the following target parameter: $m_{1}^{NO}\left(x, u\right) \coloneqq \mathbb{E}\left[Y_{1}^{*} \left\vert X = x, U = u, S_{0} = 0, S_{1} = 1 \right.\right]$, which is equal to $\Delta_{Y}^{NO}$ according to equation \eqref{DeltaY_NO}. Following the same steps of the proof of Proposition \ref{boundsY1Proposition}, I can show that:
\begin{corollary}\label{boundsother}
	Suppose that the $m_{0}^{Y}\left(x, u\right)$, $m_{1}^{Y}\left(x, u\right)$, $m_{0}^{S}\left(x, u\right)$ and $\Delta_{S}\left(x, u\right)$ are point identified.
	
	Under assumptions \ref{ind}-\ref{support}, \ref{bounded}.1 and \ref{increasing_sample_selection}, the bounds for $m_{1}^{NO}\left(x, u\right)$ are given by
	\begin{equation}
	\underline{m_{1}^{NO}}\left(x, u\right) \coloneqq \underline{y}^{*} \leq m_{1}^{NO}\left(x, u\right) \leq \dfrac{m_{1}^{Y}\left(x, u\right) - \underline{y}^{*} \cdot m_{0}^{S}\left(x, u\right)}{\Delta_{S}\left(x, u\right)} \eqqcolon \overline{m_{1}^{NO}}\left(x, u\right).
	\end{equation}
	
	Under assumptions \ref{ind}-\ref{support}, \ref{bounded}.2 and \ref{increasing_sample_selection}, the bounds for $m_{1}^{NO}\left(x, u\right)$ are given by
	\begin{equation}
	\underline{m_{1}^{NO}}\left(x, u\right) \coloneqq \dfrac{m_{1}^{Y}\left(x, u\right) - \overline{y}^{*} \cdot m_{0}^{S}\left(x, u\right)}{\Delta_{S}\left(x, u\right)} \leq m_{1}^{NO}\left(x, u\right) \leq \overline{y}^{*} \eqqcolon \overline{m_{1}^{NO}}\left(x, u\right).
	\end{equation}
	
	Under assumptions \ref{ind}-\ref{support}, \ref{bounded}.3 (sub-case (a) or (b)) and \ref{increasing_sample_selection}, the bounds for $m_{1}^{NO}\left(x, u\right)$ are given by
	\begin{equation}
	\underline{m_{1}^{NO}}\left(x, u\right) \coloneqq \dfrac{m_{1}^{Y}\left(x, u\right) - \overline{y}^{*} \cdot m_{0}^{S}\left(x, u\right)}{\Delta_{S}\left(x, u\right)} \leq m_{1}^{NO}\left(x, u\right) \leq \dfrac{m_{1}^{Y}\left(x, u\right) - \underline{y}^{*} \cdot m_{0}^{S}\left(x, u\right)}{\Delta_{S}\left(x, u\right)} \eqqcolon \overline{m_{1}^{NO}}\left(x, u\right).
	\end{equation}
\end{corollary}

Following the same proof of Theorem \ref{sharpbounds} (see Remark 2 at the end of Appendix \ref{proofsharp3}), I can also show that:
\begin{proposition}
	Suppose that the functions $m_{0}^{Y}$, $m_{1}^{Y}$, $m_{0}^{S}$ and $\Delta_{S}$ are point identified at every pair $\left(x, u\right) \in \mathcal{X} \times \left[0, 1\right]$. Under assumptions \ref{ind}-\ref{support}, \ref{bounded} (sub-cases 1, 2, 3(a) or 3(b)) and \ref{increasing_sample_selection}, the bounds $\underline{m_{1}^{NO}}$ and $\overline{m_{1}^{NO}}$, given by Proposition \ref{boundsother}, are pointwise sharp, i.e., for any $\overline{u} \in \left[0, 1\right]$, $\overline{x} \in \mathcal{X}$ and $\gamma\left(\overline{x}, \overline{u}\right) \in \left(\underline{m_{1}^{NO}}\left(\overline{x}, \overline{u}\right), \overline{m_{1}^{NO}}\left(\overline{x}, \overline{u}\right)\right)$, there exist random variables $\left(\tilde{Y}_{0}^{*}, \tilde{Y}_{1}^{*}, \tilde{U}, \tilde{V}\right)$ such that
	\begin{equation}
	\tilde{m}_{1}^{NO}\left(\overline{x}, \overline{u}\right) \coloneqq \mathbb{E}\left[\tilde{Y}_{1}^{*} \left\vert X = \overline{x}, \tilde{U} = \overline{u}, \tilde{S}_{0} = 0, \tilde{S}_{1} = 1 \right.\right] = \gamma\left(\overline{x}, \overline{u}\right),
	\end{equation}
	\begin{equation}
	\mathbb{P}\left[\left. \left(\tilde{Y}_{0}^{*}, \tilde{Y}_{1}^{*}, \tilde{V}\right) \in \mathcal{Y}^{*} \times \mathcal{Y}^{*} \times \left[0, 1\right] \right\vert X = \overline{x}, \tilde{U} = u \right] = 1 \text{ for any } u \in \left[0, 1\right],
	\end{equation}
	and
	\begin{equation}
	F_{\tilde{Y}, \tilde{D}, \tilde{S}, Z, X}\left(y, d, s, z, \overline{x} \right) = F_{Y, D, S, Z, X} \left(y, d, s, z, \overline{x}\right)
	\end{equation}
	for any $\left(y, d, s, z\right) \in \mathbb{R}^{4}$,	where $\tilde{D} \coloneqq \mathbf{1}\left\lbrace P\left(X, Z\right) \geq \tilde{U}\right\rbrace$, $\tilde{S}_{0} = \mathbf{1}\left\lbrace Q\left(0, X\right) \geq \tilde{V}\right\rbrace$, $\tilde{S}_{1} = \mathbf{1}\left\lbrace Q\left(1, X\right) \geq \tilde{V}\right\rbrace$, $\tilde{Y}_{0} = \tilde{S}_{0} \cdot \tilde{Y}_{0}^{*}$, $\tilde{Y}_{1} = \tilde{S}_{1} \cdot \tilde{Y}_{1}^{*}$ and $\tilde{Y} = \tilde{D} \cdot \tilde{Y}_{1} + \left(1 - \tilde{D}\right) \cdot \tilde{Y}_{0}$.
\end{proposition}

Finally, following the same proof of Proposition \ref{partialnecessary}, I can also show that:
\begin{proposition}
	Suppose that the functions $m_{0}^{Y}$, $m_{1}^{Y}$, $m_{0}^{S}$ and $\Delta_{S}$ are point identified at every pair $\left(x, u\right) \in \mathcal{X} \times \left[0, 1\right]$. Impose assumptions \ref{ind}-\ref{support} and \ref{increasing_sample_selection}. If $\mathcal{Y}^{*} = \mathbb{R}$, then, for any $\overline{u} \in \left[0, 1\right]$, $\overline{x} \in \mathcal{X}$ and $\gamma\left(\overline{x}, \overline{u}\right) \in \mathbb{R}$, there exist random variables $\left(\tilde{Y}_{0}^{*}, \tilde{Y}_{1}^{*}, \tilde{U}, \tilde{V}\right)$ such that
	\begin{equation}
	\tilde{m}_{1}^{NO}\left(\overline{x}, \overline{u}\right) \coloneqq \mathbb{E}\left[\tilde{Y}_{1}^{*} \left\vert X = \overline{x}, \tilde{U} = \overline{u}, \tilde{S}_{0} = 0, \tilde{S}_{1} = 1 \right.\right] = \gamma\left(\overline{x}, \overline{u}\right),
	\end{equation}
	\begin{equation}
	\mathbb{P}\left[\left. \left(\tilde{Y}_{0}^{*}, \tilde{Y}_{1}^{*}, \tilde{V}\right) \in \mathcal{Y}^{*} \times \mathcal{Y}^{*} \times \left[0, 1\right] \right\vert X = \overline{x}, \tilde{U} = u \right] = 1 \text{ for any } u \in \left[0, 1\right],
	\end{equation}
	and
	\begin{equation}
	F_{\tilde{Y}, \tilde{D}, \tilde{S}, Z, X}\left(y, d, s, z, \overline{x} \right) = F_{Y, D, S, Z, X} \left(y, d, s, z, \overline{x}\right)
	\end{equation}
	for any $\left(y, d, s, z\right) \in \mathbb{R}^{4}$,	where $\tilde{D} \coloneqq \mathbf{1}\left\lbrace P\left(X, Z\right) \geq \tilde{U}\right\rbrace$, $\tilde{S}_{0} = \mathbf{1}\left\lbrace Q\left(0, X\right) \geq \tilde{V}\right\rbrace$, $\tilde{S}_{1} = \mathbf{1}\left\lbrace Q\left(1, X\right) \geq \tilde{V}\right\rbrace$, $\tilde{Y}_{0} = \tilde{S}_{0} \cdot \tilde{Y}_{0}^{*}$, $\tilde{Y}_{1} = \tilde{S}_{1} \cdot \tilde{Y}_{1}^{*}$ and $\tilde{Y} = \tilde{D} \cdot \tilde{Y}_{1} + \left(1 - \tilde{D}\right) \cdot \tilde{Y}_{0}$.
\end{proposition}

\pagebreak

%%%%%%%%%%%%%%%%%%%%%%%%%%%%%%%%%%%%%%%%%
% Decreasing Sample Selection
%%%%%%%%%%%%%%%%%%%%%%%%%%%%%%%%%%%%%%%%%
\setcounter{table}{0}
\renewcommand\thetable{C.\arabic{table}}

\setcounter{figure}{0}
\renewcommand\thefigure{C.\arabic{figure}}

\setcounter{equation}{0}
\renewcommand\theequation{C.\arabic{equation}}

\setcounter{theorem}{0}
\renewcommand\thetheorem{C.\arabic{theorem}}

\section{Negative Treatment Effect on the Selection Indicator}\label{decreasing_sample_selection}
Even when sample selection is monotone (equation \eqref{selection}), Assumption \ref{increasing_sample_selection} may be invalid in some empirical applications. In particular, it might be the case that the following assumption holds:
\begin{assumption}\label{decreasingassumption}
	Treatment has a negative effect on the sample selection indicator for all individuals, i.e., $Q\left(0, x\right) > Q\left(1, x\right) > 0$ for any $x \in \mathcal{X}$.
\end{assumption}
I stress that this assumption is testable according to \cite{Machado2018}.

With straightforward modifications to the proofs of Corollary \ref{MTEbounds}, Theorem \ref{sharpbounds} and Proposition \ref{partialnecessary} (see the proofs of Propositions \ref{agnosticsharp} and \ref{agnosticnecessary}), I can show that the target parameter in Section \ref{bounds} can be bounded, that its bounds are sharp and that it is impossible to derive bounds for the target parameter with only assumptions \ref{ind}-\ref{support} and \ref{decreasingassumption}. First, I state a result that is analogous to Corollary \ref{MTEbounds}.
\begin{corollary}\label{propositiondecreasing}
	Fix $u \in \left[0, 1\right]$ and $x \in \mathcal{X}$ arbitrarily. Suppose that the $m_{0}^{Y}\left(x, u\right)$, $m_{1}^{Y}\left(x, u\right)$, $m_{0}^{S}\left(x, u\right)$ and $\Delta_{S}\left(x, u\right)$ are point identified.
	
	Under Assumptions \ref{ind}-\ref{support}, \ref{bounded}.1 and \ref{decreasingassumption}, the bounds for $\Delta_{Y^{*}}^{OO}\left(x, u\right)$ are given by
	\begin{equation}
	\Delta_{Y^{*}}^{OO}\left(x, u\right) \geq \dfrac{m_{1}^{Y}\left(x, u\right)}{m_{1}^{S}\left(x, u\right)} - \dfrac{m_{0}^{Y}\left(x, u\right) - \underline{y}^{*} \cdot \left(- \Delta_{S}\left(x, u\right)\right)}{m_{1}^{S}\left(x, u\right)} \eqqcolon \underline{\Lambda_{Y^{*}}^{OO}}\left(x, u\right)
	\end{equation}
	and
	\begin{equation}
	\Delta_{Y^{*}}^{OO}\left(x, u\right) \leq \dfrac{m_{1}^{Y}\left(x, u\right)}{m_{1}^{S}\left(x, u\right)} - \underline{y}^{*}  \eqqcolon \overline{\Lambda_{Y^{*}}^{OO}}\left(x, u\right).
	\end{equation}
	
	Under Assumptions \ref{ind}-\ref{support}, \ref{bounded}.2 and \ref{decreasingassumption}, the bounds for $\Delta_{Y^{*}}^{OO}\left(x, u\right)$ are given by
	\begin{equation}
	\Delta_{Y^{*}}^{OO}\left(x, u\right) \geq \dfrac{m_{1}^{Y}\left(x, u\right)}{m_{1}^{S}\left(x, u\right)} - \overline{y}^{*} \eqqcolon \underline{\Lambda_{Y^{*}}^{OO}}\left(x, u\right)
	\end{equation}
	and
	\begin{equation}
	\Delta_{Y^{*}}^{OO}\left(x, u\right) \leq \dfrac{m_{1}^{Y}\left(x, u\right)}{m_{1}^{S}\left(x, u\right)} - \dfrac{m_{0}^{Y}\left(x, u\right) - \overline{y}^{*} \cdot \left(- \Delta_{S}\left(x, u\right)\right) }{m_{1}^{S}\left(x, u\right)} \eqqcolon \overline{\Lambda_{Y^{*}}^{OO}}\left(x, u\right).
	\end{equation}
	
	Under Assumptions \ref{ind}-\ref{support}, \ref{bounded}.3 (sub-case (a) or (b)) and \ref{decreasingassumption}, the bounds for $\Delta_{Y^{*}}^{OO}\left(x, u\right)$ are given by
	\begin{equation}
	\Delta_{Y^{*}}^{OO}\left(x, u\right) \geq \dfrac{m_{1}^{Y}\left(x, u\right)}{m_{1}^{S}\left(x, u\right)} - \min\left\lbrace \dfrac{m_{0}^{Y}\left(x, u\right) - \underline{y}^{*} \cdot \left(- \Delta_{S}\left(x, u\right)\right)}{m_{1}^{S}\left(x, u\right)}, \overline{y}^{*}\right\rbrace \eqqcolon \underline{\Lambda_{Y^{*}}^{OO}}\left(x, u\right)
	\end{equation}
	and
	\begin{equation}
	\Delta_{Y^{*}}^{OO}\left(x, u\right) \leq \dfrac{m_{1}^{Y}\left(x, u\right)}{m_{1}^{S}\left(x, u\right)} - \max\left\lbrace \dfrac{m_{0}^{Y}\left(x, u\right) - \overline{y}^{*} \cdot \left(- \Delta_{S}\left(x, u\right)\right)}{m_{1}^{S}\left(x, u\right)}, \underline{y}^{*}\right\rbrace \eqqcolon \overline{\Lambda_{Y^{*}}^{OO}}\left(x, u\right).
	\end{equation}
\end{corollary}

Second, I state a result that is analogous to Theorem \ref{sharpbounds}.

\begin{proposition}
	Suppose that the functions $m_{0}^{Y}$, $m_{1}^{Y}$, $m_{0}^{S}$ and $\Delta_{S}$ are point identified at every pair $\left(x, u\right) \in \mathcal{X} \times \left[0, 1\right]$. Under Assumptions \ref{ind}-\ref{support}, \ref{bounded} (sub-cases 1, 2, 3(a) or 3(b)) and \ref{decreasingassumption}, the bounds $\underline{\Lambda_{Y^{*}}^{OO}}$ and $\overline{\Lambda_{Y^{*}}^{OO}}$, given by Proposition \ref{propositiondecreasing}, are pointwise sharp, i.e., for any $\overline{u} \in \left[0, 1\right]$, $\overline{x} \in \mathcal{X}$ and $\delta\left(\overline{x}, \overline{u}\right) \in \left(\underline{\Lambda_{Y^{*}}^{OO}}\left(\overline{x}, \overline{u}\right), \overline{\Lambda_{Y^{*}}^{OO}}\left(\overline{x}, \overline{u}\right)\right)$, there exist random variables $\left(\tilde{Y}_{0}^{*}, \tilde{Y}_{1}^{*}, \tilde{U}, \tilde{V}\right)$ such that
	\begin{equation}
	\Delta_{\tilde{Y}^{*}}^{OO}\left(\overline{x}, \overline{u}\right) \coloneqq \mathbb{E}\left[\tilde{Y}_{1}^{*} - \tilde{Y}_{0}^{*} \left\vert X = \overline{x}, \tilde{U} = \overline{u}, \tilde{S}_{0} = 1, \tilde{S}_{1} = 1 \right.\right] = \delta\left(\overline{x}, \overline{u}\right),
	\end{equation}
	\begin{equation}
	\mathbb{P}\left[\left. \left(\tilde{Y}_{0}^{*}, \tilde{Y}_{1}^{*}, \tilde{V}\right) \in \mathcal{Y}^{*} \times \mathcal{Y}^{*} \times \left[0, 1\right] \right\vert X = \overline{x}, \tilde{U} = u \right] = 1 \text{ for any } u \in \left[0, 1\right],
	\end{equation}
	and
	\begin{equation}
	F_{\tilde{Y}, \tilde{D}, \tilde{S}, Z, X}\left(y, d, s, z, \overline{x} \right) = F_{Y, D, S, Z, X} \left(y, d, s, z, \overline{x}\right)
	\end{equation}
	for any $\left(y, d, s, z\right) \in \mathbb{R}^{4}$,	where $\tilde{D} \coloneqq \mathbf{1}\left\lbrace P\left(X, Z\right) \geq \tilde{U}\right\rbrace$, $\tilde{S}_{0} = \mathbf{1}\left\lbrace Q\left(0, X\right) \geq \tilde{V}\right\rbrace$, $\tilde{S}_{1} = \mathbf{1}\left\lbrace Q\left(1, X\right) \geq \tilde{V}\right\rbrace$, $\tilde{Y}_{0} = \tilde{S}_{0} \cdot \tilde{Y}_{0}^{*}$, $\tilde{Y}_{1} = \tilde{S}_{1} \cdot \tilde{Y}_{1}^{*}$ and $\tilde{Y} = \tilde{D} \cdot \tilde{Y}_{1} + \left(1 - \tilde{D}\right) \cdot \tilde{Y}_{0}$.
\end{proposition}

Finally, I state a result that is analogous to Proposition \ref{partialnecessary}.

\begin{proposition}
	Suppose that the functions $m_{0}^{Y}$, $m_{1}^{Y}$, $m_{0}^{S}$ and $\Delta_{S}$ are point identified at every pair $\left(x, u\right) \in \mathcal{X} \times \left[0, 1\right]$. Impose Assumptions \ref{ind}-\ref{support} and \ref{decreasingassumption}. If $\mathcal{Y}^{*} = \mathbb{R}$, then, for any $\overline{u} \in \left[0, 1\right]$, $\overline{x} \in \mathcal{X}$ and $\delta\left(\overline{x}, \overline{u}\right) \in \mathbb{R}$, there exist random variables $\left(\tilde{Y}_{0}^{*}, \tilde{Y}_{1}^{*}, \tilde{U}, \tilde{V}\right)$ such that
	\begin{equation}
	\Delta_{\tilde{Y}^{*}}^{OO}\left(\overline{x}, \overline{u}\right) \coloneqq \mathbb{E}\left[\tilde{Y}_{1}^{*} - \tilde{Y}_{0}^{*} \left\vert X = \overline{x}, \tilde{U} = \overline{u}, \tilde{S}_{0} = 1, \tilde{S}_{1} = 1 \right.\right] = \delta\left(\overline{x}, \overline{u}\right),
	\end{equation}
	\begin{equation}
	\mathbb{P}\left[\left. \left(\tilde{Y}_{0}^{*}, \tilde{Y}_{1}^{*}, \tilde{V}\right) \in \mathcal{Y}^{*} \times \mathcal{Y}^{*} \times \left[0, 1\right] \right\vert X = \overline{x}, \tilde{U} = u \right] = 1 \text{ for any } u \in \left[0, 1\right],
	\end{equation}
	and
	\begin{equation}
	F_{\tilde{Y}, \tilde{D}, \tilde{S}, Z, X}\left(y, d, s, z, \overline{x} \right) = F_{Y, D, S, Z, X} \left(y, d, s, z, \overline{x}\right)
	\end{equation}
	for any $\left(y, d, s, z\right) \in \mathbb{R}^{4}$,	where $\tilde{D} \coloneqq \mathbf{1}\left\lbrace P\left(X, Z\right) \geq \tilde{U}\right\rbrace$, $\tilde{S}_{0} = \mathbf{1}\left\lbrace Q\left(0, X\right) \geq \tilde{V}\right\rbrace$, $\tilde{S}_{1} = \mathbf{1}\left\lbrace Q\left(1, X\right) \geq \tilde{V}\right\rbrace$, $\tilde{Y}_{0} = \tilde{S}_{0} \cdot \tilde{Y}_{0}^{*}$, $\tilde{Y}_{1} = \tilde{S}_{1} \cdot \tilde{Y}_{1}^{*}$ and $\tilde{Y} = \tilde{D} \cdot \tilde{Y}_{1} + \left(1 - \tilde{D}\right) \cdot \tilde{Y}_{0}$.
\end{proposition}

\pagebreak

%%%%%%%%%%%%%%%%%%%%%%%%%%%%%%%%%%%%%%%%%
% Monotone Sample Selection
%%%%%%%%%%%%%%%%%%%%%%%%%%%%%%%%%%%%%%%%%
\setcounter{table}{0}
\renewcommand\thetable{D.\arabic{table}}

\setcounter{figure}{0}
\renewcommand\thefigure{D.\arabic{figure}}

\setcounter{equation}{0}
\renewcommand\theequation{D.\arabic{equation}}

\setcounter{theorem}{0}
\renewcommand\thetheorem{D.\arabic{theorem}}

\section{Monotone Sample Selection}\label{agnostic}
Depending on the results of the test proposed by \cite{Machado2018}, a researcher may want to be agnostic about the direction of the monotone selection problem and impose only equation \eqref{selection}, while ruling out uninteresting cases. In this situation, it is reasonable to assume:
\begin{assumption}\label{agnosticassumption}
	Treatment has a monotone effect on the sample selection indicator for all individuals, i.e., \textbf{either} (i) $Q\left(1, x\right) > Q\left(0, x\right) > 0$ for any $x \in \mathcal{X}$ \textbf{or} (ii) $Q\left(0, x\right) > Q\left(1, x\right) > 0$ for any $x \in \mathcal{X}$.
\end{assumption}
Note that Assumption \ref{agnosticassumption} only strengthens equation \eqref{selection} by ruling out the theoretically uninteresting cases mentioned after Assumption \eqref{increasing_sample_selection}.

By combining Corollaries \ref{MTEbounds} and \ref{propositiondecreasing}, I find that:
\begin{corollary}\label{agnosticprop}
Fix $u \in \left[0, 1\right]$ and $x \in \mathcal{X}$ arbitrarily. Suppose that the $m_{0}^{Y}\left(x, u\right)$, $m_{1}^{Y}\left(x, u\right)$, $m_{0}^{S}\left(x, u\right)$ and $\Delta_{S}\left(x, u\right)$ are point identified. Under Assumptions \ref{ind}-\ref{support}, \ref{bounded} and \ref{agnosticassumption}, the bounds for $\Delta_{Y^{*}}^{OO}\left(x, u\right)$ are given by
\begin{align}
\underline{\Upsilon_{Y^{*}}^{OO}}\left(x, u\right) & \coloneqq \min \left\lbrace \underline{\Delta_{Y^{*}}^{OO}}\left(x, u\right), \underline{\Lambda_{Y^{*}}^{OO}}\left(x, u\right) \right\rbrace \nonumber \\
& \leq \Delta_{Y^{*}}^{OO}\left(x, u\right) \\
& \leq \max \left\lbrace \overline{\Delta_{Y^{*}}^{OO}}\left(x, u\right), \overline{\Lambda_{Y^{*}}^{OO}}\left(x, u\right) \right\rbrace \eqqcolon \overline{\Upsilon_{Y^{*}}^{OO}}\left(x, u\right) \nonumber
\end{align}
\end{corollary}

Moreover, these bounds are also pointwise sharp:\footnote{The proof of propositions \ref{agnosticsharp} and \ref{agnosticnecessary} are located at the end of Appendix \ref{agnostic}.}
\begin{proposition}\label{agnosticsharp}
Suppose that the functions $m_{0}^{Y}$, $m_{1}^{Y}$, $m_{0}^{S}$ and $\Delta_{S}$ are point identified at every pair $\left(x, u\right) \in \mathcal{X} \times \left[0, 1\right]$. Under Assumptions \ref{ind}-\ref{support}, \ref{bounded} (sub-cases 1, 2, 3(a) or 3(b)) and \ref{agnosticassumption}, the bounds $\underline{\Upsilon_{Y^{*}}^{OO}}$ and $\overline{\Upsilon_{Y^{*}}^{OO}}$, given by Corollary \ref{agnosticprop}, are pointwise sharp, i.e., for any $\overline{u} \in \left[0, 1\right]$, $\overline{x} \in \mathcal{X}$ and $\delta\left(\overline{x}, \overline{u}\right) \in \left(\underline{\Upsilon_{Y^{*}}^{OO}}\left(\overline{x}, \overline{u}\right), \overline{\Upsilon_{Y^{*}}^{OO}}\left(\overline{x}, \overline{u}\right)\right)$, there exist random variables $\left(\tilde{Y}_{0}^{*}, \tilde{Y}_{1}^{*}, \tilde{U}, \tilde{V}\right)$ such that
\begin{equation}\label{agnosticfaketarget}
\Delta_{\tilde{Y}^{*}}^{OO}\left(\overline{x}, \overline{u}\right) \coloneqq \mathbb{E}\left[\tilde{Y}_{1}^{*} - \tilde{Y}_{0}^{*} \left\vert X = \overline{x}, \tilde{U} = \overline{u}, \tilde{S}_{0} = 1, \tilde{S}_{1} = 1 \right.\right] = \delta\left(\overline{x}, \overline{u}\right),
\end{equation}
\begin{equation}\label{agnosticcorrectsupport}
\mathbb{P}\left[\left. \left(\tilde{Y}_{0}^{*}, \tilde{Y}_{1}^{*}, \tilde{V}\right) \in \mathcal{Y}^{*} \times \mathcal{Y}^{*} \times \left[0, 1\right] \right\vert X = \overline{x}, \tilde{U} = u \right] = 1 \text{ for any } u \in \left[0, 1\right],
\end{equation}
and
\begin{equation}\label{agnosticDataRestriction}
F_{\tilde{Y}, \tilde{D}, \tilde{S}, Z, X}\left(y, d, s, z, \overline{x} \right) = F_{Y, D, S, Z, X} \left(y, d, s, z, \overline{x}\right)
\end{equation}
for any $\left(y, d, s, z\right) \in \mathbb{R}^{4}$,	where $\tilde{D} \coloneqq \mathbf{1}\left\lbrace P\left(X, Z\right) \geq \tilde{U}\right\rbrace$, $\tilde{S}_{0} = \mathbf{1}\left\lbrace Q\left(0, X\right) \geq \tilde{V}\right\rbrace$, $\tilde{S}_{1} = \mathbf{1}\left\lbrace Q\left(1, X\right) \geq \tilde{V}\right\rbrace$, $\tilde{Y}_{0} = \tilde{S}_{0} \cdot \tilde{Y}_{0}^{*}$, $\tilde{Y}_{1} = \tilde{S}_{1} \cdot \tilde{Y}_{1}^{*}$ and $\tilde{Y} = \tilde{D} \cdot \tilde{Y}_{1} + \left(1 - \tilde{D}\right) \cdot \tilde{Y}_{0}$.
\end{proposition}

Finally, I state an impossibility result that is analogous to Proposition \ref{partialnecessary}.

\begin{proposition}\label{agnosticnecessary}
	Suppose that the functions $m_{0}^{Y}$, $m_{1}^{Y}$, $m_{0}^{S}$ and $\Delta_{S}$ are point identified at every pair $\left(x, u\right) \in \mathcal{X} \times \left[0, 1\right]$. Impose assumptions \ref{ind}-\ref{support} and \ref{agnosticassumption}. If $\mathcal{Y}^{*} = \mathbb{R}$, then, for any $\overline{u} \in \left[0, 1\right]$, $\overline{x} \in \mathcal{X}$ and $\delta\left(\overline{x}, \overline{u}\right) \in \mathbb{R}$, there exist random variables $\left(\tilde{Y}_{0}^{*}, \tilde{Y}_{1}^{*}, \tilde{U}, \tilde{V}\right)$ such that
	\begin{equation}\label{agnosticfaketargetP}
	\Delta_{\tilde{Y}^{*}}^{OO}\left(\overline{x}, \overline{u}\right) \coloneqq \mathbb{E}\left[\tilde{Y}_{1}^{*} - \tilde{Y}_{0}^{*} \left\vert X = \overline{x}, \tilde{U} = \overline{u}, \tilde{S}_{0} = 1, \tilde{S}_{1} = 1 \right.\right] = \delta\left(\overline{x}, \overline{u}\right),
	\end{equation}
	\begin{equation}\label{agnosticcorrectsupportP}
	\mathbb{P}\left[\left. \left(\tilde{Y}_{0}^{*}, \tilde{Y}_{1}^{*}, \tilde{V}\right) \in \mathcal{Y}^{*} \times \mathcal{Y}^{*} \times \left[0, 1\right] \right\vert X = \overline{x}, \tilde{U} = u \right] = 1 \text{ for any } u \in \left[0, 1\right],
	\end{equation}
	and
	\begin{equation}\label{agnosticDataRestrictionP}
	F_{\tilde{Y}, \tilde{D}, \tilde{S}, Z, X}\left(y, d, s, z, \overline{x} \right) = F_{Y, D, S, Z, X} \left(y, d, s, z, \overline{x}\right)
	\end{equation}
	for any $\left(y, d, s, z\right) \in \mathbb{R}^{4}$,	where $\tilde{D} \coloneqq \mathbf{1}\left\lbrace P\left(X, Z\right) \geq \tilde{U}\right\rbrace$, $\tilde{S}_{0} = \mathbf{1}\left\lbrace Q\left(0, X\right) \geq \tilde{V}\right\rbrace$, $\tilde{S}_{1} = \mathbf{1}\left\lbrace Q\left(1, X\right) \geq \tilde{V}\right\rbrace$, $\tilde{Y}_{0} = \tilde{S}_{0} \cdot \tilde{Y}_{0}^{*}$, $\tilde{Y}_{1} = \tilde{S}_{1} \cdot \tilde{Y}_{1}^{*}$ and $\tilde{Y} = \tilde{D} \cdot \tilde{Y}_{1} + \left(1 - \tilde{D}\right) \cdot \tilde{Y}_{0}$.
\end{proposition}

\begin{proof}[Proof of Proposition \ref{agnosticsharp}]
I only prove Proposition \ref{agnosticsharp} under Assumption \ref{bounded}.3 (sub-cases (a) and (b)).The proofs of Proposition \ref{agnosticsharp} under assumptions \ref{bounded}.1 and \ref{bounded}.2 are trivial modifications of the proof presented below.

Fix $\overline{u} \in \left[0, 1\right]$, $\overline{x} \in \mathcal{X}$ and $\delta\left(\overline{x}, \overline{u}\right) \in \left(\underline{\Upsilon_{Y^{*}}^{OO}}\left(\overline{x}, \overline{u}\right), \overline{\Upsilon_{Y^{*}}^{OO}}\left(\overline{x}, \overline{u}\right)\right)$ arbitrarily. For brevity, define
\begin{align*}
\alpha\left(\overline{x}, \overline{u}\right) & \coloneqq \mathbf{1}\left\lbrace Q\left(1, x\right) > Q\left(0, x\right) \right\rbrace \cdot \left(\delta\left(\overline{x}, \overline{u}\right) + \dfrac{m_{0}^{Y}\left(\overline{x}, \overline{u}\right)}{m_{0}^{S}\left(\overline{x}, \overline{u}\right)}\right) \\
& \hspace{30pt} + \mathbf{1}\left\lbrace Q\left(1, x\right) < Q\left(0, x\right) \right\rbrace \cdot \left(- \delta\left(\overline{x}, \overline{u}\right) + \dfrac{m_{1}^{Y}\left(\overline{x}, \overline{u}\right)}{m_{1}^{S}\left(\overline{x}, \overline{u}\right)}\right),
\end{align*}
\begin{align*}
\gamma\left(\overline{x}, \overline{u}\right) & \coloneqq \mathbf{1}\left\lbrace Q\left(1, x\right) > Q\left(0, x\right) \right\rbrace \cdot \left(\dfrac{m_{1}^{Y}\left(\overline{x}, \overline{u}\right) - \alpha\left(\overline{x}, \overline{u}\right) \cdot m_{0}^{S}\left(\overline{x}, \overline{u}\right)}{\Delta_{S}\left(\overline{x}, \overline{u}\right)}\right) \\
& \hspace{30pt} + \mathbf{1}\left\lbrace Q\left(1, x\right) < Q\left(0, x\right) \right\rbrace \cdot \left(\dfrac{m_{0}^{Y}\left(\overline{x}, \overline{u}\right) - \alpha\left(\overline{x}, \overline{u}\right) \cdot m_{1}^{S}\left(\overline{x}, \overline{u}\right)}{-\Delta_{S}\left(\overline{x}, \overline{u}\right)}\right),
\end{align*}
\begin{equation*}
\underline{Q}\left(x\right) = \min \left\lbrace Q\left(0, x\right), Q\left(1, x\right) \right\rbrace,
\end{equation*}
\begin{equation*}
\overline{Q}\left(x\right) = \max \left\lbrace Q\left(0, x\right), Q\left(1, x\right) \right\rbrace,
\end{equation*}
\begin{equation*}
\underline{m}^{S}\left(x, \overline{u}\right) = \min \left\lbrace m_{0}^{S}\left(x, \overline{u}\right), m_{1}^{S}\left(x, \overline{u}\right) \right\rbrace \text{ for any } x \in \mathcal{X},
\end{equation*}
and
\begin{equation*}
\overline{m}^{S}\left(x, \overline{u}\right) = \max \left\lbrace m_{0}^{S}\left(x, \overline{u}\right), m_{1}^{S}\left(x, \overline{u}\right) \right\rbrace \text{ for any } x \in \mathcal{X}.
\end{equation*}

Note that
\begin{equation}\label{sanity1A}
\alpha\left(\overline{x}, \overline{u}\right) \in \left(\underline{y}^{*}, \overline{y}^{*}\right),
\end{equation}
and that
\begin{equation}\label{sanity2A}
\gamma\left(\overline{x}, \overline{u}\right) \in \left(\underline{y}^{*}, \overline{y}^{*}\right).
\end{equation}

The strategy of this proof consists of defining candidate random variables $\left(\tilde{Y}_{0}^{*}, \tilde{Y}_{1}^{*}, \tilde{U}, \tilde{V}\right)$ through their joint cumulative distribution function $F_{\tilde{Y}_{0}^{*}, \tilde{Y}_{1}^{*}, \tilde{U}, \tilde{V}, Z, X}$ and then checking that equations \eqref{agnosticfaketarget}, \eqref{agnosticcorrectsupport} and \eqref{agnosticDataRestriction} are satisfied. I fix $\left(y_{0}, y_{1}, u, v, z, x\right) \in \mathbb{R}^{6}$ and define $F_{\tilde{Y}_{0}^{*}, \tilde{Y}_{1}^{*}, \tilde{U}, \tilde{V}, Z, X}$ in twelve steps:
\begin{enumerate}
	\item[Step 1.] For $x \notin \mathcal{X}$, $F_{\tilde{Y}_{0}^{*}, \tilde{Y}_{1}^{*}, \tilde{U}, \tilde{V}, Z, X}\left(y_{0}, y_{1}, u, v, z, x\right) = F_{Y_{0}^{*}, Y_{1}^{*}, U, V, Z, X}\left(y_{0}, y_{1}, u, v, z, x\right)$.
	
	\item[Step 2.] From now on, consider $x \in \mathcal{X}$. Since $$F_{\tilde{Y}_{0}^{*}, \tilde{Y}_{1}^{*}, \tilde{U}, \tilde{V}, Z, X}\left(y_{0}, y_{1}, u, v, z, x\right) = F_{\tilde{Y}_{0}^{*}, \tilde{Y}_{1}^{*}, \tilde{U}, \tilde{V}, Z \left\vert X \right.}\left(y_{0}, y_{1}, u, v, z \left\vert x \right.\right) \cdot F_{X}\left(x\right),$$ it suffices to define $F_{\tilde{Y}_{0}^{*}, \tilde{Y}_{1}^{*}, \tilde{U}, \tilde{V}, Z \left\vert X \right.}\left(y_{0}, y_{1}, u, v, z \left\vert x \right.\right)$. Moreover, I impose $$\left. Z \independent \left(\tilde{Y}_{0}^{*}, \tilde{Y}_{1}^{*}, \tilde{U}, \tilde{V} \right) \right\vert X$$ by writing $$F_{\tilde{Y}_{0}^{*}, \tilde{Y}_{1}^{*}, \tilde{U}, \tilde{V}, Z \left\vert X \right.}\left(y_{0}, y_{1}, u, v, z \left\vert x \right.\right) = F_{\tilde{Y}_{0}^{*}, \tilde{Y}_{1}^{*}, \tilde{U}, \tilde{V}\left\vert X \right.}\left(y_{0}, y_{1}, u, v \left\vert x \right.\right) \cdot F_{Z \left\vert X \right.}\left(z \left\vert x \right.\right),$$ implying that it is sufficient to define $F_{\tilde{Y}_{0}^{*}, \tilde{Y}_{1}^{*}, \tilde{U}, \tilde{V}\left\vert X \right.}\left(y_{0}, y_{1}, u, v \left\vert x \right.\right)$.
	
	\item[Step 3.] For $u \notin \left[0, 1\right]$, I define $F_{\tilde{Y}_{0}^{*}, \tilde{Y}_{1}^{*}, \tilde{U}, \tilde{V}\left\vert X \right.}\left(y_{0}, y_{1}, u, v \left\vert x \right.\right) = F_{Y_{0}^{*}, Y_{1}^{*}, U, V\left\vert X \right.}\left(y_{0}, y_{1}, u, v \left\vert x \right.\right)$.
	
	\item[Step 4.] From now on, consider $u \in \left[0, 1\right]$. Since $$F_{\tilde{Y}_{0}^{*}, \tilde{Y}_{1}^{*}, \tilde{U}, \tilde{V}\left\vert X \right.}\left(y_{0}, y_{1}, u, v \left\vert x \right.\right) = F_{\tilde{Y}_{0}^{*}, \tilde{Y}_{1}^{*}, \tilde{V}\left\vert X, \tilde{U} \right.}\left(y_{0}, y_{1}, v \left\vert x, u \right.\right) \cdot F_{\tilde{U}\left\vert X \right.}\left(u \left\vert x \right.\right),$$ it suffices to define $F_{\tilde{Y}_{0}^{*}, \tilde{Y}_{1}^{*}, \tilde{V}\left\vert X, \tilde{U} \right.}\left(y_{0}, y_{1}, v \left\vert x, u \right.\right)$ and $F_{\tilde{U}\left\vert X \right.}\left(u \left\vert x \right.\right)$.
	
	\item[Step 5.] I define $F_{\tilde{U}\left\vert X \right.}\left(u \left\vert x \right.\right) = F_{U\left\vert X \right.}\left(u \left\vert x \right.\right) = u$.
	
	\item[Step 6.] For any $u \neq \overline{u}$, I define $F_{\tilde{Y}_{0}^{*}, \tilde{Y}_{1}^{*}, \tilde{V}\left\vert X, \tilde{U} \right.}\left(y_{0}, y_{1}, v \left\vert x, u\right.\right) = F_{Y_{0}^{*}, Y_{1}^{*}, V\left\vert X, U \right.}\left(y_{0}, y_{1}, v \left\vert x, u\right.\right)$.
	
	\item[Step 7.] For any $v \notin \left[0, 1\right]$, I define $F_{\tilde{Y}_{0}^{*}, \tilde{Y}_{1}^{*}, \tilde{V}\left\vert X, \tilde{U} \right.}\left(y_{0}, y_{1}, v \left\vert x, \overline{u}\right.\right) = F_{Y_{0}^{*}, Y_{1}^{*}, V\left\vert X, U \right.}\left(y_{0}, y_{1}, v \left\vert x, \overline{u}\right.\right)$.
	
	\item[Step 8.] From now on, assume that $v \in \left[0, 1\right]$. Since $$F_{\tilde{Y}_{0}^{*}, \tilde{Y}_{1}^{*}, \tilde{V}\left\vert X, \tilde{U} \right.}\left(y_{0}, y_{1}, v \left\vert x, \overline{u}\right.\right) = F_{\tilde{Y}_{0}^{*}, \tilde{Y}_{1}^{*}\left\vert X, \tilde{U}, \tilde{V} \right.}\left(y_{0}, y_{1}\left\vert x, \overline{u}, v\right.\right) \cdot F_{\tilde{V}\left\vert X, \tilde{U} \right.}\left(v \left\vert x, \overline{u}\right.\right),$$ it is sufficient to define $F_{\tilde{Y}_{0}^{*}, \tilde{Y}_{1}^{*}\left\vert X, \tilde{U}, \tilde{V} \right.}\left(y_{0}, y_{1}\left\vert x, \overline{u}, v\right.\right)$ and $F_{\tilde{V}\left\vert X, \tilde{U} \right.}\left(v \left\vert x, \overline{u}\right.\right)$.
	
	\item[Step 9.] I define
	$$
	F_{\tilde{V}\left\vert X, \tilde{U} \right.}\left(v \left\vert x, \overline{u}\right.\right) = \left\lbrace
	\begin{array}{cl}
	\underline{m}^{S}\left(x, \overline{u}\right) \cdot \dfrac{v}{\underline{Q}\left(x\right)} & \text{if } v \leq \underline{Q}\left(x\right) \\
	& \\
	\underline{m}^{S}\left(x, \overline{u}\right) + \left(\overline{m}^{S}\left(x, \overline{u}\right) - \underline{m}^{S}\left(x, \overline{u}\right)\right) \cdot \dfrac{v - \underline{Q}\left(x\right)}{\overline{Q}\left(x\right) - \underline{Q}\left(x\right)} & \text{if } \underline{Q}\left(x\right) < v \leq \overline{Q}\left(x\right) \\
	& \\
	\overline{m}^{S}\left(x, \overline{u}\right) + \left(1 - \overline{m}^{S}\left(x, \overline{u}\right)\right)\dfrac{v - \overline{Q}\left(x\right)}{1 - \overline{Q}\left(x\right)} & \text{if } \overline{Q}\left(x\right) < v
	\end{array}
	\right..
	$$
	
	\item[Step 10.] I write $F_{\tilde{Y}_{0}^{*}, \tilde{Y}_{1}^{*}\left\vert X, \tilde{U}, \tilde{V} \right.}\left(y_{0}, y_{1}\left\vert x, \overline{u}, v\right.\right) = F_{\tilde{Y}_{0}^{*}\left\vert X, \tilde{U}, \tilde{V} \right.}\left(y_{0}\left\vert x, \overline{u}, v\right.\right) \cdot F_{ \tilde{Y}_{1}^{*}\left\vert X, \tilde{U}, \tilde{V} \right.}\left(y_{1}\left\vert x, \overline{u}, v\right.\right)$, implying that I can separately define $F_{\tilde{Y}_{0}^{*}\left\vert X, \tilde{U}, \tilde{V} \right.}\left(y_{0}\left\vert x, \overline{u}, v\right.\right)$ and $F_{ \tilde{Y}_{1}^{*}\left\vert X, \tilde{U}, \tilde{V} \right.}\left(y_{1}\left\vert x, \overline{u}, v\right.\right)$.
	
	\item[Step 11.] When $Q\left(1, x\right) > Q\left(0, x\right)$ and $\mathcal{Y}^{*}$ is a bounded interval (sub-case (a) in Assumption \ref{bounded}.3), I define
	$$
	F_{\tilde{Y}_{0}^{*}\left\vert X, \tilde{U}, \tilde{V} \right.}\left(y_{0}\left\vert x, \overline{u}, v\right.\right) = \left\lbrace
	\begin{array}{cl}
	\mathbf{1}\left\lbrace y_{0} \geq \dfrac{m_{0}^{Y}\left(\overline{x}, \overline{u}\right)}{m_{0}^{S}\left(\overline{x}, \overline{u}\right)} \right\rbrace & \text{if } v \leq \underline{Q}\left(x\right) \\
	---------- & ------- \\
	\mathbf{1}\left\lbrace y_{0} \geq \dfrac{\underline{y}^{*} + \overline{y}^{*}}{2} \right\rbrace & \text{if } \underline{Q}\left(x\right) < v
	\end{array}
	\right..
	$$
	
	When $Q\left(1, x\right) > Q\left(0, x\right)$ and $\overline{y}^{*} = \max \left\lbrace y \in \mathcal{Y}^{*} \right\rbrace$ and $\underline{y}^{*} = \min \left\lbrace y \in \mathcal{Y}^{*} \right\rbrace$ (sub-case (b) in Assumption \ref{bounded}.3), I define
	$$
	F_{\tilde{Y}_{0}^{*}\left\vert X, \tilde{U}, \tilde{V} \right.}\left(y_{0}\left\vert x, \overline{u}, v\right.\right) = \left\lbrace
	\begin{array}{cl}
	0 & \text{if } y_{0} < \underline{y}^{*} \text{ and } v \leq \underline{Q}\left(x\right) \\
	& \\
	1 - \dfrac{\dfrac{m_{0}^{Y}\left(\overline{x}, \overline{u}\right)}{m_{0}^{S}\left(\overline{x}, \overline{u}\right)} - \underline{y}^{*}}{\overline{y}^{*} - \underline{y}^{*}} & \text{if } \underline{y}^{*} \leq y_{0} < \overline{y}^{*} \text{ and } v \leq \underline{Q}\left(x\right) \\
	& \\
	1 & \text{if } \overline{y}^{*} \leq y_{0} \text{ and } v \leq \underline{Q}\left(x\right) \\
	---------- & -------------- \\
	\mathbf{1}\left\lbrace y_{0} \geq \overline{y}^{*} \right\rbrace & \text{if } \underline{Q}\left(x\right) < v
	\end{array}
	\right..
	$$
	which are valid cumulative distribution functions because $\dfrac{m_{0}^{Y}\left(\overline{x}, \overline{u}\right)}{m_{0}^{S}\left(\overline{x}, \overline{u}\right)} \in \left[\underline{y}^{*}, \overline{y}^{*}\right]$.
	
	When $Q\left(1, x\right) < Q\left(0, x\right)$ and $\mathcal{Y}^{*}$ is a bounded interval (sub-case (a) in Assumption \ref{bounded}.3), I define
	$$
	F_{\tilde{Y}_{0}^{*}\left\vert X, \tilde{U}, \tilde{V} \right.}\left(y_{0}\left\vert x, \overline{u}, v\right.\right) = \left\lbrace
	\begin{array}{cl}
	\mathbf{1}\left\lbrace y_{0} \geq \alpha\left(\overline{x}, \overline{u}\right) \right\rbrace & \text{if } v \leq \underline{Q}\left(x\right) \\
	-------- & ----------- \\
	\mathbf{1}\left\lbrace y_{0} \geq \gamma\left(\overline{x}, \overline{u}\right) \right\rbrace & \text{if } \underline{Q}\left(x\right) < v \leq \overline{Q}\left(x\right) \\
	-------- & ----------- \\
	\mathbf{1}\left\lbrace y_{0} \geq \dfrac{\underline{y}^{*} + \overline{y}^{*}}{2} \right\rbrace & \text{if } \overline{Q}\left(x\right) < v
	\end{array}
	\right..
	$$
	
	When $Q\left(1, x\right) < Q\left(0, x\right)$ and $\overline{y}^{*} = \max \left\lbrace y \in \mathcal{Y}^{*} \right\rbrace$ and $\underline{y}^{*} = \min \left\lbrace y \in \mathcal{Y}^{*} \right\rbrace$ (sub-case (b) in Assumption \ref{bounded}.3), I define
	$$
	F_{\tilde{Y}_{0}^{*}\left\vert X, \tilde{U}, \tilde{V} \right.}\left(y_{0}\left\vert x, \overline{u}, v\right.\right) = \left\lbrace
	\begin{array}{cl}
	0 & \text{if } y_{0} < \underline{y}^{*} \text{ and } v \leq \underline{Q}\left(x\right) \\
	& \\
	1 - \dfrac{\alpha\left(\overline{x}, \overline{u}\right) - \underline{y}^{*}}{\overline{y}^{*} - \underline{y}^{*}} & \text{if } \underline{y}^{*} \leq y_{0} < \overline{y}^{*} \text{ and } v \leq \underline{Q}\left(x\right) \\
	& \\
	1 & \text{if } \overline{y}^{*} \leq y_{0} \text{ and } v \leq \underline{Q}\left(x\right) \\
	-------- & ------------------ \\
	0 & \text{if } y_{0} < \underline{y}^{*} \text{ and } \underline{Q}\left(x\right) < v \leq \overline{Q}\left(x\right) \\
	& \\
	1 - \dfrac{\gamma\left(\overline{x}, \overline{u}\right) - \underline{y}^{*}}{\overline{y}^{*} - \underline{y}^{*}} & \text{if } \underline{y}^{*} \leq y_{0} < \overline{y}^{*} \text{ and } \underline{Q}\left(x\right) < v \leq \overline{Q}\left(x\right) \\
	& \\
	1 & \text{if } \overline{y}^{*} \leq y_{0} \text{ and } \underline{Q}\left(x\right) < v \leq \overline{Q}\left(x\right) \\
	-------- & ------------------ \\
	\mathbf{1}\left\lbrace y_{0} \geq \overline{y}^{*} \right\rbrace & \text{if } \overline{Q}\left(x\right) < v
	\end{array}
	\right..
	$$
	which are valid cumulative distribution functions because of equations \eqref{sanity1A} and \eqref{sanity2A}.
	
	\item[Step 12.] When $Q\left(1, x\right) > Q\left(0, x\right)$ and $\mathcal{Y}^{*}$ is a bounded interval (sub-case (a) in Assumption \ref{bounded}.3), I define
	$$
	F_{\tilde{Y}_{1}^{*}\left\vert X, \tilde{U}, \tilde{V} \right.}\left(y_{1}\left\vert x, \overline{u}, v\right.\right) = \left\lbrace
	\begin{array}{cl}
	\mathbf{1}\left\lbrace y_{1} \geq \alpha\left(\overline{x}, \overline{u}\right) \right\rbrace & \text{if } v \leq \underline{Q}\left(x\right) \\
	-------- & ----------- \\
	\mathbf{1}\left\lbrace y_{1} \geq \gamma\left(\overline{x}, \overline{u}\right) \right\rbrace & \text{if } \underline{Q}\left(x\right) < v \leq \overline{Q}\left(x\right) \\
	-------- & ----------- \\
	\mathbf{1}\left\lbrace y_{1} \geq \dfrac{\underline{y}^{*} + \overline{y}^{*}}{2} \right\rbrace & \text{if } \overline{Q}\left(x\right) < v
	\end{array}
	\right..
	$$
	
	When $Q\left(1, x\right) > Q\left(0, x\right)$ and $\overline{y}^{*} = \max \left\lbrace y \in \mathcal{Y}^{*} \right\rbrace$ and $\underline{y}^{*} = \min \left\lbrace y \in \mathcal{Y}^{*} \right\rbrace$ (sub-case (b) in Assumption \ref{bounded}.3), I define
	$$
	F_{\tilde{Y}_{1}^{*}\left\vert X, \tilde{U}, \tilde{V} \right.}\left(y_{1}\left\vert x, \overline{u}, v\right.\right) = \left\lbrace
	\begin{array}{cl}
	0 & \text{if } y_{1} < \underline{y}^{*} \text{ and } v \leq \underline{Q}\left(x\right) \\
	& \\
	1 - \dfrac{\alpha\left(\overline{x}, \overline{u}\right) - \underline{y}^{*}}{\overline{y}^{*} - \underline{y}^{*}} & \text{if } \underline{y}^{*} \leq y_{1} < \overline{y}^{*} \text{ and } v \leq \underline{Q}\left(x\right) \\
	& \\
	1 & \text{if } \overline{y}^{*} \leq y_{1} \text{ and } v \leq \underline{Q}\left(x\right) \\
	-------- & ------------------ \\
	0 & \text{if } y_{1} < \underline{y}^{*} \text{ and } \underline{Q}\left(x\right) < v \leq \overline{Q}\left(x\right) \\
	& \\
	1 - \dfrac{\gamma\left(\overline{x}, \overline{u}\right) - \underline{y}^{*}}{\overline{y}^{*} - \underline{y}^{*}} & \text{if } \underline{y}^{*} \leq y_{1} < \overline{y}^{*} \text{ and } \underline{Q}\left(x\right) < v \leq \overline{Q}\left(x\right) \\
	& \\
	1 & \text{if } \overline{y}^{*} \leq y_{1} \text{ and } \underline{Q}\left(x\right) < v \leq \overline{Q}\left(x\right) \\
	-------- & ------------------ \\
	\mathbf{1}\left\lbrace y_{1} \geq \overline{y}^{*} \right\rbrace & \text{if } \overline{Q}\left(x\right) < v
	\end{array}
	\right..
	$$
	which are valid cumulative distribution functions because of equations \eqref{sanity1} and \eqref{sanity2}.
	
	When $Q\left(1, x\right) < Q\left(0, x\right)$ and $\mathcal{Y}^{*}$ is a bounded interval (sub-case (a) in Assumption \ref{bounded}.3), I define
	$$
	F_{\tilde{Y}_{1}^{*}\left\vert X, \tilde{U}, \tilde{V} \right.}\left(y_{1}\left\vert x, \overline{u}, v\right.\right) = \left\lbrace
	\begin{array}{cl}
	\mathbf{1}\left\lbrace y_{1} \geq \dfrac{m_{1}^{Y}\left(\overline{x}, \overline{u}\right)}{m_{1}^{S}\left(\overline{x}, \overline{u}\right)} \right\rbrace & \text{if } v \leq \underline{Q}\left(x\right) \\
	---------- & ------- \\
	\mathbf{1}\left\lbrace y_{1} \geq \dfrac{\underline{y}^{*} + \overline{y}^{*}}{2} \right\rbrace & \text{if } \underline{Q}\left(x\right) < v
	\end{array}
	\right..
	$$
	
	When $Q\left(1, x\right) < Q\left(0, x\right)$ and $\overline{y}^{*} = \max \left\lbrace y \in \mathcal{Y}^{*} \right\rbrace$ and $\underline{y}^{*} = \min \left\lbrace y \in \mathcal{Y}^{*} \right\rbrace$ (sub-case (b) in Assumption \ref{bounded}.3), I define
	$$
	F_{\tilde{Y}_{1}^{*}\left\vert X, \tilde{U}, \tilde{V} \right.}\left(y_{1}\left\vert x, \overline{u}, v\right.\right) = \left\lbrace
	\begin{array}{cl}
	0 & \text{if } y_{1} < \underline{y}^{*} \text{ and } v \leq \underline{Q}\left(x\right) \\
	& \\
	1 - \dfrac{\dfrac{m_{1}^{Y}\left(\overline{x}, \overline{u}\right)}{m_{1}^{S}\left(\overline{x}, \overline{u}\right)} - \underline{y}^{*}}{\overline{y}^{*} - \underline{y}^{*}} & \text{if } \underline{y}^{*} \leq y_{1} < \overline{y}^{*} \text{ and } v \leq \underline{Q}\left(x\right) \\
	& \\
	1 & \text{if } \overline{y}^{*} \leq y_{1} \text{ and } v \leq \underline{Q}\left(x\right) \\
	---------- & -------------- \\
	\mathbf{1}\left\lbrace y_{1} \geq \overline{y}^{*} \right\rbrace & \text{if } \underline{Q}\left(x\right) < v
	\end{array}
	\right..
	$$
	which are valid cumulative distribution functions because $\dfrac{m_{1}^{Y}\left(\overline{x}, \overline{u}\right)}{m_{1}^{S}\left(\overline{x}, \overline{u}\right)} \in \left[\underline{y}^{*}, \overline{y}^{*}\right]$.
\end{enumerate}

Having defined the joint cumulative distribution function $F_{\tilde{Y}_{0}^{*}, \tilde{Y}_{1}^{*}, \tilde{U}, \tilde{V}, Z, X}$, note that equations \eqref{sanity1A} and \eqref{sanity2A}, the facts $\dfrac{m_{0}^{Y}\left(\overline{x}, \overline{u}\right)}{m_{0}^{S}\left(\overline{x}, \overline{u}\right)} \in \left[\underline{y}^{*}, \overline{y}^{*}\right]$ and $\dfrac{m_{1}^{Y}\left(\overline{x}, \overline{u}\right)}{m_{1}^{S}\left(\overline{x}, \overline{u}\right)} \in \left[\underline{y}^{*}, \overline{y}^{*}\right]$, and steps 7-12 ensure that equation \eqref{agnosticcorrectsupport} holds.

Now, I show, in three steps, that equation \eqref{agnosticfaketarget} holds.
\begin{enumerate}	
	\item[Step 13.] Observe that
	\begin{align}
	& \mathbb{E}\left[\tilde{Y}_{1}^{*} \left\vert X = \overline{x}, \tilde{U} = \overline{u},  \tilde{S}_{0} = 1, \tilde{S}_{1} = 1 \right.\right] \nonumber \\
	& \hspace{30pt} = \mathbf{1}\left\lbrace Q\left(1, x\right) > Q\left(0, x\right) \right\rbrace  \cdot \alpha\left(\overline{x}, \overline{u}\right) + \mathbf{1}\left\lbrace Q\left(1, x\right) < Q\left(0, x\right) \right\rbrace \cdot \dfrac{m_{1}^{Y}\left(\overline{x}, \overline{u}\right)}{m_{1}^{S}\left(\overline{x}, \overline{u}\right)}.
	\end{align}
	
	\item[Step 14.] Notice that
	\begin{align}
	& \mathbb{E}\left[\tilde{Y}_{0}^{*} \left\vert X = \overline{x}, \tilde{U} = \overline{u},  \tilde{S}_{0} = 1, \tilde{S}_{1} = 1 \right.\right] \nonumber \\
	& \hspace{30pt} = \mathbf{1}\left\lbrace Q\left(1, x\right) > Q\left(0, x\right) \right\rbrace  \cdot \dfrac{m_{0}^{Y}\left(\overline{x}, \overline{u}\right)}{m_{0}^{S}\left(\overline{x}, \overline{u}\right)} + \mathbf{1}\left\lbrace Q\left(1, x\right) < Q\left(0, x\right) \right\rbrace \cdot \alpha\left(\overline{x}, \overline{u}\right) .
	\end{align}
	
	\item[Step 15.] Note that Steps 13 and 14 imply that
	\begin{equation*}
	\Delta_{\tilde{Y}^{*}}^{OO}\left(\overline{x}, \overline{u}\right) \coloneqq \mathbb{E}\left[\tilde{Y}_{1}^{*} - \tilde{Y}_{0}^{*} \left\vert X = \overline{x}, \tilde{U} = \overline{u}, \tilde{S}_{0} = 1, \tilde{S}_{1} = 1 \right.\right] = \delta\left(\overline{x}, \overline{u}\right),
	\end{equation*}
	ensuring that equation \eqref{agnosticfaketarget} holds.
\end{enumerate}

Finally, to show that equation \eqref{agnosticDataRestriction} holds, it suffices to follow steps 16 and 17 in Appendix \ref{proofsharp3}.

I can then conclude that Proposition \ref{agnosticsharp} is true.
\end{proof}

\begin{proof}[Proof of Proposition \ref{agnosticnecessary}]
This proof is essentially the same proof of Proposition \ref{agnosticsharp} under Assumption \ref{bounded}.3.(a). Fix $\overline{u} \in \left[0, 1\right]$, $\overline{x} \in \mathcal{X}$ and $\delta\left(\overline{x}, \overline{u}\right) \in \mathbb{R}$ arbitrarily. For brevity, define
\begin{align*}
\alpha\left(\overline{x}, \overline{u}\right) & \coloneqq \mathbf{1}\left\lbrace Q\left(1, x\right) > Q\left(0, x\right) \right\rbrace \cdot \left(\delta\left(\overline{x}, \overline{u}\right) + \dfrac{m_{0}^{Y}\left(\overline{x}, \overline{u}\right)}{m_{0}^{S}\left(\overline{x}, \overline{u}\right)}\right) \\
& \hspace{30pt} + \mathbf{1}\left\lbrace Q\left(1, x\right) < Q\left(0, x\right) \right\rbrace \cdot \left(- \delta\left(\overline{x}, \overline{u}\right) + \dfrac{m_{1}^{Y}\left(\overline{x}, \overline{u}\right)}{m_{1}^{S}\left(\overline{x}, \overline{u}\right)}\right),
\end{align*}
and
\begin{align*}
\gamma\left(\overline{x}, \overline{u}\right) & \coloneqq \mathbf{1}\left\lbrace Q\left(1, x\right) > Q\left(0, x\right) \right\rbrace \cdot \left(\dfrac{m_{1}^{Y}\left(\overline{x}, \overline{u}\right) - \alpha\left(\overline{x}, \overline{u}\right) \cdot m_{0}^{S}\left(\overline{x}, \overline{u}\right)}{\Delta_{S}\left(\overline{x}, \overline{u}\right)}\right) \\
& \hspace{30pt} + \mathbf{1}\left\lbrace Q\left(1, x\right) < Q\left(0, x\right) \right\rbrace \cdot \left(\dfrac{m_{0}^{Y}\left(\overline{x}, \overline{u}\right) - \alpha\left(\overline{x}, \overline{u}\right) \cdot m_{1}^{S}\left(\overline{x}, \overline{u}\right)}{-\Delta_{S}\left(\overline{x}, \overline{u}\right)}\right).
\end{align*}
Note that $\alpha\left(\overline{x}, \overline{u}\right) \in \mathbb{R} = \mathcal{Y}^{*}$ and $\gamma\left(\overline{x}, \overline{u}\right) \in \mathbb{R} = \mathcal{Y}^{*}$.

I define the random variables $\left(\tilde{Y}_{0}^{*}, \tilde{Y}_{1}^{*}, \tilde{U}, \tilde{V}\right)$ using the joint cumulative distribution function $F_{\tilde{Y}_{0}^{*}, \tilde{Y}_{1}^{*}, \tilde{U}, \tilde{V}, Z, X}$ described by steps 1-12 in the last proof for the case of convex support $\mathcal{Y}^{*}$. Note that equation \eqref{agnosticcorrectsupportP} is trivially true when $\mathcal{Y}^{*} = \mathbb{R}$. Moreover, equations \eqref{agnosticfaketargetP} and \eqref{agnosticDataRestrictionP} are valid by the argument described in the last proof.

I can then conclude that Proposition \ref{agnosticnecessary} is true.
\end{proof}

\pagebreak

%%%%%%%%%%%%%%%%%%%%%%%%%%%%%%%%%%%%%%%%%
% Smoothness restrictions
%%%%%%%%%%%%%%%%%%%%%%%%%%%%%%%%%%%%%%%%%
\setcounter{table}{0}
\renewcommand\thetable{E.\arabic{table}}

\setcounter{figure}{0}
\renewcommand\thefigure{E.\arabic{figure}}

\setcounter{equation}{0}
\renewcommand\theequation{E.\arabic{equation}}

\setcounter{theorem}{0}
\renewcommand\thetheorem{E.\arabic{theorem}}

\section{Uninformative Bounds with Non-monotone Sample Selection}\label{nomonotonicity}

In the main text and in Appendices \ref{decreasing_sample_selection} and \ref{agnostic}, I impose some monotonicity condition on the sample selection problem through equation \eqref{selection}. However, in some empirical applications, this assumption may be invalid. For example, in the short run, a job training program may move some individuals from unemployment to employment by increasing their human capital or from employment to unemployment by decreasing their labor market experience. Since this is a frequent feature in empirical economics, it is important to understand what can be discovered about the marginal treatment effect when sample selection is not monotone. To do so, I drop equation \eqref{selection} and impose equation \eqref{treatment}, Assumptions \ref{ind}-\ref{support}, a small generalization of Assumption \ref{bounded}
\begin{assumption}\label{boundednomonotone}
I assume that $\underline{y}^{*}$ and $\overline{y}^{*}$ are known, and that
\begin{enumerate}
	\item $\underline{y}^{*} = - \infty$, $\overline{y}^{*} = \infty$ and $\mathcal{Y}^{*} = \mathbb{R}$, or
	
	\item $\underline{y}^{*} > - \infty$, $\overline{y}^{*} = \infty$ and $\mathcal{Y}^{*}$ is an interval, or
	
	\item $\underline{y}^{*} = - \infty$, $\overline{y}^{*} < \infty$ and $\mathcal{Y}^{*}$ is an interval, or
	
	\item $\underline{y}^{*} > - \infty$, $\overline{y}^{*} < \infty$ and
	\begin{enumerate}
		\item $\mathcal{Y}^{*}$ is an interval or
		
		\item $\underline{y}^{*} \in \mathcal{Y}^{*}$ and $\overline{y}^{*} \in \mathcal{Y}^{*}$.
	\end{enumerate}
\end{enumerate}
\end{assumption}
I also impose mild regularity conditions to ensure that all objects are well-defined:
\begin{assumption}\label{regularity}
For any $x \in \mathcal{X}$ and $u \in \left[0 , 1\right]$,
\begin{equation}\label{groupOO}
\mathbb{P}\left[S_{0} = 1, S_{1} = 1 \right] > 0,
\end{equation}
\begin{equation}\label{groupON}
\mathbb{P}\left[S_{0} = 1, S_{1} = 0 \right] > 0,
\end{equation}
\begin{equation}\label{groupNO}
\mathbb{P}\left[S_{0} = 0, S_{1} = 1 \right] > 0,
\end{equation}
\begin{equation}\label{regularity1}
\overline{y}^{*} \cdot m_{d}^{S}\left(x, u\right) - m_{d}^{Y}\left(x, u\right) > 0 \text{ for any } d \in \left\lbrace 0, 1 \right\rbrace,
\end{equation}
and
\begin{equation}\label{regularity2}
m_{d}^{Y}\left(x, u\right) - \underline{y}^{*} \cdot m_{d}^{S}\left(x, u\right) > 0 \text{ for any } d \in \left\lbrace 0, 1 \right\rbrace.
\end{equation}
\end{assumption}
Observe that conditions \eqref{regularity1} and \eqref{regularity2} are implied by a non-degenerate conditional distribution for each potential outcome of interest. Most importantly, the above assumptions are sufficient to construct bounds for the $ITT^{OO}$ (\cite{Horowitz2000}) and for the $LATE^{OO}$ \citep[section 2.4]{Chen2015} that are shorter than the entire support of the treatment effect.

I, now, show that, differently from the $ITT^{OO}$ and the $LATE^{OO}$, the bounds for the $MTE^{OO}$ on the outcome of interest (equation \eqref{target}) without equation \eqref{selection} are uninformative, i.e., the bounds without monotone sample selection are equal to $\left(\underline{y}^{*} - \overline{y}^{*}, \overline{y}^{*} - \underline{y}^{*}\right)$. Formally, I have that:

\begin{proposition}\label{monotonenecessary}
Suppose that the functions $m_{0}^{Y}$, $m_{1}^{Y}$, $m_{0}^{S}$ and $\Delta_{S}$ are point identified at every pair $\left(x, u\right) \in \mathcal{X} \times \left[0, 1\right]$. Impose equation \eqref{treatment} and assumptions \ref{ind}-\ref{support} and \ref{boundednomonotone}-\ref{regularity}. Then, for any $\overline{u} \in \left[0, 1\right]$, $\overline{x} \in \mathcal{X}$ and $\delta\left(\overline{x}, \overline{u}\right) \in \left(\underline{y}^{*} - \overline{y}^{*}, \overline{y}^{*} - \underline{y}^{*}\right)$, there exist random variables $\left(\tilde{Y}_{0}^{*}, \tilde{Y}_{1}^{*}, \tilde{U}, \tilde{S}_{0}, \tilde{S}_{1}\right)$ such that
\begin{equation}\label{faketargetnomonotone}
\Delta_{\tilde{Y}^{*}}^{OO}\left(\overline{x}, \overline{u}\right) \coloneqq \mathbb{E}\left[\tilde{Y}_{1}^{*} - \tilde{Y}_{0}^{*} \left\vert X = \overline{x}, \tilde{U} = \overline{u}, \tilde{S}_{0} = 1, \tilde{S}_{1} = 1 \right.\right] = \delta\left(\overline{x}, \overline{u}\right),
\end{equation}
\begin{equation}\label{correctsupportnomonotone}
\mathbb{P}\left[\left. \left(\tilde{Y}_{0}^{*}, \tilde{Y}_{1}^{*}, \tilde{S}_{0}, \tilde{S}_{1}\right) \in \mathcal{Y}^{*} \times \mathcal{Y}^{*} \times \left\lbrace 0, 1 \right\rbrace \times \left\lbrace 0, 1 \right\rbrace \right\vert X = \overline{x}, \tilde{U} = u \right] = 1 \text{ for any } u \in \left[0, 1\right],
\end{equation}
and
\begin{equation}\label{DataRestrictionnomonotone}
F_{\tilde{Y}, \tilde{D}, \tilde{S}, Z, X}\left(y, d, s, z, \overline{x} \right) = F_{Y, D, S, Z, X} \left(y, d, s, z, \overline{x}\right)
\end{equation}
for any $\left(y, d, s, z\right) \in \mathbb{R}^{4}$, where $\tilde{D} \coloneqq \mathbf{1}\left\lbrace P\left(X, Z\right) \geq \tilde{U}\right\rbrace$, $\tilde{S} = \tilde{D} \cdot \tilde{S}_{1} + \left(1 - \tilde{D}\right) \cdot \tilde{S}_{0}$, $\tilde{Y}_{0} = \tilde{S}_{0} \cdot \tilde{Y}_{0}^{*}$, $\tilde{Y}_{1} = \tilde{S}_{1} \cdot \tilde{Y}_{1}^{*}$ and $\tilde{Y} = \tilde{D} \cdot \tilde{Y}_{1} + \left(1 - \tilde{D}\right) \cdot \tilde{Y}_{0}$.
\end{proposition}

\begin{proof}[Proof of Proposition \ref{monotonenecessary}]
I only prove Proposition \ref{monotonenecessary} under assumption \ref{boundednomonotone}.4 (sub-cases (a) or (b)) because this is the more demanding case and because the other cases are trivial extensions of this one.

Fix $\overline{u} \in \left[0, 1\right]$, $\overline{x} \in \mathcal{X}$ and $\delta\left(\overline{x}, \overline{u}\right) \in \left(\underline{y}^{*} - \overline{y}^{*}, \overline{y}^{*} - \underline{y}^{*}\right)$ arbitrarily. For brevity, define $\left(\alpha_{0}\left(\overline{x}, \overline{u}\right), \alpha_{1}\left(\overline{x}, \overline{u}\right)\right) \in \left(\underline{y}^{*}, \overline{y}^{*}\right)^{2}$ such that $\delta\left(\overline{x}, \overline{u}\right) = \alpha_{1}\left(\overline{x}, \overline{u}\right) - \alpha_{0}\left(\overline{x}, \overline{u}\right)$, $$\pi\left(\overline{x}, \overline{u}\right) \coloneqq \dfrac{1}{2} \cdot \min\limits_{d \in \left\lbrace 0, 1 \right\rbrace} \left\lbrace \min \left\lbrace m_{d}^{S}\left(\overline{x}, \overline{u}\right), \dfrac{\overline{y}^{*} \cdot m_{d}^{S}\left(\overline{x}, \overline{u}\right) - m_{d}^{Y}\left(\overline{x}, \overline{u}\right)}{\overline{y}^{*} - \alpha_{d}\left(\overline{x}, \overline{u}\right)}, \dfrac{m_{d}^{Y}\left(\overline{x}, \overline{u}\right) - \underline{y}^{*} \cdot m_{d}^{S}\left(\overline{x}, \overline{u}\right)}{\alpha_{d}\left(\overline{x}, \overline{u}\right) - \underline{y}^{*}}\right\rbrace \right\rbrace,$$ $$\gamma_{0}\left(\overline{x}, \overline{u}\right) \coloneqq \dfrac{m_{0}^{Y}\left(\overline{x}, \overline{u}\right) - \alpha_{0}\left(\overline{x}, \overline{u}\right) \cdot \pi\left(\overline{x}, \overline{u}\right)}{m_{0}^{S}\left(\overline{x}, \overline{u}\right) - \pi\left(\overline{x}, \overline{u}\right)} \text{ and } \gamma_{1}\left(\overline{x}, \overline{u}\right) \coloneqq \dfrac{m_{1}^{Y}\left(\overline{x}, \overline{u}\right) - \alpha_{1}\left(\overline{x}, \overline{u}\right) \cdot \pi\left(\overline{x}, \overline{u}\right)}{m_{1}^{S}\left(\overline{x}, \overline{u}\right) - \pi\left(\overline{x}, \overline{u}\right)}.$$ Note that, by construction, $$\min\left\lbrace m_{1}^{S}\left(\overline{x}, \overline{u}\right) + m_{0}^{S}\left(\overline{x}, \overline{u}\right), 1 \right\rbrace > \pi\left(\overline{x}, \overline{u}\right) > 0 \text{ and } \left(\gamma_{0}\left(\overline{x}, \overline{u}\right), \gamma_{1}\left(\overline{x}, \overline{u}\right)\right) \in \left(\underline{y}^{*}, \overline{y}^{*}\right)^{2}.$$

The strategy of this proof consists of defining candidate random variables $\left(\tilde{Y}_{0}^{*}, \tilde{Y}_{1}^{*}, \tilde{U}, \tilde{S}_{0}, \tilde{S}_{1}\right)$ through their joint cumulative distribution function $F_{\tilde{Y}_{0}^{*}, \tilde{Y}_{1}^{*}, \tilde{U}, \tilde{S}_{0}, \tilde{S}_{1}, Z, X}$ and then checking that equations \eqref{faketargetnomonotone}, \eqref{correctsupportnomonotone} and \eqref{DataRestrictionnomonotone} are satisfied. I fix $\left(y_{0}, y_{1}, u, s_{0}, s_{1}, z, x\right) \in \mathbb{R}^{7}$ and define $F_{\tilde{Y}_{0}^{*}, \tilde{Y}_{1}^{*}, \tilde{U}, \tilde{S}_{0}, \tilde{S}_{1}, Z, X}$ in twelve steps:
\begin{enumerate}
	\item[Step 1.] For $x \notin \mathcal{X}$, $F_{\tilde{Y}_{0}^{*}, \tilde{Y}_{1}^{*}, \tilde{U}, \tilde{S}_{0}, \tilde{S}_{1}, Z, X}\left(y_{0}, y_{1}, u, s_{0}, s_{1}, z, x\right) = F_{Y_{0}^{*}, Y_{1}^{*}, U, S_{0}, S_{1}, Z, X}\left(y_{0}, y_{1}, u, s_{0}, s_{1}, z, x\right)$.
	
	\item[Step 2.] From now on, consider $x \in \mathcal{X}$. Since $$F_{\tilde{Y}_{0}^{*}, \tilde{Y}_{1}^{*}, \tilde{U}, \tilde{S}_{0}, \tilde{S}_{1}, Z, X}\left(y_{0}, y_{1}, u, s_{0}, s_{1}, z, x\right) = F_{\tilde{Y}_{0}^{*}, \tilde{Y}_{1}^{*}, \tilde{U}, \tilde{S}_{0}, \tilde{S}_{1}, Z \left\vert X \right.}\left(y_{0}, y_{1}, u, s_{0}, s_{1}, z \left\vert x \right.\right) \cdot F_{X}\left(x\right),$$ it suffices to define $F_{\tilde{Y}_{0}^{*}, \tilde{Y}_{1}^{*}, \tilde{U}, \tilde{S}_{0}, \tilde{S}_{1}, Z, X}\left(y_{0}, y_{1}, u, s_{0}, s_{1}, z, x\right)$. Moreover, I impose $$\left. Z \independent \left(\tilde{Y}_{0}^{*}, \tilde{Y}_{1}^{*}, \tilde{U}, \tilde{S}_{0}, \tilde{S}_{1} \right) \right\vert X$$ by writing $$F_{\tilde{Y}_{0}^{*}, \tilde{Y}_{1}^{*}, \tilde{U}, \tilde{S}_{0}, \tilde{S}_{1}, Z, X}\left(y_{0}, y_{1}, u, s_{0}, s_{1}, z, x\right) = F_{\tilde{Y}_{0}^{*}, \tilde{Y}_{1}^{*}, \tilde{U}, \tilde{S}_{0}, \tilde{S}_{1}\left\vert X \right.}\left(y_{0}, y_{1}, u, s_{0}, s_{1} \left\vert x \right.\right) \cdot F_{Z \left\vert X \right.}\left(z \left\vert x \right.\right),$$ implying that it is sufficient to define $F_{\tilde{Y}_{0}^{*}, \tilde{Y}_{1}^{*}, \tilde{U}, \tilde{S}_{0}, \tilde{S}_{1}\left\vert X \right.}\left(y_{0}, y_{1}, u, s_{0}, s_{1} \left\vert x \right.\right)$.
	
	\item[Step 3.] For $u \notin \left[0, 1\right]$, I define $F_{\tilde{Y}_{0}^{*}, \tilde{Y}_{1}^{*}, \tilde{U}, \tilde{S}_{0}, \tilde{S}_{1}\left\vert X \right.}\left(y_{0}, y_{1}, u, s_{0}, s_{1} \left\vert x \right.\right) = F_{Y_{0}^{*}, Y_{1}^{*}, U, S_{0}, S_{1}\left\vert X \right.}\left(y_{0}, y_{1}, u, s_{0}, s_{1} \left\vert x \right.\right)$.
	
	\item[Step 4.] From now on, consider $u \in \left[0, 1\right]$. Since $$F_{\tilde{Y}_{0}^{*}, \tilde{Y}_{1}^{*}, \tilde{U}, \tilde{S}_{0}, \tilde{S}_{1}\left\vert X \right.}\left(y_{0}, y_{1}, u, s_{0}, s_{1} \left\vert x \right.\right) = F_{\tilde{Y}_{0}^{*}, \tilde{Y}_{1}^{*}, \tilde{S}_{0}, \tilde{S}_{1}\left\vert X, \tilde{U} \right.}\left(y_{0}, y_{1}, s_{0}, s_{1} \left\vert x, u \right.\right) \cdot F_{\tilde{U}\left\vert X \right.}\left(u \left\vert x \right.\right),$$ it suffices to define $F_{\tilde{Y}_{0}^{*}, \tilde{Y}_{1}^{*}, \tilde{S}_{0}, \tilde{S}_{1}\left\vert X, \tilde{U} \right.}\left(y_{0}, y_{1}, s_{0}, s_{1} \left\vert x, u \right.\right)$ and $F_{\tilde{U}\left\vert X \right.}\left(u \left\vert x \right.\right)$.
	
	\item[Step 5.] I define $F_{\tilde{U}\left\vert X \right.}\left(u \left\vert x \right.\right) = F_{U\left\vert X \right.}\left(u \left\vert x \right.\right) = u$.
	
	\item[Step 6.] For any $u \neq \overline{u}$, I define $F_{\tilde{Y}_{0}^{*}, \tilde{Y}_{1}^{*}, \tilde{S}_{0}, \tilde{S}_{1}\left\vert X, \tilde{U} \right.}\left(y_{0}, y_{1}, s_{0}, s_{1} \left\vert x, u \right.\right) = F_{Y_{0}^{*}, Y_{1}^{*}, S_{0}, S_{1}\left\vert X, U \right.}\left(y_{0}, y_{1}, s_{0}, s_{1} \left\vert x, u \right.\right)$.
	
	\item[Step 7.] For any $\left(s_{0}, s_{1}\right) \notin \left\lbrace 0, 1 \right\rbrace^{2}$, I define $F_{\tilde{Y}_{0}^{*}, \tilde{Y}_{1}^{*}, \tilde{S}_{0}, \tilde{S}_{1}\left\vert X, \tilde{U} \right.}\left(y_{0}, y_{1}, s_{0}, s_{1} \left\vert x, \overline{u}\right.\right) = F_{Y_{0}^{*}, Y_{1}^{*}, S_{0}, S_{1}\left\vert X, U \right.}\left(y_{0}, y_{1}, s_{0}, s_{1} \left\vert x, \overline{u}\right.\right)$.
	
	\item[Step 8.] From now on, consider $\left(s_{0}, s_{1}\right) \in \left\lbrace 0, 1 \right\rbrace^{2}$. Since $$F_{\tilde{Y}_{0}^{*}, \tilde{Y}_{1}^{*}, \tilde{S}_{0}, \tilde{S}_{1}\left\vert X, \tilde{U} \right.}\left(y_{0}, y_{1}, s_{0}, s_{1} \left\vert x, \overline{u}\right.\right) = F_{\tilde{Y}_{0}^{*}, \tilde{Y}_{1}^{*}\left\vert X, \tilde{U}, \tilde{S}_{0}, \tilde{S}_{1} \right.}\left(y_{0}, y_{1}\left\vert x, \overline{u}, s_{0}, s_{1}\right.\right) \cdot F_{\tilde{S}_{0}, \tilde{S}_{1}\left\vert X, \tilde{U} \right.}\left(s_{0}, s_{1} \left\vert x, \overline{u}\right.\right),$$ it is sufficient to define $F_{\tilde{Y}_{0}^{*}, \tilde{Y}_{1}^{*}\left\vert X, \tilde{U}, \tilde{S}_{0}, \tilde{S}_{1} \right.}\left(y_{0}, y_{1}\left\vert x, \overline{u}, s_{0}, s_{1}\right.\right)$ and $F_{\tilde{S}_{0}, \tilde{S}_{1}\left\vert X, \tilde{U} \right.}\left(s_{0}, s_{1} \left\vert x, \overline{u}\right.\right)$.
	
	\item[Step 9.] I define $F_{\tilde{S}_{0}, \tilde{S}_{1}\left\vert X, \tilde{U} \right.}\left(s_{0}, s_{1} \left\vert x, \overline{u}\right.\right)$ by writing
	 $$\mathbb{P}\left[\left. \tilde{S}_{0} = 1, \tilde{S}_{1} = 1 \right\vert X = x, \tilde{U} = \overline{u} \right] = \dfrac{\pi\left(\overline{x}, \overline{u}\right)}{m_{1}^{S}\left(\overline{x}, \overline{u}\right) + m_{0}^{S}\left(\overline{x}, \overline{u}\right) - \pi\left(\overline{x}, \overline{u}\right)} \in \left(0, 1\right),$$
	 $$\mathbb{P}\left[\left. \tilde{S}_{0} = 1, \tilde{S}_{1} = 0 \right\vert X = x, \tilde{U} = \overline{u} \right] = \dfrac{m_{0}^{S}\left(\overline{x}, \overline{u}\right) - \pi\left(\overline{x}, \overline{u}\right)}{m_{1}^{S}\left(\overline{x}, \overline{u}\right) + m_{0}^{S}\left(\overline{x}, \overline{u}\right) - \pi\left(\overline{x}, \overline{u}\right)} \in \left(0, 1\right),$$
	 $$\mathbb{P}\left[\left. \tilde{S}_{0} = 0, \tilde{S}_{1} = 1 \right\vert X = x, \tilde{U} = \overline{u} \right] = \dfrac{m_{1}^{S}\left(\overline{x}, \overline{u}\right) - \pi\left(\overline{x}, \overline{u}\right)}{m_{1}^{S}\left(\overline{x}, \overline{u}\right) + m_{0}^{S}\left(\overline{x}, \overline{u}\right) - \pi\left(\overline{x}, \overline{u}\right)} \in \left(0, 1\right), \text{ and}$$
	 $$\mathbb{P}\left[\left. \tilde{S}_{0} = 0, \tilde{S}_{1} = 0 \right\vert X = x, \tilde{U} = \overline{u} \right] = 0.$$
	
	\item[Step 10.] I write $F_{\tilde{Y}_{0}^{*}, \tilde{Y}_{1}^{*}\left\vert X, \tilde{U}, \tilde{S}_{0}, \tilde{S}_{1} \right.}\left(y_{0}, y_{1}\left\vert x, \overline{u}, s_{0}, s_{1}\right.\right) = F_{\tilde{Y}_{0}^{*}\left\vert X, \tilde{U}, \tilde{S}_{0}, \tilde{S}_{1} \right.}\left(y_{0}\left\vert x, \overline{u}, s_{0}, s_{1}\right.\right) \cdot F_{\tilde{Y}_{1}^{*}\left\vert X, \tilde{U}, \tilde{S}_{0}, \tilde{S}_{1} \right.}\left(y_{1}\left\vert x, \overline{u}, s_{0}, s_{1}\right.\right)$, implying that I can separately define $F_{\tilde{Y}_{0}^{*}\left\vert X, \tilde{U}, \tilde{S}_{0}, \tilde{S}_{1} \right.}\left(y_{0}\left\vert x, \overline{u}, s_{0}, s_{1}\right.\right)$ and $F_{\tilde{Y}_{1}^{*}\left\vert X, \tilde{U}, \tilde{S}_{0}, \tilde{S}_{1} \right.}\left(y_{1}\left\vert x, \overline{u}, s_{0}, s_{1}\right.\right)$.
	
	\item[Step 11.] When $\mathcal{Y}^{*}$ is a bounded interval (sub-case (a) in Assumption \ref{bounded}.3), I define
	$$
	F_{\tilde{Y}_{0}^{*}\left\vert X, \tilde{U}, \tilde{S}_{0}, \tilde{S}_{1} \right.}\left(y_{0}\left\vert x, \overline{u}, s_{0}, s_{1}\right.\right) = \left\lbrace
	\begin{array}{cl}
	\mathbf{1}\left\lbrace y_{0} \geq \alpha_{0}\left(\overline{x}, \overline{u}\right) \right\rbrace & \text{if } \left(s_{0}, s_{1}\right) = \left(1, 1\right) \\
	---------- & ------------ \\
	\mathbf{1}\left\lbrace y_{0} \geq \gamma_{0}\left(\overline{x}, \overline{u}\right) \right\rbrace & \text{if } \left(s_{0}, s_{1}\right) = \left(1, 0\right) \\
	---------- & ------------ \\
	\mathbf{1}\left\lbrace y_{0} \geq \dfrac{\underline{y}^{*} + \overline{y}^{*}}{2} \right\rbrace & \text{if } \left(s_{0}, s_{1}\right) \in \left\lbrace \left(0, 0\right), \left(0, 1\right) \right\rbrace
	\end{array}
	\right..
	$$
	
	When $\overline{y}^{*} = \max \left\lbrace y \in \mathcal{Y}^{*} \right\rbrace$ and $\underline{y}^{*} = \min \left\lbrace y \in \mathcal{Y}^{*} \right\rbrace$ (sub-case (b) in Assumption \ref{bounded}.3), I define
	$$
	F_{\tilde{Y}_{0}^{*}\left\vert X, \tilde{U}, \tilde{V} \right.}\left(y_{0}\left\vert x, \overline{u}, v\right.\right) = \left\lbrace
	\begin{array}{cl}
	0 & \text{if } y_{0} < \underline{y}^{*} \text{ and } \left(s_{0}, s_{1}\right) = \left(1, 1\right) \\
	& \\
	1 - \dfrac{\alpha_{0}\left(\overline{x}, \overline{u}\right) - \underline{y}^{*}}{\overline{y}^{*} - \underline{y}^{*}} & \text{if } \underline{y}^{*} \leq y_{0} < \overline{y}^{*} \text{ and } \left(s_{0}, s_{1}\right) = \left(1, 1\right) \\
	& \\
	1 & \text{if } \overline{y}^{*} \leq y_{0} \text{ and } \left(s_{0}, s_{1}\right) = \left(1, 1\right) \\
	---------- & ---------------- \\
	0 & \text{if } y_{0} < \underline{y}^{*} \text{ and } \left(s_{0}, s_{1}\right) = \left(1, 0\right) \\
	& \\
	1 - \dfrac{\gamma_{0}\left(\overline{x}, \overline{u}\right) - \underline{y}^{*}}{\overline{y}^{*} - \underline{y}^{*}} & \text{if } \underline{y}^{*} \leq y_{0} < \overline{y}^{*} \text{ and } \left(s_{0}, s_{1}\right) = \left(1, 0\right) \\
	& \\
	1 & \text{if } \overline{y}^{*} \leq y_{0} \text{ and } \left(s_{0}, s_{1}\right) = \left(1, 0\right) \\
	---------- & ---------------- \\
	\mathbf{1}\left\lbrace y_{0} \geq \overline{y}^{*} \right\rbrace & \left(s_{0}, s_{1}\right) \in \left\lbrace \left(0, 0\right), \left(0, 1\right) \right\rbrace
	\end{array}
	\right..
	$$
	which are valid cumulative distribution functions because $\alpha_{0}\left(\overline{x}, \overline{u}\right) \in \left(\underline{y}^{*}, \overline{y}^{*}\right)$ and $\gamma_{0}\left(\overline{x}, \overline{u}\right) \in \left(\underline{y}^{*}, \overline{y}^{*}\right)$.
	
	\item[Step 12.] When $\mathcal{Y}^{*}$ is a bounded interval (sub-case (a) in Assumption \ref{bounded}.3), I define
	$$
	F_{\tilde{Y}_{1}^{*}\left\vert X, \tilde{U}, \tilde{S}_{0}, \tilde{S}_{1} \right.}\left(y_{1}\left\vert x, \overline{u}, s_{0}, s_{1}\right.\right) = \left\lbrace
	\begin{array}{cl}
	\mathbf{1}\left\lbrace y_{1} \geq \alpha_{1}\left(\overline{x}, \overline{u}\right) \right\rbrace & \text{if } \left(s_{0}, s_{1}\right) = \left(1, 1\right) \\
	---------- & ------------ \\
	\mathbf{1}\left\lbrace y_{1} \geq \gamma_{1}\left(\overline{x}, \overline{u}\right) \right\rbrace & \text{if } \left(s_{0}, s_{1}\right) = \left(0, 1\right) \\
	---------- & ------------ \\
	\mathbf{1}\left\lbrace y_{1} \geq \dfrac{\underline{y}^{*} + \overline{y}^{*}}{2} \right\rbrace & \text{if } \left(s_{0}, s_{1}\right) \in \left\lbrace \left(0, 0\right), \left(1, 0\right) \right\rbrace
	\end{array}
	\right..
	$$
	
	When $\overline{y}^{*} = \max \left\lbrace y \in \mathcal{Y}^{*} \right\rbrace$ and $\underline{y}^{*} = \min \left\lbrace y \in \mathcal{Y}^{*} \right\rbrace$ (sub-case (b) in Assumption \ref{bounded}.3), I define
	$$
	F_{\tilde{Y}_{1}^{*}\left\vert X, \tilde{U}, \tilde{V} \right.}\left(y_{1}\left\vert x, \overline{u}, v\right.\right) = \left\lbrace
	\begin{array}{cl}
	0 & \text{if } y_{1} < \underline{y}^{*} \text{ and } \left(s_{0}, s_{1}\right) = \left(1, 1\right) \\
	& \\
	1 - \dfrac{\alpha_{1}\left(\overline{x}, \overline{u}\right) - \underline{y}^{*}}{\overline{y}^{*} - \underline{y}^{*}} & \text{if } \underline{y}^{*} \leq y_{1} < \overline{y}^{*} \text{ and } \left(s_{0}, s_{1}\right) = \left(1, 1\right) \\
	& \\
	1 & \text{if } \overline{y}^{*} \leq y_{1} \text{ and } \left(s_{0}, s_{1}\right) = \left(1, 1\right) \\
	---------- & ---------------- \\
	0 & \text{if } y_{1} < \underline{y}^{*} \text{ and } \left(s_{0}, s_{1}\right) = \left(0, 1\right) \\
	& \\
	1 - \dfrac{\gamma_{1}\left(\overline{x}, \overline{u}\right) - \underline{y}^{*}}{\overline{y}^{*} - \underline{y}^{*}} & \text{if } \underline{y}^{*} \leq y_{1} < \overline{y}^{*} \text{ and } \left(s_{0}, s_{1}\right) = \left(0, 1\right) \\
	& \\
	1 & \text{if } \overline{y}^{*} \leq y_{1} \text{ and } \left(s_{0}, s_{1}\right) = \left(0, 1\right) \\
	---------- & ---------------- \\
	\mathbf{1}\left\lbrace y_{1} \geq \overline{y}^{*} \right\rbrace & \left(s_{0}, s_{1}\right) \in \left\lbrace \left(0, 0\right), \left(1, 0\right) \right\rbrace
	\end{array}
	\right..
	$$
	which are valid cumulative distribution functions because $\alpha_{1}\left(\overline{x}, \overline{u}\right) \in \left(\underline{y}^{*}, \overline{y}^{*}\right)$ and $\gamma_{1}\left(\overline{x}, \overline{u}\right) \in \left(\underline{y}^{*}, \overline{y}^{*}\right)$.
\end{enumerate}

Having defined the joint cumulative distribution function $F_{\tilde{Y}_{0}^{*}, \tilde{Y}_{1}^{*}, \tilde{U}, \tilde{S}_{0}, \tilde{S}_{1}, Z, X}$, note that steps 7-12 ensure that equation \eqref{correctsupportnomonotone} holds.

Now, observe equation \eqref{faketargetnomonotone} holds because steps 11 and 12 ensure that $\alpha_{1}\left(\overline{x}, \overline{u}\right) = \mathbb{E}\left[\tilde{Y}_{1}^{*} \left\vert X = \overline{x}, \tilde{U} = \overline{u},  \tilde{S}_{0} = 1, \tilde{S}_{1} = 1 \right.\right]$ and $\alpha_{0}\left(\overline{x}, \overline{u}\right) = \mathbb{E}\left[\tilde{Y}_{0}^{*} \left\vert X = \overline{x}, \tilde{U} = \overline{u},  \tilde{S}_{0} = 1, \tilde{S}_{1} = 1 \right.\right]$.

Finally, equation \eqref{DataRestrictionnomonotone} holds according to the same argument described at the end of appendix \ref{proofsharp3}.

I can then conclude that Proposition \ref{monotonenecessary} is true.
\end{proof}

\pagebreak

%%%%%%%%%%%%%%%%%%%%%%%%%%%%%%%%%%%%%%%%%
% Mean Dominance
%%%%%%%%%%%%%%%%%%%%%%%%%%%%%%%%%%%%%%%%%
\setcounter{table}{0}
\renewcommand\thetable{F.\arabic{table}}

\setcounter{figure}{0}
\renewcommand\thefigure{F.\arabic{figure}}

\setcounter{equation}{0}
\renewcommand\theequation{F.\arabic{equation}}

\setcounter{theorem}{0}
\renewcommand\thetheorem{F.\arabic{theorem}}

\section{MTE bounds under a Mean Dominance Assumption} \label{MeanDominance}
Here, I modify the Mean Dominance Assumption \ref{meandominanceG} by changing the direction of the inequality, i.e., I assume that:
\begin{assumption}\label{meandominanceL}
	The potential outcome when treated within the always-observed subpopulation is less than or equal to the same parameter within the observed-only-when-treated subpopulation:
	\begin{equation*}
	\mathbb{E}\left[Y_{1}^{*} \left\vert X = x, U = u, S_{0} = 1, S_{1} = 1 \right.\right] \leq \mathbb{E}\left[Y_{1}^{*} \left\vert X = x, U = u, S_{0} = 0, S_{1} = 1 \right.\right]
	\end{equation*}
	for any $x \in \mathcal{X}$ and $u \in \left[0, 1\right]$.
\end{assumption}

Note that assumption \ref{meandominanceL} implies that $\Delta_{Y}^{NO}\left(x, u\right) \geq \dfrac{m_{1}^{Y}\left(x, u\right)}{m_{1}^{S}\left(x, u\right)} \geq \mathbb{E}\left[Y_{1}^{*} \left\vert X = x, U = u, S_{0} = 1, S_{1} = 1 \right.\right]$. As a consequence, by following the same steps of the proof of Corollary \ref{boundmeandomG}, I can derive:

\begin{corollary}\label{boundmeandomL}
	Fix $u \in \left[0, 1\right]$ and $x \in \mathcal{X}$ arbitrarily. Suppose that the $m_{0}^{Y}\left(x, u\right)$, $m_{1}^{Y}\left(x, u\right)$, $m_{0}^{S}\left(x, u\right)$ and $\Delta_{S}\left(x, u\right)$ are point identified.
	
	Under assumptions \ref{ind}-\ref{support}, \ref{bounded}.1, \ref{increasing_sample_selection} and \ref{meandominanceL}, $\Delta_{Y^{*}}^{OO}\left(x, u\right)$ must satisfy
	\begin{equation}
	\Delta_{Y^{*}}^{OO}\left(x, u\right) \geq \underline{y}^{*} - \dfrac{m_{0}^{Y}\left(x, u\right)}{m_{0}^{S}\left(x, u\right)} \eqqcolon \underline{\Delta_{Y^{*}}^{OO}}\left(x, u\right)
	\end{equation}
	and
	\begin{equation}
	\Delta_{Y^{*}}^{OO}\left(x, u\right) \leq \dfrac{m_{1}^{Y}\left(x, u\right)}{m_{1}^{S}\left(x, u\right)} - \dfrac{m_{0}^{Y}\left(x, u\right)}{m_{0}^{S}\left(x, u\right)} \eqqcolon \overline{\Delta_{Y^{*}}^{OO}}\left(x, u\right).
	\end{equation}
	
	Under assumptions \ref{ind}-\ref{support}, \ref{bounded}.2, \ref{increasing_sample_selection} and \ref{meandominanceL}, $\Delta_{Y^{*}}^{OO}\left(x, u\right)$ must satisfy
	\begin{equation}
	\Delta_{Y^{*}}^{OO}\left(x, u\right) \geq \dfrac{m_{1}^{Y}\left(x, u\right) - \overline{y}^{*} \cdot \Delta_{S}\left(x, u\right)}{m_{0}^{S}\left(x, u\right)} - \dfrac{m_{0}^{Y}\left(x, u\right)}{m_{0}^{S}\left(x, u\right)} \eqqcolon \underline{\Delta_{Y^{*}}^{OO}}\left(x, u\right)
	\end{equation}
	and
	\begin{equation}
	\Delta_{Y^{*}}^{OO}\left(x, u\right) \leq \dfrac{m_{1}^{Y}\left(x, u\right)}{m_{1}^{S}\left(x, u\right)} - \dfrac{m_{0}^{Y}\left(x, u\right)}{m_{0}^{S}\left(x, u\right)} \eqqcolon \overline{\Delta_{Y^{*}}^{OO}}\left(x, u\right).
	\end{equation}
	
	Under assumptions \ref{ind}-\ref{support}, \ref{bounded}.3 (sub-case (a) or (b)), \ref{increasing_sample_selection} and \ref{meandominanceL}, $\Delta_{Y^{*}}^{OO}\left(x, u\right)$ must satisfy
	\begin{equation}
	\Delta_{Y^{*}}^{OO}\left(x, u\right) \geq \max\left\lbrace \dfrac{m_{1}^{Y}\left(x, u\right) - \overline{y}^{*} \cdot \Delta_{S}\left(x, u\right)}{m_{0}^{S}\left(x, u\right)}, \underline{y}^{*}\right\rbrace - \dfrac{m_{0}^{Y}\left(x, u\right)}{m_{0}^{S}\left(x, u\right)} \eqqcolon \underline{\Delta_{Y^{*}}^{OO}}\left(x, u\right)
	\end{equation}
	and
	\begin{equation}
	\Delta_{Y^{*}}^{OO}\left(x, u\right) \leq \dfrac{m_{1}^{Y}\left(x, u\right)}{m_{1}^{S}\left(x, u\right)} - \dfrac{m_{0}^{Y}\left(x, u\right)}{m_{0}^{S}\left(x, u\right)} \eqqcolon \overline{\Delta_{Y^{*}}^{OO}}\left(x, u\right).
	\end{equation}
	
	When $\mathcal{Y}^{*} = \mathbb{R}$ and assumptions \ref{ind}-\ref{support}, \ref{increasing_sample_selection} and \ref{meandominanceL} hold, $\Delta_{Y^{*}}^{OO}\left(x, u\right)$ must satisfy
	\begin{equation}
	\Delta_{Y^{*}}^{OO}\left(x, u\right) \geq - \infty \eqqcolon \underline{\Delta_{Y^{*}}^{OO}}\left(x, u\right)
	\end{equation}
	and
	\begin{equation}\label{uppermean}
	\Delta_{Y^{*}}^{OO}\left(x, u\right) \leq \dfrac{m_{1}^{Y}\left(x, u\right)}{m_{1}^{S}\left(x, u\right)} - \dfrac{m_{0}^{Y}\left(x, u\right)}{m_{0}^{S}\left(x, u\right)} \eqqcolon \overline{\Delta_{Y^{*}}^{OO}}\left(x, u\right).
	\end{equation}
\end{corollary}

The bounds in corollary \ref{boundmeandomL} can be identified using the strategies that were described in Sections \ref{interval} and \ref{discrete}. Furthermore, I can derive a result similar to Proposition \ref{sharpboundsmeanG}:

\begin{proposition}\label{sharpboundsmeanL}
	Suppose that the functions $m_{0}^{Y}$, $m_{1}^{Y}$, $m_{0}^{S}$, $m_{1}^{S}$ and $\Delta_{S}$ are point identified at every pair $\left(x, u\right) \in \mathcal{X} \times \left[0, 1\right]$. Under assumptions \ref{ind}-\ref{support}, \ref{increasing_sample_selection} and \ref{meandominanceL}, the bounds $\underline{\Delta_{Y^{*}}^{OO}}$ and $\overline{\Delta_{Y^{*}}^{OO}}$, given by corollary \ref{boundmeandomL}, are pointwise sharp, i.e., for any $\overline{u} \in \left[0, 1\right]$, $\overline{x} \in \mathcal{X}$ and $\delta\left(\overline{x}, \overline{u}\right) \in \left(\underline{\Delta_{Y^{*}}^{OO}}\left(\overline{x}, \overline{u}\right), \overline{\Delta_{Y^{*}}^{OO}}\left(\overline{x}, \overline{u}\right)\right)$, there exist random variables $\left(\tilde{Y}_{0}^{*}, \tilde{Y}_{1}^{*}, \tilde{U}, \tilde{V}\right)$ such that
	\begin{equation}\label{faketargetmeanL}
	\Delta_{\tilde{Y}^{*}}^{OO}\left(\overline{x}, \overline{u}\right) \coloneqq \mathbb{E}\left[\tilde{Y}_{1}^{*} - \tilde{Y}_{0}^{*} \left\vert X = \overline{x}, \tilde{U} = \overline{u}, \tilde{S}_{0} = 1, \tilde{S}_{1} = 1 \right.\right] = \delta\left(\overline{x}, \overline{u}\right),
	\end{equation}
	\begin{equation}\label{correctsupportmeanL}
	\mathbb{P}\left[\left. \left(\tilde{Y}_{0}^{*}, \tilde{Y}_{1}^{*}, \tilde{V}\right) \in \mathcal{Y}^{*} \times \mathcal{Y}^{*} \times \left[0, 1\right] \right\vert X = \overline{x}, \tilde{U} = u \right] = 1 \text{ for any } u \in \left[0, 1\right],
	\end{equation}
	\begin{equation}\label{fakemeandominanceL}
	\mathbb{E}\left[\tilde{Y}_{1}^{*} \left\vert X = \overline{x}, \tilde{U} = \overline{u}, \tilde{S}_{0} = 1, \tilde{S}_{1} = 1 \right.\right] \leq \mathbb{E}\left[\tilde{Y}_{1}^{*} \left\vert X = \overline{x}, \tilde{U} = \overline{u}, \tilde{S}_{0} = 0, \tilde{S}_{1} = 1 \right.\right],
	\end{equation}
	and
	\begin{equation}\label{DataRestrictionmeanL}
	F_{\tilde{Y}, \tilde{D}, \tilde{S}, Z, X}\left(y, d, s, z, \overline{x} \right) = F_{Y, D, S, Z, X} \left(y, d, s, z, \overline{x}\right)
	\end{equation}
	for any $\left(y, d, s, z\right) \in \mathbb{R}^{4}$,	where $\tilde{D} \coloneqq \mathbf{1}\left\lbrace P\left(X, Z\right) \geq \tilde{U}\right\rbrace$, $\tilde{S}_{0} = \mathbf{1}\left\lbrace Q\left(0, X\right) \geq \tilde{V}\right\rbrace$, $\tilde{S}_{1} = \mathbf{1}\left\lbrace Q\left(1, X\right) \geq \tilde{V}\right\rbrace$, $\tilde{Y}_{0} = \tilde{S}_{0} \cdot \tilde{Y}_{0}^{*}$, $\tilde{Y}_{1} = \tilde{S}_{1} \cdot \tilde{Y}_{1}^{*}$ and $\tilde{Y} = \tilde{D} \cdot \tilde{Y}_{1} + \left(1 - \tilde{D}\right) \cdot \tilde{Y}_{0}$.
\end{proposition}

The proof of Proposition \ref{sharpboundsmeanL} is symmetric to the proof of Proposition \ref{sharpboundsmeanG} (Appendix \ref{proofsharpboundsmeanG}).

\pagebreak

%%%%%%%%%%%%%%%%%%%%%%%%%%%%%%%%%%%%%%%%%
% Smoothness restrictions
%%%%%%%%%%%%%%%%%%%%%%%%%%%%%%%%%%%%%%%%%
\setcounter{table}{0}
\renewcommand\thetable{G.\arabic{table}}

\setcounter{figure}{0}
\renewcommand\thefigure{G.\arabic{figure}}

\setcounter{equation}{0}
\renewcommand\theequation{G.\arabic{equation}}

\setcounter{theorem}{0}
\renewcommand\thetheorem{G.\arabic{theorem}}

\section{Sharpness and Impossibility Results with Smoothness Restrictions}\label{discontinuous}

In the main text, I imposed no smoothness condition on the joint distribution of $\left(Y_{0}^{*}, Y_{1}^{*}, U, V, Z, X\right)$. Here, I impose the following smoothness condition:
\begin{assumption}\label{smoothness}
The conditional cumulative distribution functions $F_{V \left\vert X, U \right.}$ are $F_{Y_{0}^{*}, Y_{1}^{*} \left\vert X, U, V \right.}$ are continuous functions of the value of U.
\end{assumption}

As a consequence of this new assumption, Theorem \ref{sharpbounds} and Proposition \ref{partialnecessary} have to be modified to accommodate infinitesimal violations of the data restriction and to ensure that the extra model restrictions imposed by assumption \ref{smoothness} are also satisfied.

\begin{proposition}\label{sharpboundsS}
	Suppose that the functions $m_{0}^{Y}$, $m_{1}^{Y}$, $m_{0}^{S}$ and $\Delta_{S}$ are point identified at every pair $\left(x, u\right) \in \mathcal{X} \times \left[0, 1\right]$. Under Assumptions \ref{ind}-\ref{support}, \ref{bounded} (sub-cases 1, 2, 3(a) or 3(b)), \ref{increasing_sample_selection} and \ref{smoothness}, the bounds $\underline{\Delta_{Y^{*}}^{OO}}$ and $\overline{\Delta_{Y^{*}}^{OO}}$, given by Corollary \ref{MTEbounds} are infinitesimally pointwise sharp, i.e., for any $\epsilon \in \mathbb{R}_{++}$, $\overline{u} \in \left[0, 1\right]$, $\overline{x} \in \mathcal{X}$ and $\delta\left(\overline{x}, \overline{u}\right) \in \left(\underline{\Delta_{Y^{*}}^{OO}}\left(\overline{x}, \overline{u}\right), \overline{\Delta_{Y^{*}}^{OO}}\left(\overline{x}, \overline{u}\right)\right)$, there exist random variables $\left(\tilde{Y}_{0}^{*}, \tilde{Y}_{1}^{*}, \tilde{U}, \tilde{V}\right)$ such that
	\begin{equation}\label{faketargetS}
		\Delta_{\tilde{Y}^{*}}^{OO}\left(\overline{x}, \overline{u}\right) \coloneqq \mathbb{E}\left[\tilde{Y}_{1}^{*} - \tilde{Y}_{0}^{*} \left\vert X = \overline{x}, \tilde{U} = \overline{u}, \tilde{S}_{0} = 1, \tilde{S}_{1} = 1 \right.\right] = \delta\left(\overline{x}, \overline{u}\right),
	\end{equation}
	\begin{equation}\label{correctsupportS}
		\mathbb{P}\left[\left. \left(\tilde{Y}_{0}^{*}, \tilde{Y}_{1}^{*}, \tilde{V}\right) \in \mathcal{Y}^{*} \times \mathcal{Y}^{*} \times \left[0, 1\right] \right\vert X = \overline{x}, \tilde{U} = u \right] = 1 \text{ for any } u \in \left[0, 1\right],
	\end{equation}
	\begin{equation}\label{smoothV}
	F_{\tilde{V} \left\vert X, \tilde{U} \right.} \text{ is a continuous function of the value of } \tilde{U},
	\end{equation}
	\begin{equation}\label{smoothY}
	F_{\tilde{Y}_{0}^{*}, \tilde{Y}_{1}^{*} \left\vert X, \tilde{U}, \tilde{V} \right.} \text{ is a continuous function of the value of } \tilde{U},
	\end{equation}
	and
	\begin{equation}\label{DataRestrictionS}
		\left\vert F_{\tilde{Y}, \tilde{D}, \tilde{S}, Z, X}\left(y, d, s, z, \overline{x} \right) - F_{Y, D, S, Z, X} \left(y, d, s, z, \overline{x}\right) \right\vert \leq \epsilon
	\end{equation}
	for any $\left(y, d, s, z\right) \in \mathbb{R}^{4}$,	where $\tilde{D} \coloneqq \mathbf{1}\left\lbrace P\left(X, Z\right) \geq \tilde{U}\right\rbrace$, $\tilde{S}_{0} = \mathbf{1}\left\lbrace Q\left(0, X\right) \geq \tilde{V}\right\rbrace$, $\tilde{S}_{1} = \mathbf{1}\left\lbrace Q\left(1, X\right) \geq \tilde{V}\right\rbrace$, $\tilde{Y}_{0} = \tilde{S}_{0} \cdot \tilde{Y}_{0}^{*}$, $\tilde{Y}_{1} = \tilde{S}_{1} \cdot \tilde{Y}_{1}^{*}$ and $\tilde{Y} = \tilde{D} \cdot \tilde{Y}_{1} + \left(1 - \tilde{D}\right) \cdot \tilde{Y}_{0}$.
\end{proposition}

\begin{proposition}\label{partialnecessaryS}
	Suppose that the functions $m_{0}^{Y}$, $m_{1}^{Y}$, $m_{0}^{S}$ and $\Delta_{S}$ are point identified at every pair $\left(x, u\right) \in \mathcal{X} \times \left[0, 1\right]$. Impose Assumptions \ref{ind}-\ref{support}, \ref{increasing_sample_selection} and \ref{smoothness}. If $\mathcal{Y}^{*} = \mathbb{R}$, then, for any $\epsilon \in \mathbb{R}_{++}$, $\overline{u} \in \left[0, 1\right]$, $\overline{x} \in \mathcal{X}$ and $\delta\left(\overline{x}, \overline{u}\right) \in \mathbb{R}$, there exist random variables $\left(\tilde{Y}_{0}^{*}, \tilde{Y}_{1}^{*}, \tilde{U}, \tilde{V}\right)$ such that
	\begin{equation}\label{faketargetPS}
		\Delta_{\tilde{Y}^{*}}^{OO}\left(\overline{x}, \overline{u}\right) \coloneqq \mathbb{E}\left[\tilde{Y}_{1}^{*} - \tilde{Y}_{0}^{*} \left\vert X = \overline{x}, \tilde{U} = \overline{u}, \tilde{S}_{0} = 1, \tilde{S}_{1} = 1 \right.\right] = \delta\left(\overline{x}, \overline{u}\right),
	\end{equation}
	\begin{equation}\label{correctsupportPS}
		\mathbb{P}\left[\left. \left(\tilde{Y}_{0}^{*}, \tilde{Y}_{1}^{*}, \tilde{V}\right) \in \mathcal{Y}^{*} \times \mathcal{Y}^{*} \times \left[0, 1\right] \right\vert X = \overline{x}, \tilde{U} = u \right] = 1 \text{ for any } u \in \left[0, 1\right],
	\end{equation}
	\begin{equation}\label{smoothVP}
	F_{\tilde{V} \left\vert X, \tilde{U} \right.} \text{ is a continuous function of the value of } \tilde{U},
	\end{equation}
	\begin{equation}\label{smoothYP}
	F_{\tilde{Y}_{0}^{*}, \tilde{Y}_{1}^{*} \left\vert X, \tilde{U}, \tilde{V} \right.} \text{ is a continuous function of the value of } \tilde{U},
	\end{equation}
	and
	\begin{equation}\label{DataRestrictionPS}
		\left\vert F_{\tilde{Y}, \tilde{D}, \tilde{S}, Z, X}\left(y, d, s, z, \overline{x} \right) - F_{Y, D, S, Z, X} \left(y, d, s, z, \overline{x}\right) \right\vert \leq \epsilon
	\end{equation}
	for any $\left(y, d, s, z\right) \in \mathbb{R}^{4}$,	where $\tilde{D} \coloneqq \mathbf{1}\left\lbrace P\left(X, Z\right) \geq \tilde{U}\right\rbrace$, $\tilde{S}_{0} = \mathbf{1}\left\lbrace Q\left(0, X\right) \geq \tilde{V}\right\rbrace$, $\tilde{S}_{1} = \mathbf{1}\left\lbrace Q\left(1, X\right) \geq \tilde{V}\right\rbrace$, $\tilde{Y}_{0} = \tilde{S}_{0} \cdot \tilde{Y}_{0}^{*}$, $\tilde{Y}_{1} = \tilde{S}_{1} \cdot \tilde{Y}_{1}^{*}$ and $\tilde{Y} = \tilde{D} \cdot \tilde{Y}_{1} + \left(1 - \tilde{D}\right) \cdot \tilde{Y}_{0}$.
\end{proposition}

The proofs of propositions \ref{sharpboundsS} and \ref{partialnecessaryS} are below. They are small modification of the previous proofs.

\begin{proof}[Proof of Proposition \ref{sharpboundsS}]
I only prove Proposition \ref{sharpboundsS} under Assumption \ref{bounded}.3 (sub-cases (a) and (b)).The proofs of Proposition \ref{sharpboundsS} under assumptions \ref{bounded}.1 and \ref{bounded}.2 are trivial modifications of the proof presented below.

Fix any $\overline{u} \in \left[0, 1\right]$, any $\overline{x} \in \mathcal{X}$, any $\delta\left(\overline{x}, \overline{u}\right) \in \left(\underline{\Delta_{Y^{*}}^{OO}}\left(\overline{x}, \overline{u}\right), \overline{\Delta_{Y^{*}}^{OO}}\left(\overline{x}, \overline{u}\right)\right)$ and any $\epsilon \in \mathbb{R}_{++}$ such that $\min \left\lbrace \overline{u} - \dfrac{\epsilon}{2 \cdot F_{X}\left(\overline{x}\right)}, 1 - \left(\overline{u} - \dfrac{\epsilon}{2 \cdot F_{X}\left(\overline{x}\right)}\right) \right\rbrace > 0$. For brevity, define $\alpha\left(\overline{x}, \overline{u}\right) \coloneqq \delta\left(\overline{x}, \overline{u}\right) + \dfrac{m_{0}^{Y}\left(\overline{x}, \overline{u}\right)}{m_{0}^{S}\left(\overline{x}, \overline{u}\right)}$, $\gamma\left(\overline{x}, \overline{u}\right) \coloneqq \dfrac{m_{1}^{Y}\left(\overline{x}, \overline{u}\right) - \alpha\left(\overline{x}, \overline{u}\right) \cdot m_{0}^{S}\left(\overline{x}, \overline{u}\right)}{\Delta_{S}\left(\overline{x}, \overline{u}\right)}$ and $\overline{\epsilon} \coloneqq \dfrac{\epsilon}{2 \cdot F_{X}\left(\overline{x}\right)}$.

Note that
\begin{equation}\label{sanity1S}
\begin{array}{cll}
& \delta\left(\overline{x}, \overline{u}\right) & \in \left(\underline{\Delta_{Y^{*}}^{OO}}\left(\overline{x}, \overline{u}\right), \overline{\Delta_{Y^{*}}^{OO}}\left(\overline{x}, \overline{u}\right)\right) \\
& & \\
\Leftrightarrow & \alpha\left(\overline{x}, \overline{u}\right) & \in \left(\max\left\lbrace \dfrac{m_{1}^{Y}\left(x, u\right) - \overline{y}^{*} \cdot \Delta_{S}\left(x, u\right)}{m_{0}^{S}\left(x, u\right)}, \underline{y}^{*}\right\rbrace, \right. \\
& & \\
& & \hspace{30pt} \left. \min\left\lbrace \dfrac{m_{1}^{Y}\left(x, u\right) - \underline{y}^{*} \cdot \Delta_{S}\left(x, u\right)}{m_{0}^{S}\left(x, u\right)}, \overline{y}^{*}\right\rbrace\right) \\
& & \\
& & \subseteq \left(\underline{y}^{*}, \overline{y}^{*}\right),
\end{array}
\end{equation}
and that
\begin{equation}\label{sanity2S}
\begin{array}{cll}
& \alpha\left(\overline{x}, \overline{u}\right) & \in \left(\dfrac{m_{1}^{Y}\left(x, u\right) - \overline{y}^{*} \cdot \Delta_{S}\left(x, u\right)}{m_{0}^{S}\left(x, u\right)}, \dfrac{m_{1}^{Y}\left(x, u\right) - \underline{y}^{*} \cdot \Delta_{S}\left(x, u\right)}{m_{0}^{S}\left(x, u\right)} \right) \\
& & \\
\Leftrightarrow & \gamma\left(\overline{x}, \overline{u}\right) & \in \left(\underline{y}^{*}, \overline{y}^{*}\right).
\end{array}
\end{equation}

The strategy of this proof consists of defining candidate random variables $\left(\tilde{Y}_{0}^{*}, \tilde{Y}_{1}^{*}, \tilde{U}, \tilde{V}\right)$ through their joint cumulative distribution function $F_{\tilde{Y}_{0}^{*}, \tilde{Y}_{1}^{*}, \tilde{U}, \tilde{V}, Z, X}$ and then checking that conditions \eqref{faketargetS}-\eqref{DataRestrictionS} are satisfied. I fix $\left(y_{0}, y_{1}, u, v, z, x\right) \in \mathbb{R}^{6}$ and define $F_{\tilde{Y}_{0}^{*}, \tilde{Y}_{1}^{*}, \tilde{U}, \tilde{V}, Z, X}$ in fourteen steps:
\begin{enumerate}
	\item[Step 1.] For $x \notin \mathcal{X}$, $F_{\tilde{Y}_{0}^{*}, \tilde{Y}_{1}^{*}, \tilde{U}, \tilde{V}, Z, X}\left(y_{0}, y_{1}, u, v, z, x\right) = F_{Y_{0}^{*}, Y_{1}^{*}, U, V, Z, X}\left(y_{0}, y_{1}, u, v, z, x\right)$.
	
	\item[Step 2.] From now on, consider $x \in \mathcal{X}$. Since $$F_{\tilde{Y}_{0}^{*}, \tilde{Y}_{1}^{*}, \tilde{U}, \tilde{V}, Z, X}\left(y_{0}, y_{1}, u, v, z, x\right) = F_{\tilde{Y}_{0}^{*}, \tilde{Y}_{1}^{*}, \tilde{U}, \tilde{V}, Z \left\vert X \right.}\left(y_{0}, y_{1}, u, v, z \left\vert x \right.\right) \cdot F_{X}\left(x\right),$$ it suffices to define $F_{\tilde{Y}_{0}^{*}, \tilde{Y}_{1}^{*}, \tilde{U}, \tilde{V}, Z \left\vert X \right.}\left(y_{0}, y_{1}, u, v, z \left\vert x \right.\right)$. Moreover, I impose $$\left. Z \independent \left(\tilde{Y}_{0}^{*}, \tilde{Y}_{1}^{*}, \tilde{U}, \tilde{V} \right) \right\vert X$$ by writing $$F_{\tilde{Y}_{0}^{*}, \tilde{Y}_{1}^{*}, \tilde{U}, \tilde{V}, Z \left\vert X \right.}\left(y_{0}, y_{1}, u, v, z \left\vert x \right.\right) = F_{\tilde{Y}_{0}^{*}, \tilde{Y}_{1}^{*}, \tilde{U}, \tilde{V}\left\vert X \right.}\left(y_{0}, y_{1}, u, v \left\vert x \right.\right) \cdot F_{Z \left\vert X \right.}\left(z \left\vert x \right.\right),$$ implying that it is sufficient to define $F_{\tilde{Y}_{0}^{*}, \tilde{Y}_{1}^{*}, \tilde{U}, \tilde{V}\left\vert X \right.}\left(y_{0}, y_{1}, u, v \left\vert x \right.\right)$.
	
	\item[Step 3.] For $u \notin \left[0, 1\right]$, I define $F_{\tilde{Y}_{0}^{*}, \tilde{Y}_{1}^{*}, \tilde{U}, \tilde{V}\left\vert X \right.}\left(y_{0}, y_{1}, u, v \left\vert x \right.\right) = F_{Y_{0}^{*}, Y_{1}^{*}, U, V\left\vert X \right.}\left(y_{0}, y_{1}, u, v \left\vert x \right.\right)$.
	
	\item[Step 4.] From now on, consider $u \in \left[0, 1\right]$. Since $$F_{\tilde{Y}_{0}^{*}, \tilde{Y}_{1}^{*}, \tilde{U}, \tilde{V}\left\vert X \right.}\left(y_{0}, y_{1}, u, v \left\vert x \right.\right) = F_{\tilde{Y}_{0}^{*}, \tilde{Y}_{1}^{*}, \tilde{V}\left\vert X, \tilde{U} \right.}\left(y_{0}, y_{1}, v \left\vert x, u \right.\right) \cdot F_{\tilde{U}\left\vert X \right.}\left(u \left\vert x \right.\right),$$ it suffices to define $F_{\tilde{Y}_{0}^{*}, \tilde{Y}_{1}^{*}, \tilde{V}\left\vert X, \tilde{U} \right.}\left(y_{0}, y_{1}, v \left\vert x, u \right.\right)$ and $F_{\tilde{U}\left\vert X \right.}\left(u \left\vert x \right.\right)$.
	
	\item[Step 5.] I define $F_{\tilde{U}\left\vert X \right.}\left(u \left\vert x \right.\right) = F_{U\left\vert X \right.}\left(u \left\vert x \right.\right) = u$.
	
	\item[Step 6.] For any $u \notin \left(\overline{u} - \overline{\epsilon}, \overline{u} + \overline{\epsilon}\right)$, I define $F_{\tilde{Y}_{0}^{*}, \tilde{Y}_{1}^{*}, \tilde{V}\left\vert X, \tilde{U} \right.}\left(y_{0}, y_{1}, v \left\vert x, u\right.\right) = F_{Y_{0}^{*}, Y_{1}^{*}, V\left\vert X, U \right.}\left(y_{0}, y_{1}, v \left\vert x, u\right.\right)$.
	
	\item[Step 7.] For any $v \notin \left[0, 1\right]$, I define $F_{\tilde{Y}_{0}^{*}, \tilde{Y}_{1}^{*}, \tilde{V}\left\vert X, \tilde{U} \right.}\left(y_{0}, y_{1}, v \left\vert x, \overline{u}\right.\right) = F_{Y_{0}^{*}, Y_{1}^{*}, V\left\vert X, U \right.}\left(y_{0}, y_{1}, v \left\vert x, \overline{u}\right.\right)$.
	
	\item[Step 8.] From now on, consider $v \in \left[0, 1\right]$. Since $$F_{\tilde{Y}_{0}^{*}, \tilde{Y}_{1}^{*}, \tilde{V}\left\vert X, \tilde{U} \right.}\left(y_{0}, y_{1}, v \left\vert x, \overline{u}\right.\right) = F_{\tilde{Y}_{0}^{*}, \tilde{Y}_{1}^{*}\left\vert X, \tilde{U}, \tilde{V} \right.}\left(y_{0}, y_{1}\left\vert x, \overline{u}, v\right.\right) \cdot F_{\tilde{V}\left\vert X, \tilde{U} \right.}\left(v \left\vert x, \overline{u}\right.\right),$$ it is sufficient to define $F_{\tilde{Y}_{0}^{*}, \tilde{Y}_{1}^{*}\left\vert X, \tilde{U}, \tilde{V} \right.}\left(y_{0}, y_{1}\left\vert x, \overline{u}, v\right.\right)$ and $F_{\tilde{V}\left\vert X, \tilde{U} \right.}\left(v \left\vert x, \overline{u}\right.\right)$.
	
	\item[Step 9.] I define
	$$
	F_{\tilde{V}\left\vert X, \tilde{U} \right.}\left(v \left\vert x, \overline{u}\right.\right) = \left\lbrace
	\begin{array}{cl}
	m_{0}^{S}\left(x, \overline{u}\right) \cdot \dfrac{v}{Q\left(0, x\right)} & \text{if } v \leq Q\left(0, x\right) \\
	& \\
	m_{0}^{S}\left(x, \overline{u}\right) + \Delta_{S}\left(x, \overline{u}\right) \cdot \dfrac{v - Q\left(0, x\right)}{Q\left(1, x\right) - Q\left(0, x\right)} & \text{if } Q\left(0, x\right) < v \leq Q\left(1, x\right) \\
	& \\
	m_{1}^{S}\left(x, \overline{u}\right) + \left(1 - m_{1}^{S}\left(x, \overline{u}\right)\right)\dfrac{v - Q\left(1, x\right)}{1 - Q\left(1, x\right)} & \text{if } Q\left(1, x\right) < v
	\end{array}
	\right..
	$$
	
	\item[Step 10.] For any $u \in \left(\overline{u} - \overline{\epsilon}, \overline{u} \right)$, I define
	$$
	F_{\tilde{V}\left\vert X, \tilde{U} \right.}\left(v \left\vert x, u \right.\right) = F_{\tilde{V}\left\vert X, \tilde{U} \right.}\left(v \left\vert x, \overline{u} - \overline{\epsilon}\right.\right) \cdot \left(\dfrac{\overline{u} - u}{\overline{\epsilon}}\right) + F_{\tilde{V}\left\vert X, \tilde{U} \right.}\left(v \left\vert x, \overline{u}\right.\right) \cdot \left(\dfrac{ u - \overline{u} + \overline{\epsilon}}{\overline{\epsilon}}\right),
	$$
	which are valid cumulative distribution functions because a convex combination of cumulative distribution functions is a cumulative distribution function.
	
	For any $u \in \left(\overline{u}, \overline{u} + \overline{\epsilon} \right)$, I define
	$$
	F_{\tilde{V}\left\vert X, \tilde{U} \right.}\left(v \left\vert x, u \right.\right) = F_{\tilde{V}\left\vert X, \tilde{U} \right.}\left(v \left\vert x, \overline{u}\right.\right) \cdot \left(\dfrac{\overline{u} + \overline{\epsilon} - u}{\overline{\epsilon}}\right) + F_{\tilde{V}\left\vert X, \tilde{U} \right.}\left(v \left\vert x, \overline{u} + \overline{\epsilon} \right.\right) \cdot \left(\dfrac{u - \overline{u}}{\overline{\epsilon}}\right),
	$$
	which are valid cumulative distribution functions because a convex combination of cumulative distribution functions is a cumulative distribution function.
	
	Note that $F_{\tilde{V}\left\vert X, \tilde{U} \right.}$ is a continuous function of the value of $\tilde{U}$, i.e., it satisfies restriction \eqref{smoothV}.
		
	\item[Step 11.] I write $F_{\tilde{Y}_{0}^{*}, \tilde{Y}_{1}^{*}\left\vert X, \tilde{U}, \tilde{V} \right.}\left(y_{0}, y_{1}\left\vert x, \overline{u}, v\right.\right) = F_{\tilde{Y}_{0}^{*}\left\vert X, \tilde{U}, \tilde{V} \right.}\left(y_{0}\left\vert x, \overline{u}, v\right.\right) \cdot F_{ \tilde{Y}_{1}^{*}\left\vert X, \tilde{U}, \tilde{V} \right.}\left(y_{1}\left\vert x, \overline{u}, v\right.\right)$, implying that I can separately define $F_{\tilde{Y}_{0}^{*}\left\vert X, \tilde{U}, \tilde{V} \right.}\left(y_{0}\left\vert x, \overline{u}, v\right.\right)$ and $F_{ \tilde{Y}_{1}^{*}\left\vert X, \tilde{U}, \tilde{V} \right.}\left(y_{1}\left\vert x, \overline{u}, v\right.\right)$.
	
	\item[Step 12.] When $\mathcal{Y}^{*}$ is a bounded interval (sub-case (a) in Assumption \ref{bounded}.3), I define
	$$
	F_{\tilde{Y}_{0}^{*}\left\vert X, \tilde{U}, \tilde{V} \right.}\left(y_{0}\left\vert x, \overline{u}, v\right.\right) = \left\lbrace
	\begin{array}{cl}
	\mathbf{1}\left\lbrace y_{0} \geq \dfrac{m_{0}^{Y}\left(\overline{x}, \overline{u}\right)}{m_{0}^{S}\left(\overline{x}, \overline{u}\right)} \right\rbrace & \text{if } v \leq Q\left(0, x\right) \\
	---------- & ------- \\
	\mathbf{1}\left\lbrace y_{0} \geq \dfrac{\underline{y}^{*} + \overline{y}^{*}}{2} \right\rbrace & \text{if } Q\left(0, x\right) < v
	\end{array}
	\right..
	$$
	
	When $\overline{y}^{*} = \max \left\lbrace y \in \mathcal{Y}^{*} \right\rbrace$ and $\underline{y}^{*} = \min \left\lbrace y \in \mathcal{Y}^{*} \right\rbrace$ (sub-case (b) in Assumption \ref{bounded}.3), I define
	$$
	F_{\tilde{Y}_{0}^{*}\left\vert X, \tilde{U}, \tilde{V} \right.}\left(y_{0}\left\vert x, \overline{u}, v\right.\right) = \left\lbrace
	\begin{array}{cl}
	0 & \text{if } y_{0} < \underline{y}^{*} \text{ and } v \leq Q\left(0, x\right) \\
	& \\
	1 - \dfrac{\dfrac{m_{0}^{Y}\left(\overline{x}, \overline{u}\right)}{m_{0}^{S}\left(\overline{x}, \overline{u}\right)} - \underline{y}^{*}}{\overline{y}^{*} - \underline{y}^{*}} & \text{if } \underline{y}^{*} \leq y_{0} < \overline{y}^{*} \text{ and } v \leq Q\left(0, x\right) \\
	& \\
	1 & \text{if } \overline{y}^{*} \leq y_{0} \text{ and } v \leq Q\left(0, x\right) \\
	---------- & -------------- \\
	\mathbf{1}\left\lbrace y_{0} \geq \overline{y}^{*} \right\rbrace & \text{if } Q\left(0, x\right) < v
	\end{array}
	\right..
	$$
	which are valid cumulative distribution functions because $\dfrac{m_{0}^{Y}\left(\overline{x}, \overline{u}\right)}{m_{0}^{S}\left(\overline{x}, \overline{u}\right)} \in \left[\underline{y}^{*}, \overline{y}^{*}\right]$.
		
	\item[Step 13.] When $\mathcal{Y}^{*}$ is a bounded interval (sub-case (a) in Assumption \ref{bounded}.3), I define
	$$
	F_{\tilde{Y}_{1}^{*}\left\vert X, \tilde{U}, \tilde{V} \right.}\left(y_{1}\left\vert x, \overline{u}, v\right.\right) = \left\lbrace
	\begin{array}{cl}
	\mathbf{1}\left\lbrace y_{1} \geq \alpha\left(\overline{x}, \overline{u}\right) \right\rbrace & \text{if } v \leq Q\left(0, x\right) \\
	-------- & ----------- \\
	\mathbf{1}\left\lbrace y_{1} \geq \gamma\left(\overline{x}, \overline{u}\right) \right\rbrace & \text{if } Q\left(0, x\right) < v \leq Q\left(1, x\right) \\
	-------- & ----------- \\
	\mathbf{1}\left\lbrace y_{1} \geq \dfrac{\underline{y}^{*} + \overline{y}^{*}}{2} \right\rbrace & \text{if } Q\left(1, x\right) < v
	\end{array}
	\right..
	$$
	
	When $\overline{y}^{*} = \max \left\lbrace y \in \mathcal{Y}^{*} \right\rbrace$ and $\underline{y}^{*} = \min \left\lbrace y \in \mathcal{Y}^{*} \right\rbrace$ (sub-case (b) in Assumption \ref{bounded}.3), I define
	$$
	F_{\tilde{Y}_{1}^{*}\left\vert X, \tilde{U}, \tilde{V} \right.}\left(y_{1}\left\vert x, \overline{u}, v\right.\right) = \left\lbrace
	\begin{array}{cl}
	0 & \text{if } y_{1} < \underline{y}^{*} \text{ and } v \leq Q\left(0, x\right) \\
	& \\
	1 - \dfrac{\alpha\left(\overline{x}, \overline{u}\right) - \underline{y}^{*}}{\overline{y}^{*} - \underline{y}^{*}} & \text{if } \underline{y}^{*} \leq y_{1} < \overline{y}^{*} \text{ and } v \leq Q\left(0, x\right) \\
	& \\
	1 & \text{if } \overline{y}^{*} \leq y_{1} \text{ and } v \leq Q\left(0, x\right) \\
	-------- & ------------------ \\
	0 & \text{if } y_{1} < \underline{y}^{*} \text{ and } Q\left(0, x\right) < v \leq Q\left(1, x\right) \\
	& \\
	1 - \dfrac{\gamma\left(\overline{x}, \overline{u}\right) - \underline{y}^{*}}{\overline{y}^{*} - \underline{y}^{*}} & \text{if } \underline{y}^{*} \leq y_{1} < \overline{y}^{*} \text{ and } Q\left(0, x\right) < v \leq Q\left(1, x\right) \\
	& \\
	1 & \text{if } \overline{y}^{*} \leq y_{1} \text{ and } Q\left(0, x\right) < v \leq Q\left(1, x\right) \\
	-------- & ------------------ \\
	\mathbf{1}\left\lbrace y_{1} \geq \overline{y}^{*} \right\rbrace & \text{if } Q\left(1, x\right) < v
	\end{array}
	\right..
	$$
	which are valid cumulative distribution functions because of equations \eqref{sanity1S} and \eqref{sanity2S}.
	
	\item[Step 14.] For any $u \in \left(\overline{u} - \overline{\epsilon}, \overline{u} \right)$, I define
	\begin{align*}
	F_{\tilde{Y}_{0}^{*}, \tilde{Y}_{1}^{*}\left\vert X, \tilde{U}, \tilde{V} \right.}\left(y_{0}, y_{1}\left\vert x, u, v\right.\right) & = F_{\tilde{Y}_{0}^{*}, \tilde{Y}_{1}^{*}\left\vert X, \tilde{U}, \tilde{V} \right.}\left(y_{0}, y_{1}\left\vert x, \overline{u} - \overline{\epsilon}, v\right.\right)  \cdot \left(\dfrac{\overline{u} - u}{\overline{\epsilon}}\right) \\
	& \hspace{30pt} + F_{\tilde{Y}_{0}^{*}, \tilde{Y}_{1}^{*}\left\vert X, \tilde{U}, \tilde{V} \right.}\left(y_{0}, y_{1}\left\vert x, \overline{u}, v\right.\right) \cdot \left(\dfrac{ u - \overline{u} + \overline{\epsilon}}{\overline{\epsilon}}\right),
	\end{align*}
	which are valid cumulative distribution functions because a convex combination of cumulative distribution functions is a cumulative distribution function.
	
	For any $u \in \left(\overline{u}, \overline{u} + \overline{\epsilon} \right)$, I define
	\begin{align*}
	F_{\tilde{Y}_{0}^{*}, \tilde{Y}_{1}^{*}\left\vert X, \tilde{U}, \tilde{V} \right.}\left(y_{0}, y_{1}\left\vert x, u, v\right.\right) & = F_{\tilde{Y}_{0}^{*}, \tilde{Y}_{1}^{*}\left\vert X, \tilde{U}, \tilde{V} \right.}\left(y_{0}, y_{1}\left\vert x, \overline{u}, v\right.\right) \cdot \left(\dfrac{\overline{u} + \overline{\epsilon} - u}{\overline{\epsilon}}\right) \\
	& \hspace{30pt} + F_{\tilde{Y}_{0}^{*}, \tilde{Y}_{1}^{*}\left\vert X, \tilde{U}, \tilde{V} \right.}\left(y_{0}, y_{1}\left\vert x, \overline{u} + \overline{\epsilon}, v\right.\right) \left(\dfrac{u - \overline{u}}{\overline{\epsilon}}\right),
	\end{align*}
	which are valid cumulative distribution functions because a convex combination of cumulative distribution functions is a cumulative distribution function.
	
	Note that $	F_{\tilde{Y}_{0}^{*}, \tilde{Y}_{1}^{*}\left\vert X, \tilde{U}, \tilde{V} \right.}$ is a continuous function of the value of $\tilde{U}$, i.e., it satisfies restriction \eqref{smoothY}.
\end{enumerate}

Having defined the joint cumulative distribution function $F_{\tilde{Y}_{0}^{*}, \tilde{Y}_{1}^{*}, \tilde{U}, \tilde{V}, Z, X}$, note that equations \eqref{sanity1S} and \eqref{sanity2S}, $\dfrac{m_{0}^{Y}\left(\overline{x}, \overline{u}\right)}{m_{0}^{S}\left(\overline{x}, \overline{u}\right)} \in \left[\underline{y}^{*}, \overline{y}^{*}\right]$ and steps 7-14 ensure that equation \eqref{correctsupportS} holds.

Now, I show, in three steps, that equation \eqref{faketargetS} holds.
\begin{enumerate}	
	\item[Step 15.] Observe that
	\begin{align}
	& \mathbb{E}\left[\tilde{Y}_{1}^{*} \left\vert X = \overline{x}, \tilde{U} = \overline{u},  \tilde{S}_{0} = 1, \tilde{S}_{1} = 1 \right.\right] \nonumber \\
	& \hspace{30pt} = \mathbb{E}\left[\tilde{Y}_{1}^{*} \left\vert X = \overline{x}, \tilde{U} = \overline{u},  Q\left(0, \overline{x}\right) \geq \tilde{V} \right.\right] \nonumber \\
	& \hspace{30pt} = \dfrac{\mathbb{E}\left[\mathbf{1}\left\lbrace Q\left(0, \overline{x}\right) \geq \tilde{V} \right\rbrace \cdot \tilde{Y}_{1}^{*} \left\vert X = \overline{x}, \tilde{U} = \overline{u} \right.\right]}{\mathbb{P}\left[Q\left(0, \overline{x}\right) \geq \tilde{V} \left\vert X = \overline{x}, \tilde{U} = \overline{u} \right.\right]} \nonumber \\
	& \hspace{30pt} = \dfrac{\mathbb{E}\left[\mathbf{1}\left\lbrace Q\left(0, \overline{x}\right) \geq \tilde{V} \right\rbrace \cdot \mathbb{E}\left[\tilde{Y}_{1}^{*} \left\vert X = \overline{x}, \tilde{U} = \overline{u}, \tilde{V} \right. \right] \left\vert X = \overline{x}, \tilde{U} = \overline{u} \right.\right]}{\mathbb{P}\left[Q\left(0, \overline{x}\right) \geq \tilde{V} \left\vert X = \overline{x}, \tilde{U} = \overline{u} \right.\right]} \nonumber \\
	& \hspace{30pt} = \dfrac{\mathop{\mathlarger{\int\limits_{0}^{Q\left(0, \overline{x}\right)}}} \mathbb{E}\left[\tilde{Y}_{1}^{*} \left\vert X = \overline{x}, \tilde{U} = \overline{u}, \tilde{V} = v \right. \right] \, \text{d}F_{\tilde{V} \left\vert X, \tilde{U} \right.}\left(v \left\vert x, \overline{u} \right.\right)}{\mathbb{P}\left[Q\left(0, \overline{x}\right) \geq \tilde{V} \left\vert X = \overline{x}, \tilde{U} = \overline{u} \right.\right]} \nonumber \\
	& \hspace{30pt} = \dfrac{\mathop{\mathlarger{\int\limits_{0}^{Q\left(0, \overline{x}\right)}}} \alpha\left(\overline{x}, \overline{u}\right) \, \text{d}F_{\tilde{V} \left\vert X, \tilde{U} \right.}\left(v \left\vert x, \overline{u} \right.\right)}{\mathbb{P}\left[Q\left(0, \overline{x}\right) \geq \tilde{V} \left\vert X = \overline{x}, \tilde{U} = \overline{u} \right.\right]} \nonumber \\
	& \hspace{50pt} \text{by step 13} \nonumber \\
	& \label{step14aS} \hspace{30pt} = \alpha\left(\overline{x}, \overline{u}\right).
	\end{align}
	
	\item[Step 16.] Notice that
	\begin{align}
	& \mathbb{E}\left[\tilde{Y}_{0}^{*} \left\vert X = \overline{x}, \tilde{U} = \overline{u},  \tilde{S}_{0} = 1, \tilde{S}_{1} = 1 \right.\right] \nonumber \\
	& \hspace{30pt} = \mathbb{E}\left[\tilde{Y}_{0}^{*} \left\vert X = \overline{x}, \tilde{U} = \overline{u},  Q\left(0, \overline{x}\right) \geq \tilde{V} \right.\right] \nonumber \\
	& \hspace{30pt} = \dfrac{\mathbb{E}\left[\mathbf{1}\left\lbrace Q\left(0, \overline{x}\right) \geq \tilde{V} \right\rbrace \cdot \tilde{Y}_{0}^{*} \left\vert X = \overline{x}, \tilde{U} = \overline{u} \right.\right]}{\mathbb{P}\left[Q\left(0, \overline{x}\right) \geq \tilde{V} \left\vert X = \overline{x}, \tilde{U} = \overline{u} \right.\right]} \nonumber \\
	& \hspace{30pt} = \dfrac{\mathbb{E}\left[\mathbf{1}\left\lbrace Q\left(0, \overline{x}\right) \geq \tilde{V} \right\rbrace \cdot \mathbb{E}\left[\tilde{Y}_{0}^{*} \left\vert X = \overline{x}, \tilde{U} = \overline{u}, \tilde{V} \right. \right] \left\vert X = \overline{x}, \tilde{U} = \overline{u} \right.\right]}{\mathbb{P}\left[Q\left(0, \overline{x}\right) \geq \tilde{V} \left\vert X = \overline{x}, \tilde{U} = \overline{u} \right.\right]} \nonumber \\
	& \hspace{30pt} = \dfrac{\mathop{\mathlarger{\int\limits_{0}^{Q\left(0, \overline{x}\right)}}} \mathbb{E}\left[\tilde{Y}_{0}^{*} \left\vert X = \overline{x}, \tilde{U} = \overline{u}, \tilde{V} = v \right. \right] \, \text{d}F_{\tilde{V} \left\vert X, \tilde{U} \right.}\left(v \left\vert x, \overline{u} \right.\right)}{\mathbb{P}\left[Q\left(0, \overline{x}\right) \geq \tilde{V} \left\vert X = \overline{x}, \tilde{U} = \overline{u} \right.\right]} \nonumber \\
	& \hspace{30pt} = \dfrac{ \mathop{\mathlarger{\mathlarger{\mathop{\mathlarger{\int\limits_{0}^{Q\left(0, \overline{x}\right)}}}}}} \dfrac{m_{0}^{Y}\left(\overline{x}, \overline{u}\right)}{m_{0}^{S}\left(\overline{x}, \overline{u}\right)} \, \text{d}F_{\tilde{V} \left\vert X, \tilde{U} \right.}\left(v \left\vert x, \overline{u} \right.\right)}{\mathbb{P}\left[Q\left(0, \overline{x}\right) \geq \tilde{V} \left\vert X = \overline{x}, \tilde{U} = \overline{u} \right.\right]} \nonumber \\
	& \hspace{50pt} \text{by step 12} \nonumber \\
	& \label{step15aS} \hspace{30pt} = \dfrac{m_{0}^{Y}\left(\overline{x}, \overline{u}\right)}{m_{0}^{S}\left(\overline{x}, \overline{u}\right)}.
	\end{align}
	
	\item[Step 17.] Note that
	\begin{align*}
	\Delta_{\tilde{Y}^{*}}^{OO}\left(\overline{x}, \overline{u}\right) & \coloneqq \mathbb{E}\left[\tilde{Y}_{1}^{*} - \tilde{Y}_{0}^{*} \left\vert X = \overline{x}, \tilde{U} = \overline{u}, \tilde{S}_{0} = 1, \tilde{S}_{1} = 1 \right.\right] \\
	& = \mathbb{E}\left[\tilde{Y}_{1}^{*}\left\vert X = \overline{x}, \tilde{U} = \overline{u}, \tilde{S}_{0} = 1, \tilde{S}_{1} = 1 \right.\right] \\
	& \hspace{20pt} - \mathbb{E}\left[\tilde{Y}_{0}^{*}\left\vert X = \overline{x}, \tilde{U} = \overline{u}, \tilde{S}_{0} = 1, \tilde{S}_{1} = 1 \right.\right] \\
	& = \alpha\left(\overline{x}, \overline{u}\right) - \dfrac{m_{0}^{Y}\left(\overline{x}, \overline{u}\right)}{m_{0}^{S}\left(\overline{x}, \overline{u}\right)} \\
	& \hspace{20pt} \text{by equations } \eqref{step14aS} \text{ and } \eqref{step15aS} \\
	& = \delta\left(\overline{x}, \overline{u}\right) \\
	& \hspace{20pt} \text{by the definition of } \alpha\left(\overline{x}, \overline{u}\right),
	\end{align*}
	ensuring that equation \eqref{faketargetS} holds.
\end{enumerate}

Finally, I show, in four steps, that equation \eqref{DataRestrictionS} holds.
\begin{enumerate}
	\item[Step 18.] Fix $\left(y, d, s, z\right) \in \mathbb{R}^{4}$ arbitrarily and observe that expression \eqref{DataRestrictionS} can be simplified to:
	\begin{align}
	& \left\vert F_{\tilde{Y}, \tilde{D}, \tilde{S}, Z, X}\left(y, d, s, z, \overline{x} \right) - F_{Y, D, S, Z, X} \left(y, d, s, z, \overline{x}\right) \right\vert \leq \epsilon \nonumber \\
	\Leftrightarrow & \left\vert F_{\tilde{Y}, \tilde{D}, \tilde{S}, Z \left\vert X \right.}\left(y, d, s, z \left\vert \overline{x} \right.\right) \cdot F_{X}\left(\overline{x}\right) - F_{Y, D, S, Z \left\vert X \right.}\left(y, d, s, z \left\vert \overline{x} \right.\right) \cdot F_{X}\left(\overline{x}\right) \right\vert \leq \epsilon \nonumber \\
	\Leftrightarrow & \left\vert F_{\tilde{Y}, \tilde{D}, \tilde{S}, Z \left\vert X \right.}\left(y, d, s, z \left\vert \overline{x} \right.\right) - F_{Y, D, S, Z \left\vert X \right.}\left(y, d, s, z \left\vert \overline{x} \right.\right) \right\vert \leq \dfrac{\epsilon}{F_{X}\left(\overline{x}\right)} \nonumber \\
	\Leftrightarrow & \label{DataRestrictionSimplifiedS} \left\vert F_{\tilde{Y}, \tilde{D}, \tilde{S}, Z \left\vert X \right.}\left(y, d, s, z \left\vert \overline{x} \right.\right) - F_{Y, D, S, Z \left\vert X \right.}\left(y, d, s, z \left\vert \overline{x} \right.\right) \right\vert \leq 2 \cdot \overline{\epsilon} \\
	& \text{by the definition of } \overline{\epsilon}. \nonumber
	\end{align}
	
	\item[Step 19.] Notice that
	\begin{small}
	\begin{align}
	& F_{\tilde{Y}, \tilde{D}, \tilde{S}, Z \left\vert X \right.}\left(y, d, s, z \left\vert \overline{x} \right.\right) - F_{Y, D, S, Z \left\vert X \right.}\left(y, d, s, z \left\vert \overline{x} \right.\right) \nonumber \\
	& \hspace{30pt} = \mathbb{E}\left[\left.\mathbf{1}\left\lbrace \left(\tilde{Y}, \tilde{D}, \tilde{S}, Z\right) \leq \left(y, d, s, z\right) \right\rbrace \right\vert X = \overline{x} \right] - \mathbb{E}\left[\left.\mathbf{1}\left\lbrace \left(Y, D, S, Z\right) \leq \left(y, d, s, z\right) \right\rbrace \right\vert X = \overline{x} \right] \nonumber \\
	& \hspace{30pt} = \int \mathbf{1}\left\lbrace \left(\tilde{Y}, \tilde{D}, \tilde{S}, Z\right) \leq \left(y, d, s, z\right) \right\rbrace \, \text{d} F_{\tilde{Y}_{0}^{*}, \tilde{Y}_{1}^{*}, \tilde{U}, \tilde{V}, Z \left\vert X \right.}\left(y_{0}, y_{1}, u, v, z \left\vert \overline{x} \right.\right) \nonumber \\
	& \hspace{50pt} - \int \mathbf{1}\left\lbrace \left(Y, D, S, Z\right) \leq \left(y, d, s, z\right) \right\rbrace \, \text{d} F_{Y_{0}^{*}, Y_{1}^{*}, U, V, Z \left\vert X \right.}\left(y_{0}, y_{1}, u, v, z \left\vert \overline{x} \right.\right) \nonumber \\
	& \hspace{30pt} = \int \left[\mathbf{1}\left\lbrace \left(\tilde{Y}, \tilde{D}, \tilde{S}, Z\right) \leq \left(y, d, s, z\right) \right\rbrace \cdot \mathbf{1}\left\lbrace u \notin \left(\overline{u} - \overline{\epsilon}, \overline{u} + \overline{\epsilon} \right) \right\rbrace\right] \, \text{d} F_{\tilde{Y}_{0}^{*}, \tilde{Y}_{1}^{*}, \tilde{U}, \tilde{V}, Z \left\vert X \right.}\left(y_{0}, y_{1}, u, v, z \left\vert \overline{x} \right.\right) \nonumber \\
	& \hspace{50pt} + \int \left[\mathbf{1}\left\lbrace \left(\tilde{Y}, \tilde{D}, \tilde{S}, Z\right) \leq \left(y, d, s, z\right) \right\rbrace \cdot \mathbf{1}\left\lbrace u \in \left(\overline{u} - \overline{\epsilon}, \overline{u} + \overline{\epsilon} \right) \right\rbrace\right] \, \text{d} F_{\tilde{Y}_{0}^{*}, \tilde{Y}_{1}^{*}, \tilde{U}, \tilde{V}, Z \left\vert X \right.}\left(y_{0}, y_{1}, u, v, z \left\vert \overline{x} \right.\right) \nonumber \\
	& \hspace{50pt} - \int \left[\mathbf{1}\left\lbrace \left(Y, D, S, Z\right) \leq \left(y, d, s, z\right) \right\rbrace \cdot \mathbf{1}\left\lbrace u \notin \left(\overline{u} - \overline{\epsilon}, \overline{u} + \overline{\epsilon} \right) \right\rbrace\right] \, \text{d} F_{Y_{0}^{*}, Y_{1}^{*}, U, V, Z \left\vert X \right.}\left(y_{0}, y_{1}, u, v, z \left\vert \overline{x} \right.\right) \nonumber \\
	& \hspace{50pt} - \int \left[\mathbf{1}\left\lbrace \left(Y, D, S, Z\right) \leq \left(y, d, s, z\right) \right\rbrace \cdot \mathbf{1}\left\lbrace u \in \left(\overline{u} - \overline{\epsilon}, \overline{u} + \overline{\epsilon} \right) \right\rbrace\right] \, \text{d} F_{Y_{0}^{*}, Y_{1}^{*}, U, V, Z \left\vert X \right.}\left(y_{0}, y_{1}, u, v, z \left\vert \overline{x} \right.\right) \nonumber \\
	& \hspace{50pt} \text{by linearity of the Lebesgue Integral} \nonumber \\
	& \hspace{30pt} = \int \left[\mathbf{1}\left\lbrace \left(Y, D, S, Z\right) \leq \left(y, d, s, z\right) \right\rbrace \cdot \mathbf{1}\left\lbrace u \notin \left(\overline{u} - \overline{\epsilon}, \overline{u} + \overline{\epsilon} \right) \right\rbrace\right] \, \text{d} F_{Y_{0}^{*}, Y_{1}^{*}, U, V, Z \left\vert X \right.}\left(y_{0}, y_{1}, u, v, z \left\vert \overline{x} \right.\right) \nonumber \\
	& \hspace{50pt} + \int \left[\mathbf{1}\left\lbrace \left(\tilde{Y}, \tilde{D}, \tilde{S}, Z\right) \leq \left(y, d, s, z\right) \right\rbrace \cdot \mathbf{1}\left\lbrace u \in \left(\overline{u} - \overline{\epsilon}, \overline{u} + \overline{\epsilon} \right) \right\rbrace\right] \, \text{d} F_{\tilde{Y}_{0}^{*}, \tilde{Y}_{1}^{*}, \tilde{U}, \tilde{V}, Z \left\vert X \right.}\left(y_{0}, y_{1}, u, v, z \left\vert \overline{x} \right.\right) \nonumber \\
	& \hspace{50pt} - \int \left[\mathbf{1}\left\lbrace \left(Y, D, S, Z\right) \leq \left(y, d, s, z\right) \right\rbrace \cdot \mathbf{1}\left\lbrace u \notin \left(\overline{u} - \overline{\epsilon}, \overline{u} + \overline{\epsilon} \right) \right\rbrace\right] \, \text{d} F_{Y_{0}^{*}, Y_{1}^{*}, U, V, Z \left\vert X \right.}\left(y_{0}, y_{1}, u, v, z \left\vert \overline{x} \right.\right) \nonumber \\
	& \hspace{50pt} - \int \left[\mathbf{1}\left\lbrace \left(Y, D, S, Z\right) \leq \left(y, d, s, z\right) \right\rbrace \cdot \mathbf{1}\left\lbrace u \in \left(\overline{u} - \overline{\epsilon}, \overline{u} + \overline{\epsilon} \right) \right\rbrace\right] \, \text{d} F_{Y_{0}^{*}, Y_{1}^{*}, U, V, Z \left\vert X \right.}\left(y_{0}, y_{1}, u, v, z \left\vert \overline{x} \right.\right) \nonumber \\
	& \hspace{50pt} \text{by steps 2-6} \nonumber \\
	& \hspace{30pt} = \int \left[\mathbf{1}\left\lbrace \left(\tilde{Y}, \tilde{D}, \tilde{S}, Z\right) \leq \left(y, d, s, z\right) \right\rbrace \cdot \mathbf{1}\left\lbrace u \in \left(\overline{u} - \overline{\epsilon}, \overline{u} + \overline{\epsilon} \right) \right\rbrace\right] \, \text{d} F_{\tilde{Y}_{0}^{*}, \tilde{Y}_{1}^{*}, \tilde{U}, \tilde{V}, Z \left\vert X \right.}\left(y_{0}, y_{1}, u, v, z \left\vert \overline{x} \right.\right) \nonumber \\
	& \hspace{50pt} - \int \left[\mathbf{1}\left\lbrace \left(Y, D, S, Z\right) \leq \left(y, d, s, z\right) \right\rbrace \cdot \mathbf{1}\left\lbrace u \in \left(\overline{u} - \overline{\epsilon}, \overline{u} + \overline{\epsilon} \right) \right\rbrace\right] \, \text{d} F_{Y_{0}^{*}, Y_{1}^{*}, U, V, Z \left\vert X \right.}\left(y_{0}, y_{1}, u, v, z \left\vert \overline{x} \right.\right) \nonumber \\
	& \hspace{30pt} \leq \int \mathbf{1}\left\lbrace u \in \left(\overline{u} - \overline{\epsilon}, \overline{u} + \overline{\epsilon} \right) \right\rbrace \, \text{d} F_{\tilde{Y}_{0}^{*}, \tilde{Y}_{1}^{*}, \tilde{U}, \tilde{V}, Z \left\vert X \right.}\left(y_{0}, y_{1}, u, v, z \left\vert \overline{x} \right.\right) \nonumber \\
	& \hspace{30pt} = \int \mathbf{1}\left\lbrace u \in \left(\overline{u} - \overline{\epsilon}, \overline{u} + \overline{\epsilon} \right) \right\rbrace \, \text{d} F_{\tilde{U} \left\vert X \right.}\left(u \left\vert \overline{x} \right.\right) \nonumber \\
	& \hspace{30pt} = 2 \cdot \overline{\epsilon} \nonumber \\
	& \hspace{50pt} \text{by step 5.} \nonumber
	\end{align}
	\end{small}

	\item[Step 20.] Following the same procedure of step 19, I have that: 
	\begin{small}
		\begin{align}
		& F_{\tilde{Y}, \tilde{D}, \tilde{S}, Z \left\vert X \right.}\left(y, d, s, z \left\vert \overline{x} \right.\right) - F_{Y, D, S, Z \left\vert X \right.}\left(y, d, s, z \left\vert \overline{x} \right.\right) \nonumber \\
		& \hspace{30pt} = \int \left[\mathbf{1}\left\lbrace \left(\tilde{Y}, \tilde{D}, \tilde{S}, Z\right) \leq \left(y, d, s, z\right) \right\rbrace \cdot \mathbf{1}\left\lbrace u \in \left(\overline{u} - \overline{\epsilon}, \overline{u} + \overline{\epsilon} \right) \right\rbrace\right] \, \text{d} F_{\tilde{Y}_{0}^{*}, \tilde{Y}_{1}^{*}, \tilde{U}, \tilde{V}, Z \left\vert X \right.}\left(y_{0}, y_{1}, u, v, z \left\vert \overline{x} \right.\right) \nonumber \\
		& \hspace{50pt} - \int \left[\mathbf{1}\left\lbrace \left(Y, D, S, Z\right) \leq \left(y, d, s, z\right) \right\rbrace \cdot \mathbf{1}\left\lbrace u \in \left(\overline{u} - \overline{\epsilon}, \overline{u} + \overline{\epsilon} \right) \right\rbrace\right] \, \text{d} F_{Y_{0}^{*}, Y_{1}^{*}, U, V, Z \left\vert X \right.}\left(y_{0}, y_{1}, u, v, z \left\vert \overline{x} \right.\right) \nonumber \\
		& \hspace{30pt} \geq - \int \mathbf{1}\left\lbrace u \in \left(\overline{u} - \overline{\epsilon}, \overline{u} + \overline{\epsilon} \right) \right\rbrace \, \text{d} F_{Y_{0}^{*}, Y_{1}^{*}, U, V, Z \left\vert X \right.}\left(y_{0}, y_{1}, u, v, z \left\vert \overline{x} \right.\right) \nonumber \\
		& \hspace{30pt} = - \int \mathbf{1}\left\lbrace u \in \left(\overline{u} - \overline{\epsilon}, \overline{u} + \overline{\epsilon} \right) \right\rbrace \, \text{d} F_{U \left\vert X \right.}\left(u \left\vert \overline{x} \right.\right) \nonumber \\
		& \hspace{30pt} = - 2 \cdot \overline{\epsilon} \nonumber
		\end{align}
	\end{small}
	
	\item[Step 21.] Combining steps 19 and 20, I find that
	$$
	\left\vert F_{\tilde{Y}, \tilde{D}, \tilde{S}, Z \left\vert X \right.}\left(y, d, s, z \left\vert \overline{x} \right.\right) - F_{Y, D, S, Z \left\vert X \right.}\left(y, d, s, z \left\vert \overline{x} \right.\right) \right\vert \leq 2 \cdot \overline{\epsilon},
	$$
	implying equation \eqref{DataRestrictionS} according to equation \eqref{DataRestrictionSimplifiedS}.
\end{enumerate}

I can then conclude that Proposition \ref{sharpboundsS} is true.
\end{proof}

\begin{proof}[Proof of Proposition \ref{partialnecessaryS}]
This proof is essentially the same proof of Proposition \ref{sharpboundsS} under Assumption \ref{bounded}.3.(a). Fix any $\overline{u} \in \left[0, 1\right]$, any $\overline{x} \in \mathcal{X}$, any $\delta\left(\overline{x}, \overline{u}\right) \in \left(\underline{\Delta_{Y^{*}}^{OO}}\left(\overline{x}, \overline{u}\right), \overline{\Delta_{Y^{*}}^{OO}}\left(\overline{x}, \overline{u}\right)\right)$ and any $\epsilon \in \mathbb{R}_{++}$ such that $\min \left\lbrace \overline{u} - \dfrac{\epsilon}{2 \cdot F_{X}\left(\overline{x}\right)}, 1 - \left(\overline{u} - \dfrac{\epsilon}{2 \cdot F_{X}\left(\overline{x}\right)}\right) \right\rbrace > 0$. For brevity, define $\alpha\left(\overline{x}, \overline{u}\right) \coloneqq \delta\left(\overline{x}, \overline{u}\right) + \dfrac{m_{0}^{Y}\left(\overline{x}, \overline{u}\right)}{m_{0}^{S}\left(\overline{x}, \overline{u}\right)}$, $\gamma\left(\overline{x}, \overline{u}\right) \coloneqq \dfrac{m_{1}^{Y}\left(\overline{x}, \overline{u}\right) - \alpha\left(\overline{x}, \overline{u}\right) \cdot m_{0}^{S}\left(\overline{x}, \overline{u}\right)}{\Delta_{S}\left(\overline{x}, \overline{u}\right)}$ and $\overline{\epsilon} \coloneqq \dfrac{\epsilon}{2 \cdot F_{X}\left(\overline{x}\right)}$. Note that $\alpha\left(\overline{x}, \overline{u}\right) \in \mathbb{R} = \mathcal{Y}^{*}$ and $\gamma\left(\overline{x}, \overline{u}\right) \in \mathbb{R} = \mathcal{Y}^{*}$.

I define the random variables $\left(\tilde{Y}_{0}^{*}, \tilde{Y}_{1}^{*}, \tilde{U}, \tilde{V}\right)$ using the joint cumulative distribution function $F_{\tilde{Y}_{0}^{*}, \tilde{Y}_{1}^{*}, \tilde{U}, \tilde{V}, Z, X}$ described by steps 1-14 in the proof of Proposition \ref{sharpboundsS} for the case of convex support $\mathcal{Y}^{*}$. Note that equation \eqref{correctsupportPS} is trivially true when $\mathcal{Y}^{*} = \mathbb{R}$. Moreover, equations \eqref{faketargetPS} and \eqref{DataRestrictionPS} are valid by the argument described in steps 15-21 in the previous proof.

I can then conclude that Proposition \ref{partialnecessaryS} is true.
\end{proof}

\pagebreak

%%%%%%%%%%%%%%%%%%%%%%%%%%%%%%%%%%%%%%%%%
% Monte Carlo Simulation
%%%%%%%%%%%%%%%%%%%%%%%%%%%%%%%%%%%%%%%%%
\setcounter{table}{0}
\renewcommand\thetable{H.\arabic{table}}

\setcounter{figure}{0}
\renewcommand\thefigure{H.\arabic{figure}}

\setcounter{equation}{0}
\renewcommand\theequation{H.\arabic{equation}}

\setcounter{theorem}{0}
\renewcommand\thetheorem{H.\arabic{theorem}}

\section{Monte Carlo Simulations}\label{montecarlo}

My empirical analysis uses two new tools in order partially identify the marginal treatment effects on wages for the always-employed population ($MTE^{OO}$): the sharp bounds (Section \ref{bounds}) and the restricted version of the parametric estimation strategy proposed \cite{Brinch2017} (Subsection \ref{parametric}). Given the novelty of these methods, it is useful to implement a Monte Carlo Simulation in order to check whether the above methods work reasonably well in finite samples. In particular, I design six data-generating processes (DGPs) that capture important features of the Job Corp Training Program (JCTP) dataset and, using 1,000 simulations, estimate the coverage rate of the confidence intervals used to analyze the wage effect of the JCTP (Section \ref{empiricalresults}.) The first three DGPs satisfy the linearity assumptions imposed by the parametric estimation method, while the last three DGPs have non-linear marginal treatment response functions for employment and hourly labor earnings. The latter are useful to understand how my partial identification strategy behaves under model mis-specification.

In Subsection \ref{DGP}, I describe each one of the six DGPs used in this Monte Carlo exercise, while, in subsection \ref{MCresults}, I describe the results from my simulations.

\subsection{Data Generating Processes}\label{DGP}

All six data-generating processes have 7,531 observations, the same number as in the Non-Hispanic subsample of the JCTP. The dummy variable $Z$ indicates treatment assignment and is equal to $1$ with probability $0.605$, the same probability of a Non-Hispanic person being assigned to the treatment in my empirical application. To create the dummy variable $D$ that indicates treatment take-up, I use a random variable $U \sim Uniform\left[0,1\right]$ and the propensity score function (see Equation \eqref{treatment}) as $P\left(0\right) = 0.047$ and $P\left(0\right) = 0.737$, the same values of Table \ref{PreliminaryEffects}. Although potential employment status $S_{0}$ and $S_{1}$ and potential wages $Y_{0}^{*}$ and $Y_{1}^{*}$ follow different distributions in each DGP, employment and wages are always independent after conditioning on the latent heterogeneity in this Monte Carlo study, i.e., $\left. \left(S_{0}, S_{1}\right) \independent \left(Y_{0}^{*}, Y_{1}^{*}\right) \right\vert U$ for any DGP. I impose this restrictive condition so that I can easily write the marginal treatment response ($MTR$) function of hourly labor earnings as the product between the $MTR$ functions of employment and wages, i.e., $m_{d}^{Y}\left(u\right) = m_{d}^{S}\left(u\right) \cdot m_{d}^{Y^{*}}\left(u\right)$ for any $u \in \left[0, 1\right]$ and $d \in \left\lbrace 0, 1 \right\rbrace$. Moreover, the Mean Dominance Assumption \ref{meandominanceG} holds with equality in all DGPs. Finally, there are no covariates in this simulation study since they are not used in my empirical application.

\subsubsection{Design 1}

Potential employment status $\left(S_{0}, S_{1}\right)$ are generated following equation \eqref{selection} with $V \sim Uniform\left[0,1\right]$, $V \independent U$, $Q\left(0\right) = 0.564$ and $Q\left(1\right) = 0.613$, where $Q\left(z\right)$ is equal to the employment probability of a Non-Hispanic person being employed conditioning on treatment assignment $z \in \left\lbrace 0, 1 \right\rbrace$ in the JCTP sample. Consequently, the $MTR$ functions for employment are constant.

Potential wages $\left(Y_{0}^{*}, Y_{1}^{*}\right)$ are generated by $Y_{0}^{*} = 7.72 + \eta$ and $Y_{1}^{*} = Y_{0}^{*} + 0.61$, where $\eta \sim Uniform\left[-2, 2\right]$, $7.72$ is the average observed hourly wage of the Non-Hispanics assigned to the control group in the JCTP sample, and $0.61$ is the estimated lower bound on the $ATE^{OO}$ (Table \ref{ATEW}). Consequently, the $MTR$ functions for hourly wages are constant.

Since the $MTR$ functions for employment and hourly wages are constant, the $MTR$ function for hourly labor earnings is also constant. 

\subsubsection{Design 2}

Potential employment status $\left(S_{0}, S_{1}\right)$ are generated based on Design 1.

Potential untreated wage $Y_{0}^{*}$ is generated based on Design 1, while potential treated wage $Y_{1}^{*}$ is generated by $Y_{1}^{*} = Y_{0}^{*} + 2 \cdot 0.61 \cdot U$. Consequently, the $MTR$ function for treated hourly wages is linear.

Since the $MTR$ functions for employment are constant and the $MTR$ function for treated hourly wages is linear, the $MTR$ function for treated hourly labor earnings is linear.

\subsubsection{Design 3}

Potential employment status $\left(S_{0}, S_{1}\right)$ are generated to ensure that $S_{1} \geq S_{0}$ and that the true $MTR$ functions are equal to the estimated $MTR$ function in the JCTP sutdy (Table \ref{bmw2017}). Consequently, the $MTR$ functions for employment are linear.

Potential wages $\left(Y_{0}^{*}, Y_{1}^{*}\right)$ are generated based on Design 1.

Since the $MTR$ functions for employment are linear and the $MTR$ functions for hourly wages are constant, the $MTR$ functions for hourly labor earnings are linear.

\subsubsection{Design 4}

Potential employment status $\left(S_{0}, S_{1}\right)$ are generated based on Design 3.

Potential wages $\left(Y_{0}^{*}, Y_{1}^{*}\right)$ are generated based on Design 2.

Since the $MTR$ functions for employment are linear and the $MTR$ function for treated hourly wages is linear, the $MTR$ function for treated hourly labor earnings is quadratic.

\subsubsection{Design 5}

Potential employment status $\left(S_{0}, S_{1}\right)$ are generated following equation \eqref{selection} with $Q\left(0\right) = 0.706481$, $Q\left(1\right) = 0.873880$ and $\left. V \right\vert U \sim Beta\left[0.000468 + 1.079615 \cdot U, 0.873059 \cdot U\right]$, where the parameters of the Beta distribution and the values $Q\left(d\right)$ for any $d \in \left\lbrace 0, 1 \right\rbrace$ are chosen so that the true $MTR$ functions on employment match the estimated $MTR$ functions on employment (Table \ref{bmw2017}) when the latent heterogeneity variable is equal to the propensity score values. Note that the true $MTR$ functions for employment are non-linear.

Potential wages $\left(Y_{0}^{*}, Y_{1}^{*}\right)$ are generated based on Design 1.

Since the $MTR$ functions for employment are non-linear, the $MTR$ functions for hourly labor earnings are non-linear.

\subsubsection{Design 6}

Potential employment status $\left(S_{0}, S_{1}\right)$ are generated based on Design 5.

Potential wages $\left(Y_{0}^{*}, Y_{1}^{*}\right)$ are generated based on Design 2.

Since the $MTR$ functions for employment are non-linear, the $MTR$ functions for hourly labor earnings are non-linear.

\subsection{Monte Carlo Results}\label{MCresults}

The focus of this subsection is whether the two types of confidence intervals used in the empirical application (Subsection \ref{empiricalresults}) contain the true marginal treatment effect on wages for the always-employed population. To analyze this question, I report the pointwise coverage rate using 1,000 Monte Carlo simulations: while Figure \ref{bootstrap} reports the pointwise coverage rate of Bootstrap 90\%-Confidence Intervals for each data-generating process, Figure \ref{IMCI} reports the pointwise coverage rate of 90\%-Confidence Intervals based on \cite{Imbens2004} for each data-generating process. The solid lines are associated with bounds that do not impose the Mean Dominance Assumption \ref{meandominanceG} (Corollary \ref{MTEbounds}), while the dashed lines are associated with  bounds that impose the Mean Dominance Assumption \ref{meandominanceG} (Corollary \ref{boundmeandomG}). Since the results for the Bootstrap 90\%-Confidence Intervals are very similar to the results for the 90\%-Confidence Intervals based on \cite{Imbens2004}, I focus on the latter. Moreover, since the bounds that impose the Mean Dominance Assumption \ref{meandominanceG} are tighter than the ones that do not impose this assumption, I only discuss the results associated with Corollary \ref{boundmeandomG}. 

\begin{figure}[!htbp] 
	\caption{Coverage Rate: Bootstrap 90\%-Confidence Intervals} \label{bootstrap}
	
	\begin{center}
	\subfloat[Design 1\label{CI01}]{\includegraphics[width = .3 \columnwidth]{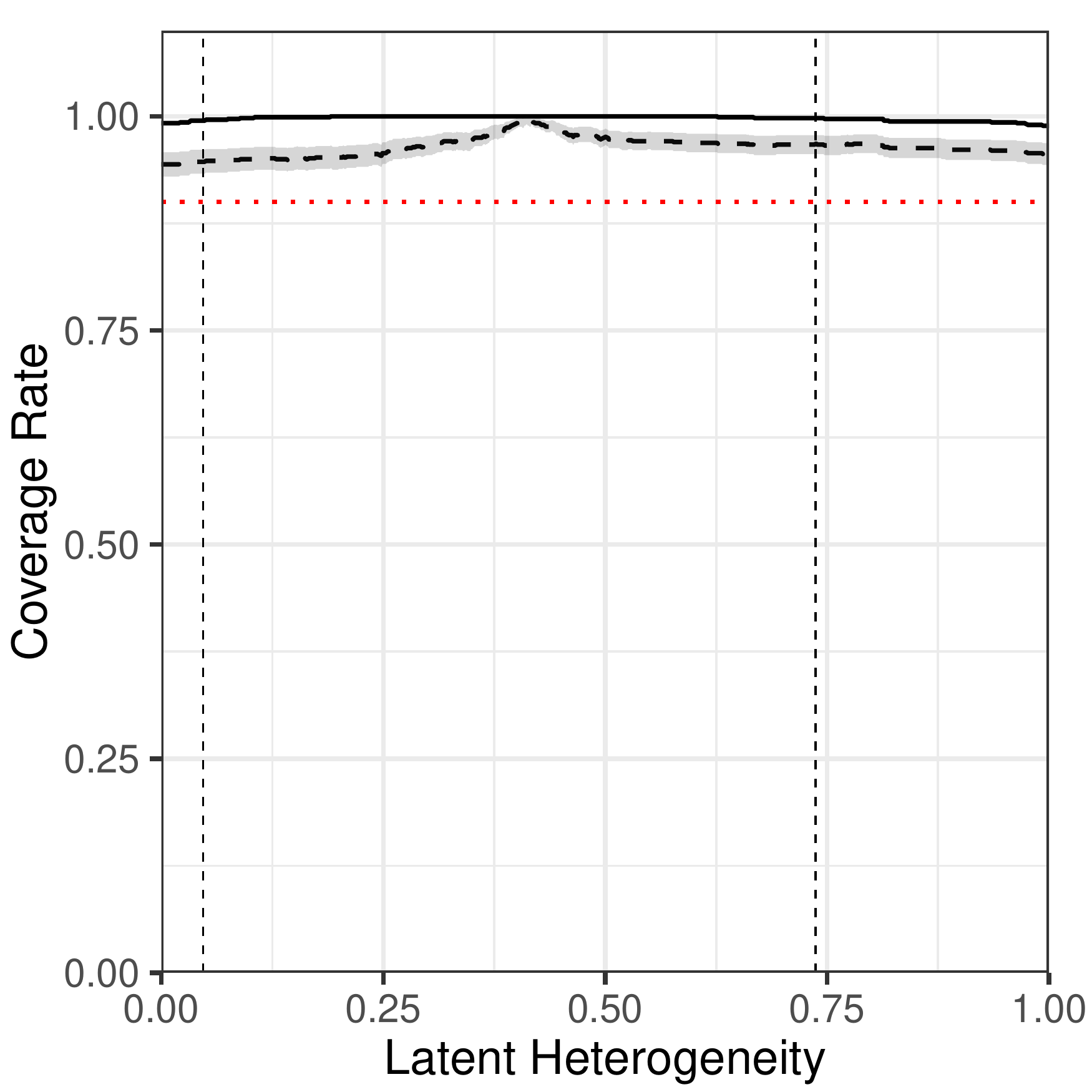}} \quad	
	\subfloat[Design 3\label{CI03}]{\includegraphics[width = .3 \columnwidth]{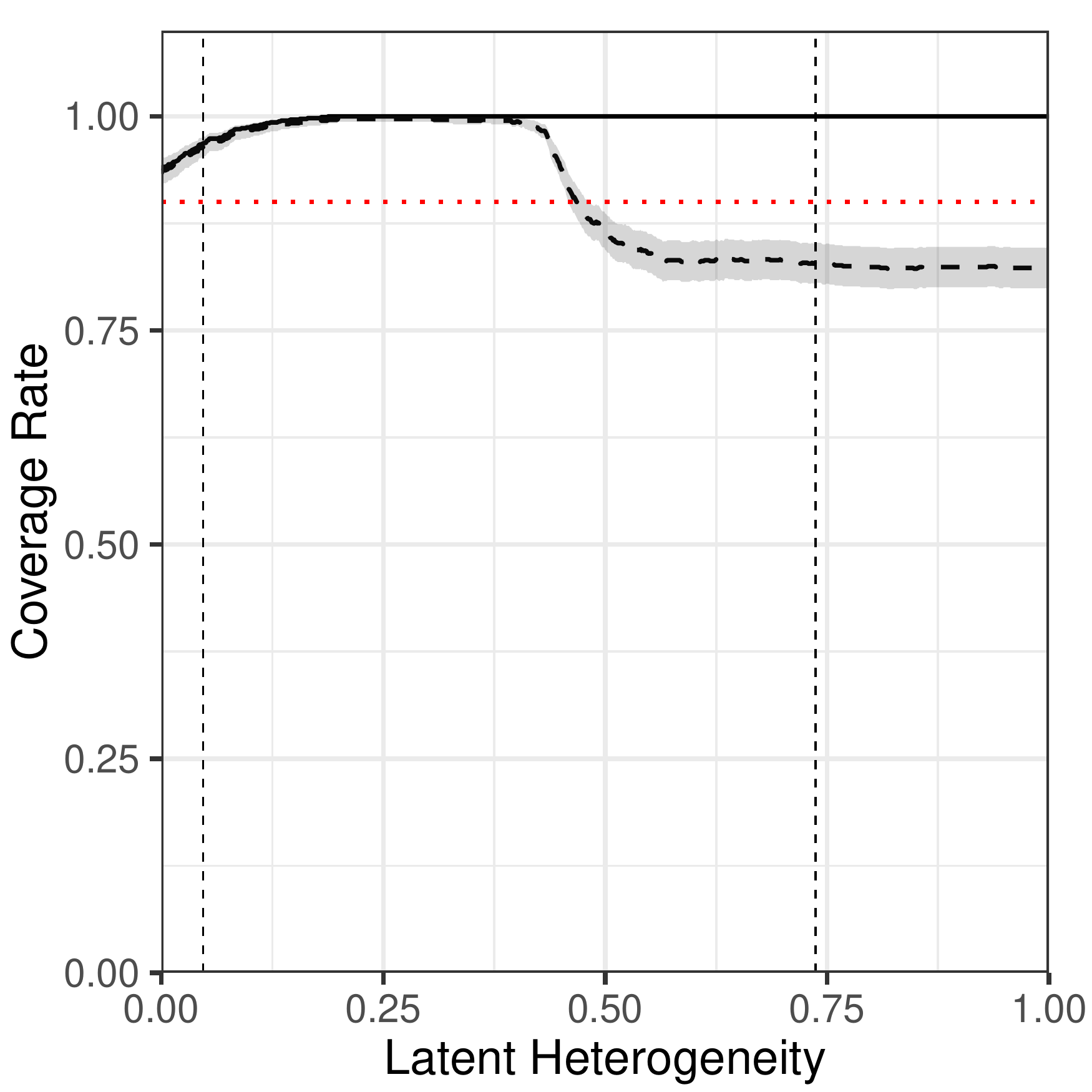}} \quad
	\subfloat[Design 5\label{CI05}]{\includegraphics[width = .3 \columnwidth]{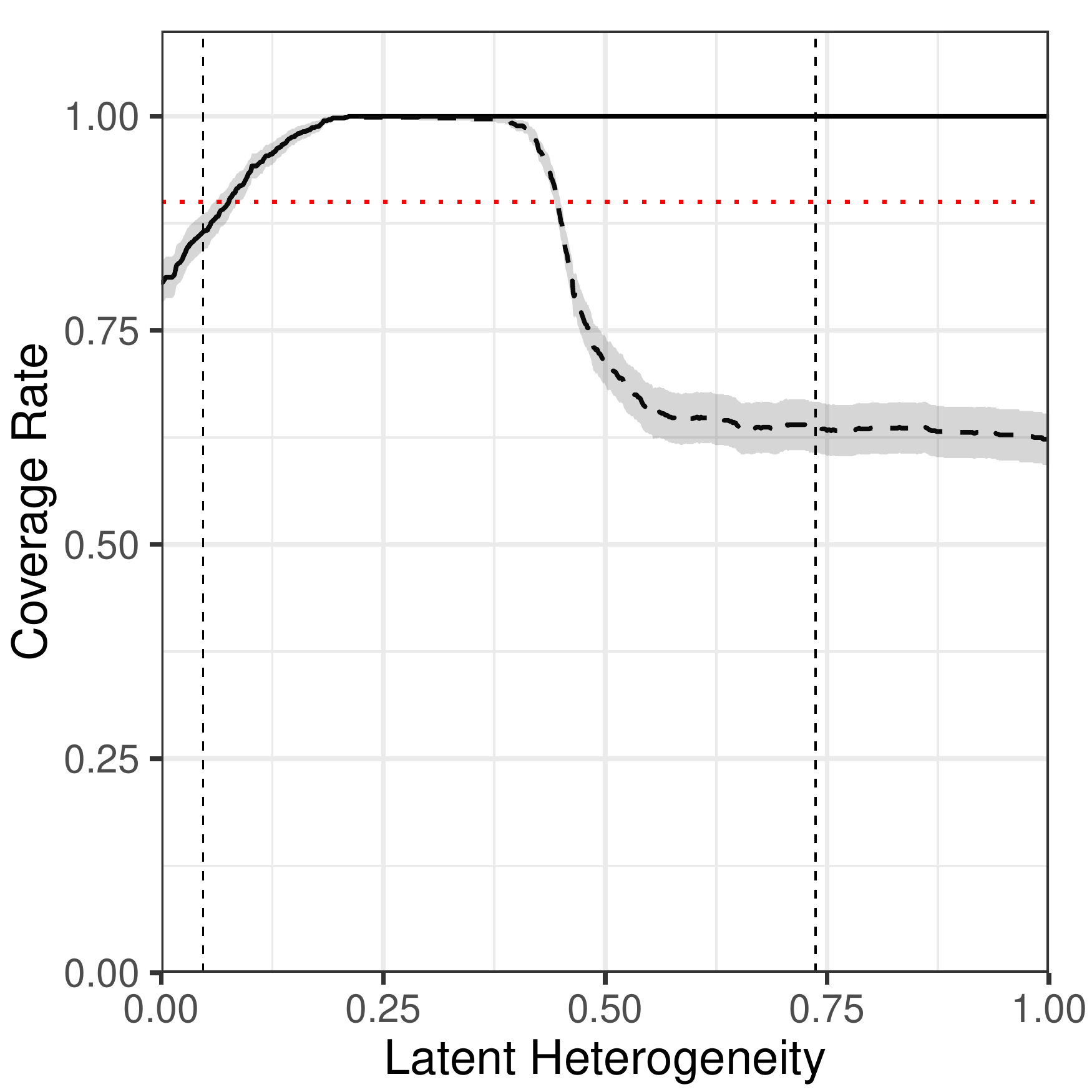}} \\
	\subfloat[Design 2\label{CI02}]{\includegraphics[width = .3 \columnwidth]{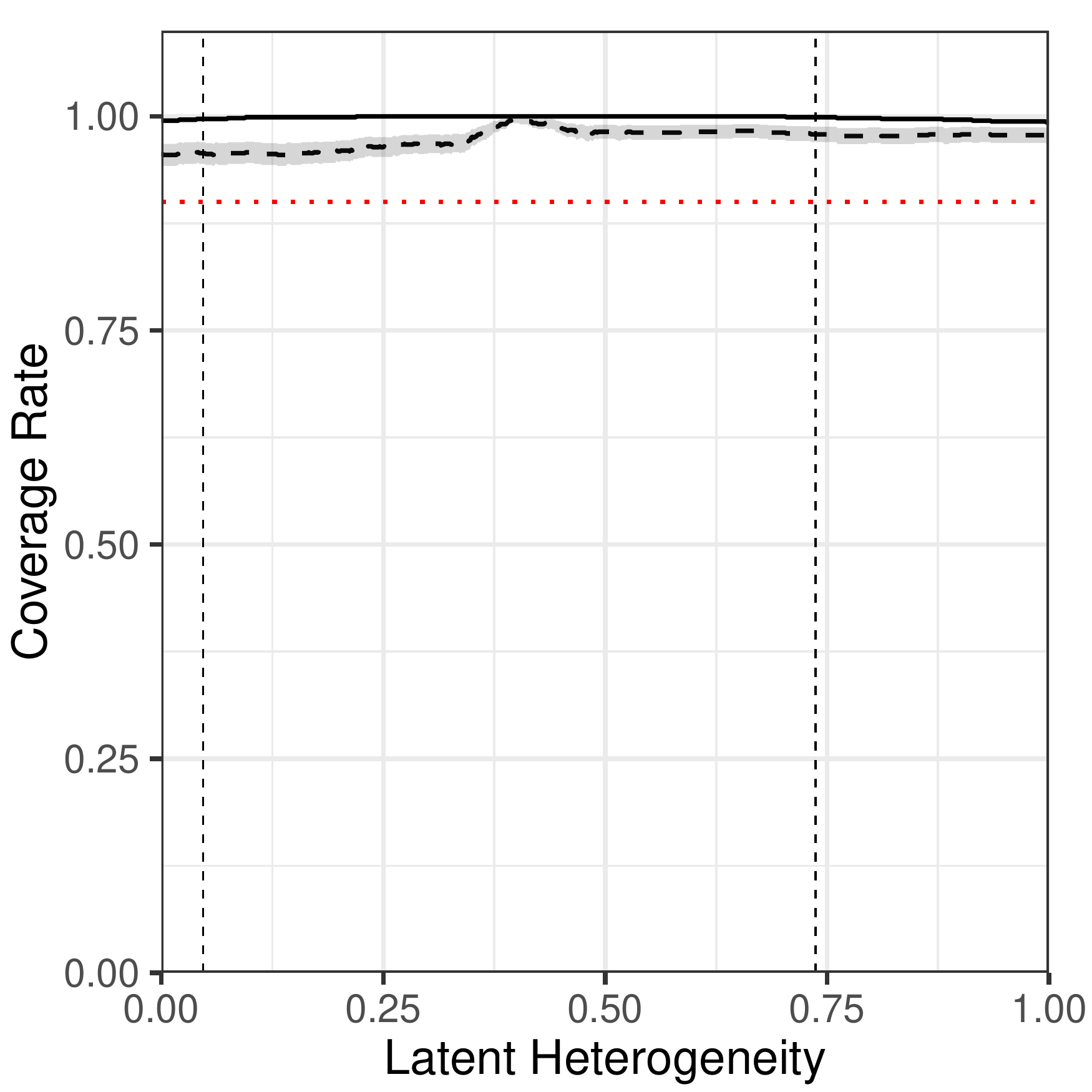}} \quad	
	\subfloat[Design 4\label{CI04}]{\includegraphics[width = .3 \columnwidth]{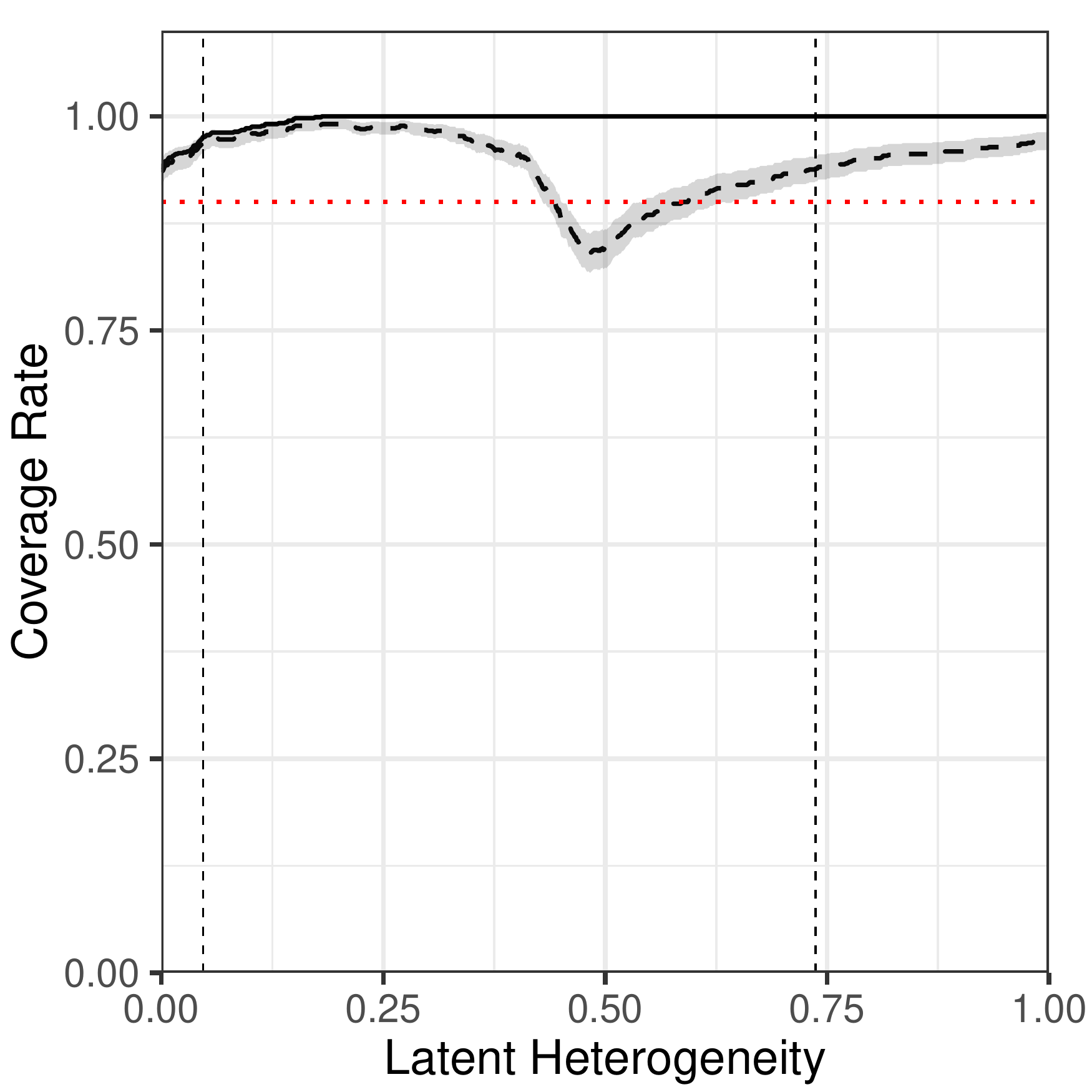}} \quad
	\subfloat[Design 6\label{CI06}]{\includegraphics[width = .3 \columnwidth]{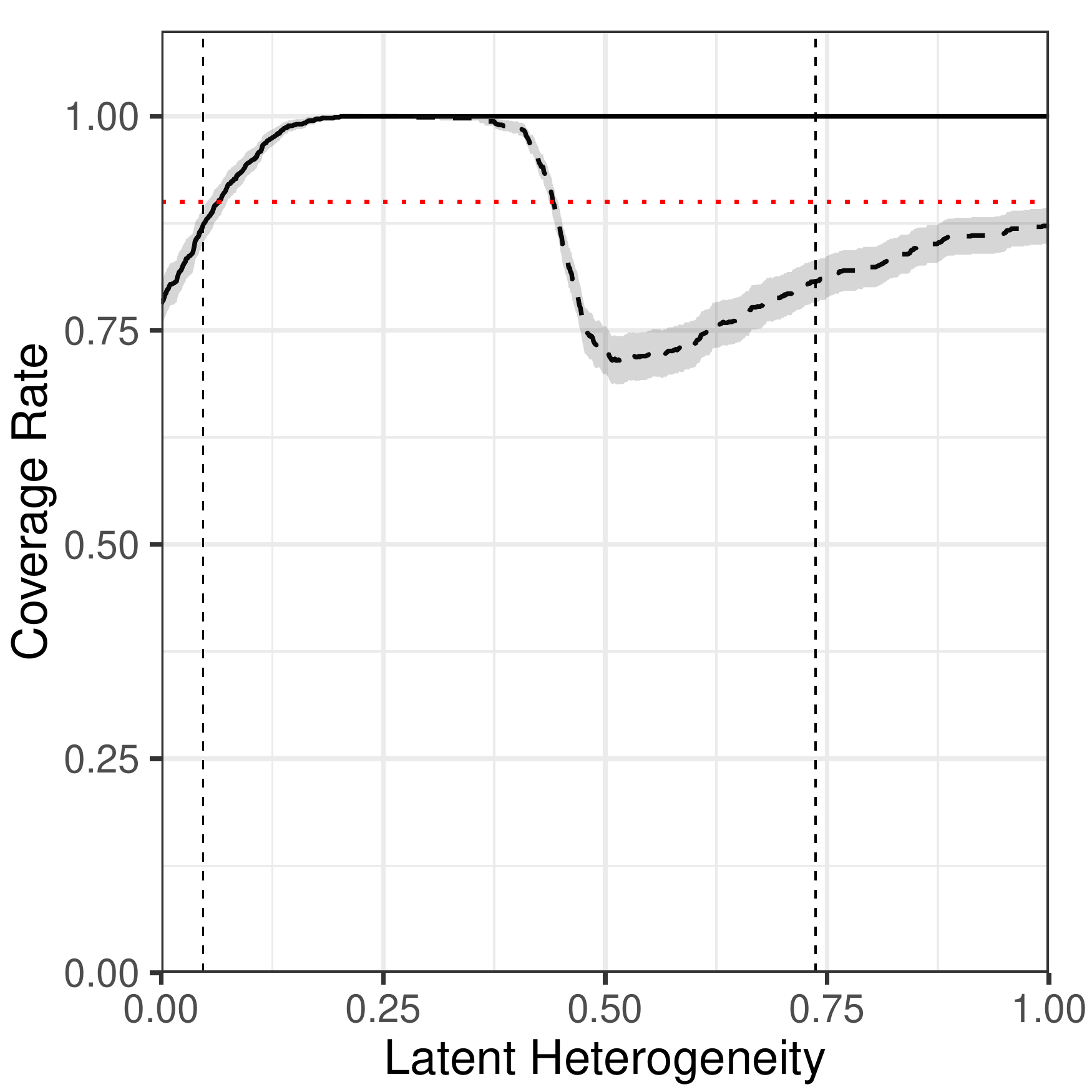}} \\
	\end{center}
	\footnotesize{Notes: The solid lines are the share of bootstrapped pointwise confidence intervals that contain the true parameter when the Mean Dominance Assumption \ref{meandominanceG} is not imposed. The dashed lines are the share of bootstrapped pointwise confidence intervals that contain the true parameter when the Mean Dominance Assumption \ref{meandominanceG} is imposed. Bootstrapped confidence intervals are based in 5,000 repetitions and the Monte Carlo results are based on 1,000 simulated datasets. The gray areas are pointwise 95\%-confidence intervals around the coverage rate when the Mean Dominance Assumption \ref{meandominanceG} is imposed and they measure simulation uncertainty. To make the figures easier to visualize, such confidence intervals are not shown when the Mean Dominance Assumption \ref{meandominanceG} is not imposed. The vertical dotted lines represent the population values of the propensity score $P\left[\left. D = 1 \right \vert Z = z\right] \text{ for any } z \in \left\lbrace 0, 1 \right\rbrace$. The red dotted lines denote the nominal coverage rate of 90\%.}
\end{figure}

For Designs 1 and 2 (which satisfy the linearity assumptions of the parametric estimation procedure detailed in Subsection \ref{parametric}), the coverage rate for the confidence interval proposed by \cite{Imbens2004} is above the nominal confidence level. This finding is not surprising in light of Proposition 1 by \cite{Stoye2009}, who shows that such confidence intervals have an asymptotic coverage rate that is at least the nominal confidence level.

For Design 3, I find a surprising negative result. Even though the $MTR$ functions are linear for this DGP, the coverage rate is below the nominal confidence level for many values of the latent heterogeneity. A even more surprising but positive result is the coverage rate for Design 4. Although the $MTR$ function for treated hourly labor earnings is quadratic for this DGP, the coverage rate is above the nominal confidence level for most values of the latent heterogeneity.

Finally, for Designs 5 and 6, I find that the 90\%-Confidence Intervals based on \cite{Imbens2004} severely under-cover the true $MTE$ function for most values of the latent heterogeneity. This negative result is not surprising because the $MTR$ functions of those DGPs are not linear.

\begin{figure}[!htbp] 
	\caption{Coverage Rate: 90\%-Confidence Intervals based on \cite{Imbens2004}} \label{IMCI}
	
	\begin{center}
		\subfloat[Design 1\label{IMCI01}]{\includegraphics[width = .3 \columnwidth]{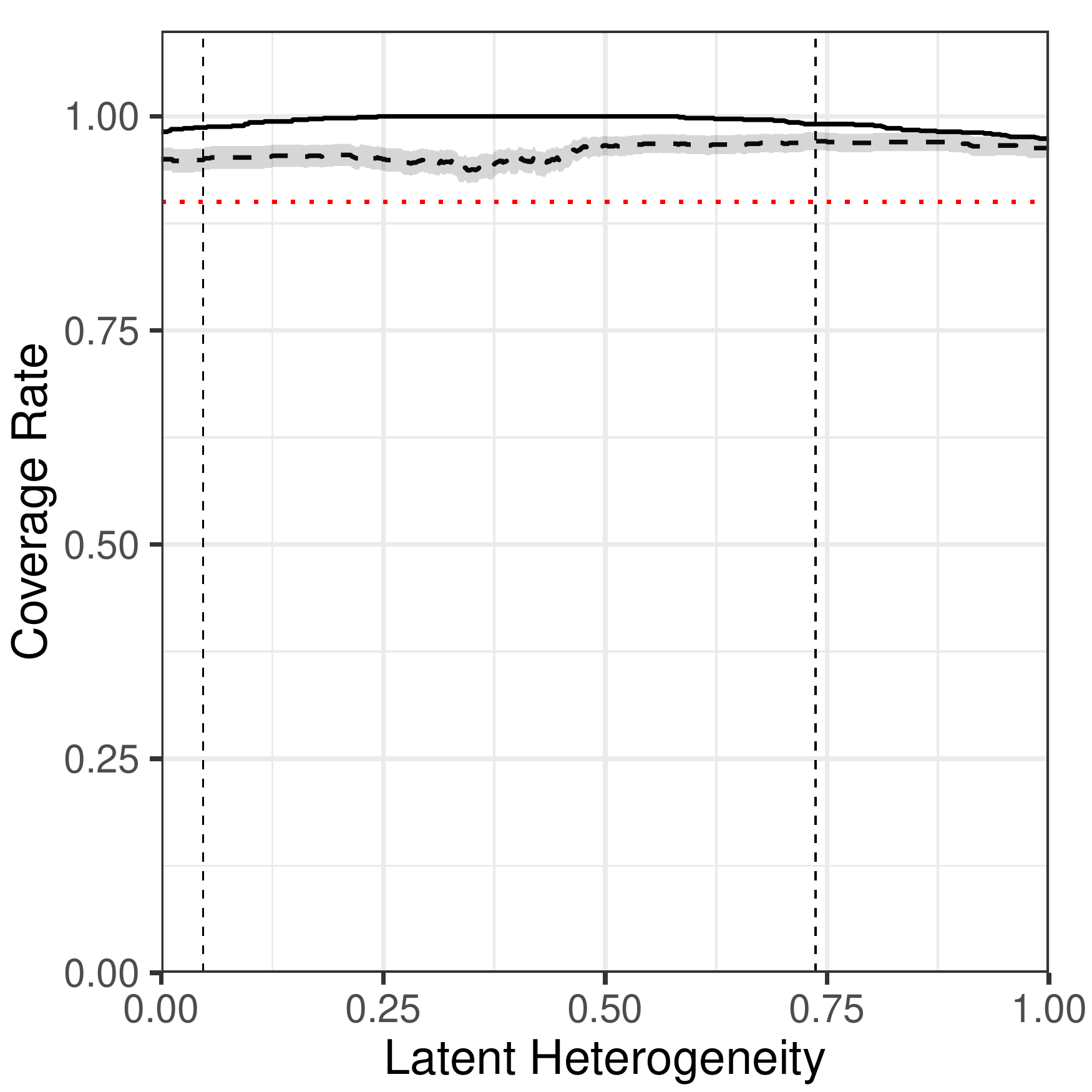}} \quad	
		\subfloat[Design 3\label{IMCI03}]{\includegraphics[width = .3 \columnwidth]{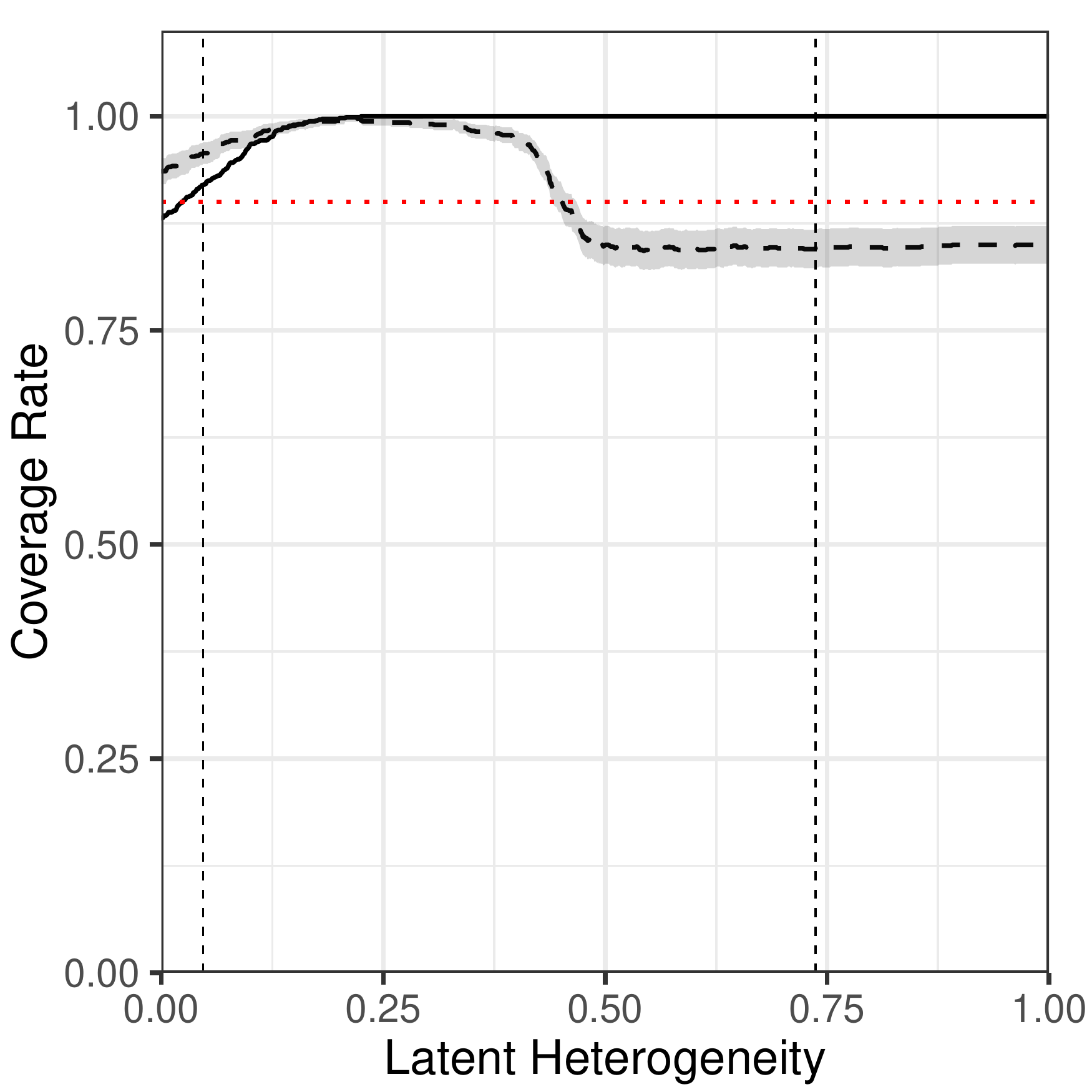}} \quad
		\subfloat[Design 5\label{IMCI05}]{\includegraphics[width = .3 \columnwidth]{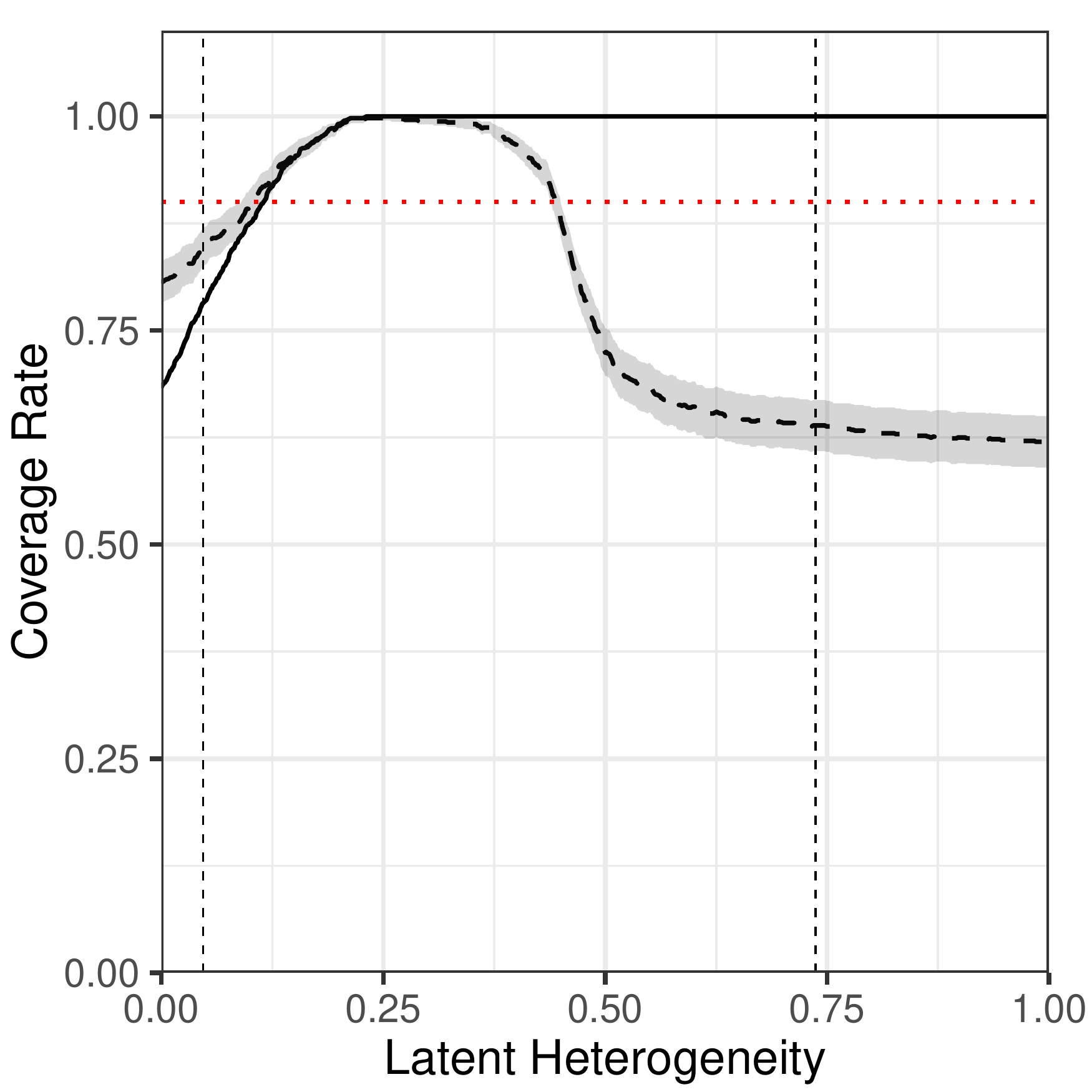}} \\
		\subfloat[Design 2\label{IMCI02}]{\includegraphics[width = .3 \columnwidth]{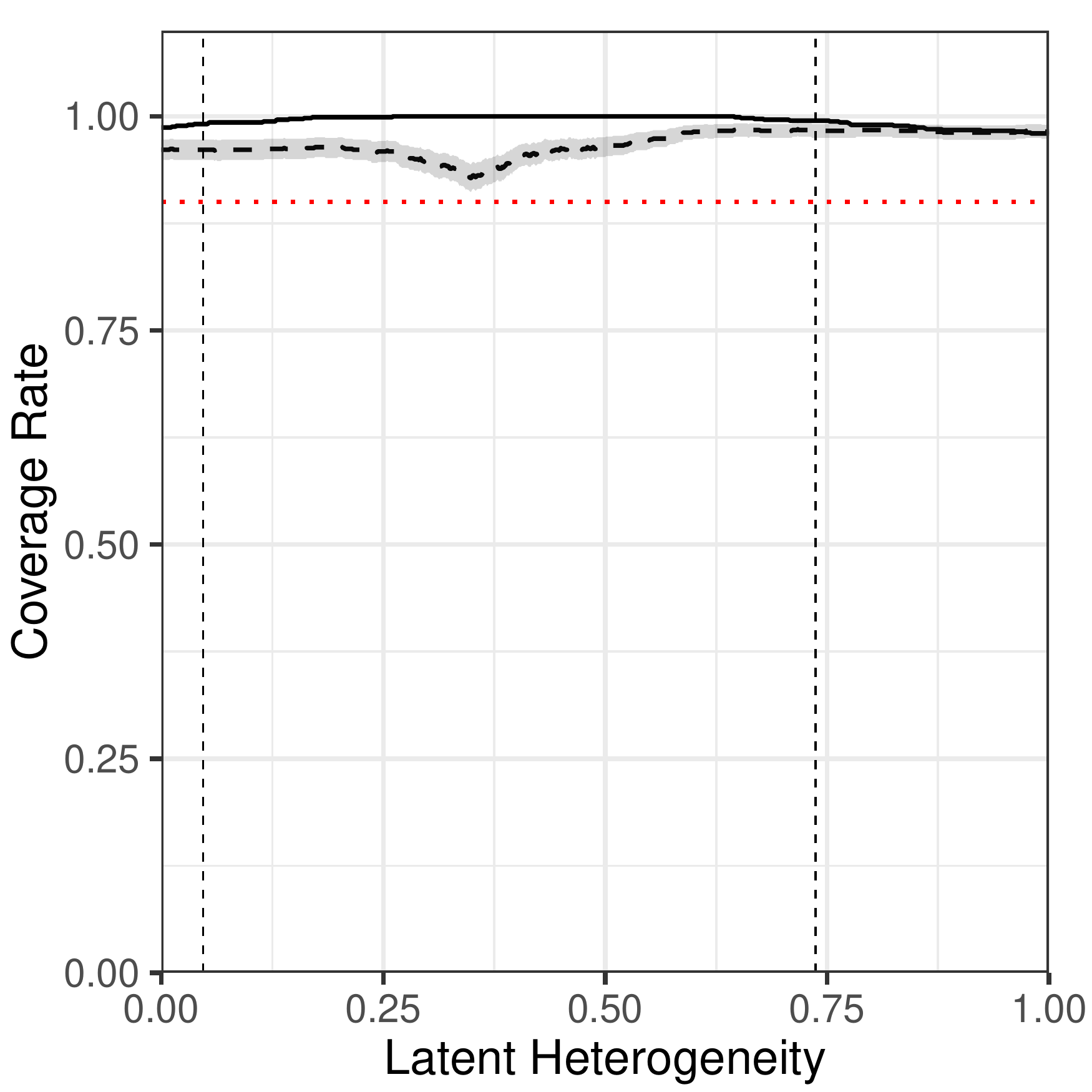}} \quad	
		\subfloat[Design 4\label{IMCI04}]{\includegraphics[width = .3 \columnwidth]{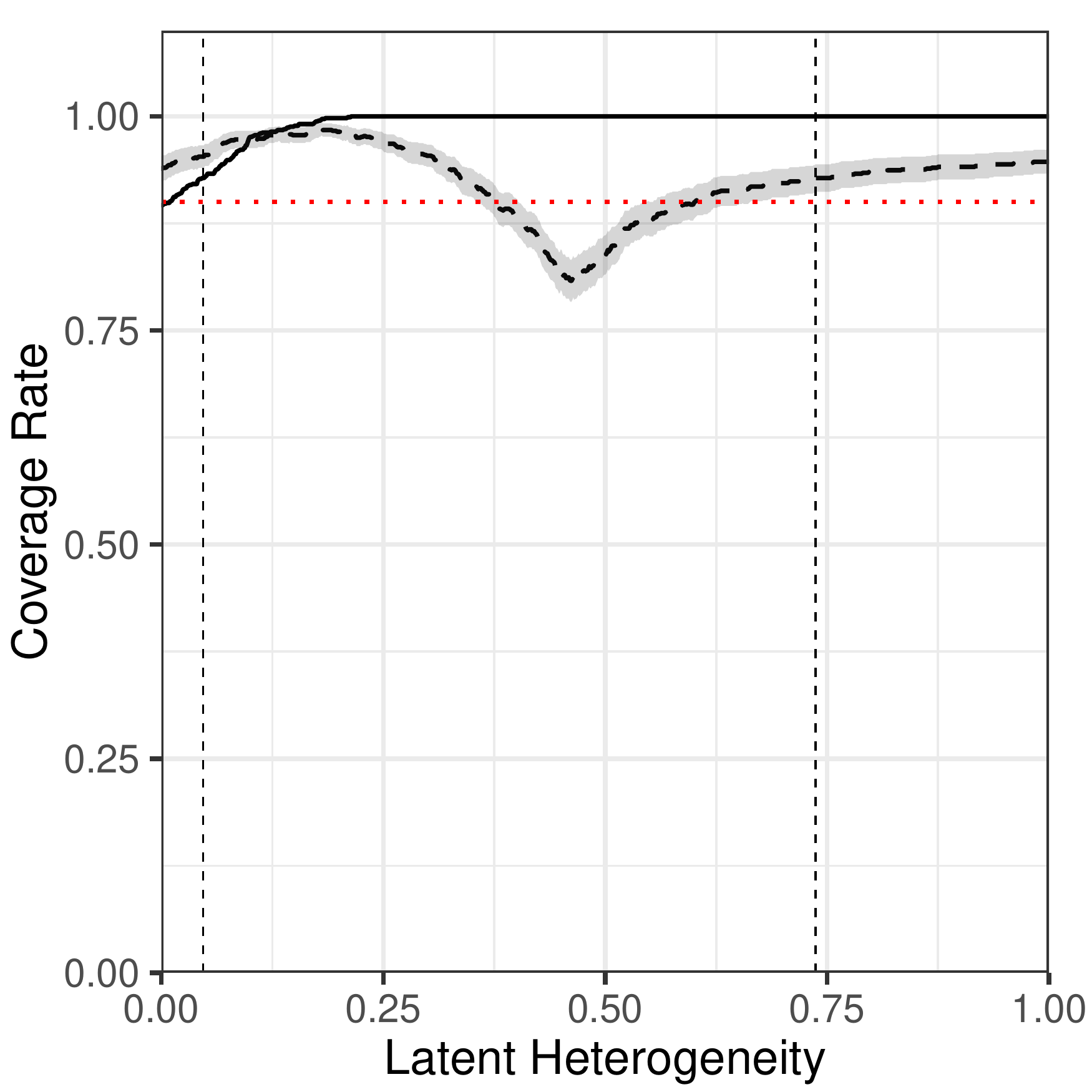}} \quad
		\subfloat[Design 6\label{IMCI06}]{\includegraphics[width = .3 \columnwidth]{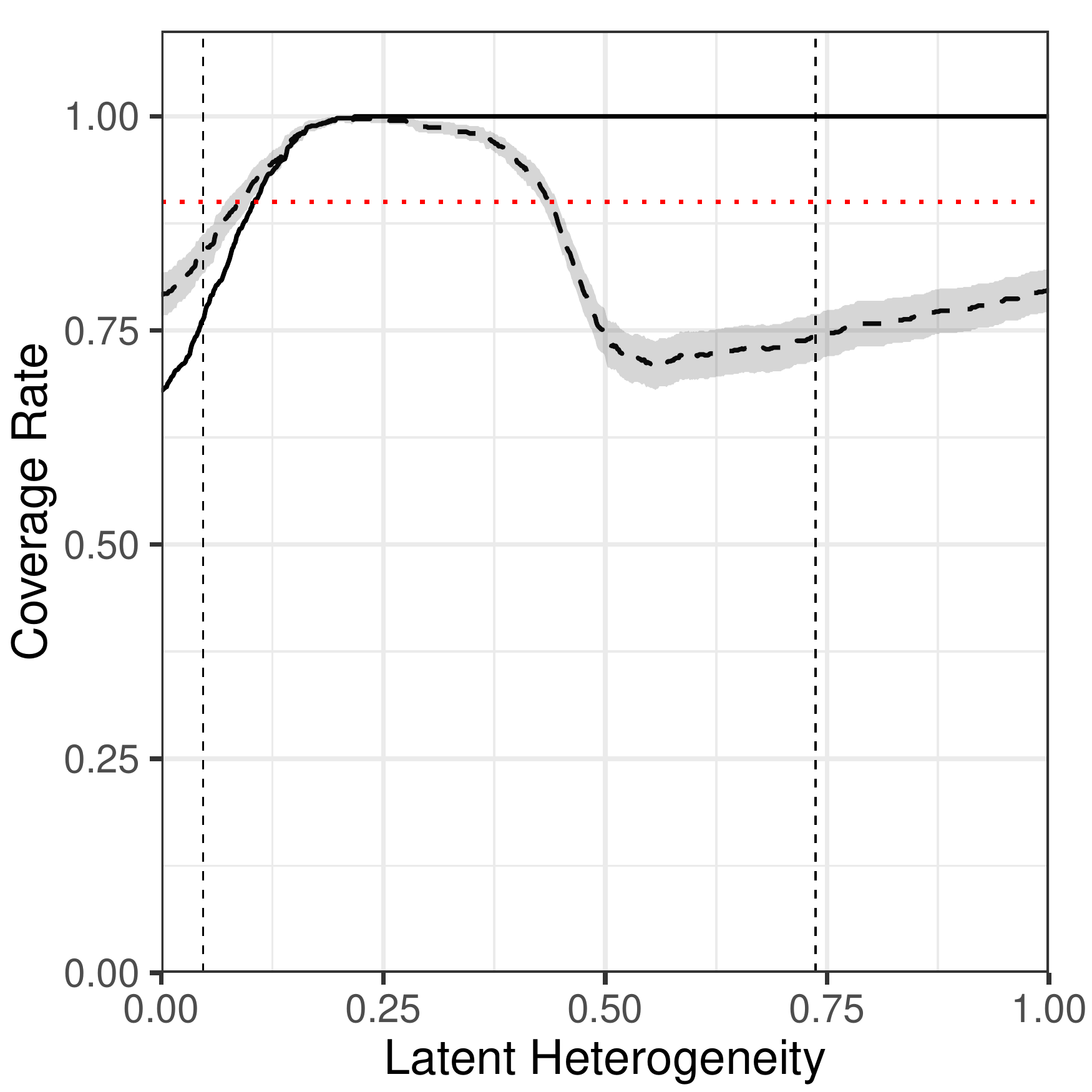}} \\
	\end{center}
	\footnotesize{Notes: The solid lines are the share of pointwise confidence intervals based on \cite{Imbens2004} that contain the true parameter when the Mean Dominance Assumption \ref{meandominanceG} is not imposed. The dashed lines are the share of pointwise confidence intervals based on \cite{Imbens2004} that contain the true parameter when the Mean Dominance Assumption \ref{meandominanceG} is imposed. Confidence intervals based on \cite{Imbens2004} are computed using 5,000 bootstrap repetitions and the Monte Carlo results are based on 1,000 simulated datasets. The gray areas are pointwise 95\%-confidence intervals around the coverage rate when the Mean Dominance Assumption \ref{meandominanceG} is imposed and they measure simulation uncertainty. To make the figures easier to visualize, such confidence intervals are not shown when the Mean Dominance Assumption \ref{meandominanceG} is not imposed. The vertical dotted lines represent the population values of the propensity score $P\left[\left. D = 1 \right \vert Z = z\right] \text{ for any } z \in \left\lbrace 0, 1 \right\rbrace$. The red dotted lines denote the nominal coverage rate of 90\%.}
\end{figure}

\end{document}